\documentclass[english, 11pt, letter]{article}
\usepackage{amsmath,amssymb,amsthm}
\usepackage{mathtools}
\usepackage{natbib}
\usepackage[nodisplayskipstretch]{setspace}
\usepackage{bm}
\usepackage{threeparttable}
\usepackage{graphicx}
\usepackage{enumerate}
\usepackage{commath}			
\usepackage{bibentry}
\usepackage{booktabs}
\usepackage[titletoc,title]{appendix}
\usepackage{multirow}
\usepackage{arydshln}
\usepackage{tabularx}
\usepackage{dsfont}
\usepackage{rotating} 
\usepackage[font={small}]{subcaption}
\usepackage[left=1in, right=1in, top=1in, bottom=1in]{geometry}

\usepackage[colorinlistoftodos,textsize=small]{todonotes}
\allowdisplaybreaks

\usepackage{url}
\usepackage{xr-hyper}  
\usepackage[
	breaklinks,
	colorlinks=true,
	pagebackref=false,
	pdfauthor={},%
	pdftitle={},%
	citecolor=blue,
	linkcolor=blue,
	urlcolor=blue
	]{hyperref}      

\newtheorem{lem}{Lemma}
\newtheorem{assumption}{Assumption}
\newtheorem{CGassumption}{CCC--GARCH Assumption}
\newtheorem{Vassumption}{VAR Assumption}
\newtheorem{prop}{Proposition}

\newtheorem{thm}{Theorem}

\newtheoremstyle{remark2}{1ex}{1ex}%
      {}
      {}
      {\bf}
      {.}
      {5pt}
      {\thmname{#1}\thmnumber{ #2}\thmnote{ \slshape{(#3)}}} 
\theoremstyle{remark2}
\newtheorem{rem}{Remark}

\newtheorem{example}{Example}

\newtheoremstyle{remark3}{1ex}{1ex}%
      {\it}
      {}
      {\bf}
      {.}
      {5pt}
      {\thmname{#1}\thmnumber{.#2}\thmnote{ \slshape{(#3)}}} 
\theoremstyle{remark3}
\newtheorem{lemV}{Lemma V}
\newtheorem{lemC}{Lemma C}

\newcommand{\1}{\mathds{1}}

\usepackage{babel}
\input{ee.sty}

\@ifundefined{bibfont}{  } 
    {  }

\makeatletter
\renewenvironment{proof}[1][\bfseries\proofname]{\par
   \pushQED{\qed}%
   \normalfont \topsep6\p@\@plus6\p@\relax
   \trivlist
   \item[\hskip\labelsep
     #1\@addpunct{:}]\ignorespaces
}{%
   \popQED\endtrivlist\@endpefalse
}
\makeatother

\newcommand{\Comments}{1}
\newcommand{\mynote}[2]{\ifnum\Comments=1\textcolor{#1}{#2}\fi}
\newcommand{\mytodo}[2]{\ifnum\Comments=1%
  \todo[linecolor=#1!80!black,backgroundcolor=#1,bordercolor=#1!80!black]{#2}\fi}

\ifnum\Comments=1               
\fi

\ifnum\Comments=1               
\fi

\renewcommand\appendixpagename{Appendix}
\newcommand{\CoVaR}{\operatorname{CoVaR}}

\newcommand{\VaR}{\operatorname{VaR}}

\newcommand{\full}{\operatorname{full}}

\newcommand{\inter}{\operatorname{int}}
\newcommand{\D}{\,\mathrm{d}}
\newcommand{\lex}{\operatorname{lex}}
\renewcommand{\E}{\mathbb{E}}
\renewcommand{\P}{\mathbb{P}}
\renewcommand{\a}{\alpha}
\renewcommand{\b}{\beta}

\newcommand{\Vv}{\mathsf{\mathbf{v}}}
\newcommand{\Cc}{\mathsf{\mathbf{c}}}

\newcommand{\init}{\boldsymbol{\mathcal{I}}_0}

\begin{document}

\baselineskip18pt
\renewcommand\floatpagefraction{.9}
\renewcommand\topfraction{.9}
\renewcommand\bottomfraction{.9}
\renewcommand\textfraction{.1}
\setcounter{totalnumber}{50}
\setcounter{topnumber}{50}
\setcounter{bottomnumber}{50}
\abovedisplayskip1.5ex plus1ex minus1ex
\belowdisplayskip1.5ex plus1ex minus1ex
\abovedisplayshortskip1.5ex plus1ex minus1ex
\belowdisplayshortskip1.5ex plus1ex minus1ex

\title{Dynamic CoVaR Modeling and Estimation\thanks{Both authors gratefully acknowledge support of the Deutsche Forschungsgemeinschaft (DFG, German Research Foundation) through grants 502572912 (first author), and 460479886 and 531866675 (second author).
		Replication material for the simulations and the empirical application in Section \ref{sec:EmpiricalApplication} is available on Github under \href{https://github.com/TimoDimi/replication_CoQR}{https://github.com/TimoDimi/replication\_CoQR}. 
		It draws on the corresponding open source package \texttt{SystemicRisk} \citep{SystemicRisk_package} implemented in \texttt{R} \citep{R2024}.}}

\author{
	Timo Dimitriadis\thanks{Alfred-Weber-Institute for Economics, Heidelberg University, Bergheimer Str.~58, D--69115 Heidelberg, Germany, and Heidelberg Institute for Theoretical Studies, \href{mailto:timo.dimitriadis@awi.uni-heidelberg.de}{timo.dimitriadis@awi.uni-heidelberg.de}.}
\and 
	Yannick Hoga\thanks{Faculty of Economics and Business Administration, University of Duisburg-Essen, Universit\"atsstra\ss e 12, D--45117 Essen, Germany, \href{mailto:yannick.hoga@vwl.uni-due.de}{yannick.hoga@vwl.uni-due.de}.}
}

\date{\today}
\maketitle

\begin{abstract}
	\noindent 
	The popular systemic risk measure CoVaR (conditional Value-at-Risk) and its variants are widely used in economics and finance.
	In this article, we propose joint dynamic forecasting models for the Value-at-Risk (VaR) and CoVaR.
	The CoVaR version we consider is defined as a large quantile of one variable (e.g., losses in the financial system) conditional on some other variable (e.g., losses in a bank's shares) being in distress.
	We introduce a two-step M-estimator for the model parameters drawing on recently proposed bivariate scoring functions for the pair (VaR, CoVaR).
	We prove consistency and asymptotic normality of our parameter estimator and analyze its finite-sample properties in simulations.
	Finally, we apply a specific subclass of our dynamic forecasting models, which we call \textit{CoCAViaR models}, to log-returns of large US banks.
	A formal forecast comparison shows that our CoCAViaR models generate CoVaR predictions which are superior to forecasts issued from current benchmark models.\\
	
	\noindent \textbf{Keywords:} CoVaR, Estimation, Forecasting, Modeling, Systemic Risk \\
	\noindent \textbf{JEL classification:} C14 (Semiparametric and Nonparametric Methods); C22 (Time-Series Models); C58 (Financial Econometrics); G17 (Financial Forecasting and Simulation); G32 (Financial Risk and Risk Management)
\end{abstract}


\section{Motivation}

\onehalfspacing

Since the introduction of the Value-at-Risk (VaR), risk forecasts have become a key input in financial decision making. For instance, VaR and Expected Shortfall (ES) forecasts are now routinely used for setting capital requirements of financial institutions under the Basel framework \citep{LS21}. 
Consequently, a huge literature on forecasting VaR and ES has emerged \citep[see, e.g.,][]{MF00,EM04,NR16,Mas17,PZC19, DimiHalbleib2021}. 
By definition, these measures are primarily designed to assess the risks faced by banks in isolation. Thus, these measures are well-suited to address microprudential objectives in banking regulation, that is, to limit the risk taking of individual institutions.

However, in the aftermath of the financial crisis of 2007--09, macroprudential objectives have gained importance on the regulatory agenda \citep{AER12,Aea17}. 
While also attempting to curb the risk taking of individual financial institutions, the macroprudential approach additionally takes into account the commonality of risk exposures among banks. 
To do so, a measure of interconnectedness of banks is now used under the Basel framework of the \citet[SCO40]{BCBSBF19} to determine the global systemically important banks (G-SIBs), which are subjected to higher capital requirements. 
This new focus has spurred the development of \textit{systemic} risk measures to accurately measure the interlinkages for which VaR and ES are unsuitable.
By now, a plethora of systemic risk measures is available \citep{GK11,CIM13,AB16,Aea17}. 

One of the most popular systemic risk measures is the conditional VaR (CoVaR) of \citet{AB16}. 
It is defined as a quantile of a  financial loss (e.g., of an entire market), given that a reference asset (e.g., a systemically important bank) is in distress, where the latter is taken to mean an exceedance of the reference asset's VaR.  
Note that this definition corresponds to the slightly redefined version of \citeauthor{AB16}'s \citeyearpar{AB16} CoVaR, which is due to \citet{GT13} and relies on the broader ``stress event'' of a VaR \textit{exceedance} \citep[instead of a VaR \textit{equality} of the reference asset, as originally proposed by][]{AB16}. 
Section~\ref{sec:CoVaR} outlines our reasons for doing so.

While forecasting models and asymptotic properties of forecasts are well explored for VaR and ES \citep{Cea07,GS08,WZ16,Hog18+}, the models for CoVaR forecasting are hitherto rather ad-hoc and little is known about the consistency of the forecasts.
For instance, DCC--GARCH models of \citet{Eng02}---which are popular due to their ability to accurately forecast conditional variance-covariance matrices \citep{LRV12,CM14}---may be used to generate CoVaR forecasts. 
However, as \citet[p.~620]{FZ16b} point out, ``[n]o formally established asymptotic results exist for the full estimation of the DCC [...] models''; see also \citet{DFL18}. 
This renders their application in forecasting systemic risk measures questionable. 
Moreover, there are additional problems associated with their use in systemic risk forecasting that stem from the non-uniqueness of the decomposition of the variance-covariance matrix (see Section~\ref{sec:CoVaRModels} and Appendix~\ref{CoVaR forecasting with multivariate GARCH models}). 
Recently, \citet{francq2025inference} and \citet{Hog25} derive asymptotic properties for forecasts of systemic risk, as measured by the CoVaR and the marginal expected shortfall (MES).
However, the multivariate GARCH-type models they consider for prediction do not allow for a dynamic conditional correlation structure.
In sum, there is a need for dynamic multivariate models that can deliver accurate systemic risk forecasts with strong theoretical underpinnings. 

This paper fills this gap. Specifically, the first main contribution of this paper is to introduce such dynamic models for the CoVaR, which we call \emph{CoCAViaR} models, as they nest the classical CAViaR models of \citet{EM04}.
In the spirit of \citet{EM04} and \citet{PZC19}, we model the pair (VaR, CoVaR) depending on past financial losses, lagged (VaR, CoVaR) model values, and (possibly) external covariates.
Our models are semiparametric in the sense that the quantities of interest (VaR and CoVaR) are modeled parametrically, yet no further (parametric or otherwise) assumptions are placed on the conditional distribution.

There are new challenges in modeling systemic risk vis-\`{a}-vis modeling univariate quantities such as VaR \citep{EM04,WhiteKimManganelli2015,Catania2022} and ES \citep{PZC19}. 
Modeling VaR or (VaR, ES) requires to specify the univariate dynamics only. Yet, since CoVaR measures the interlinkage of bivariate random variables, it becomes necessary to model the co-movements as well. 
We explore different models for doing so, and identify several competitive performers in our empirical application.

Our second main contribution is to propose an estimator for the model parameters and derive its large sample properties. The main technical hurdle to overcome is that---unlike VaR and (VaR, ES)---the pair (VaR, CoVaR) fails to be elicitable, such that no real-valued scoring function exists that is uniquely minimized by the true report \citep{FH24}. 
This renders standard M-estimation---adopted for VaR and ES models by \citet{EM04}, \citet{DimiBayer2019} and \citet{PZC19}---infeasible \citep[Theorem 1]{DFZ_CharMest}. 
Instead, we exploit the multi-objective elicitability of (VaR, CoVaR) \citep{FH24}. 
This property suggests a two-step M-estimator. In the first step, the score in the VaR component is minimized, and then the CoVaR score is minimized in the second step. 
We show that this leads to a consistent estimator.
In Appendices~\ref{Asymptotic Normality} and \ref{avar}, we also establish asymptotic normality of our two-step M-estimator and propose valid inference based on consistent estimation of the asymptotic variance-covariance matrix.
Our proofs show how to deal with non-smooth and discontinuous objective functions in the context of two-step M-estimation for dynamic models (based on past model values).
We speculate that our proof strategy may also be used elsewhere and, therefore, may be of independent interest.

In Appendix~\ref{sec:StartingValues}, we also show consistency of our estimator for parameters from (a fairly general class of) models that are based on misspecified initial values (which arise, e.g., when the true model is initialized in the infinite past).
While this issue is well-explored for estimators of GARCH-type models \citep{FZ10}, it is not treated for current semiparametric VaR and ES models that face the additional difficulty of non-smooth objective functions  \citep{EM04,WhiteKimManganelli2015,PZC19,Catania2022}.


Simulations confirm the good finite-sample properties of our two-step estimator. 
While the parameters of the CoCAViaR models can be estimated accurately in realistic settings, reliable inference for the model parameters requires large sample sizes or moderate probability levels.
However, as the main use of our CoCAViaR models is in forecasting, inference for the model parameters is of lesser importance than consistent estimates.



In the empirical application, we compare CoVaR forecasts issued from our CoCAViaR models with those from benchmark DCC--GARCH models. We do so for the four most systemically risky US banks according to the \cite{FSB22}, whose impact on a broader market index is assessed by CoVaR. Our various CoCAViaR model specifications tend to outperform the DCC--GARCH benchmarks, particularly in the CoVaR forecasts but also in the VaR component. Often, these differences in predictive ability are also statistically significant, as judged by the \citet{DM95}-type comparative backtest of \citet{FH24}. The superiority of our proposals may be explained by two reasons. First, our CoCAViaR specifications are specifically designed to model the quantities of interest (VaR and CoVaR), whereas multivariate GARCH processes focus on modeling the complete predictive distribution. Second, our estimation technique is tailored to provide an accurate description of the (VaR, CoVaR) evolution. In particular, the estimator is not too strongly influenced by center-of-the-distribution observations as, e.g., standard estimators of multivariate GARCH models.


The rest of the paper is structured as follows. 
Section~\ref{Modeling and Estimation} formally introduces the CoVaR, our modeling framework and the appertaining parameter estimator. 
Section~\ref{Asymptotic Properties} gives large sample results for our estimator.
We illustrate the finite-sample properties of our estimator in Section~\ref{sec:Simulations}. Section~\ref{sec:EmpiricalApplication} presents the empirical application and the final Section~\ref{sec:Conclusion} concludes. 
The main proofs as well as additional details are given in the Online Appendix (containing Sections~\ref{sec:thm1}--\ref{Computation of Risk Measure Forecasts}), where we also verify our main assumptions for CCC--GARCH models and vector autoregressive (VAR) models. 

\section{Modeling and Estimation}\label{Modeling and Estimation}

\subsection{The CoVaR}
\label{sec:CoVaR}

Throughout the paper, we consider a sample of size $n\in\mathbb{N}$ of the bivariate series $\big\{(X_t,Y_t)^\prime\big\}_{t\in\mathbb{N}}$, which is defined on the probability space $\big(\Omega, \mathcal{F}, \P\big)$. 
Specifically, $Y_t$ stands for the log-losses of interest (e.g., system-wide losses in the financial system) and $X_t$ are the log-losses of some reference position (e.g., the losses of a bank's shares).
Here, log-losses are simply the negated log-returns.
The information set that the forecaster is interested in conditioning on at time $(t-1)$ is $\mathcal{F}_{t-1}=\sigma\big((X_{t-1},Y_{t-1})^\prime,\mZ_{t-1},\ldots,(X_{1},Y_{1})^\prime,\mZ_{1},\init \big)$.
Here, the variable $\mZ_t$ contains some (possibly multivariate) exogenous covariates, and $\init$ represents known and (possibly multivariate) pre-sample information that is used for model initialization at time $t=1$ (see below for details).

For $\beta\in[0,1)$, we define $\VaR_{\beta}(F)=F^{\leftarrow}(\b)$ to be the $\beta$-quantile of the distribution $F$.
Then, the conditional VaR is simply $\VaR_{t,\b}=\VaR_{\b}(F_{X_t\mid\mathcal{F}_{t-1}})$, where $F_{X_t\mid\mathcal{F}_{t-1}}$ denotes the conditional distribution of $X_t$ given $\mathcal{F}_{t-1}$.
The stress event considered in the definition of the CoVaR is that the loss of the reference position exceeds its VaR, i.e., $\{X_t\geq\VaR_{t,\b}\}$.
With our orientation of $X_t$ denoting financial losses, we commonly consider values for $\beta$ close to one, such as $\b=0.95$.
For $\a\in(0,1)$ we define $\CoVaR_{t,\a|\b}=\CoVaR_{\a|\b}(F_{X_t,Y_t\mid\mathcal{F}_{t-1}})$, where $F_{X_t,Y_t\mid\mathcal{F}_{t-1}}$ denotes the distribution of $(X_t,Y_t)^\prime\mid\mathcal{F}_{t-1}$ and
\[
	\CoVaR_{\a|\b}(F_{X,Y})=\VaR_\a(F_{Y\mid X \geq \VaR_\b(F_{X})})
\]
for a joint distribution function $F_{X,Y}$ with marginals $F_{X}$ and $F_{Y}$. 
Again, we usually consider values of $\a$ close to one.
If $\a=\b$ we simply write $\CoVaR_{t,\a} = \CoVaR_{t,\a|\a}$, and for $\beta=0$ we simply have $\CoVaR_{t,\a|\b} = \VaR_{\a}(F_{Y_t\mid\mathcal{F}_{t-1}})$.

As pointed out in the Motivation, our CoVaR definition follows \citet{GT13} and deviates from the original one of \citet{AB16}.
The latter authors use $\{X=\VaR_{\b}(F_X)\}$ as the stress event instead of $\{X\geq\VaR_{\b}(F_X)\}$, such that $\CoVaR_{\a|\b}^{=}(F_{X,Y})=\VaR_{\a}(F_{Y\mid X=\VaR_{\b}(F_X)})$.
Our reasons for adopting the ``inequality version'' of the CoVaR are as follows.
First, on an intuitive level, conditioning on $\{X\geq\VaR_{\b}(F_X)\}$ instead of $\{X=\VaR_{\b}(F_X)\}$ better captures ``tail risks''.
Second, $\CoVaR_{\a|\b}(F_{X,Y})$ is dependence consistent in the sense of \citet{MS14}, while $\CoVaR_{\a|\b}^{=}(F_{X,Y})$ is not \citep{BDJ17}.
Third, a further advantage of the redefined CoVaR is its multi-objective elicitability shown by \citet{FH24}. 
This property allows CoVaR forecasts to be meaningfully evaluated via so-called backtests, which are, however, not available for \citeauthor{AB16}'s \citeyearpar{AB16} $\CoVaR_{\a|\b}^{=}(F_{X,Y})$.
Fourth, (dynamic models for) $\CoVaR_{\a|\b}^{=}(F_{X,Y})$ would be harder to estimate than $\CoVaR_{\a|\b}(F_{X,Y})$ without imposing further parametric assumptions.
This is because the $\CoVaR_{\a|\b}^{=}(F_{X,Y})$ is the $\a$-quantile of the distribution of $Y\mid X=\VaR_{\b}(F_X)$, which requires non-parametric methods with their usually slower convergence rates.
For all these reasons, we work with the definition of \citet{GT13}, which---next to the authors mentioned above---is also used by \citet{AFS13}, \citet{BC19}, \citet{NZ20} and \citet{CR22}.

\subsection{Dynamic Models for VaR and CoVaR}
\label{sec:CoVaRModels}

An attractive class of models for financial data---and, hence, also for CoVaR modeling---are multivariate GARCH processes, such as the DCC--GARCH of \citet{Eng02} or its corrected version by \citet{Aie13}.
These model the conditional covariance matrix $\mH_t := \Var\big((X_t,Y_t)^\prime\mid\mathcal{F}_{t-1}\big)$ dynamically and are regularly used for volatility forecasting.
However, in order to generate CoVaR forecasts, one requires the ``square-root matrix'' $\mSigma_t$ of $\mH_t$, satisfying $\mSigma_t\mSigma_t^\prime = \mH_t$.
For this matrix decomposition there exist infinitely many possibilities, such as the one based on the symmetric eigenvalue decomposition ($\mSigma_t^{s}$) or the lower triangular matrix of the Cholesky decomposition ($\mSigma_t^{l}$).
The problem is that for non-spherically distributed shocks, each possibility, while implying the \textit{same} variance-covariance dynamics $\mH_t$, may imply \textit{different} values for the CoVaR (forecasts).
Thus, in principle, a given single GARCH-type model for $\mH_t$ may be consistent with an infinite number of CoVaR forecasts (depending on the choice of $\mSigma_t$), thus creating an unsatisfactory ambiguity when applied to CoVaR forecasting. 
We provide additional details on this ambiguity in Appendix~\ref{CoVaR forecasting with multivariate GARCH models}, where we also show that this identification problem for the CoVaR can only arise for non-spherically distributed shocks.

This deficiency of multivariate GARCH models underlines the necessity to construct explicit CoVaR models as we do in the following.
In the spirit of \citet{EM04} and \citet{PZC19}, we consider semiparametric models for the pair (VaR, CoVaR) of the general form
\begin{align}
	\label{eqn:GeneralModel}
	\begin{pmatrix}
		v_t(\vtheta^{v})\\
		c_t(\vtheta^{c})
	\end{pmatrix}
	=\begin{pmatrix}
		v_t\big((X_{t-1},Y_{t-1})^\prime,\mZ_{t-1},\ldots,(X_{1},Y_{1})^\prime,\mZ_{1}, \init; \vtheta^{v}\big)\\
		c_t\big((X_{t-1},Y_{t-1})^\prime,\mZ_{t-1},\ldots,(X_{1},Y_{1})^\prime,\mZ_{1}, \init; \vtheta^{c}\big)
	\end{pmatrix}, \qquad t\in\mathbb{N},
\end{align}
where $\vtheta^v$ and $\vtheta^c$ are generic parameters from some parameter spaces $\mTheta^v\subset\mathbb{R}^{p}$ and $\mTheta^c\subset\mathbb{R}^q$, respectively.
	For $t=1$, equation \eqref{eqn:GeneralModel} is taken to mean that the model only depends on the initialization information $\init$ (introduced in Section~\ref{sec:CoVaR}) and (possibly) the respective parameter $\vtheta^v$ or $\vtheta^c$.
Throughout the paper, we assume that the underlying data-generating process is such that the model in \eqref{eqn:GeneralModel} is correctly specified for (VaR, CoVaR).
That is, there exist true parameters $\vtheta^v_0\in\mTheta^v$ and $\vtheta^c_0\in\mTheta^c$, such that 
\begin{align}
	\label{eqn:TrueModelParameters}
	\begin{pmatrix}
		\VaR_{t,\beta}\\
		\CoVaR_{t,\alpha\mid\beta}
	\end{pmatrix}
	= \begin{pmatrix}
		v_t(\vtheta_{0}^{v})\\
		c_t(\vtheta_{0}^{c})
	\end{pmatrix},\qquad t\in\mathbb{N},
\end{align}
almost surely (a.s.).

Before offering some comments on our general framework, we give a specific instance of a model satisfying \eqref{eqn:GeneralModel} in Example~\ref{ex:1} now.
A discussion of the conditions under which \eqref{eqn:TrueModelParameters} holds for the models of Example~\ref{ex:1} follows in Remark~\ref{rem:Init}.

\begin{example}
	\label{ex:1}
	The general formulation in \eqref{eqn:GeneralModel} allows for \emph{dynamic} models, where (VaR, CoVaR) may depend on their past model values.
	E.g., it nests ARMA-type models of the form
	\begin{align}
		\label{eqn:SAVCoCAViaRModelClass}
		\begin{pmatrix} v_t(\vtheta^v) \\ c_t(\vtheta^c) \end{pmatrix}
		= \vomega + 
		\mA \begin{pmatrix} |X_{t-1}| \\ |Y_{t-1}| \end{pmatrix} +
		\mB \begin{pmatrix} v_{t-1}(\vtheta^v) \\ c_{t-1}(\vtheta^c) \end{pmatrix}
	\end{align}
	with parameters $\vomega \in \mathbb{R}^2$, $\mA, \mB \in  \mathbb{R}^{2 \times 2}$ that can be collected in $\vtheta^v$ and $\vtheta^c$.
	We call model \eqref{eqn:SAVCoCAViaRModelClass} ``ARMA-type'' because an ARMA model $X_t=c+\phi X_{t-1} + \varepsilon_{t} + \theta\varepsilon_{t-1}$ with martingale difference errors $\varepsilon_t$ implies the dynamics $\mu_t=c+(\phi+\theta)Y_{t-1} + \theta\mu_{t-1}$ for the conditional mean $\mu_t=\E_{t-1}[Y_t]$. This dynamic mean recursion is, of course, very similar to the (VaR, CoVaR) recursion in \eqref{eqn:SAVCoCAViaRModelClass}.
	Borrowing terminology from \cite{EM04}, we call the models in \eqref{eqn:SAVCoCAViaRModelClass} \emph{Symmetric Absolute Value (SAV) CoCAViaR} models.
	Of course, our restriction in \eqref{eqn:GeneralModel} that $v_t(\cdot)$ ($c_t(\cdot)$) only depends on $\vtheta^v$ ($\vtheta^c$) implies that the matrix $\mB=(B_{ij})_{i,j=1,2}$ must be diagonal (i.e., $B_{12}=B_{21} = 0$) to facilitate two-step M-estimation.
\end{example}

A few comments on our general modeling approach given in \eqref{eqn:GeneralModel} and \eqref{eqn:TrueModelParameters} are in order.
First, while we assume the model in \eqref{eqn:GeneralModel} to be correctly specified, we make no assumption on its stationarity.
Indeed, because the information set $\mathcal{F}_{t-1}$ increases with $t$, $v_t(\vtheta^v)$ and $c_t(\vtheta^c)$ will in general be non-stationary, even for the true parameters (i.e., for $\vtheta^v=\vtheta_0^v$ and $\vtheta^c=\vtheta_0^c$).

Second, our model in \eqref{eqn:GeneralModel} is semiparametric in the sense that---while the model dynamics are governed by the parameters $\vtheta^v$ and $\vtheta^c$---we impose no additional assumptions on the conditional distribution of $(X_t,Y_t)^\prime\mid\mathcal{F}_{t-1}$.

Third, as the $v_t(\cdot)$ and $c_t(\cdot)$ functions may vary with time $t$, our modeling framework is sufficiently flexible to allow for the inclusion of lagged (VaR, CoVaR) as in Example~\ref{ex:1}.
This is important because including lags of model values often leads to a better predictive performance, as is the case for GARCH models, which improve upon ARCH models by including lags of volatility in the variance equation.

Fourth, the VaR model in \eqref{eqn:GeneralModel} generalizes classical (time series) quantile regressions based on lagged values of $X_t$, $Y_t$ and the covariate vector $\mZ_{t}$ \citep{EM04,OH16}. 
Also, the ``VAR for VaR'' models (of dimension two) of \cite{WhiteKimManganelli2015} are nested by choosing $\beta=0$ such that the CoVaR simply becomes the VaR of $Y_t$.

%

Fifth, the fact that the VaR model $v_t(\cdot)$ does not depend on the CoVaR parameter $\vtheta^c$ renders our two-step estimator (introduced in Section~\ref{sec:ParameterEstimation} below) feasible. 
Likewise, we assume that the CoVaR model $c_t(\cdot)$ does not include the VaR parameter $\vtheta^v$.
Our two-step estimator would remain feasible if we dropped this requirement.
This would, however, come at the cost of a further explosion of technicality in the proofs, without any offsetting improvements (in terms of predictive accuracy) of the resulting models; see the empirical application in Section~\ref{sec:EmpiricalApplication} for evidence on this.

Sixth, besides using the full history of lagged values of $X_t$, $Y_t$ and $\mZ_t$, our models must be initialized at time $t=1$.
For this,  we use the initialization information $\init$ that may be thought of as capturing the conditions at which the observable history started. 
E.g., when $X_t$ or $Y_t$ denote returns on JPMorgan Chase shares from 2001 onwards, then $\init$ may capture general market conditions on the first day the merged bank JPMorgan Chase started trading in 2001.
We give specific possibilities for the initialization of the SAV CoCAViaR models in the following example.

\setcounter{example}{0}
\begin{example}[continued]
Among others, the SAV CoCAViaR models in \eqref{eqn:SAVCoCAViaRModelClass} can be initialized in one of the following ways:
\begin{enumerate}[(a)] 
	\item 
	\label{item:InitConst}
	The initial information is $\init = \vk\in\mathbb{R}^2$, where $\vk$ is the known correct starting value of the model, i.e., $(v_1(\vtheta^v), c_1(\vtheta^c))^\prime = \vk$ for all $\vtheta^v\in\mTheta^v$ and $\vtheta^c\in\mTheta^c$. 
	Then, the recursion in \eqref{eqn:SAVCoCAViaRModelClass} holds for $t\geq2$.
	
	
	\item 
	\label{item:InitParameter}
	\sloppy
	The forecaster has no initial information, such that $\init$ is empty, and makes the modeling assumption that $(v_1(\vtheta^v), c_1(\vtheta^c))^\prime = \vomega$.
	Once again, the recursion in \eqref{eqn:SAVCoCAViaRModelClass} holds for $t\geq2$.

	\item 
	\label{item:InitInfinitePast}
	The forecaster is interested in conditioning on the \emph{infinite} past, such that $\init = (X_0,Y_0, X_{-1},Y_{-1},\ldots )^\prime$. 
	Hence, the recursion in \eqref{eqn:SAVCoCAViaRModelClass} holds for $t\in\mathbb{Z}$.
	This is the most complicated setting for practical parameter estimation.
	
\end{enumerate}
Of course, similar possible initializations as under \eqref{item:InitConst}--\eqref{item:InitInfinitePast} also apply to the general models of the form \eqref{eqn:GeneralModel}.
We discuss the implication of item \eqref{item:InitInfinitePast} on the properties of our estimator in more detail in Appendix~\ref{sec:StartingValues}.
Notice however that---as part of the modeling exercise---any assumption on how the (VaR, CoVaR) process originated ultimately remains an approximation to the true underlying dynamics. 
\end{example}

The SAV CoCAViaR models of Example~\ref{ex:1} are, among others, correctly specified (in the sense of \eqref{eqn:TrueModelParameters}) under a type of \citeauthor{Jea98}'s \citeyearpar{Jea98} extended constant conditional correlation (ECCC) GARCH process, which has been studied extensively in the literature \citep{HT04,CK10}:
\begin{equation}
	\label{eqn:ECCCmodel}
	\begin{pmatrix}X_t \\ Y_t\end{pmatrix} =
	\begin{pmatrix}\sigma_{X,t}\,\varepsilon_{X,t} \\ \sigma_{Y,t} \, \varepsilon_{Y,t}\end{pmatrix},\qquad
	\begin{pmatrix} \sigma_{X,t} \\ \sigma_{Y,t}  \end{pmatrix}
	= \widetilde{\vomega} + 
	\widetilde{\mA} \begin{pmatrix} |X_{t-1}| \\ |Y_{t-1}| \end{pmatrix} +
	\widetilde{\mB} \begin{pmatrix} \sigma_{X,t-1} \\ \sigma_{Y,t-1} \end{pmatrix},
\end{equation}
where $\widetilde{\vomega}\in\mathbb{R}^2$, $\widetilde{\mA}, \widetilde{\mB}\in\mathbb{R}^{2\times2}$.
In \eqref{eqn:ECCCmodel}, the $\mathcal{F}_{t-1}$-measurable conditional volatilities $\sigma_{X,t}$ and $\sigma_{Y,t}$ are independent of the independent and identically distributed (i.i.d.) innovations
$\big(\varepsilon_{X,t}, \varepsilon_{Y,t} \big)' \overset{\text{i.i.d.}}{\sim} F \big( \vzeros, \mSigma \big)$, where $F$ denotes a generic absolutely continuous, bivariate distribution with zero mean and covariance matrix $\mSigma=\begin{psmallmatrix}1 & \rho\\ \rho & 1 \end{psmallmatrix}$, $\rho\in(-1,1)$.
Equation \eqref{eqn:ECCCmodel} deviates from \citet{Jea98} by using \textit{absolute} (instead of \textit{squared}) returns as drivers of volatility, since these have more predictive content for volatility \citep{FG07}.
Model~\eqref{eqn:ECCCmodel} implies the (VaR, CoVaR) dynamics in \eqref{eqn:SAVCoCAViaRModelClass}, where the true parameters $\vomega_0, \mA_0, \mB_0$ of the model in \eqref{eqn:SAVCoCAViaRModelClass} arise as transformations of $\widetilde{\vomega}, \widetilde{\mA}, \widetilde{\mB}$.\footnote{
	Multiplying the rows of the volatility equation in \eqref{eqn:ECCCmodel} by the true VaR and CoVaR of the innovations respectively gives that
	$\vtheta_0^v  
	= \big( \omega_{1,0}, A_{11,0} , A_{12,0}, B_{11,0}, B_{12,0} \big)^\prime
	= \big( v_\varepsilon \widetilde{\omega}_1, v_\varepsilon \widetilde{A}_{11}, v_\varepsilon \widetilde{A}_{12}, \widetilde{B}_{11}, v_\varepsilon/ c_\varepsilon \widetilde{B}_{12} \big)^\prime$ and
	$\vtheta_0^c  
	= \big( \omega_{2,0}, A_{21,0} , A_{22,0}, B_{21,0}, B_{22,0} \big)^\prime
	= \big( c_\varepsilon \widetilde{\omega}_2, c_\varepsilon \widetilde{A}_{12}, c_\varepsilon \widetilde{A}_{22},  c_\varepsilon / v_\varepsilon \widetilde{B}_{21}, \widetilde{B}_{22} \big)^\prime$,
	where $v_\varepsilon$ is the $\beta$-VaR of $\varepsilon_{X,t}$ and  $c_\varepsilon$ the $\alpha|\beta$-CoVaR of the pair $(\varepsilon_{X,t}, \varepsilon_{Y,t})^\prime$, whose analytical form is given in \citet[Theorem~3.1~(b)]{MS14}.\label{FN1}
}




\begin{rem}\label{rem:Init}
The following initializations of the above ECCC--GARCH model in \eqref{eqn:ECCCmodel} lead to correctly specified (VaR, CoVaR) models in \eqref{eqn:SAVCoCAViaRModelClass} initialized according to items~\eqref{item:InitConst}--\eqref{item:InitInfinitePast} in (the continuation of) Example~\ref{ex:1}:
\begin{enumerate}[(a)] 

		\item The recursion in \eqref{eqn:ECCCmodel} holds for $t\geq2$ and $(X_1, Y_1)^\prime=(\sigma_{X,1} \varepsilon_{X,1},\sigma_{Y,1} \varepsilon_{Y,1})^\prime$ with known constant $(\sigma_{X,1},\sigma_{Y,1} )^\prime = \widetilde{\vk} \in\mathbb{R}^2$ and known $(v_{\varepsilon},c_{\varepsilon})^\prime$ from footnote~\ref{FN1}.
		Then, $(v_1(\vtheta^v), c_1(\vtheta^c))^\prime$ (i.e., the unconditional (VaR, CoVaR) model of $(X_1, Y_1)^\prime$) is given by some known $\vk\in\mathbb{R}^2$ and the model initialization for \eqref{eqn:SAVCoCAViaRModelClass} in Example~\ref{ex:1}~\eqref{item:InitConst} is correct.
		
		\item\label{it:CCC ii} The recursion in \eqref{eqn:ECCCmodel} holds for $t\geq2$ and $(X_1, Y_1)^\prime=(\sigma_{X,1} \varepsilon_{X,1}, \sigma_{Y,1} \varepsilon_{Y,1})^\prime$ with $(\sigma_{X,1}, \sigma_{Y,1})^\prime = \widetilde{\vomega} \in \mathbb{R}^2$. 
		Then, $(v_1(\vtheta^v), c_1(\vtheta^c))^\prime$ (i.e., the unconditional (VaR, CoVaR) model of $(X_1, Y_1)^\prime$) is given by $\vomega\in\mathbb{R}^2$ and the model initialization for \eqref{eqn:SAVCoCAViaRModelClass} in Example~\ref{ex:1}~\eqref{item:InitParameter} is correct.

		\item The recursion in \eqref{eqn:ECCCmodel} holds for all $t\in\mathbb{Z}$.
		Then, the recursion in \eqref{eqn:SAVCoCAViaRModelClass} also holds for all $t\in\mathbb{Z}$ as under Example~\ref{ex:1}~\eqref{item:InitInfinitePast}.
		
\end{enumerate}
In each of the above cases, the information set $\mathcal{F}_t$ depends on the chosen $\init$, and the true parameters $(\vtheta_0^v, \vtheta_0^c)$ of the SAV CoCAViaR models are as in footnote~\ref{FN1}.
\end{rem}

While ECCC--GARCH models are primarily designed to model the complete predictive distribution of $(X_t,Y_t)^\prime\mid\mathcal{F}_{t-1}$, they can also be used to predict VaR and CoVaR (via \eqref{eqn:SAVCoCAViaRModelClass} and the parameter transformation given in footnote~\ref{FN1}). 
When \eqref{eqn:ECCCmodel} describes the true data-generating process for $(X_t, Y_t)^\prime$, then there will be little difference between the (VaR, CoVaR) forecasts by our model in Example~\ref{ex:1} and those issued by the ECCC--GARCH model (assuming consistent parameter estimates in both cases). However, when \eqref{eqn:ECCCmodel} does not describe the underlying dynamics of $(X_t, Y_t)^\prime$, the two approaches to (VaR, CoVaR) modeling may yield widely different predictions for reasons discussed in more detail in Section~\ref{sec:ForecastingApplication2}.

\begin{rem}
We pointed out above that DCC--GARCH models may suffer from an identification problem for the CoVaR, because the ``square root'' $\mSigma_t$ of the conditional covariance matrix $\mH_t=\mSigma_{t}\mSigma_t^\prime$ is not uniquely defined.
In contrast, CCC-type GARCH models considered by---among others---\citet{Jea98} (and us in \eqref{eqn:ECCCmodel}) do not have this problem, since they directly model $\mSigma_t$.
Of course, $\mSigma_t$ is much more likely to be misspecified in CCC-type models, but the identification issue is circumvented.
\citet{francq2025inference} exploit this convenient feature of constant conditional correlations to derive asymptotic properties of CoVaR forecasts; see in particular their equation~(1).
Similarly, \citet{Hog25} proves asymptotic normality for the MES based on a CCC--GARCH-type filter and a MES estimator from extreme value theory.
\end{rem}

In our simulations in Section~\ref{sec:Simulations}, we use (a restricted version of) model \eqref{eqn:ECCCmodel} to generate data to assess how well our two-step M-estimator performs.
In our empirical application, we consider various SAV CoCAViaR specifications based on zero restrictions in $\mA$ and $\mB$, as well as generalizations to so-called ``asymmetric slope'' models based on the positive and negative components of $X_t$ and $Y_t$ instead of on their absolute values.



\subsection{Parameter Estimation}
\label{sec:ParameterEstimation}

We now introduce estimators of the unknown model parameters $\vtheta_{0}^{v}$ and $\vtheta_{0}^{c}$. As we consider M-estimation in this paper, we require a to-be-minimized objective (or also: scoring) function. However, as pointed out in the Motivation, there is no \textit{real-valued} scoring function associated with the pair (VaR, CoVaR). \citet[Theorem 1]{DFZ_CharMest} show that the existence of such a (consistent) scoring function is a necessary condition for consistent M-estimation of semiparametric models.

To overcome this drawback, \citet{FH24} propose a \textit{$\mathbb{R}^2$-valued} scoring function in the closely related context of forecast evaluation.
To be able to compare forecasts, $\mathbb{R}^2$ has to be equipped with an order, and \citet{FH24} show that the lexicographic order is suitable for that purpose. Specifically, they prove that (under some regularity conditions) the expectation of the $\mathbb{R}^2$-valued scoring function
\begin{equation}\label{eq:loss}
	\mS\Bigg(\begin{pmatrix}v\\ c\end{pmatrix}, \begin{pmatrix}x\\ y\end{pmatrix}\Bigg) = \begin{pmatrix}S^{\VaR}(v,x)\\ S^{\CoVaR}\big((v,c)^\prime, (x,y)^\prime\big) \end{pmatrix}=
	\begin{pmatrix}
		[\1_{\{x\leq v\}}-\beta][v-x]\\
		\1_{\{x>v\}}[\1_{\{y\leq c\}}-\alpha][c-y]
	\end{pmatrix}
\end{equation}
is minimized by the true VaR and CoVaR with respect to the lexicographic order. That is, for all $v,c\in\mathbb{R}$,
\[
	\E\Bigg[\mS\Bigg(\begin{pmatrix}\VaR_\beta(F_X)\\ \CoVaR_{\a\mid\b}(F_{X,Y})\end{pmatrix}, \begin{pmatrix}X\\ Y\end{pmatrix}\Bigg)\Bigg]\preceq_{\lex}\E\Bigg[\mS\Bigg(\begin{pmatrix}v\\ c\end{pmatrix}, \begin{pmatrix}X\\ Y\end{pmatrix}\Bigg)\Bigg],
\]
where $(x_1,x_2)^\prime\preceq_{\lex} (y_1,y_2)^\prime$ if $x_1<y_1$ or ($x_1=y_1$ and $x_2\leq y_2$). Note that $S^{\VaR}(\cdot,\cdot)$ in \eqref{eq:loss} is the standard tick-loss function known from quantile regression. Clearly, $S^{\CoVaR}(\cdot,\cdot)$ is similar in structure to $S^{\VaR}(\cdot,\cdot)$, with the only difference being the indicator $\1_{\{x>v\}}$ that restricts the evaluation to observations with VaR exceedances in the first component. 
In the related literature, these scoring functions are often also called ``loss functions'' \citep{Gne11}. To avoid confusion with financial ``losses'', we adhere to the term ``scoring function''.

The definition of the lexicographic order suggests the following two-step M-estimator of $\vtheta_{0}^{v}$ and $\vtheta_{0}^{c}$. In the first step, we estimate $\vtheta_0^v$ via
\begin{align}
	\label{eqn:MestVaR}
	\widehat{\vtheta}_{n}^{v} = \argmin_{\vtheta^{v}\in\mTheta^{v}}\frac{1}{n}\sum_{t=1}^{n}S^{\VaR}\big(v_t(\vtheta^{v}), X_t\big),
\end{align}
such that the parameters are chosen to minimize the average empirical score in the VaR component over the VaR parameter space $\mTheta^{v}$.
For the VaR model $v_t(\cdot)$, this is the quantile regression estimator of \citet{EM04}.

With the estimate $\widehat{\vtheta}_n^v$ at hand, the lexicographic order then suggests to minimize the average empirical score in the second component to estimate $\vtheta_0^c$ via
\begin{align}
	\label{eqn:MestCoVaR}
	\widehat{\vtheta}_{n}^{c} = \argmin_{\vtheta^{c}\in\mTheta^{c}}\frac{1}{n}\sum_{t=1}^{n}S^{\CoVaR}\Big(\big(v_t(\widehat{\vtheta}_{n}^{v}), c_t(\vtheta^{c})\big)^\prime, \big(X_t, Y_t\big)^\prime\Big).
\end{align}
For this two-step estimator to be feasible, the requirement that the VaR evolution does not depend on $\vtheta^c$ is essential.
Appendix~\ref{Asymptotic Normality} shows that the presence of $\widehat{\vtheta}_n^v$ in the second-stage minimization impacts the asymptotic variance of $\widehat{\vtheta}_{n}^{c}$. Of course, this is usually the case for two-step estimators \citep[Section 6]{NeweyMcFadden1994}. 
A two-step estimator similar in spirit to \eqref{eqn:MestVaR}--\eqref{eqn:MestCoVaR} is employed in \citet{DimiHogaMES2024} for their \textit{static} regressions under adverse conditions.

Note that the estimators $\widehat{\vtheta}_{n}^{v}$ and $\widehat{\vtheta}_{n}^{c}$ depend on the assumed initialization information in $\init$ through the model functions $v_t(\vtheta^{v})$ and $c_t(\vtheta^{c})$.
For instance, the parameter estimators $\widehat{\vtheta}_n^v$ and $\widehat{\vtheta}_n^c$ for the model in Example~\ref{ex:1} rely on initial values for $v_1(\vtheta^v)$ and $c_1(\vtheta^c)$, which depend on whether the model was initialized according to \eqref{item:InitConst} or \eqref{item:InitParameter}.
We mention that in case \eqref{item:InitInfinitePast}---when $\init$ contains the infinite past---the estimators cannot be computed as $v_1(\vtheta^v)$ and $c_1(\vtheta^c)$ depend on infinitely many lags. In this case, one has to rely on a truncated information set for estimation.

Next, Section~\ref{Consistency} shows consistency of the parameter estimators in the case where the full $\init$ can be used for estimation (such as \eqref{item:InitConst}--\eqref{item:InitParameter}).
The case where $\init$ needs to be truncated (such as \eqref{item:InitInfinitePast}) is treated in Appendix~\ref{sec:StartingValues}.

\section{Asymptotic Properties of the Estimators}\label{Asymptotic Properties}

\subsection{Consistency}\label{Consistency}

Here, we show consistency of our two-step M-estimators $\widehat{\vtheta}_{n}^{v}$ and $\widehat{\vtheta}_{n}^{c}$. 
To do so, we introduce several regularity conditions in Assumption~\ref{ass:cons}.
Throughout the paper, we let $K\in(0,\infty)$ denote some large universal constant not depending on $n$, $t$ or any other values.
Further, $\Vert\vx\Vert$ denotes the Euclidean norm when $\vx$ is a vector, and the Frobenius norm when $\vx$ is matrix-valued. 
The joint cumulative distribution function (c.d.f.) of $(X_{t}, Y_t)^\prime\mid\mathcal{F}_{t-1}$ is denoted by $F_t(\cdot,\cdot)$, and its Lebesgue density (which we assume exists) by $f_t(\cdot,\cdot)$. 
Similarly, $F_t^{W}(\cdot)$ ($f_t^{W}(\cdot)$) denotes the distribution (density) function of $W_t\mid\mathcal{F}_{t-1}$ for $W\in\{X,Y\}$. 
For a sufficiently smooth function $\mathbb{R}^{p}\ni\vtheta\mapsto f(\vtheta) \in \mathbb{R}$, we denote the $(p\times1)$-gradient by $\nabla f(\vtheta)$, its transpose by $\nabla' f(\vtheta)$ and the $(p\times p)$-Hessian by $\nabla^2 f(\vtheta)$. 

\begin{assumption}\label{ass:cons}
$ $ \vspace{-0.3cm}
\renewcommand{\theenumi}{(\roman{enumi})}
\begin{enumerate}
	\item\label{it:id} 
	The models $v_t(\cdot)$ and $c_t(\cdot)$ are correctly specified in the sense of \eqref{eqn:TrueModelParameters}. 
	
	\item\label{it:str stat} $\big\{(X_t, Y_t, \mZ_t^\prime)^\prime\big\}_{t\in\mathbb{N}}$ is strictly stationary.
	
	\item\label{it:cond dist} For all $t\in\mathbb{N}$, $F_t(\cdot,\cdot)$ belongs to a class of distributions on $\mathbb{R}^2$ that possesses a positive Lebesgue density for all $(x,y)^\prime\in\mathbb{R}^2$ such that $F_t(x,y)\in(0,1)$. 
	
	\item\label{it:Lipschitz cons} For all $t\in\mathbb{N}$, $\big|f_t^{X}(x)-f_t^{X}(x^\prime)\big|\leq K|x-x^\prime|$, $\big|f_t^{Y}(y)-f_t^{Y}(y^\prime)\big|\leq K|y-y^\prime|$, and $\big|\partial_2 F_t(x,y)-\partial_2 F_t(x,y^\prime)\big|\leq K|y-y^\prime|$ for all $x,x^\prime,y,y^\prime\in\mathbb{R}$.
	
	\item\label{it:cond dens} For all $t\in\mathbb{N}$ it holds that $f_t^{X}(\cdot)\leq K$.
	Moreover, there exist some $f_1,f_2>0$ such that for some constants $p_1,p_2>0$, $\P\Big\{f_t^{X}\big(v_t(\vtheta_0^v)\big)> f_1\Big\}>p_1$ and $\P\Big\{\int_{v_t(\vtheta_0^v)}^{\infty}f_t\big(x, c_t(\vtheta_0^c)\big)\D x> f_2\Big\}>p_2$ hold for all $t\in\mathbb{N}$.
	
	\item\label{it:compact} $\mTheta=\mTheta^{v}\times\mTheta^{c}$ is compact, where $\mTheta^{v}\subset\mathbb{R}^p$ and $\mTheta^{c}\subset\mathbb{R}^q$.
	
	\item\label{it:ULLN} $\Big\{\mS\Big(\big(v_t(\vtheta^v), c_t(\vtheta^c)\big)^\prime, (X_t, Y_t)^\prime\Big)\Big\}_{t\in\mathbb{N}}$ obeys the uniform law of large numbers (ULLN) in $\vtheta=(\vtheta^{v\prime},\vtheta^{c\prime})^\prime\in\mTheta$, i.e., as $n\to\infty$,
	\[
		\sup_{\vtheta\in\mTheta}\Bigg\Vert\frac{1}{n}\sum_{t=1}^{n}\bigg\{\mS\Big(\big(v_t(\vtheta^v), c_t(\vtheta^c)\big)^\prime, (X_t, Y_t)^\prime\Big)- \E\Big[\mS\Big(\big(v_t(\vtheta^v), c_t(\vtheta^c)\big)^\prime, (X_t, Y_t)^\prime\Big)\Big]\bigg\}\Bigg\Vert=o_{\P}(1).
	\]
	
	\item\label{it:unique id} 
	\sloppy
	For every $\xi>0$, there exists a $\tau>0$ such that (a) if $\big\Vert\vtheta^v-\vtheta_0^v\big\Vert\geq\xi$, then $\liminf_{n\to\infty}\frac{1}{n}\sum_{t=1}^{n}\P\big\{|v_t(\vtheta^v) - v_t(\vtheta_0^v)|>\tau\mid f_t^{X}\big(v_t(\vtheta_0^v)\big)> f_1\big\}>0$ and (b) if $\big\Vert\vtheta^c-\vtheta_0^c\big\Vert\geq\xi$, then $\liminf_{n\to\infty}\frac{1}{n}\sum_{t=1}^{n}\P\big\{|c_t(\vtheta^c) - c_t(\vtheta_0^c)|>\tau\mid \int_{v_t(\vtheta_0^v)}^{\infty}f_t\big(x, c_t(\vtheta_0^c)\big)\D x> f_2\big\}>0$, where $f_1$ and $f_2$ are from Assumption~\ref{ass:cons}~\ref{it:cond dens}.
	
	\item\label{it:smooth} For all $t\in\mathbb{N}$, $v_t(\vtheta^v)$ and $c_t(\vtheta^{c})$ are $\mathcal{F}_{t-1}$-measurable and a.s.\ continuous in $\vtheta^v$ and $\vtheta^{c}$.
	
	\item\label{it:diff c} For all $t\in\mathbb{N}$, $v_t(\vtheta^v)$ is a.s.~continuously differentiable on $\inter(\mTheta^v)$.
	
	\item\label{it:bound} There exists a neighborhood of $\vtheta_0^v$, such that $\big\Vert\nabla v_t(\vtheta^v)\big\Vert\leq V_1(\mathcal{F}_{t-1})$ for all elements $\vtheta^v$ of that neighborhood and all $t\in\mathbb{N}$. Furthermore, for all $t\in\mathbb{N}$ it holds that $|v_t(\vtheta^v)|\leq V(\mathcal{F}_{t-1})$ for all $\vtheta^{v}\in\mTheta^{v}$, and $|c_t(\vtheta^c)|\leq C(\mathcal{F}_{t-1})$ for all $\vtheta^{c}\in\mTheta^{c}$.
	
	\item\label{it:mom bounds cons} For all $t\in\mathbb{N}$ it holds that $\E\big[V(\mathcal{F}_{t-1})\big]\leq K$, $\E\big[V_1(\mathcal{F}_{t-1})\big]\leq K$, $\E\big[V_1(\mathcal{F}_{t-1})C(\mathcal{F}_{t-1})\big]\leq K$, $\E|X_t|\leq K$, $\E|Y_t|^{1+\iota}\leq K$ for some $\iota>0$.
\end{enumerate}
\end{assumption}

Since some of the conditions in Assumption~\ref{ass:cons} are high-level, we verify these for CCC--GARCH and VAR models in Appendices~\ref{sec:verification} and \ref{sec:verificationVAR}, respectively.
In Assumption~\ref{ass:cons}, item~\ref{it:id} ensures a correctly specified VaR and CoVaR model, including a correct model initialization at $t=1$ through $\mathcal{I}_0$, which can, e.g., be achieved by one of the possibilities \eqref{item:InitConst}--\eqref{item:InitInfinitePast} discussed in Example~\ref{ex:1} and Remark~\ref{rem:Init}.
We treat arbitrarily initialized models in Appendix~\ref{sec:StartingValues}.
Item~\ref{it:str stat} is a standard stationarity condition for the observables. 
Yet notice that the (VaR, CoVaR) process $\big(v_t(\cdot), c_t(\cdot)\big)^\prime$ does not have to be stationary and---in general---is non-stationary due to the expanding information set $\mathcal{F}_{t-1}$ for processes initialized in the finite past.
Item~\ref{it:cond dist} ensures---among other things---strict (multi-objective) consistency of the scoring function given in \eqref{eq:loss}; see \citet[Theorem~4.2]{FH24} for details.
Uniform Lipschitz conditions on derivatives of the joint conditional c.d.f.~are imposed in item~\ref{it:Lipschitz cons}. 
The next item~\ref{it:cond dens} provides certain boundedness conditions on the conditional probability density function of $X_t$.
Compactness of the parameter space in \ref{it:compact} is a standard requirement in extremum estimation; see \citet{NeweyMcFadden1994}. 
Assumption~\ref{ass:cons}~\ref{it:ULLN} is also a standard condition, imposed in similar form by \citet[Assumption~C6]{EM04}, \citet[Assumption~1~(A)]{PZC19} and \citet[Assumption~4]{Catania2022}. 
In Appendices~\ref{Ass1} and \ref{Ass1 VAR}, we show how this condition can be verified based on Theorem~21.9 in \citet{Dav94} for an underlying CCC--GARCH and a linear VAR process, respectively.
Item~\ref{it:unique id} for the VaR model is identical to condition (ID) in \citet{Wei91} and, loosely speaking, ensures that $v_t(\vtheta^v)$ differs sufficiently from $v_t(\vtheta_0^v)$ with positive (conditional) probability for any $\vtheta^v \not= \vtheta_0^v$.
This condition is required to establish the unique identifiability of the parameter $\vtheta_0^v$ in the sense of \citet[Definition~2.1]{Whi80}.
Condition (b) of item~\ref{it:unique id} is the analogous condition for the CoVaR model, required for the unique identifiability of $\vtheta_0^c$.
The final items~\ref{it:smooth}--\ref{it:mom bounds cons} are smoothness and moment conditions. 
We mention that the moment bounds in item \ref{it:mom bounds cons} on the observables, $\E|X_t|\leq K$ and $\E|Y_t|^{1+\iota}\leq K$, are sufficiently weak to be practically always satisfied in financial and economic applications.

Our first main theoretical result establishes the consistency of $\widehat{\vtheta}_{n}^{v}$ and $\widehat{\vtheta}_{n}^{c}$.

\begin{thm}\label{thm:cons}
Suppose Assumption~\ref{ass:cons} holds. Then, as $n\to\infty$, $\widehat{\vtheta}_{n}^{v}\overset{\P}{\longrightarrow}\vtheta_0^{v}$ and $\widehat{\vtheta}_{n}^{c}\overset{\P}{\longrightarrow}\vtheta_0^{c}$.
\end{thm}

\begin{proof}
See Appendix~\ref{sec:thm1}.
\end{proof}

The proof of Theorem~\ref{thm:cons} exploits an ``in probability'' version of Lemma~2.2 in \citet{Whi80} (also see \citet[Theorem 3.4]{Whi96} and the proof of Theorem 1 in \citet{Wei91}), which does not require stationarity of the true model values $v_t(\vtheta_0^v)$ and $c_t(\vtheta_0^c)$.
The first result that $\widehat{\vtheta}_{n}^{v}\overset{\P}{\longrightarrow}\vtheta_0^{v}$ is essentially a version of Theorem~1 in \citet{EM04}; the only difference being that our regularity conditions are more involved, since we also show that $\widehat{\vtheta}_{n}^{c}\overset{\P}{\longrightarrow}\vtheta_0^{c}$.
The proof of this latter result is complicated by the fact that the CoVaR score 
\[
	S^{\CoVaR}\big((v_t(\vtheta^v),c_t(\vtheta^c))^\prime, (X_t,Y_t)^\prime\big)=\1_{\{X_t>v_t(\vtheta^v)\}}\big[\1_{\{Y_t\leq c_t(\vtheta^c)\}}-\alpha\big]\big[c_t(\vtheta^c)-Y_t\big]
\]
is discontinuous in the VaR parameter $\vtheta^v$. 
This fact necessitates many of the regularity conditions imposed in Assumption~\ref{ass:cons}.

\subsection{Additional Asymptotic Results}

We provide additional asymptotic theory in the Online Appendix, which we only briefly outline here.
First, following \citet{EM04}, \citet{WhiteKimManganelli2015}, \citet{PZC19} and \citet{Catania2022}, Theorem~\ref{thm:cons} requires correct use of initialization information in the computation of the estimator.
When using arbitrary (and, hence, possibly incorrect) initial values, Theorem~\ref{prop:StartingValues} in Appendix~\ref{Asymptotic Normality} shows that consistency of the estimator is retained under some weak contraction condition on the models.
This condition is trivially satisfied for the models we use in the simulations and the application; see \eqref{eqn:SAVCoCAViaRModelClass} and \eqref{eqn:ASCoCAViaRModelClass}.

Second, we also establish asymptotic normality of $(\widehat{\vtheta}_n^v, \widehat{\vtheta}_n^c)$ in Theorem~\ref{thm:an} in Appendix~\ref{Asymptotic Normality} under the additional Assumption~\ref{ass:an}, which we verify for CCC--GARCH and VAR models in Appendices~\ref{sec:verification} and \ref{sec:verificationVAR}, respectively.

Third, we propose estimators of the asymptotic variance of the limiting distribution of $(\widehat{\vtheta}_n^v, \widehat{\vtheta}_n^c)$ and prove their consistency in Theorem~\ref{thm:avar} in Appendix~\ref{avar}.
This allows for feasible inference on the model parameters.

\section{Simulations}
\label{sec:Simulations}

\begin{sidewaystable}[p!]
	\centering
	\footnotesize
	\begin{tabular}{rr l rrrrr l rrrrr l rrrrr}
		\toprule
		\multicolumn{2}{l}{\textbf{VaR}} & &  \multicolumn{5}{c}{$\omega_1$} & & \multicolumn{5}{c}{$A_{11}$} & & \multicolumn{5}{c}{$B_{11}$}  \\
		\cmidrule{4-8}  	\cmidrule{10-14} 	 \cmidrule{16-20} 
		$\alpha,\beta$ & $n$ & &
		Bias & M Bias & $\hat \sigma_\text{emp}$ & $\hat \sigma_\text{asy}$ & CI  & &
		Bias & M Bias & $\hat \sigma_\text{emp}$ & $\hat \sigma_\text{asy}$ & CI  & &
		Bias & M Bias & $\hat \sigma_\text{emp}$ & $\hat \sigma_\text{asy}$ & CI  \\
		\midrule
		\multirow{4}{*}{0.90} 
		& 500 &  & 0.0157 & 0.0012 & 0.055 & 0.084 & 0.98 &  & 0.0114 & $-$0.0007 & 0.110 & 0.142 & 0.94 &  & $-$0.0523 & $-$0.0065 & 0.193 & 0.265 & 0.97 \\ 
		& 1000 &  & 0.0105 & 0.0008 & 0.041 & 0.073 & 0.99 &  & 0.0092 & 0.0025 & 0.078 & 0.100 & 0.95 &  & $-$0.0357 & $-$0.0067 & 0.145 & 0.222 & 0.98 \\ 
		& 2000 &  & 0.0064 & 0.0008 & 0.028 & 0.060 & 0.99 &  & 0.0045 & 0.0006 & 0.055 & 0.071 & 0.95 &  & $-$0.0208 & $-$0.0045 & 0.099 & 0.179 & 0.99 \\ 
		& 4000 &  & 0.0021 & $-$0.0002 & 0.017 & 0.045 & 1.00 &  & 0.0015 & $-$0.0003 & 0.039 & 0.051 & 0.96 &  & $-$0.0072 & $-$0.0016 & 0.062 & 0.133 & 0.99 \\ 
		\addlinespace                 
		\multirow{4}{*}{0.95} 
		& 500 &  & 0.0242 & 0.0013 & 0.083 & 0.103 & 0.98 &  & 0.0110 & $-$0.0026 & 0.147 & 0.192 & 0.93 &  & $-$0.0585 & $-$0.0054 & 0.211 & 0.250 & 0.96 \\ 
		& 1000 &  & 0.0154 & 0.0011 & 0.057 & 0.095 & 0.98 &  & 0.0122 & 0.0014 & 0.106 & 0.136 & 0.93 &  & $-$0.0388 & $-$0.0075 & 0.149 & 0.218 & 0.97 \\ 
		& 2000 &  & 0.0086 & 0.0013 & 0.037 & 0.080 & 0.99 &  & 0.0065 & 0.0029 & 0.076 & 0.097 & 0.94 &  & $-$0.0215 & $-$0.0074 & 0.099 & 0.179 & 0.98 \\ 
		& 4000 &  & 0.0043 & 0.0009 & 0.024 & 0.061 & 1.00 &  & 0.0020 & $-$0.0007 & 0.053 & 0.069 & 0.95 &  & $-$0.0098 & $-$0.0031 & 0.067 & 0.137 & 0.99 \\ 
		\midrule
		\multicolumn{2}{l}{\textbf{CoVaR}} & &  \multicolumn{5}{c}{$\omega_2$} & & \multicolumn{5}{c}{$A_{22}$} & & \multicolumn{5}{c}{$B_{22}$}  \\
		\cmidrule{4-8}  	\cmidrule{10-14} 	 \cmidrule{16-20} 
		$\alpha,\beta$ & $n$ & &
		Bias & M Bias & $\hat \sigma_\text{emp}$ & $\hat \sigma_\text{asy}$ & CI  & &
		Bias & M Bias & $\hat \sigma_\text{emp}$ & $\hat \sigma_\text{asy}$ & CI  & &
		Bias & M Bias & $\hat \sigma_\text{emp}$ & $\hat \sigma_\text{asy}$ & CI  \\
		\midrule
		\multirow{4}{*}{0.90}  
		& 500 &  & 0.1341 & 0.0156 & 0.223 & 0.171 & 0.85 &  & 0.0568 & 0.0197 & 0.453 & 0.486 & 0.90 &  & $-$0.4141 & $-$0.1090 & 0.628 & 0.521 & 0.81 \\ 
		& 1000 &  & 0.1136 & 0.0075 & 0.208 & 0.147 & 0.86 &  & 0.0539 & 0.0245 & 0.300 & 0.342 & 0.90 &  & $-$0.3467 & $-$0.0642 & 0.588 & 0.442 & 0.83 \\ 
		& 2000 &  & 0.0776 & 0.0035 & 0.173 & 0.129 & 0.89 &  & 0.0407 & 0.0181 & 0.203 & 0.246 & 0.91 &  & $-$0.2364 & $-$0.0354 & 0.492 & 0.380 & 0.88 \\ 
		& 4000 &  & 0.0353 & 0.0009 & 0.111 & 0.103 & 0.94 &  & 0.0294 & 0.0133 & 0.137 & 0.173 & 0.94 &  & $-$0.1117 & $-$0.0122 & 0.328 & 0.304 & 0.93 \\ 
		\addlinespace                                     
		\multirow{4}{*}{0.95} 
		& 500 &  & 0.2319 & 0.0889 & 0.329 & 0.293 & 0.82 &  & 0.1006 & 0.0278 & 0.976 & 0.973 & 0.88 &  & $-$0.5528 & $-$0.3042 & 0.663 & 0.651 & 0.77 \\ 
		& 1000 &  & 0.2148 & 0.0581 & 0.316 & 0.255 & 0.82 &  & 0.1094 & 0.0423 & 0.673 & 0.693 & 0.86 &  & $-$0.4896 & $-$0.2205 & 0.640 & 0.559 & 0.78 \\ 
		& 2000 &  & 0.1720 & 0.0218 & 0.295 & 0.223 & 0.84 &  & 0.0775 & 0.0365 & 0.449 & 0.507 & 0.87 &  & $-$0.3816 & $-$0.1059 & 0.596 & 0.487 & 0.82 \\ 
		& 4000 &  & 0.1234 & 0.0100 & 0.254 & 0.195 & 0.87 &  & 0.0707 & 0.0314 & 0.305 & 0.365 & 0.90 &  & $-$0.2759 & $-$0.0501 & 0.527 & 0.422 & 0.85 \\ 
		\bottomrule
	\end{tabular}
	\caption{Simulation results for the six parameter CoCAViaR model based on the CCC--GARCH model in \eqref{eqn:ECCCmodel} and $M=5000$ simulation replications.
			The columns ``Bias'' show the average bias and the columns ``M Bias'' the median bias of the parameter estimates. 
			The columns ``$\hat \sigma_\text{emp}$'' report the empirical standard deviation of the parameter estimates, ``$\hat \sigma_\text{asy}$'' the median of the estimated standard deviations, and the columns ``CI'' show the coverage rates of $95\%$-confidence intervals.}
	\label{tab:SimResultsCoCAViaR6p}
\end{sidewaystable}

Here, we consider estimation of a dynamic SAV CoCAViaR model given in \eqref{eqn:SAVCoCAViaRModelClass}.
For this, we simulate $\big\{(X_t, Y_t)^\prime\big\}_{t=1,\ldots,n}$ from the absolute value ECCC--GARCH model in \eqref{eqn:ECCCmodel} with diagonal $\widetilde{\mA}$ and $\widetilde{\mB}$, such that Assumptions~\ref{ass:cons} and \ref{ass:an} are met (see Appendix~\ref{sec:verification} for the verification).
We choose the parameter values $\widetilde{\vomega} = (0.04,\ 0.02)^\prime$, $\widetilde{\mA} = \begin{psmallmatrix} 0.1 & 0 \\ 0 & 0.15 \end{psmallmatrix}$,  $\widetilde{\mB} = \begin{psmallmatrix} 0.8 & 0 \\ 0 & 0.75 \end{psmallmatrix}$ and  let $F$ be the multivariate (marginally standardized) $t$-distribution with $\nu = 8$ degrees of freedom and a residual correlation of $\rho = 0.5$. 
We initialize the volatility process at $t=1$ with the starting values $\big(\sigma_{X,1}, \sigma_{Y,1}\big)^\prime = \widetilde{\vomega}$, hence complying with item \eqref{it:CCC ii} in Remark~\ref{rem:Init}.
 

The off-diagonal zero-restrictions in $\widetilde{\mA}$ and $\widetilde{\mB}$ result in (more or less) the classic CCC--GARCH model of \citet{Bol90}.
Recall that $\widetilde{B}_{12} = 0$ is essential for our two-step M-estimator, whereas $\widetilde{B}_{21} = 0$ just facilitates the derivation of the asymptotic theory. 
We estimate the SAV-diag CoCAViaR model, which arises for diagonal $\mA$ and $\mB$ in \eqref{eqn:SAVCoCAViaRModelClass}. 
(Table~\ref{tab:ModelsApplication2} below provides a complete nomenclature of CoCAViaR models considered in this paper.)
Therefore, the to-be-estimated parameters are $\vtheta^v = (\omega_1, A_{11}, B_{11})^\prime$ and $\vtheta^c = (\omega_2, A_{22}, B_{22})^\prime$, whose true values can be obtained from $\widetilde{\vomega}$, $\widetilde{\mA}$ and $\widetilde{\mB}$ as in footnote~\ref{FN1}.

Table~\ref{tab:SimResultsCoCAViaR6p} shows simulation results for the dynamic CoCAViaR model based on $M=5000$ replications for the choices $\alpha = \beta \in \{0.9,\ 0.95\}$ and sample sizes $n \in \{500,\ 1000,\ 2000,\ 4000\}$.
We correctly initialize the model recursion in the numerical estimation process by using the model parameters $\vomega$ as described in item \eqref{item:InitParameter} of Example~\ref{ex:1}.
We show in Online Supplement~\ref{sec:Initializations} that the results are almost unchanged when using an (incorrect) initialization with a constant value as discussed in item \eqref{item:InitConst} of Example~\ref{ex:1}.
The asymptotic variance-covariance matrices are estimated as detailed in Appendix~\ref{avar}; see in particular Remark~\ref{rem:asvar} of that appendix.
A formal description of the table columns is given in the table caption.

The columns reporting the (average and median) bias and the standard deviations confirm the consistency of the estimator from Theorem~\ref{thm:cons}. 
In general, the VaR parameters are estimated with smaller empirical bias than the CoVaR parameters, which is not surprising given that CoVaR is further out in the tail and, hence, subject to larger estimation uncertainty.
Furthermore, the average bias is often larger than the median bias, indicating that the empirical distributions of the parameter estimates are still subject to some skewness or outliers.
Notice that even for the largest sample size of $n=4000$ for our choice of $\alpha = \beta = 0.95$, the CoVaR model is essentially estimated as a $95\%$-quantile based on an effective sample size of only $\tilde n = (1-\beta) n = 200$ observations, which is an inherently difficult task.
We further see that sample sizes of around 2000 days are required to reliably estimate the models for these extreme levels. This is especially true for the CoVaR parameters.

Next, we assess the precision of the standard errors and coverage of the confidence intervals for the parameters derived from Theorems~\ref{thm:an} and \ref{thm:avar} (see Appendix~\ref{Asymptotic Normality} and \ref{avar}).
The results for the estimated standard deviations and the confidence interval coverage rates in Table~\ref{tab:SimResultsCoCAViaR6p} show that asymptotic variance-covariance estimation is a very difficult task for (Co)CAViaR models.
The empirical standard errors are somewhat overestimated for the VaR parameters in Table \ref{tab:SimResultsCoCAViaR6p}, whereas they are interestingly more accurate for the CoVaR parameters.
While the confidence intervals for the VaR parameters are rather conservative, the ones for CoVaR display some undercoverage for the extreme probability levels of $\alpha =  \beta = 0.95$, and exhibit almost correct coverage for $\alpha =  \beta = 0.9$.
Appendix~\ref{sec:SimBandwidthChoice} illustrates that simple adjustments of the bandwidth choices do not result in meaningful improvements, which demonstrates the need for future research on improving the estimation accuracy of the asymptotic variance-covariance matrix for dynamic (Co)CAViaR models.
We mention however that for us, the main purpose of these models lies in prediction (see the next section), where inference is of lesser importance.

\section{Applied Systemic Risk Forecasting for US G-SIBs}
\label{sec:EmpiricalApplication}

For the empirical forecasting application, we use daily close-to-close log-losses from January 4, 2000 until December 31, 2021 with a total of $n=5535$ trading days, obtained from the financial data provider Refinitiv.
We use data for the S\&P~500 index as our $Y_t$, and for $X_t$, we use Bank of America (BAC), Citigroup (C), Goldman Sachs (GS), JPMorgan Chase (JPM) and the S\&P~500 Financials (SPF) that represents the financial sector of the S\&P~500.  
The individual financial institutions are the four systemically most risky US banks according to the \cite{FSB22}. 
The information set of interest is $\mathcal{F}_{t-1}=\sigma(X_{t-1}, Y_{t-1}, \ldots, X_1, Y_1)$ (such that $\init$ is the empty set).
Then, the quantity to be forecasted, i.e., $\CoVaR_{t,\alpha\mid\beta}$, measures the spillover risk of the financial system/institution to the overall economy conditional on the current state of the market (embodied by $\mathcal{F}_{t-1}$).
We focus on the probability levels $\alpha = \beta = 0.95$ and estimate all models using a rolling window with estimation samples of length 3000 days. To reduce the computational burden, we only re-estimate all entertained models (introduced below) every 100 days.

In Section~\ref{sec:ForecastingApplication1}, we present predictions and parameter estimates for six different CoCAViaR models, while Section~\ref{sec:ForecastingApplication2} compares these specifications with six distinct DCC--GARCH CoVaR forecasts.

\subsection{CoVaR Forecasting with CoCAViaR Models}
\label{sec:ForecastingApplication1}

We consider six forecasting models from the CoCAViaR class in this section.
The first three candidate models are from the general class of SAV CoCAViaR models given in \eqref{eqn:SAVCoCAViaRModelClass}. The acronym \textit{SAV} indicates that the driving forces of the models are the absolute values of the log-losses, $|X_{t-1}|$ and $|Y_{t-1}|$.
The top panel of Table \ref{tab:ModelsApplication2} summarizes which covariates are included in each of the employed SAV CoCAViaR model specifications.
The suffix ``diag'' in the first model indicates that the off-diagonal elements of the parameter matrices $\mA$ and $\mB$ are set to zero.
Similarly, ``full'' indicates that the full specification of  \eqref{eqn:SAVCoCAViaRModelClass} is used (only restricting $B_{12} = 0$ to facilitate two-step M-estimation) and the suffix ``fullA'' indicates that the full matrix $\mA$ is considered, while $\mB$ is restricted to be a diagonal matrix.

Notice that the SAV-full model is not covered by our modeling framework \eqref{eqn:GeneralModel}, because the CoVaR model depends on lags of $v_t(\vtheta^v)$, such that $c_t(\vtheta^c,\vtheta^v)$ is a function of the VaR parameters $\vtheta^v$ as well. Nonetheless, we include this specification here to show that the gains in forecast accuracy of this extension may be non-existent in practice.

\begin{table}[tb]
	\centering
	\small
	\begin{tabular}{ll c ccc c ccc c cc}
		\toprule
		 & & & \multicolumn{10}{c}{Covariates} \\
		 \cmidrule{4-13}  
		 \multicolumn{2}{l}{CoCAViaR Model}  &   & $|X_{t-1}|$ & $X_{t-1}^+$ & $X_{t-1}^-$ && $|Y_{t-1}|$ & $Y_{t-1}^+$ & $Y_{t-1}^-$ & & $v_{t-1}(\cdot)$ & $ c_{t-1}(\cdot)$\\
		\midrule
		\multirow{2}{*}{SAV-diag} & $v_{t}(\cdot)$  & & $\bullet$& & & & &&& & $\bullet$  & \\
		& $c_{t}(\cdot)$  &  && & & &$\bullet$  &&& && $\bullet$  \\
		\midrule
		\multirow{2}{*}{SAV-fullA} & $v_{t}(\cdot)$  & & $\bullet$& & & & $\bullet$ &&& & $\bullet$  & \\
		& $c_{t}(\cdot)$  &  &$\bullet$& & & &$\bullet$  &&& && $\bullet$  \\
		\midrule
		\multirow{2}{*}{SAV-full} & $v_{t}(\cdot)$  & & $\bullet$& & & & $\bullet$ &&& & $\bullet$  & \\
		& $c_{t}(\cdot)$  &  &$\bullet$& & & &$\bullet$  &&& & $\bullet$ & $\bullet$  \\
		\midrule
		\midrule
		\multirow{2}{*}{AS-pos} & $v_{t}(\cdot)$  & & & $\bullet$&  &  &&$\bullet$ &&& $\bullet$  & \\
		& $c_{t}(\cdot)$  &  & & $\bullet$ & &  && $\bullet$& &&  & $\bullet$  \\
		\midrule
		\multirow{2}{*}{AS-signs} & $v_{t}(\cdot)$  & & & $\bullet$& $\bullet$ &  &&&&& $\bullet$  & \\
		& $c_{t}(\cdot)$  &  & & $\bullet$ & $\bullet$ &  && $\bullet$&$\bullet$&&  & $\bullet$  \\
		\midrule
		\multirow{2}{*}{AS-mixed} & $v_{t}(\cdot)$  & & & $\bullet$& $\bullet$ &  &$\bullet$ &&&& $\bullet$  & \\
		& $c_{t}(\cdot)$  &  & $\bullet$ & & &  && $\bullet$&$\bullet$&&  & $\bullet$  \\
		\bottomrule
	\end{tabular}
	\caption{In this table the symbol $\bullet$ indicates which covariates are used in the six CoCAViaR specifications employed in the empirical application. All models additionally contain an intercept.}
	\label{tab:ModelsApplication2}
\end{table}

Along the lines of \cite{EM04}, we extend the CoCAViaR model class by using signed values of $X_t$ and $Y_t$ to the \emph{Asymmetric Slope (AS) CoCAViaR} models,
\begin{align}
	\label{eqn:ASCoCAViaRModelClass}
	\begin{pmatrix} v_t(\vtheta^v) \\ c_t(\vtheta^c) \end{pmatrix}
	= \vomega + 
	\mA^+ \begin{pmatrix} X_{t-1}^+ \\ Y_{t-1}^+ \end{pmatrix} +
	\mA^- \begin{pmatrix} X_{t-1}^- \\ Y_{t-1}^- \end{pmatrix} +
	\mB \begin{pmatrix} v_{t-1}(\vtheta^v) \\ c_{t-1}(\vtheta^c) \end{pmatrix},
\end{align}
where $\vomega \in \mathbb{R}^2, \mA^+, \mA^-, \mB \in \mathbb{R}^{2 \times 2}$ are collected in the parameter vectors $\vtheta^v$ and $\vtheta^c$. 
Here, we define $x^+ = \max(x,0)$ and $x^- = -\min(x,0)$ for $x \in \mathbb{R}$.
Parameter equalities in $\mA^+$ and $\mA^-$ can be used to generate absolute values $|X_{t-1}|$ and $|Y_{t-1}|$ in \eqref{eqn:ASCoCAViaRModelClass}. 
Intuitively, the positive values of $X_{t-1}$ and $Y_{t-1}$ (i.e., financial losses in our orientation) are expected to contribute more to the future VaR and CoVaR than their negative values.
This is much like large losses often have more predictive content for volatility than equally large gains in the GJR--GARCH models of \citet{GJR93}.
The three model suffixes ``pos'', ``signs'' and ``mixed'' in the bottom panel of Table~\ref{tab:ModelsApplication2} imply that for ``pos'' only the positive components are included, for ``signs'' both positive and negative components are used, and ``mixed'' includes a mix of positive, negative and absolute value losses (see Table~\ref{tab:ModelsApplication2} for details).

\begin{table}[tb]
	\centering
	\scriptsize
	\begin{tabular}{ll c cccc c ccc c cc}
		\toprule
		& & & \multicolumn{11}{c}{Estimated Model Parameters} \\
		\cmidrule{4-14}  
		\multicolumn{2}{l}{CoCAViaR Model}  &   & 1 & $|X_{t-1}|$ & $X_{t-1}^+$ & $X_{t-1}^-$ && $|Y_{t-1}|$ & $Y_{t-1}^+$ & $Y_{t-1}^-$ & & $v_{t-1}(\cdot)$ & $ c_{t-1}(\cdot)$\\
		\midrule
		\multirow{4}{*}{SAV-diag} & \multirow{2}{*}{$v_t(\cdot)$} & & 0.034  & 0.121 & & & & &&& & 0.932  & \\
		&   & & (0.276)  & (0.074) & & & & &&& & (0.101)  & \\
		\cmidrule{2-14}  
		& \multirow{2}{*}{$c_t(\cdot)$}  &  & 0.060 && & & & 0.750 &&& && 0.834  \\
		&  &  & (0.826) && & & & (0.620)  &&& && (0.161)  \\
		\midrule
		\multirow{4}{*}{SAV-fullA} & \multirow{2}{*}{$v_t(\cdot)$} & & 0.024  & 0.114 & & & & 0.102 &&& & 0.915  & \\
		&   & & (0.291)  & (0.085) & & & & (0.172) &&& & (0.104)  & \\
		\cmidrule{2-14}  
		& \multirow{2}{*}{$c_t(\cdot)$}  &  & 0.202 & 0.112 & & & & 0.428 &&& && 0.798  \\
		&  &  & (1.045) &(0.237)& & & & (0.636)  &&& && (0.332)  \\
		\midrule
		\multirow{4}{*}{SAV-full} & \multirow{2}{*}{$v_t(\cdot)$} & & 0.024  & 0.114 & & & & 0.102 &&& & 0.915  & \\
		&   & & (0.291)  & (0.085) & & & & (0.172) &&& & (0.104)  & \\
		\cmidrule{2-14}  
		& \multirow{2}{*}{$c_t(\cdot)$}  &  & 1.486 & 0.151 & & & & $-0.433$ &&& & $1.506$ & $-0.882$  \\
		&  &  & (0.600) &(0.266)& & & & (0.206)  &&& &(0.427) & (0.170)  \\
		\midrule
		\midrule
		\multirow{4}{*}{AS-pos} & \multirow{2}{*}{$v_t(\cdot)$} & & 0.013  &  & 0.128 & & &  & 0.065 && & 0.958  & \\
		&   && (0.150)  &  & (0.074) & & &  & (0.185) && & (0.053)  & \\
		\cmidrule{2-14}  
		& \multirow{2}{*}{$c_t(\cdot)$}  &  & 0.087 &  & 0.073 & & & & 0.627 && && 0.892  \\
		&  &  & (0.361) && (0.379) & & & & (0.482)  && && (0.104)  \\
		\midrule
		\multirow{4}{*}{AS-signs} & \multirow{2}{*}{$v_t(\cdot)$} & & 0.046  &  & 0.198 & 0.088 & &  & && & 0.918  & \\
		&   && (0.320)  &  & (0.111) & (0.227) & &  & && & (0.118)  & \\
		\cmidrule{2-14}  
		& \multirow{2}{*}{$c_t(\cdot)$}  &  & 0.189 &  & $-0.003$ & 0.116 & & & 0.550 & $-0.255$ & && 0.893  \\
		&  &  & (0.861) && (0.438) & (0.102) & & & (0.152)  & (0.611)& && (0.118)  \\
		\midrule
		\multirow{4}{*}{AS-mixed} & \multirow{2}{*}{$v_t(\cdot)$}   & & 0.040 & &  0.187 & 0.089  &  & 0.038 &&&& 0.914  & \\
		&  & & (0.288)& &  (0.122) & (0.127)  &  & (0.213) &&&& (0.116)  & \\
		\cmidrule{2-14}  
		& \multirow{2}{*}{$c_t(\cdot)$}   & &   0.208 & 0.093 & & &  && 0.567  & $-0.248$ &&  & 0.868  \\
		&  & &   (0.446) & (0.113) & & &  && (0.514) & (0.534) &&  & (0.087)  \\
		\bottomrule
	\end{tabular}
	\caption{CoCAViaR model parameter estimates based on the first estimation window consisting of 3000 trading days starting on January 4, 2000 until December 5, 2011. Estimated standard errors are given in parentheses  below the estimates.}
	\label{tab:EstimatesModelsApplication}
\end{table}

Table~\ref{tab:EstimatesModelsApplication} reports the estimated parameters together with their standard errors for the six CoCAViaR model specifications. The results are based on the initial estimation window of 3000 observations. 
All CoCAViaR models in this section are estimated using the initialization with the model parameters $\vomega$ as described in item \eqref{item:InitParameter} of Example~\ref{ex:1}.
Except for the SAV-full model, the autoregressive coefficients are all between 0.798 and 0.958, and are highly significant. The other coefficients are all barely significant, but their direction is very reasonable. The ``cross-terms'' seem to be less important in general. 
In the SAV-full model, including the lagged $v_{t-1}(\cdot)$ and $c_{t-1}(\cdot)$ in the CoVaR model results in unintuitive parameter estimates, possibly resulting from a multicollinearity problem.
This is also reflected by the model's inferior forecasting performance; see Table~\ref{tab:ApplForecastingResults} below. 
Hence, the additional generality of the SAV-full model (where $c_t(\vtheta^c,\vtheta^v)$ also depends on $\vtheta^v$) does not seem to be relevant empirically. 
Consistent with the idea of \citet{GJR93} that losses have a larger impact on volatility than do equally large gains, we find that losses are more important than gains in predicting systemic risk in the AS models; see especially $Y_{t-1}^+$ and $Y_{t-1}^-$ in the CoVaR equation of the AS-mixed model.

\begin{figure}[tb]
	\centering
	\includegraphics[width=\linewidth]{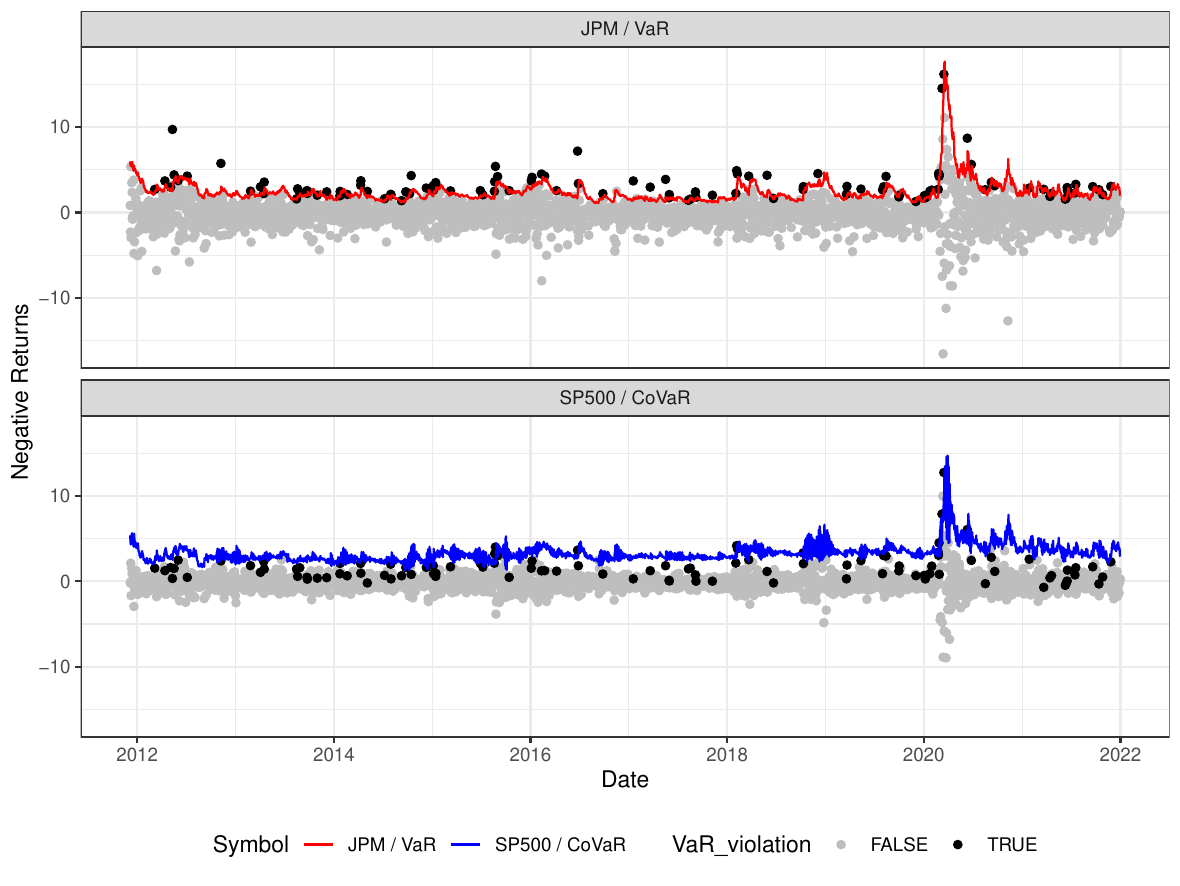}
	\caption{Out-of-sample VaR and CoVaR forecasts from the CoCAViaR-SAV-fullA model where JPMorgan Chase's log-losses are used for $X_t$ and the S\&P\,500 for $Y_t$. Log-losses on days with a VaR exceedance of JPMorgan Chase are displayed in black in both panels.}
	\label{fig:ForecastsPlot}
\end{figure}

Figure \ref{fig:ForecastsPlot} illustrates the rolling window forecasts from the SAV-fullA CoCAViaR model for the evaluation window ranging from December 6, 2011 until December 31, 2021.
The upper panel shows the log-losses of JPMorgan Chase---the systemically most important bank according to the \cite{FSB22}---together with its VaR forecasts. 
Losses exceeding the VaR forecasts are highlighted in black, which correspond to days with an (out-of-sample) stress event $\big\{X_t \ge \widehat \VaR_{t,\beta}\big\}$ in the definition of the CoVaR in Section \ref{sec:CoVaR}.
The lower panel shows the log-losses of the S\&P\,500 together with the model's CoVaR forecasts. There, the log-losses of days with a VaR exceedance (of JPMorgan Chase) are displayed in black, such that the CoVaR forecasts can be interpreted as $\alpha = 95\%$ quantile forecasts among those days with VaR exceedances.
We defer a more formal investigation of the (VaR, CoVaR) forecasts to the following section.

\subsection{Comparison with DCC--GARCH Models}
\label{sec:ForecastingApplication2}

As competitors for our six CoCAViaR models, we use six different DCC--GARCH specifications. DCC--GARCH models have attained benchmark status, because of their accurate variance-covariance matrix predictions \citep{LRV12,Laurent2013,CM14}. Particularly, we use three DCC--GARCH specifications, containing two standard DCC--GARCH(1,1) models with multivariate Gaussian and $t$-distributed innovations, respectively, and a DCC specification based on a univariate GJR--GARCH(1,1) model.
The models are abbreviated as ``DCC-n'', ``DCC-t'' and ``DCC-gjr'', respectively.
We estimate all DCC--GARCH models by maximum likelihood using the \texttt{rmgarch} package of \cite{Galanos2022} for the statistical software \texttt{R}.

As discussed in Section~\ref{sec:CoVaRModels}, different ``square roots'' $\mSigma_t$ of the conditional covariance matrix $\mH_t=\mSigma_t\mSigma_t^\prime$ may lead to different CoVaR forecasts of DCC--GARCH models (see also Appendix~\ref{CoVaR forecasting with multivariate GARCH models}).
Here, we consider VaR and CoVaR forecasts that are obtained by combining the above three DCC model specifications with a symmetric and a Cholesky decomposition (suffix ``sym'' respectively ``Chol'' after the model abbreviation) of the forecasted variance-covariance matrices, yielding a total of six sets of forecasts. 
For instance, the forecasts abbreviated ``DCC-t-Chol'' are computed from a DCC--GARCH with $t$-innovations based on the Cholesky decomposition of $\mH_t$.
We obtain the final VaR and CoVaR forecasts by multiplying $\mSigma_t$ with a nonparametric estimate of the VaR and CoVaR of the model residuals.\footnote{The forecast evaluation results in Table~\ref{tab:ApplForecastingResults} differ between the DCC-n-sym and the DCC-n-Chol models (despite the normal distribution being spherical; see Section~\ref{sec:CoVaRModels} and Appendix~\ref{Computation of Risk Measure Forecasts}) as the empirical distribution, which is used in the nonparametric estimate, is in general not spherical.}
We refer to Appendix~\ref{Computation of Risk Measure Forecasts} for details on the computation of VaR and CoVaR forecasts from DCC--GARCH models.

\begin{figure}[tb]
	\begin{center}
		\begin{subfigure}{0.49\linewidth}
			\centering
			\caption{CoCAViaR-SAV-fullA model}
			\includegraphics[width=\linewidth]{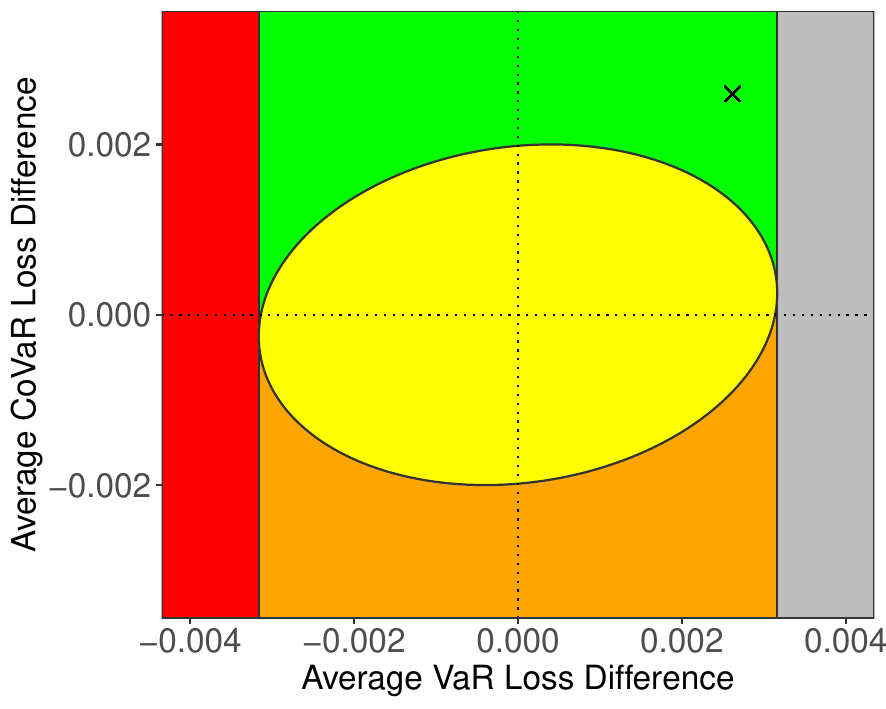}
			\label{fig:CoVaR_TrafficLight_XLF}
		\end{subfigure}
		\begin{subfigure}{0.49\linewidth}
			\centering
			\caption{CoCAViaR-AS-mixed model}
			\includegraphics[width=\linewidth]{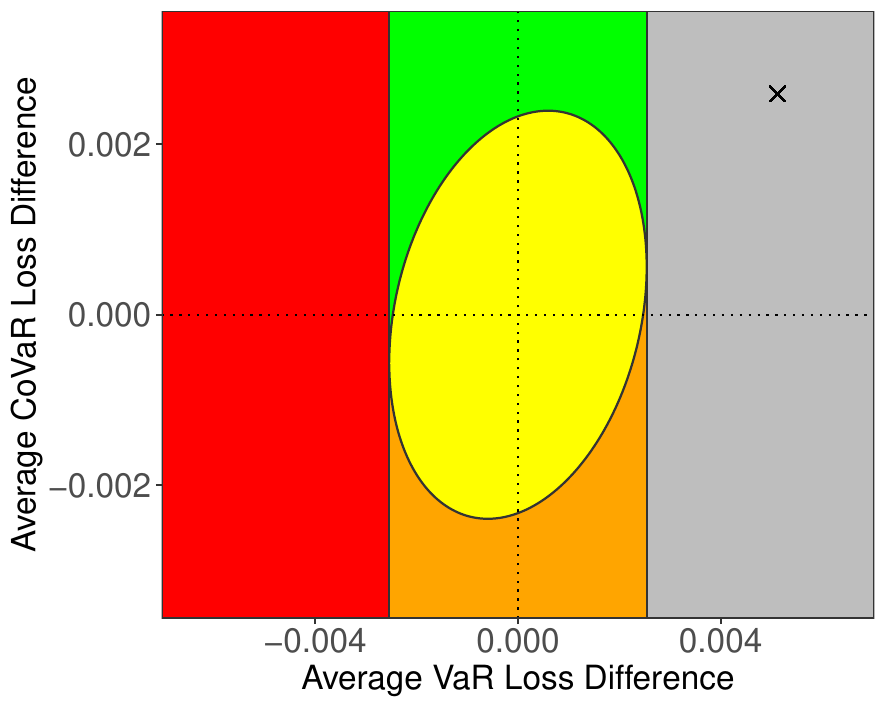}
			\label{fig:CoVaR_TrafficLight_JPM}
		\end{subfigure}
		\caption{This figure graphically illustrates the one and a half-sided forecast comparison tests for the VaR and CoVaR of \cite{FH24} based on their multi-objective scoring function in \eqref{eq:loss} for a significance level of $10\%$.   
		We use log-losses of JPMorgan Chase for $X_t$ and of the S\&P~500 for $Y_t$.
		The respective CoCAViaR models are given in the captions of the two plots and are compared against the baseline ``DCC-t-Chol'' model.}
		\label{fig:CoVaR_TrafficLight}
	\end{center}
\end{figure}

The multi-objective elicitability of (VaR, CoVaR) complicates inference on the predictive ability of the models, since the scoring function in \eqref{eq:loss} is \textit{bivariate} such that standard \citet{DM95} tests---based on \textit{scalar} scores---cannot be applied.
We follow \citet{FH24} in applying their ``one and a half-sided'' tests, which they illustrate via an extended traffic light system. 
Figure~\ref{fig:CoVaR_TrafficLight} exemplarily displays their extended traffic light approach for two models.
Doing so requires us to fix a baseline model, and we arbitrarily choose the ``DCC-t-Chol'' model for all evaluations.
Such a baseline model is necessary as classical extensions to \emph{multi-model} forecast comparison methods---such as the model confidence set of \cite{Hansen2011}---are not available in the case of bivariate (multi-objective) scoring functions.

Multi-objective elicitability implies that the scores of two competing sequences of CoVaR forecasts can only be compared if their underlying VaR forecasts perform equally well.
To obtain a reasonable finite-sample counterpart, \citet{FH24} interpret this as an insignificant score difference in the standard \cite{DM95} test for the VaR forecasts; that is, the null that the expected VaR score differences (calculated based on the first component in \eqref{eq:loss}) are equal to zero cannot be rejected.

The red zone in Figure \ref{fig:CoVaR_TrafficLight} indicates that the baseline VaR is significantly superior, and the comparison model is rejected without consideration of its CoVaR forecasts. 
The grey zone indicates the reverse, while the three remaining zones in the intermediary corridor imply insignificant score differences of the VaR forecasts.
Here, the orange zone implies that the baseline CoVaR is significantly superior (i.e., the CoVaR score differences based on the second component of \eqref{eq:loss} are smaller than zero in expectation), the green zone that the alternative model is superior and the yellow zone represents insignificant score differences.
As the baseline is the ``DCC-t-Chol'' model, a comparison with our CoCAViaR models should ideally yield results in the green zone (which indicates superior CoVaR forecasts and comparable VaR forecasts) or in the grey zone (which implies superior VaR forecasts) for the new models to have merit in practice.

\begin{table}[p!]
	\centering
	\scriptsize
	\resizebox{0.85\columnwidth}{!}{
	\begin{tabular}{lllrrrlrrrllr}
		\toprule
		&&& \multicolumn{3}{c}{VaR} && \multicolumn{3}{c}{CoVaR} &&\multicolumn{2}{c}{Inference} \\
		\cmidrule{4-6}  	\cmidrule{8-10} 	 \cmidrule{12-13} 
		$X_t$ & model &  & score & rank & hits &  & score & rank & hits &  & zone & $p$-value \\ 
		\midrule
		\multirow{12}{*}{BAC} & CoCAViaR-SAV-fullA &  & 2.074 & 6 & 4.3 &  & 5.405 & 1 & 7.3 &  & green & 0.01 \\ 
		& CoCAViaR-AS-mixed &  & 2.044 & 4 & 4.2 &  & 6.176 & 2 & 7.5 &  & grey & 0.00 \\ 
		& CoCAViaR-AS-signs &  & 2.042 & 2 & 4.6 &  & 6.253 & 3 & 5.2 &  & grey & 0.00 \\ 
		& CoCAViaR-AS-pos &  & 2.060 & 5 & 4.7 &  & 6.899 & 4 & 7.5 &  & green & 0.05 \\ 
		& CoCAViaR-SAV-diag &  & 2.089 & 12 & 4.5 &  & 7.071 & 5 & 8.7 &  & yellow & 0.14 \\ 
		& CoCAViaR-SAV-full &  & 2.074 & 6 & 4.3 &  & 7.472 & 6 & 8.3 &  & yellow & 0.35 \\ 
		& DCC-n-Chol &  & 2.082 & 9 & 4.5 &  & 8.655 & 7 & 19.5 &  & yellow & 0.28 \\ 
		& DCC-t-Chol &  & 2.087 & 10 & 4.8 &  & 8.815 & 8 & 17.2 &  &  &  \\ 
		& DCC-gjr-t-Chol &  & 2.042 & 3 & 4.7 &  & 9.165 & 9 & 15.8 &  & grey & 0.00 \\ 
		& DCC-n-sym &  & 2.081 & 8 & 4.5 &  & 9.213 & 10 & 19.5 &  & yellow & 0.33 \\ 
		& DCC-gjr-t-sym &  & 2.040 & 1 & 4.8 &  & 9.447 & 11 & 16.5 &  & grey & 0.00 \\ 
		& DCC-t-sym &  & 2.087 & 11 & 4.9 &  & 9.595 & 12 & 17.7 &  & yellow & 0.83 \\ 
		\midrule
		\multirow{12}{*}{C} & CoCAViaR-SAV-full &  & 2.143 & 10 & 4.8 &  & 6.230 & 1 & 5.0 &  & yellow & 0.13 \\ 
		& CoCAViaR-SAV-fullA &  & 2.143 & 10 & 4.8 &  & 6.286 & 2 & 4.1 &  & yellow & 0.22 \\ 
		& CoCAViaR-AS-mixed &  & 2.109 & 4 & 4.9 &  & 6.626 & 3 & 5.6 &  & green & 0.05 \\ 
		& CoCAViaR-AS-signs &  & 2.098 & 1 & 5.3 &  & 6.653 & 4 & 5.2 &  & grey & 0.00 \\ 
		& CoCAViaR-SAV-diag &  & 2.142 & 7 & 4.9 &  & 6.665 & 5 & 4.1 &  & yellow & 0.35 \\ 
		& CoCAViaR-AS-pos &  & 2.143 & 11 & 5.8 &  & 7.366 & 6 & 4.8 &  & yellow & 0.64 \\ 
		& DCC-n-sym &  & 2.140 & 6 & 4.9 &  & 7.632 & 7 & 9.6 &  & yellow & 0.29 \\ 
		& DCC-n-Chol &  & 2.138 & 5 & 5.0 &  & 7.793 & 8 & 11.0 &  & green & 0.05 \\ 
		& DCC-t-sym &  & 2.146 & 12 & 5.1 &  & 7.963 & 9 & 10.9 &  & yellow & 0.45 \\ 
		& DCC-t-Chol &  & 2.143 & 8 & 5.1 &  & 7.981 & 10 & 11.6 &  &  &  \\ 
		& DCC-gjr-t-Chol &  & 2.105 & 2 & 5.0 &  & 8.194 & 11 & 11.0 &  & grey & 0.04 \\ 
		& DCC-gjr-t-sym &  & 2.107 & 3 & 5.0 &  & 8.422 & 12 & 8.6 &  & grey & 0.03 \\ 
		\midrule
 		\multirow{12}{*}{GS} & CoCAViaR-SAV-fullA &  & 1.852 & 4 & 4.7 &  & 5.857 & 1 & 9.2 &  & green & 0.04 \\ 
		& CoCAViaR-AS-mixed &  & 1.832 & 2 & 4.7 &  & 5.971 & 2 & 4.2 &  & grey & 0.02 \\ 
		& CoCAViaR-AS-pos &  & 1.841 & 3 & 4.3 &  & 5.980 & 3 & 5.5 &  & grey & 0.02 \\ 
		& CoCAViaR-SAV-full &  & 1.852 & 4 & 4.7 &  & 6.355 & 4 & 10.1 &  & green & 0.07 \\ 
		& CoCAViaR-AS-signs &  & 1.832 & 1 & 5.0 &  & 6.453 & 5 & 4.0 &  & grey & 0.00 \\ 
		& CoCAViaR-SAV-diag &  & 1.870 & 8 & 4.9 &  & 7.382 & 6 & 6.5 &  & grey & 0.01 \\ 
		& DCC-n-sym &  & 1.891 & 9 & 4.3 &  & 7.607 & 7 & 10.0 &  & yellow & 0.16 \\ 
		& DCC-n-Chol &  & 1.895 & 11 & 4.3 &  & 7.696 & 8 & 13.0 &  & green & 0.00 \\ 
		& DCC-t-sym &  & 1.892 & 10 & 4.6 &  & 8.144 & 9 & 9.4 &  & yellow & 0.14 \\ 
		& DCC-t-Chol &  & 1.896 & 12 & 4.6 &  & 8.245 & 10 & 12.0 &  &  &  \\ 
		& DCC-gjr-t-Chol &  & 1.862 & 6 & 4.2 &  & 8.721 & 11 & 13.1 &  & grey & 0.01 \\ 
		& DCC-gjr-t-sym &  & 1.863 & 7 & 4.1 &  & 8.803 & 12 & 11.7 &  & grey & 0.02 \\ 
		\midrule
 		\multirow{12}{*}{JPM} & CoCAViaR-SAV-fullA &  & 1.715 & 6 & 4.1 &  & 5.697 & 1 & 12.5 &  & green & 0.01 \\ 
		& CoCAViaR-AS-mixed &  & 1.690 & 1 & 4.0 &  & 5.701 & 2 & 6.9 &  & grey & 0.00 \\ 
		& CoCAViaR-SAV-full &  & 1.715 & 6 & 4.1 &  & 5.926 & 3 & 10.6 &  & green & 0.02 \\ 
		& CoCAViaR-AS-signs &  & 1.692 & 2 & 4.1 &  & 6.004 & 4 & 6.7 &  & grey & 0.00 \\ 
		& CoCAViaR-AS-pos &  & 1.703 & 3 & 4.7 &  & 6.724 & 5 & 8.3 &  & grey & 0.03 \\ 
		& CoCAViaR-SAV-diag &  & 1.733 & 8 & 4.3 &  & 6.898 & 6 & 8.2 &  & yellow & 0.18 \\ 
		& DCC-n-sym &  & 1.743 & 11 & 3.7 &  & 7.196 & 7 & 17.2 &  & green & 0.03 \\ 
		& DCC-t-sym &  & 1.752 & 12 & 3.8 &  & 7.695 & 8 & 20.8 &  & red & 0.00 \\ 
		& DCC-n-Chol &  & 1.734 & 9 & 4.0 &  & 7.826 & 9 & 18.8 &  & grey & 0.01 \\ 
		& DCC-gjr-t-sym &  & 1.711 & 5 & 3.8 &  & 8.136 & 10 & 17.7 &  & grey & 0.06 \\ 
		& DCC-t-Chol &  & 1.741 & 10 & 4.1 &  & 8.290 & 11 & 21.4 &  &  &  \\ 
		& DCC-gjr-t-Chol &  & 1.706 & 4 & 3.9 &  & 8.762 & 12 & 19.2 &  & grey & 0.03 \\ 
		\midrule
 		\multirow{12}{*}{SPF} & CoCAViaR-SAV-fullA &  & 1.409 & 12 & 4.2 &  & 4.910 & 1 & 9.3 &  & green & 0.01 \\ 
		& CoCAViaR-AS-mixed &  & 1.379 & 4 & 4.5 &  & 5.025 & 2 & 5.3 &  & grey & 0.00 \\ 
		& CoCAViaR-AS-signs &  & 1.363 & 1 & 4.7 &  & 5.519 & 3 & 6.8 &  & grey & 0.00 \\ 
		& CoCAViaR-AS-pos &  & 1.398 & 5 & 4.7 &  & 5.725 & 4 & 5.9 &  & green & 0.05 \\ 
		& CoCAViaR-SAV-diag &  & 1.407 & 10 & 4.1 &  & 5.910 & 5 & 10.5 &  & green & 0.04 \\ 
		& DCC-n-sym &  & 1.400 & 7 & 4.5 &  & 7.306 & 6 & 16.7 &  & grey & 0.00 \\ 
		& DCC-gjr-t-sym &  & 1.369 & 3 & 4.3 &  & 7.621 & 7 & 14.5 &  & grey & 0.00 \\ 
		& DCC-t-sym &  & 1.406 & 9 & 4.5 &  & 7.829 & 8 & 18.3 &  & yellow & 0.25 \\ 
		& DCC-n-Chol &  & 1.399 & 6 & 4.4 &  & 7.884 & 9 & 20.5 &  & grey & 0.00 \\ 
		& DCC-t-Chol &  & 1.405 & 8 & 4.5 &  & 8.175 & 10 & 21.1 &  &  &  \\ 
		& DCC-gjr-t-Chol &  & 1.365 & 2 & 4.8 &  & 8.388 & 11 & 17.4 &  & grey & 0.00 \\ 
		& CoCAViaR-SAV-full &  & 1.409 & 12 & 4.2 &  & 8.832 & 12 & 10.3 &  & yellow & 0.68 \\ 
		\bottomrule
	\end{tabular}
}
	\caption{VaR and CoVAR forecasting results for $Y_t$ equaling S\&P~500 losses and various choices of $X_t$.
		Details on the table columns are given in the main text.}
	\label{tab:ApplForecastingResults}
\end{table}

Table~\ref{tab:ApplForecastingResults} presents results on the forecasting performance of all 12 models.
For the VaR, we report the average VaR score (multiplied by 10 for better readability) using the first component of \eqref{eq:loss}, its rank among the different models, and the ``hits'' as the percentage of days where the loss is larger than the VaR forecast, $X_t \ge \widehat \VaR_{t,\beta}$.
For the CoVaR, we report the average CoVaR score using the second component in \eqref{eq:loss} (multiplied by 1000 for better readability), the corresponding rank, and the CoVaR hits defined as the percentage of days where $Y_t \ge \widehat \CoVaR_{t,\alpha\mid\beta}$ among all days with a VaR hit, i.e., the $t$ where $X_t \ge \widehat \VaR_{t,\beta}$.
The VaR forecast should ideally be exceeded with probability $1-\beta=5\%$, and---on those days---the CoVaR forecast should be exceeded with probability $1-\alpha = 5\%$.
Hence, for reasonable VaR and CoVaR forecasts we expect the numbers in the two ``hits'' columns of Table~\ref{tab:ApplForecastingResults} to be close to 5.
The last two columns of Table~\ref{tab:ApplForecastingResults} report results for the previously described one and a half-sided test of \cite{FH24}. There we use ``DCC-t-Chol'' as the baseline model, and report the test's $p$-value together with the resulting zone of their extended traffic light system.
For each considered asset $X_t$, we sort the table rows (i.e., the models) according to their CoVaR score, as the CoVaR forecasts are of main interest in this section.

Among the CoCAViaR models, we find a better performance of the AS than the SAV models for the VaR, but a relatively comparable performance for CoVaR forecasting.
A reason for this may be that the additional VaR parameters in the AS models are estimated with more precision and, hence, the predictive content of the positive/negative parts emerges more clearly than for CoVaR, where the effective sample size is much reduced.
It is further noteworthy that the SAV-full CoCAViaR model, which includes the lagged VaR in the CoVaR equation, performs below average.
This might be caused by a high collinearity of the VaR and CoVaR, and also shows that the restriction $B_{12}=0$---which we have to impose for our two-step M-estimator---is likely to be unproblematic in practice.

Overall, we find a superior forecasting performance of the CoCAViaR models for all five employed assets for $X_t$ compared to the DCC models.
While the rankings of the VaR scores vary over the different assets, the superior forecasting performance of the CoCAViaR models is more substantial for the CoVaR.
This is supported by the CoVaR hits (corresponding to unconditional forecast calibration) being much closer to the nominal level of $5\%$ than for the DCC models, whose hit frequencies are almost all above $10\%$.
While none of the CoCAViaR models are significantly outperformed by being in the red or orange zone, many significantly outperform the baseline DCC model and are located in the green and grey zones. Also notice that among the CoCAViaR models, the SAV-full specification (not covered by our theoretical framework) performs below average. 
Therefore, our modeling framework \eqref{eqn:GeneralModel} seems to be sufficiently general to capture the main features of (systemic) risk evolution through time.

In comparing the (VaR, CoVaR) forecasts of our CoCAViaR models with those issued by the DCC--GARCH models, we have used the same scoring function \eqref{eq:loss} in the comparative backtest as we did for estimating the CoCAViaR models. 
Therefore, one may be concerned that using the same scoring function to estimate our CoCAViaR models and to evaluate the predictions, lends an unfair advantage in the forecast comparison to the CoCAViaR models. 
To alleviate such concerns, we also use a scoring function different from \eqref{eq:loss} in the comparative backtest of \citet{FH24} as a robustness check.
For brevity, the results are reported in Appendix~\ref{Forecast Comparison with Alternative Loss Function} and they show that the dominance of our CoCAViaR models still holds.

The superiority of our CoCAViaR models may be surprising, because the second-stage estimator is based on the reduced estimation sample where $\{X_t \geq \widehat{\VaR}_{t,\b} \}$.
Yet, doing so allows our models to focus on the salient features of the data in the tail, thus offsetting the drawback of the lower effective sample size.
This focus means, however, that we only model a very narrow---although for many purposes important---aspect of the conditional distribution of $(X_t,Y_t)^\prime\mid\mathcal{F}_{t-1}$ (namely the VaR and CoVaR). 
In contrast, multivariate GARCH models are capable of modeling the complete predictive distribution of $(X_t,Y_t)^\prime\mid\mathcal{F}_{t-1}$, but they may not excel at modeling each aspect of it equally well. 
This echoes \citet[p.~374]{EM04}, who conclude in the context of univariate VaR forecasting that ``[a]lthough GARCH might be a useful model for describing the evolution of volatility, the results in this article show that it might provide an unsatisfactory approximation when applied to tail estimation.''

\section{Conclusion}\label{sec:Conclusion}	

Our first main contribution is to propose a flexible, semiparametric approach for modeling (VaR, CoVaR) over time.
Since we only model (VaR, CoVaR) and not the full predictive distribution, we ``let the tails speak for themselves'' \citep{DuM83}. As we find in an empirical application on the systemic riskiness of US financial institutions, this yields models that improve upon benchmark DCC--GARCH processes in terms of predictive accuracy. 
As a second main contribution, we propose a two-step M-estimator for the model parameters and derive its large sample properties.
From an econometric perspective, our proofs are non-standard since the loss function that has to be used for the CoVaR is not only non-smooth, but also discontinuous in the VaR model (parameter).


We expect our modeling framework to have applications beyond the one considered here, for instance in predictive CoVaR regressions as in \citet{AB16}.
Furthermore, in macroeconomics, \citet{ABG19,Aea21} have recently drawn attention to tail risks and their interconnections by popularizing the Growth-at-Risk.
Hence, the models proposed in this paper could become relevant for studying interconnections of macroeconomic risks.
Much like \citet{BS21} compared the predictive accuracy of different Growth-at-Risk models, our models may be used as competitors in \textit{Co}-Growth-at-Risk comparisons. 

%
%
%
%
%

\pagebreak 
\appendix
\renewcommand\appendixpagename{Appendix}
\appendixpage

\section*{Notation}

We use the following notational conventions throughout this appendix. 
We denote by $C>0$ a large positive constant that may change from line to line. 
If not specified otherwise, all convergences are to be understood with respect to $n\to\infty$. We also write $\E_{t-1}[\cdot]=\E[\cdot\mid\mathcal{F}_{t-1}]$ and $\P_{t-1}\{\cdot\}=\P\{\cdot\mid\mathcal{F}_{t-1}\}$ for short. 
We exploit without further mention that the Frobenius norm is submultiplicative, i.e., that $\Vert\mA\mB\Vert\leq\Vert\mA\Vert\cdot\Vert\mB\Vert$ for conformable matrices $\mA$ and $\mB$. 
For a real-valued, differentiable function $f(\cdot)$, we denote the $j$-th element of the gradient $\nabla f(\cdot)$ by $\nabla_j f(\cdot)$.

We continue the labeling of theorems, examples, assumptions etc.~from the main paper.
E.g., Theorem~\ref{thm:cons} is the last theorem of the main paper and Theorem~\ref{prop:StartingValues} is the first theorem in this appendix.
Equations in this appendix carry a prefix consistent with the section numbering, such that, say, \eqref{eq:un id VaR} is the first equation in Section~\ref{sec:thm1}.
References to equations without a prefix (such as \eqref{eqn:MestCoVaR}) always refer to a formula in the main paper.

\section{Proof of Theorem~\ref{thm:cons}}
\label{sec:thm1}

\begin{proof}[{\textbf{Proof of Theorem~\ref{thm:cons}:}}] 
	Similar to \citet[Proof of Theorem~1]{Wei91}, we show consistency of the parameter estimators by invoking the ``convergence in probability'' version of Lemma~2.2 in \citet{Whi80}; cf.~also \citet[Theorem~3.4]{Whi96}.
	First, we establish that $\widehat{\vtheta}_{n}^{v}\overset{\P}{\longrightarrow}\vtheta_0^{v}$. 
	For $\vtheta^v\in\mTheta^v$, we define
	\begin{align*}
		Q_n^{v}(\vtheta^v) &= \frac{1}{n}\sum_{t=1}^{n}S^{\VaR}\big(v_t(\vtheta^v), X_t\big),\\
		\overline{Q}_n^v(\vtheta^v) &= \E\bigg[\frac{1}{n}\sum_{t=1}^{n}S^{\VaR}\big(v_t(\vtheta^v), X_t\big)\bigg].
	\end{align*}
	Recall from Assumption~\ref{ass:cons}~\ref{it:compact} that $\mTheta^v$ is compact.
	Note that $Q_n^v(\cdot)$ is a continuous function on $\mTheta^v$, due to the continuity of $v\mapsto S^{\VaR}(v,x)$ for all $x\in\mathbb{R}$ and the continuity of $v_t(\cdot)$ (by Assumption~\ref{ass:cons}~\ref{it:smooth}).
	Observe moreover that $\widehat{\vtheta}_n^v=\argmin_{\vtheta^v\in\mTheta^v}Q_n^{v}(\vtheta^v)$.

	To apply the ``in probability'' version of Lemma~2.2 in \citet{Whi80}, we have to verify that, as $n\to\infty$,
	\[
	\sup_{\vtheta^v\in\mTheta^v}\big|Q_n^v(\vtheta^v) - \overline{Q}_n^v(\vtheta^v)\big|=o_{\P}(1).
	\]
	This, however, is immediate from Assumption~\ref{ass:cons}~\ref{it:ULLN}.
	
	It remains to check that $\vtheta_0^v$ is the identifiably unique minimum of $\overline{Q}_n^{v}(\cdot)$ in the sense of \citet[Definition~2.1]{Whi80}.
	To do so, note that $\overline{Q}_n^{v}(\cdot)$ is continuous on $\mTheta^v$, because for any $\vtheta^v\in\mTheta^v$,
	\begin{align*}
		\lim_{\vtheta\to\vtheta^v}\overline{Q}_n^v(\vtheta) &= \frac{1}{n}\sum_{t=1}^{n}\lim_{\vtheta\to\vtheta^v}\E\Big[S^{\VaR}\big(v_t(\vtheta), X_t\big)\Big]\\
		&= \frac{1}{n}\sum_{t=1}^{n}\E\Big[\lim_{\vtheta\to\vtheta^v}S^{\VaR}\big(v_t(\vtheta), X_t\big)\Big]\\
		&= \frac{1}{n}\sum_{t=1}^{n}\E\Big[S^{\VaR}\big(v_t(\vtheta^v), X_t\big)\Big]\\
		&= \overline{Q}_n^v(\vtheta^v),
	\end{align*}
	where the second step follows from the dominated convergence theorem (DCT), and the third step from the continuity of the map $v\mapsto S^{\VaR}(v,x)$ together with continuity of $v_t(\cdot)$ (by Assumption~\ref{ass:cons} \ref{it:smooth}). 
	Thus, $\vtheta^v\mapsto \overline{Q}_n^v(\vtheta^v)$ is continuous. 
	Note that we may apply the DCT, since $S^{\VaR}\big(v_t(\vtheta), X_t\big)$ is dominated by 
	\[
	\big|S^{\VaR}\big(v_t(\vtheta), X_t\big)\big|\leq\big|X_t - v_t(\vtheta)\big|\leq \big|X_t\big| + \big|v_t(\vtheta)\big|\leq |X_t| + V(\mathcal{F}_{t-1}),
	\]
	where the right-hand side is integrable due to Assumption~\ref{ass:cons}~\ref{it:mom bounds cons}. 
	
	Next, we check that $\vtheta_0^v$ minimizes $\overline{Q}_n^{v}(\cdot)$.
	Theorem~4.4 of \citet{FH24} (which applies thanks to Assumption~\ref{ass:cons}~\ref{it:cond dist}) implies that $\E_{t-1}\big[S^{\VaR}\big(\,\cdot\, , X_t\big)\big]$ is minimized at $\VaR_\beta(F_{X_t\mid\mathcal{F}_{t-1}})$, which equals $v_t(\vtheta_0^v)$ by Assumption~\ref{ass:cons}~\ref{it:id}. 
	In other words,
	\[
	\E_{t-1}\big[S^{\VaR}\big(v_t(\vtheta_0^v), X_t\big)\big] \leq \E_{t-1}\big[S^{\VaR}\big(v_t(\vtheta^v), X_t\big)\big]\qquad\text{for all }\vtheta^v\neq\vtheta_0^v.
	\]
	Applying expectations to both sides gives that
	\[
	\E\Big\{\E_{t-1}\big[S^{\VaR}\big(v_t(\vtheta_0^v), X_t\big)\big]\big\} \leq \E\Big\{\E_{t-1}\big[S^{\VaR}\big(v_t(\vtheta^v), X_t\big)\big]\Big\}\qquad\text{for all }\vtheta^v\neq\vtheta_0^v,
	\]
	which, by the law of iterated expectations (LIE), is equivalent to
	\[
	\E\big[S^{\VaR}\big(v_t(\vtheta_0^v), X_t\big)\big] \leq \E\big[S^{\VaR}\big(v_t(\vtheta^v), X_t\big)\big]\qquad\text{for all }\vtheta^v\neq\vtheta_0^v.
	\]
	Therefore, $\vtheta_0^v$ minimizes $\overline{Q}_n^{v}(\cdot)$, as desired.
	
	To show that $\vtheta_0^v$ is also the \textit{identifiably unique} minimizer in the sense of \citet[Definition~2.1]{Whi80}, we proceed as follows.
	In Section~\ref{Proofs of Equations}, we show under Assumption~\ref{ass:cons} that for all $\vtheta^v\in\mTheta^v$,
	\begin{equation}\label{eq:un id VaR}
		\E\Big[S^{\VaR}\big(v_t(\vtheta^v), X_t\big) - S^{\VaR}\big(v_t(\vtheta_0^v), X_t\big)\Big]\geq\frac{1}{4}\tau^2f_1p_1\P\Big\{\big|v_t(\vtheta^v) - v_t(\vtheta_0^v)\big|>\tau\,\Big\vert\, f_t^X\big(v_t(\vtheta_0^v)\big)>f_1\Big\}
	\end{equation}
	for sufficiently small $\tau>0$, and $p_1$ and $f_1$ from Assumption~\ref{ass:cons}~\ref{it:cond dens}.
	Therefore, for every $\xi>0$ we may pick $\tau>0$ sufficiently small, such that for $\Vert\vtheta^v-\vtheta_0^v\Vert\geq\xi$ it holds that
	\[
	\liminf_{n\to\infty} \Big[\overline{Q}_n^v(\vtheta^v) - \overline{Q}_n^v(\vtheta_0^v)\Big]\geq\frac{1}{4}\tau^2 f_1 p_1 \liminf_{n\to\infty}\frac{1}{n}\sum_{t=1}^{n}\P\Big\{\big|v_t(\vtheta^v) - v_t(\vtheta_0^v)\big|>\tau\,\Big\vert\, f_t^X\big(v_t(\vtheta_0^v)\big)>f_1\Big\}>0
	\]
	by Assumption~\ref{ass:cons}~\ref{it:unique id}~(a).
	The unique identifiability of $\vtheta_0^v$ follows.

	In sum, Lemma~2.2 of \citet{Whi80} applies, such that $\widehat{\vtheta}_n^v\overset{\P}{\longrightarrow}\vtheta_0^v$.

	It remains to prove $\widehat{\vtheta}_{n}^{c}\overset{\P}{\longrightarrow}\vtheta_0^{c}$. To this end, we again invoke (the ``in probability'' version of) Lemma~2.2 of \citet{Whi80}.
	For $\vtheta^c\in\mTheta^c$, we define
	\begin{align*}
		Q_n^{c}(\vtheta^c) &= \frac{1}{n}\sum_{t=1}^{n}S^{\CoVaR}\big((v_t(\widehat{\vtheta}_n^v),c_t(\vtheta^c))^\prime, (X_t, Y_t)^\prime\big),\\
		\overline{Q}_n^{c}(\vtheta^c) &= \E\bigg[\frac{1}{n}\sum_{t=1}^{n}S^{\CoVaR}\big((v_t(\vtheta_0^v),c_t(\vtheta^c))^\prime, (X_t, Y_t)^\prime\big)\bigg].
	\end{align*}
	Once again, $\mTheta^c$ is compact by Assumption~\ref{ass:cons}~\ref{it:compact}.
	The function $\vtheta^c\mapsto Q_n^c(\vtheta^c)$ is continuous on $\mTheta^c$ due to the continuity of the map $c\mapsto S^{\CoVaR}\big((v,c)^\prime,(x,y)^\prime\big)$ together with the continuity of $c_t(\cdot)$ (by Assumption~\ref{ass:cons}~\ref{it:smooth}).
	Recall that $\widehat{\vtheta}_n^c=\argmin_{\vtheta^c\in\mTheta^c} Q_n^c(\vtheta^c)$.
	
	Now, let us verify the ULLN for the objective function, required by Lemma~2.2.
	Since $\widehat{\vtheta}_n^v\overset{\P}{\longrightarrow}\vtheta_0^v$, the event $\big\{\Vert\widehat{\vtheta}_n^v-\vtheta_0^v\Vert\leq \delta\big\}$ for any $\delta>0$ occurs with probability approaching one (w.p.a.~1), as $n\to\infty$.
	It holds on the set $\big\{\Vert\widehat{\vtheta}_n^v-\vtheta_0^v\Vert\leq \delta\big\}$ that
	\begin{align*}
		&\sup_{\vtheta^c\in\mTheta^c}\big|Q_n^c(\vtheta^c) - \overline{Q}_n^c(\vtheta^c)\big|\\
		&\leq \sup_{\substack{\vtheta^{c}\in\mTheta^{c}\\ \Vert\vtheta^v-\vtheta_0^v\Vert\leq \delta}}\bigg|\frac{1}{n}\sum_{t=1}^{n}S^{\CoVaR}\big((v_t(\vtheta^v), c_t(\vtheta^c))^\prime, (X_t, Y_t)^\prime\big)-\E\Big[S^{\CoVaR}\big((v_t(\vtheta_0^v), c_t(\vtheta^c))^\prime, (X_t, Y_t)^\prime\big)\Big]\bigg|\\	
		&\leq \sup_{\substack{\vtheta^{c}\in\mTheta^{c}\\ \Vert\vtheta^v-\vtheta_0^v\Vert\leq \delta}}\bigg|\frac{1}{n}\sum_{t=1}^{n}S^{\CoVaR}\big((v_t(\vtheta^v), c_t(\vtheta^c))^\prime, (X_t, Y_t)^\prime\big)-\E\Big[S^{\CoVaR}\big((v_t(\vtheta^v), c_t(\vtheta^c))^\prime, (X_t, Y_t)^\prime\big)\Big]\bigg|\\
		&\hspace{0.3cm} + \sup_{\substack{\vtheta^{c}\in\mTheta^{c}\\ \Vert\vtheta^v-\vtheta_0^v\Vert\leq \delta}}\bigg|\frac{1}{n}\sum_{t=1}^{n}\E\Big[S^{\CoVaR}\big((v_t(\vtheta^v), c_t(\vtheta^c))^\prime, (X_t, Y_t)^\prime\big)\Big]-\E\Big[S^{\CoVaR}\big((v_t(\vtheta_0^v), c_t(\vtheta^c))^\prime, (X_t, Y_t)^\prime\big)\Big]\bigg|\\
		&=:A_{1n} + B_{1n}.
	\end{align*}
	The ULLN from Assumption~\ref{ass:cons}~\ref{it:ULLN} immediately implies that $A_{1n}=o_{\P}(1)$.
	
	Thus, it remains to show that $B_{1n}=o(1)$. 
	Note that
	\begin{multline}\label{eq:tbb}
		B_{1n}\leq \frac{1}{n}\sum_{t=1}^{n}\E\bigg[\sup_{\substack{\vtheta^{c}\in\mTheta^{c}\\ \Vert\vtheta^v-\vtheta_0^v\Vert\leq\delta}}\Big|S^{\CoVaR}\big((v_t(\vtheta^v), c_t(\vtheta^c))^\prime, (X_t, Y_t)^\prime\big)\\
		-S^{\CoVaR}\big((v_t(\vtheta_0^v), c_t(\vtheta^c))^\prime, (X_t, Y_t)^\prime\big)\Big|\bigg].
	\end{multline}
	Choose $\delta>0$ sufficiently small, such that Assumption~\ref{ass:cons} \ref{it:bound} is satisfied.
	Define the $\mathcal{F}_{t-1}$-measurable quantities
	\begin{align*}
		\underline{\vtau} &:= \argmin_{\Vert\vtau-\vtheta_0^v\Vert\leq \delta}v_t(\vtau),\\
		\overline{\vtau} 	&:= \argmax_{\Vert\vtau-\vtheta_0^v\Vert\leq \delta}v_t(\vtau),
	\end{align*}
	which exist by continuity of $v_t(\cdot)$ (see Assumption~\ref{ass:cons}~\ref{it:smooth}).
	Write
	\begin{align*}
		\sup_{\substack{\vtheta^{c}\in\mTheta^{c}\\ \Vert\vtheta^v-\vtheta_0^v\Vert\leq\delta}}&\Big|S^{\CoVaR}\big((v_t(\vtheta^v), c_t(\vtheta^c))^\prime, (X_t, Y_t)^\prime\big)-S^{\CoVaR}\big((v_t(\vtheta_0^v), c_t(\vtheta^c))^\prime, (X_t, Y_t)^\prime\big)\Big|\\
		&= \sup_{\substack{\vtheta^{c}\in\mTheta^{c}\\ \Vert\vtheta^v-\vtheta_0^v\Vert\leq\delta}}\Big|\big[\1_{\{X_t>v_t(\vtheta^v)\}}-\1_{\{X_t>v_t(\vtheta_0^v)\}}\big]\big[\1_{\{Y_t\leq c_t(\vtheta^c)\}} - \alpha\big]\big[c_t(\vtheta^c)-Y_t\big]\Big|\\
		&\leq\big|\1_{\{X_t>v_t(\overline{\vtau})\}}-\1_{\{X_t>v_t(\underline{\vtau})\}}\big|\sup_{\vtheta^{c}\in\mTheta^{c}}\Big|\big[\1_{\{Y_t\leq c_t(\vtheta^c)\}} - \alpha\big]\big[c_t(\vtheta^c)-Y_t\big]\Big|\\
		&\leq \1_{\{v_t(\underline{\vtau})<X_t\leq v_t(\overline{\vtau})\}}\Big[\sup_{\vtheta^{c}\in\mTheta^{c}}\big|c_t(\vtheta^c)\big|+|Y_t|\Big].
	\end{align*}
	Therefore,
	\begin{align*}
		&\E\bigg[\sup_{\substack{\vtheta^{c}\in\mTheta^{c}\\ \Vert\vtheta^v-\vtheta_0^v\Vert\leq\delta}}\Big|S^{\CoVaR}\big((v_t(\vtheta^v), c_t(\vtheta^c))^\prime, (X_t, Y_t)^\prime\big)-S^{\CoVaR}\big((v_t(\vtheta_0^v), c_t(\vtheta^c))^\prime, (X_t, Y_t)^\prime\big)\Big|\bigg]\\
		&\leq \E\Big[\1_{\{v_t(\underline{\vtau})<X_t\leq v_t(\overline{\vtau})\}}\sup_{\vtheta^{c}\in\mTheta^{c}}\big|c_t(\vtheta^c)\big|\Big] + \E\Big[\1_{\{v_t(\underline{\vtau})<X_t\leq v_t(\overline{\vtau})\}}\big|Y_t\big|\Big]\\
		&=:B_{11t} + B_{12t}.
	\end{align*}
	For $B_{11t}$, we obtain using the LIE (in the first step), Assumption~\ref{ass:cons}~\ref{it:cond dens} (in the third step), the mean value theorem (in the fourth step) and Assumption~\ref{ass:cons} \ref{it:bound}--\ref{it:mom bounds cons} (in the fifth and sixth step) that
	\begin{align}
		B_{11t}	&= \E\bigg[\sup_{\vtheta^{c}\in\mTheta^{c}}\big|c_t(\vtheta^c)\big|\E_{t-1}\big\{\1_{\{v_t(\underline{\vtau})<X_t\leq v_t(\overline{\vtau})\}}\big\}\bigg]\notag\\
		&\leq \E\bigg[C(\mathcal{F}_{t-1})\int_{v_t(\underline{\vtau})}^{v_t(\overline{\vtau})}f_t^{X}(x)\D x\bigg]\notag\\
		&\leq \E\Big[C(\mathcal{F}_{t-1}) K \big\{v_t(\overline{\vtau})-v_t(\underline{\vtau})\big\}\Big]\notag\\
		&\leq K\E\Big[C(\mathcal{F}_{t-1})  \big\Vert\nabla v_{t}(\vtau^{*})\big\Vert\cdot\big\Vert\overline{\vtau}-\underline{\vtau}\big\Vert\Big]\notag\\
		&\leq K\E\big[C(\mathcal{F}_{t-1})  V_1(\mathcal{F}_{t-1})\big]2\delta\notag\\
		&\leq C\delta,\label{eq:(B.11t)}
	\end{align}
	where $\vtau^\ast$ is some mean value between $\underline{\vtau}$ and $\overline{\vtau}$.
	
	For $B_{12t}$, we apply Hölder's inequality for $r=(1+\iota)/\iota$ and $s=1+\iota$ (with $\iota>0$ from Assumption~\ref{ass:cons}~\ref{it:mom bounds cons}) to get that
	\begin{equation}\label{eq:B12t}
		B_{12t} \leq \Big\{\E\big[\1^{r}_{\{v_t(\underline{\vtau})<X_t\leq v_t(\overline{\vtau})\}}\big]\Big\}^{1/r}\Big\{\E|Y_t|^{s}\Big\}^{1/s}.
	\end{equation}
	By the LIE,
	\begin{align*}
		\E\big[\1^{r}_{\{v_t(\underline{\vtau})<X_t\leq v_t(\overline{\vtau})\}}\big]	&= \E\big[\1_{\{v_t(\underline{\vtau})<X_t\leq v_t(\overline{\vtau})\}}\big]\\
		&= \E\Big[\E_{t-1}\big\{\1_{\{v_t(\underline{\vtau})<X_t\leq v_t(\overline{\vtau})\}}\big\}\Big]\\
		&=\E\Big[F_t^{X}\big(v_t(\overline{\vtau})\big) - F_t^{X}\big(v_t(\underline{\vtau})\big)\Big]\\
		&=\E\Big[f_t^{X}\big(v_t(\vtau^\ast)\big)\big\{v_t(\overline{\vtau}) - v_t(\underline{\vtau})\big\}\Big]\\
		&\leq \E\Big[K\big\Vert \nabla v_t(\vtau^\ast)\big\Vert \cdot\big\Vert \overline{\vtau} - \underline{\vtau}\big\Vert\Big]\\
		&\leq K\E\big[V_1(\mathcal{F}_{t-1})\big]2\delta\\
		&\leq C\delta,
	\end{align*}
	where $\vtau^\ast$ (which may change from line to line) lies between $\underline{\vtau}$ and $\overline{\vtau}$, and the penultimate step follows from Assumption~\ref{ass:cons}~\ref{it:bound}.
	Plugging this into \eqref{eq:B12t} and using Assumption~\ref{ass:cons}~\ref{it:mom bounds cons} gives that
	\begin{equation}\label{eq:(B.12t)}
		B_{12t}\leq C^{1/r}\Big\{\E|Y_t|^{s}\Big\}^{1/s}\delta^{1/r}\leq C\delta^{1/r}.
	\end{equation}
	Plugging \eqref{eq:(B.11t)} and \eqref{eq:(B.12t)} into \eqref{eq:tbb} yields that $B_{1n}\leq C[\delta+\delta^{1/r}]$ for $\delta>0$ that can be chosen to be arbitrarily small.

	Recalling that the event $\big\{\Vert\widehat{\vtheta}_n^v-\vtheta_0^v\Vert\leq \delta\big\}$ occurs w.p.a.~1, as $n\to\infty$, we obtain for any $\varepsilon>0$ that
	\begin{align*}
		\P\bigg\{\sup_{\vtheta^c\in\mTheta^c}\big|Q_n^c(\vtheta^c) - \overline{Q}_n^c(\vtheta^c)\big|>\varepsilon\bigg\} &\leq \P\bigg\{\sup_{\vtheta^c\in\mTheta^c}\big|Q_n^c(\vtheta^c) - \overline{Q}_n^c(\vtheta^c)\big|>\varepsilon,\ \big\Vert\widehat{\vtheta}_n^v-\vtheta_0^v\big\Vert\leq \delta\bigg\} \\
		& \hspace{2cm} + \P\Big\{\big\Vert\widehat{\vtheta}_n^v-\vtheta_0^v\big\Vert> \delta\Big\}\\
		&\leq \P\big\{A_{1n} + B_{1n}>\varepsilon\big\}+o(1)\\
		&\leq \P\big\{A_{1n}>\varepsilon/2\big\}+\P\big\{B_{1n}>\varepsilon/2\big\}+o(1)\\
		&\leq o(1)+ \P\big\{C[\delta+\delta^{1/r}]>\varepsilon/2\big\}+o(1)\\
		&=o(1)
	\end{align*}
	if $\delta>0$ is chosen sufficiently small to satisfy $C[\delta+\delta^{1/r}]\leq\varepsilon/2$.
	This establishes the required ULLN $\sup_{\vtheta^c\in\mTheta^c}\big|Q_n^c(\vtheta^c) - \overline{Q}_n^c(\vtheta^c)\big|=o_{\P}(1)$.
	
	Now, we verify that $\vtheta_0^c$ is the identifiably unique minimum of $\overline{Q}_{n}^{c}(\cdot)$ in the sense of \citet[Definition~2.1]{Whi80}.
	To do so, first observe that $\overline{Q}_{n}^{c}(\cdot)$ is continuous on $\mTheta^c$, since for all $\vtheta^c\in\mTheta^c$,
	\begin{align*}
		\lim_{\vtheta\to\vtheta^c}\overline{Q}_n^c(\vtheta) &= \frac{1}{n}\sum_{t=1}^{n}\lim_{\vtheta\to\vtheta^c}\E\Big[S^{\CoVaR}\big((v_t(\vtheta_0^v), c_t(\vtheta))^\prime, (X_t, Y_t)^\prime\big)\Big]\\
		&= \frac{1}{n}\sum_{t=1}^{n}\E\Big[\lim_{\vtheta\to\vtheta^c}S^{\CoVaR}\big((v_t(\vtheta_0^v), c_t(\vtheta))^\prime, (X_t, Y_t)^\prime\big)\Big]\\
		&= \frac{1}{n}\sum_{t=1}^{n}\E\Big[S^{\CoVaR}\big((v_t(\vtheta_0^v), c_t(\vtheta^c))^\prime, (X_t, Y_t)^\prime\big)\Big]\\
		&= \overline{Q}_n^c(\vtheta^c),
	\end{align*}
	where the second step follows from the DCT, and the third step from the continuity of the map $c\mapsto S^{\CoVaR}\big((v,c)^\prime,(x,y)^\prime\big)$ together with the continuity of $c_t(\cdot)$ (by Assumption~\ref{ass:cons} \ref{it:smooth}). 
	Thus, $\vtheta^c\mapsto \overline{Q}_{n}^{c}(\vtheta^c)$ is continuous. 
	Note that we may apply the DCT, because $S^{\CoVaR}\big((v_t(\vtheta_0^v), c_t(\vtheta))^\prime, (X_t, Y_t)^\prime\big)$ is dominated by 
	\[
	\big|S^{\CoVaR}\big((v_t(\vtheta_0^v), c_t(\vtheta))^\prime, (X_t, Y_t)^\prime\big)\big|\leq\big|Y_t - c_t(\vtheta)\big|\leq \big|Y_t\big| + \big|c_t(\vtheta)\big|\leq |Y_t| + C(\mathcal{F}_{t-1}),
	\]
	where the final term is integrable due to Assumption~\ref{ass:cons}~\ref{it:mom bounds cons}.
	
	That $\vtheta_0^c$ minimizes $\overline{Q}_n^{c}(\cdot)$ follows from similar arguments used to prove that $\vtheta_0^v$ minimizes $\overline{Q}_n^{v}(\cdot)$.
	
	It remains to show that $\vtheta_0^c$ is even the identifiably unique minimum of $\overline{Q}_n^{c}(\cdot)$.
	Section~\ref{Proofs of Equations} shows under Assumption~\ref{ass:cons} that for all $\vtheta^c\in\mTheta^c$,
	\begin{multline}\label{eq:un id CoVaR}
		\E\Big[S^{\CoVaR}\big((v_t(\vtheta_0^v), c_t(\vtheta^c))^\prime, (X_t,Y_t)^\prime \big) - S^{\CoVaR}\big((v_t(\vtheta_0^v), c_t(\vtheta_0^c))^\prime, (X_t,Y_t)^\prime \big)\Big]\\
		\geq\frac{1}{4}\tau^2 f_2 p_2 \P\Big\{\big|c_t(\vtheta^c) - c_t(\vtheta_0^c)\big|>\tau\,\Big\vert\, \int_{v_t(\vtheta_0^v)}^{\infty}f_t\big(x,c_t(\vtheta_0^c)\big)\D x>f_2\Big\}
	\end{multline}
	for sufficiently small $\tau>0$, and $p_2$ and $f_2$ from Assumption~\ref{ass:cons}~\ref{it:cond dens}.
	Therefore, for every $\xi>0$ we may pick $\tau>0$ sufficiently small, such that for $\Vert\vtheta^c-\vtheta_0^c\Vert\geq\xi$ it holds that
	\begin{multline*}
		\liminf_{n\to\infty} \Big[\overline{Q}_n^c(\vtheta^c) - \overline{Q}_n^c(\vtheta_0^c)\Big]\\
		\geq\frac{1}{4}\tau^2 f_2p_2 \liminf_{n\to\infty}\frac{1}{n}\sum_{t=1}^{n}\P\Big\{\big|c_t(\vtheta^c) - c_t(\vtheta_0^c)\big|>\tau\,\Big\vert\, \int_{v_t(\vtheta_0^v)}^{\infty}f_t\big(x,c_t(\vtheta_0^c)\big)\D x>f_2\Big\}>0
	\end{multline*}
	by Assumption~\ref{ass:cons}~\ref{it:unique id}~(b).
	The unique identifiability of $\vtheta_0^c$ follows.
	
	Lemma~2.2 of \citet{Whi80} now implies that $\widehat{\vtheta}_n^c\overset{\P}{\longrightarrow}\vtheta_0^c$.
\end{proof}

\subsection{Proofs of Equations~\eqref{eq:un id VaR} and \eqref{eq:un id CoVaR}}\label{Proofs of Equations}

\begin{proof}[{\textbf{Proof of \eqref{eq:un id VaR}:}}]
	For brevity, define
	\begin{align*}
		\Delta S^{\VaR}(\vtheta^v) &= S^{\VaR}\big(v_t(\vtheta^v), X_t\big) - S^{\VaR}\big(v_t(\vtheta_0^v), X_t\big),\\
		\Delta v_t(\vtheta^v) &= v_t(\vtheta^v) - v_t(\vtheta_0^v).
	\end{align*}
	Then,
	\begin{align*}
		\Delta S^{\VaR}(\vtheta^v) &= \big[\1_{\{X_t\leq v_t(\vtheta^v)\}}-\beta\big]\big[v_t(\vtheta^v)-X_t\big] - \big[\1_{\{X_t\leq v_t(\vtheta_0^v)\}}-\beta\big]\big[v_t(\vtheta_0^v)-X_t\big]\\
		&=\begin{cases}
			(1-\beta)\Delta v_t(\vtheta^v),& \text{if}\ X_t\leq v_t(\vtheta^v) \wedge X_t\leq v_t(\vtheta_0^v),\\
			(1-\beta)\Delta v_t(\vtheta^v) - \big[X_t-v_t(\vtheta^v)\big],&\text{if}\ X_t\leq v_t(\vtheta^v)\wedge X_t> v_t(\vtheta_0^v),\\
			\big[X_t-v_t(\vtheta^v)\big] - \beta\Delta v_t(\vtheta^v),&\text{if}\ X_t> v_t(\vtheta^v)\wedge X_t\leq v_t(\vtheta_0^v),\\
			-\beta\Delta v_t(\vtheta^v),& \text{if}\ X_t>v_t(\vtheta^v)\wedge X_t>v_t(\vtheta_0^v).
		\end{cases}
	\end{align*}
	Therefore,
	\begin{align*}
		&\E_{t-1}\big[\Delta S^{\VaR}(\vtheta^v)\big]\\
		&= \E_{t-1}\big[\Delta S^{\VaR}(\vtheta^v)\1_{\{X_t\leq v_t(\vtheta^v),\ X_t\leq v_t(\vtheta_0^v)\}}\big] + \E_{t-1}\big[\Delta S^{\VaR}(\vtheta^v)\1_{\{v_t(\vtheta_0^v)<X_t\leq v_t(\vtheta^v)\}}\big] \\
		&\hspace{0.5cm} + \E_{t-1}\big[\Delta S^{\VaR}(\vtheta^v)\1_{\{v_t(\vtheta^v)<X_t\leq v_t(\vtheta_0^v)\}}\big] + \E_{t-1}\big[\Delta S^{\VaR}(\vtheta^v)\1_{\{X_t> v_t(\vtheta^v),\ X_t> v_t(\vtheta_0^v)\}}\big]\\
		&= (1-\beta)\Delta v_t(\vtheta^v)\P_{t-1}\big\{X_t\leq v_t(\vtheta^v),\ X_t\leq v_t(\vtheta_0^v)\big\}\\
		&\hspace{0.5cm} +(1-\beta)\Delta v_t(\vtheta^v)\P_{t-1}\big\{v_t(\vtheta_0^v)<X_t\leq v_t(\vtheta^v)\big\} - \E_{t-1}\Big[\big\{X_t - v_t(\vtheta_0^v)\big\} \1_{\{v_t(\vtheta_0^v)<X_t\leq v_t(\vtheta^v)\}}\Big]\\
		&\hspace{0.5cm} + \E_{t-1}\Big[\big\{X_t - v_t(\vtheta_0^v)\big\} \1_{\{v_t(\vtheta^v)<X_t\leq v_t(\vtheta_0^v)\}}\Big] - \beta\Delta v_t(\vtheta^v)\P_{t-1}\big\{v_t(\vtheta^v)<X_t\leq v_t(\vtheta_0^v)\big\}\\
		&\hspace{0.5cm} - \beta\Delta v_t(\vtheta^v)\P_{t-1}\big\{X_t>v_t(\vtheta^v),\ X_t>v_t(\vtheta_0^v)\big\}\\
		&= -\beta\Delta v_t(\vtheta^v) + \Delta v_t(\vtheta^v)\Big[\P_{t-1}\big\{X_t\leq v_t(\vtheta^v),\ X_t\leq v_t(\vtheta_0^v)\big\} + \P_{t-1}\big\{v_t(\vtheta_0^v)<X_t\leq v_t(\vtheta^v)\big\}\Big]\\
		&\hspace{0.5cm} - \E_{t-1}\Big[\big\{X_t - v_t(\vtheta_0^v)\big\} \1_{\{v_t(\vtheta_0^v)<X_t\leq v_t(\vtheta^v)\}}\Big] + \E_{t-1}\Big[\big\{X_t - v_t(\vtheta_0^v)\big\} \1_{\{v_t(\vtheta^v)<X_t\leq v_t(\vtheta_0^v)\}}\Big].
	\end{align*}
	The first two terms on the right-hand side simplify to
	\begin{align*}
		-\beta&\Delta v_t(\vtheta^v) + \Delta v_t(\vtheta^v)\Big[\P_{t-1}\big\{X_t\leq v_t(\vtheta^v),\ X_t\leq v_t(\vtheta_0^v)\big\} + \P_{t-1}\big\{v_t(\vtheta_0^v)<X_t\leq v_t(\vtheta^v)\big\}\Big]\\
		&= - \beta\Delta v_t(\vtheta^v) + \Delta v_t(\vtheta^v)\P_{t-1}\big\{X_t\leq v_t(\vtheta^v)\big\}\\
		&= \Delta v_t(\vtheta^v)\Big[\P_{t-1}\big\{X_t\leq v_t(\vtheta^v)\big\} - \P_{t-1}\big\{X_t\leq v_t(\vtheta_0^v)\big\}\Big]\\
		&=\Delta v_t(\vtheta^v) \int_{v_t(\vtheta_0^v)}^{v_t(\vtheta^v)} f_t^X(x)\D x\\
		&=\big[v_t(\vtheta^v) - v_t(\vtheta_0^v)\big] \int_{v_t(\vtheta_0^v)}^{v_t(\vtheta^v)} f_t^X(x)\D x \big[\1_{\{v_t(\vtheta^v)>v_t(\vtheta_0^v)\}} + \1_{\{v_t(\vtheta^v)\leq v_t(\vtheta_0^v)\}}\big]\\
		&= \1_{\{v_t(\vtheta^v)>v_t(\vtheta_0^v)\}}\big[v_t(\vtheta^v) - v_t(\vtheta_0^v)\big] \int_{v_t(\vtheta_0^v)}^{v_t(\vtheta^v)} f_t^X(x)\D x\\
		&\hspace{2cm} - \1_{\{v_t(\vtheta^v)\leq v_t(\vtheta_0^v)\}}\big[v_t(\vtheta^v) - v_t(\vtheta_0^v)\big] \int_{v_t(\vtheta^v)}^{v_t(\vtheta_0^v)} f_t^X(x)\D x.
	\end{align*}
	Hence,
	\begin{align*}
		\E_{t-1}\big[\Delta S^{\VaR}(\vtheta^v)\big] &= \1_{\{v_t(\vtheta^v)>v_t(\vtheta_0^v)\}}\big[v_t(\vtheta^v) - v_t(\vtheta_0^v)\big] \int_{v_t(\vtheta_0^v)}^{v_t(\vtheta^v)} f_t^X(x)\D x\\
		&\hspace{2cm} - \1_{\{v_t(\vtheta^v)\leq v_t(\vtheta_0^v)\}}\big[v_t(\vtheta^v) - v_t(\vtheta_0^v)\big] \int_{v_t(\vtheta^v)}^{v_t(\vtheta_0^v)} f_t^X(x)\D x\\
		&\hspace{0.5cm} - \1_{\{v_t(\vtheta_0^v)<v_t(\vtheta^v)\}}\int_{v_t(\vtheta_0^v)}^{v_t(\vtheta^v)}\big[x-v_t(\vtheta_0^v)\big]f_t^X(x)\D x\\
		&\hspace{0.5cm} + \1_{\{v_t(\vtheta^v)\leq v_t(\vtheta_0^v)\}}\int_{v_t(\vtheta^v)}^{v_t(\vtheta_0^v)}\big[x-v_t(\vtheta_0^v)\big]f_t^X(x)\D x\\
		&= \1_{\{v_t(\vtheta_0^v)< v_t(\vtheta^v)\}} \int_{v_t(\vtheta_0^v)}^{v_t(\vtheta^v)} \big[v_t(\vtheta^v) - x\big]f_t^X(x)\D x\\
		&\hspace{2cm} + \1_{\{v_t(\vtheta^v)\leq v_t(\vtheta_0^v)\}} \int_{v_t(\vtheta^v)}^{v_t(\vtheta_0^v)} \big[x - v_t(\vtheta^v)\big]f_t^X(x)\D x.
	\end{align*}
	In particular, by first decreasing the set in the indicator functions (in the first inequality), then using the conditions in the indicator functions for the integration boundaries (in the second inequality) and for the integrands (in the third inequality),  we have that for any $\tau>0$, 
	\begin{align}
		\E_{t-1}\big[\Delta S^{\VaR}(\vtheta^v)\big] &\geq  
		\1_{\{v_t(\vtheta_0^v) + \tau< v_t(\vtheta^v)\}} \int_{v_t(\vtheta_0^v)}^{v_t(\vtheta^v)} \big[v_t(\vtheta^v) - x\big]f_t^X(x)\D x\notag\\
		&\hspace{2cm} + \1_{\{v_t(\vtheta^v)+\tau\leq v_t(\vtheta_0^v)\}} \int_{v_t(\vtheta^v)}^{v_t(\vtheta_0^v)} \big[x - v_t(\vtheta^v)\big]f_t^X(x)\D x\notag\\
		&\geq \1_{\{v_t(\vtheta_0^v) + \tau < v_t(\vtheta^v)\}} \int_{v_t(\vtheta_0^v)}^{v_t(\vtheta_0^v) + \tau} \big[v_t(\vtheta^v) - x\big]f_t^X(x)\D x\notag\\
		&\hspace{2cm} + \1_{\{v_t(\vtheta^v)+\tau\leq v_t(\vtheta_0^v)\}} \int_{v_t(\vtheta_0^v)-\tau}^{v_t(\vtheta_0^v)} \big[x - v_t(\vtheta^v)\big]f_t^X(x)\D x\notag\\
		&\geq \1_{\{v_t(\vtheta_0^v) + \tau< v_t(\vtheta^v)\}} \int_{v_t(\vtheta_0^v)}^{v_t(\vtheta_0^v) + \tau} \big[\tau + \{v_t(\vtheta_0^v)- x\}\big]f_t^X(x)\D x\notag\\
		&\hspace{2cm} + \1_{\{v_t(\vtheta^v)+\tau\leq v_t(\vtheta_0^v)\}} \int_{v_t(\vtheta_0^v)-\tau}^{v_t(\vtheta_0^v)} \big[\tau - \{v_t(\vtheta_0^v)-x\}\big]f_t^X(x)\D x.\label{eq:plug in}
	\end{align}
	By Assumption~\ref{ass:cons}~\ref{it:Lipschitz cons}--\ref{it:cond dens}, we obtain that
	\begin{align*}
		f_t^{X}\big(v_t(\vtheta_0^v)+\lambda\big) &= f_t^{X}\big(v_t(\vtheta_0^v)\big) + \Big[f_t^{X}\big(v_t(\vtheta_0^v)+\lambda\big) - f_t^{X}\big(v_t(\vtheta_0^v)\big)\Big]\\
		&\geq f_t^{X}\big(v_t(\vtheta_0^v)\big) - \Big|f_t^{X}\big(v_t(\vtheta_0^v)+\lambda\big) - f_t^{X}\big(v_t(\vtheta_0^v)\big)\Big|\\
		&> f_1-K|\lambda|
	\end{align*}
	on the set $\big\{f_t^{X}\big(v_t(\vtheta_0^v)\big)>f_1\big\}$.
	Hence, for some sufficiently small $\overline{\lambda}>0$, it holds on this set that
	\[
	f_t^{X}\big(v_t(\vtheta_0^v)+\lambda\big)\geq \frac{f_1}{2}>0
	\]
	for all $|\lambda|\leq\overline{\lambda}$ and all $t\in\mathbb{N}$.
	Plugging this into \eqref{eq:plug in}, we get for $\tau\in(0,\overline{\lambda})$ that
	\begin{align*}
		\E_{t-1}\big[\Delta S^{\VaR}(\vtheta^v)\big]	& = \E_{t-1}\Big[\Delta S^{\VaR}(\vtheta^v)\big(\1_{\{f_t^{X}(v_t(\vtheta_0^v))>f_1\}} + \1_{\{f_t^{X}(v_t(\vtheta_0^v))\leq f_1\}}\big)\Big]\\
		&\geq \E_{t-1}\big[\Delta S^{\VaR}(\vtheta^v)\big]\1_{\{f_t^{X}(v_t(\vtheta_0^v))>f_1\}}\\	
		&\geq \Big(\1_{\{v_t(\vtheta_0^v) + \tau< v_t(\vtheta^v)\}} \int_{v_t(\vtheta_0^v)}^{v_t(\vtheta_0^v) + \tau} \big[\tau + \{v_t(\vtheta_0^v)- x\}\big]\frac{f_1}{2}\D x\notag\\
		&\hspace{1cm} + \1_{\{v_t(\vtheta^v)+\tau\leq v_t(\vtheta_0^v)\}} \int_{v_t(\vtheta_0^v)-\tau}^{v_t(\vtheta_0^v)} \big[\tau - \{v_t(\vtheta_0^v)-x\}\big]\frac{f_1}{2}\D x\Big)\1_{\{f_t^{X}(v_t(\vtheta_0^v))>f_1\}}\\
		&=\Big(\1_{\{v_t(\vtheta_0^v) + \tau< v_t(\vtheta^v)\}} \frac{1}{2}\tau^2\frac{f_1}{2} + \1_{\{v_t(\vtheta^v)+\tau\leq v_t(\vtheta_0^v)\}} \frac{1}{2}\tau^2\frac{f_1}{2}\Big)\1_{\{f_t^{X}(v_t(\vtheta_0^v))>f_1\}}\\
		&\geq \1_{\{|v_t(\vtheta^v) - v_t(\vtheta_0^v)|>\tau,\ f_t^{X}(v_t(\vtheta_0^v))>f_1\}}\frac{1}{4}\tau^2f_1.
	\end{align*}
	Taking the unconditional expectation gives
	\begin{align*}
		\E\big[\Delta S^{\VaR}(\vtheta^v)\big] &\geq \frac{1}{4}\tau^2f_1\P\big\{|v_t(\vtheta^v) - v_t(\vtheta_0^v)|>\tau,\ f_t^{X}(v_t(\vtheta_0^v))>f_1\big\}\\
		&=\frac{1}{4}\tau^2f_1 \P\Big\{f_t^X\big(v_t(\vtheta_0^v)\big)>f_1\Big\}\P\Big\{\big|v_t(\vtheta^v) - v_t(\vtheta_0^v)\big|>\tau \,\Big\vert\, f_t^X\big(v_t(\vtheta_0^v)\big)>f_1\Big\}\\
		&> \frac{1}{4}\tau^2f_1 p_1\P\Big\{\big|v_t(\vtheta^v) - v_t(\vtheta_0^v)\big|>\tau \,\Big\vert\, f_t^X\big(v_t(\vtheta_0^v)\big)>f_1\Big\},
	\end{align*}
	where the final step follows from Assumption~\ref{ass:cons}~\ref{it:cond dens}.
\end{proof}

\begin{proof}[{\textbf{Proof of \eqref{eq:un id CoVaR}:}}]
	The proof is similar to that of \eqref{eq:un id VaR}.
	For brevity, define
	\begin{align*}
		\Delta S^{\CoVaR}(\vtheta^c) &= S^{\CoVaR}\big((v_t(\vtheta_0^v), c_t(\vtheta^c))^\prime, (X_t,Y_t)^\prime\big) - S^{\CoVaR}\big((v_t(\vtheta_0^v), c_t(\vtheta_0^c))^\prime, (X_t,Y_t)^\prime\big),\\
		\Delta c_t(\vtheta^c) &= c_t(\vtheta^c) - c_t(\vtheta_0^c).
	\end{align*}
	Then, similarly as above,
	\begin{align*}
		\Delta S^{\CoVaR}(\vtheta^c) &= \1_{\{X_t>v_t(\vtheta_0^v)\}}\Big\{\big[\1_{\{Y_t\leq c_t(\vtheta^c)\}}-\alpha\big]\big[c_t(\vtheta^c)-Y_t\big] - \big[\1_{\{Y_t\leq c_t(\vtheta_0^c)\}}-\alpha\big]\big[c_t(\vtheta_0^c)-Y_t\big]\Big\}\\
		&=\begin{cases}
			\1_{\{X_t>v_t(\vtheta_0^v)\}}(1-\alpha)\Delta c_t(\vtheta^c),& \text{if}\ Y_t\leq c_t(\vtheta^c) \wedge Y_t\leq c_t(\vtheta_0^c),\\
			\1_{\{X_t>v_t(\vtheta_0^v)\}}\Big\{(1-\alpha)\Delta c_t(\vtheta^c) - \big[Y_t-c_t(\vtheta^c)\big]\Big\},&\text{if}\ Y_t\leq c_t(\vtheta^c)\wedge Y_t> c_t(\vtheta_0^c),\\
			\1_{\{X_t>v_t(\vtheta_0^v)\}}\Big\{\big[Y_t-c_t(\vtheta^c)\big] - \alpha\Delta c_t(\vtheta^c)\Big\},&\text{if}\ Y_t> c_t(\vtheta^c)\wedge Y_t\leq c_t(\vtheta_0^c),\\
			-\1_{\{X_t>v_t(\vtheta_0^v)\}}\alpha\Delta c_t(\vtheta^c),& \text{if}\ Y_t>c_t(\vtheta^c)\wedge Y_t>c_t(\vtheta_0^c).
		\end{cases}
	\end{align*}
	Therefore,
	\begin{align*}
		&\E_{t-1}\big[\Delta S^{\CoVaR}(\vtheta^c)\big]\\
		&= \E_{t-1}\big[\Delta S^{\CoVaR}(\vtheta^c)\1_{\{Y_t\leq c_t(\vtheta^c),\ Y_t\leq c_t(\vtheta_0^c)\}}\big] + \E_{t-1}\big[\Delta S^{\CoVaR}(\vtheta^c)\1_{\{c_t(\vtheta_0^c)<Y_t\leq c_t(\vtheta^c)\}}\big] \\
		&\hspace{0.5cm} + \E_{t-1}\big[\Delta S^{\CoVaR}(\vtheta^c)\1_{\{c_t(\vtheta^c)<Y_t\leq c_t(\vtheta_0^c)\}}\big] + \E_{t-1}\big[\Delta S^{\CoVaR}(\vtheta^c)\1_{\{Y_t> c_t(\vtheta^c),\ Y_t> c_t(\vtheta_0^c)\}}\big]\\
		&= (1-\alpha)\Delta c_t(\vtheta^c)\P_{t-1}\big\{X_t>v_t(\vtheta_0^v),\ Y_t\leq c_t(\vtheta^c),\ Y_t\leq c_t(\vtheta_0^c)\big\}\\
		&\hspace{0.5cm} +(1-\alpha)\Delta c_t(\vtheta^c)\P_{t-1}\big\{X_t>v_t(\vtheta_0^v),\ c_t(\vtheta_0^c)<Y_t\leq c_t(\vtheta^c)\big\}\\
		&\hspace{0.5cm} - \E_{t-1}\Big[\big\{Y_t - c_t(\vtheta_0^c)\big\} \1_{\{X_t>v_t(\vtheta_0^v),\ c_t(\vtheta_0^c)<Y_t\leq c_t(\vtheta^c)\}}\Big]\\
		&\hspace{0.5cm} + \E_{t-1}\Big[\big\{Y_t - c_t(\vtheta_0^c)\big\} \1_{\{X_t>v_t(\vtheta_0^v),\ c_t(\vtheta^c)<Y_t\leq c_t(\vtheta_0^c)\}}\Big] \\
		&\hspace{0.5cm}- \alpha\Delta c_t(\vtheta^c)\P_{t-1}\big\{X_t>v_t(\vtheta_0^v),\ c_t(\vtheta^c)<Y_t\leq c_t(\vtheta_0^c)\big\}\\
		&\hspace{0.5cm} - \alpha\Delta c_t(\vtheta^c)\P_{t-1}\big\{X_t>v_t(\vtheta_0^v),\ Y_t>c_t(\vtheta^c),\ Y_t>c_t(\vtheta_0^c)\big\}\\
		&= -\alpha\Delta c_t(\vtheta^c)\P_{t-1}\big\{X_t>v_t(\vtheta_0^v)\big\} + \Delta c_t(\vtheta^c)\Big[\P_{t-1}\big\{X_t>v_t(\vtheta_0^v),\ Y_t\leq c_t(\vtheta^c),\ Y_t\leq c_t(\vtheta_0^c)\big\}\\
		&\hspace{0.5cm}+ \P_{t-1}\big\{X_t>v_t(\vtheta_0^v),\ c_t(\vtheta_0^c)<Y_t\leq c_t(\vtheta^c)\big\}\Big]\\
		&\hspace{0.5cm} - \E_{t-1}\Big[\big\{Y_t - c_t(\vtheta_0^c)\big\} \1_{\{X_t>v_t(\vtheta_0^v),\ c_t(\vtheta_0^c)<Y_t\leq c_t(\vtheta^c)\}}\Big]\\
		&\hspace{0.5cm}+ \E_{t-1}\Big[\big\{Y_t - c_t(\vtheta_0^c)\big\} \1_{\{X_t>v_t(\vtheta_0^v),\ c_t(\vtheta^c)<Y_t\leq c_t(\vtheta_0^c)\}}\Big].
	\end{align*}
	The first two terms on the right-hand side simplify to
	\begin{align*}
		-\alpha&\Delta c_t(\vtheta^c)\P_{t-1}\big\{X_t>v_t(\vtheta_0^v)\big\} +  \Delta c_t(\vtheta^c)\Big[\P_{t-1}\big\{X_t>v_t(\vtheta_0^v),\ Y_t\leq c_t(\vtheta^c),\ Y_t\leq c_t(\vtheta_0^c)\big\}\\
		&\hspace{2cm} + \P_{t-1}\big\{X_t>v_t(\vtheta_0^v),\ c_t(\vtheta_0^c)<Y_t\leq c_t(\vtheta^c)\big\}\Big]\\
		&=- \Delta c_t(\vtheta^c) \P_{t-1}\big\{Y_t\leq c_t(\vtheta_0^c)\mid X_t>v_t(\vtheta_0^v)\big\}\P_{t-1}\big\{X_t>v_t(\vtheta_0^v)\big\} \\
		&\hspace{1cm}+ \Delta c_t(\vtheta^c)\P_{t-1}\big\{X_t>v_t(\vtheta_0^v),\ Y_t\leq c_t(\vtheta^c)\big\}\\
		&=\Delta c_t(\vtheta^c)\Big[\P_{t-1}\big\{X_t>v_t(\vtheta_0^v),\ Y_t\leq c_t(\vtheta^c)\big\} - \P_{t-1}\big\{X_t>v_t(\vtheta_0^v),\ Y_t\leq c_t(\vtheta_0^c)\big\}\Big]\\
		&=\Delta c_t(\vtheta^c)\int_{v_t(\vtheta_0^v)}^{\infty}\int_{c_t(\vtheta_0^c)}^{c_t(\vtheta^c)}f_t(x,y)\D y \D x\\
		&= \Delta c_t(\vtheta^c)\int_{v_t(\vtheta_0^v)}^{\infty}\int_{c_t(\vtheta_0^c)}^{c_t(\vtheta^c)}f_t(x,y)\D y \D x\big[\1_{\{c_t(\vtheta^c)>c_t(\vtheta_0^c)\}} + \1_{\{c_t(\vtheta^c)\leq c_t(\vtheta_0^c)\}}\big]\\
		&= \1_{\{c_t(\vtheta^c)>c_t(\vtheta_0^c)\}} \big[c_t(\vtheta^c)-c_t(\vtheta_0^c)\big]\int_{v_t(\vtheta_0^v)}^{\infty}\int_{c_t(\vtheta_0^c)}^{c_t(\vtheta^c)}f_t(x,y)\D y \D x\\
		&\hspace{1cm}- \1_{\{c_t(\vtheta^c)\leq c_t(\vtheta_0^c)\}} \big[c_t(\vtheta^c)-c_t(\vtheta_0^c)\big]\int_{v_t(\vtheta_0^v)}^{\infty}\int_{c_t(\vtheta^c)}^{c_t(\vtheta_0^c)}f_t(x,y)\D y \D x.
	\end{align*}
	Hence, 
	\begin{align*}
		\E_{t-1}\big[\Delta S^{\CoVaR}(\vtheta^c)\big] &= \1_{\{c_t(\vtheta^c)>c_t(\vtheta_0^c)\}} \big[c_t(\vtheta^c)-c_t(\vtheta_0^c)\big]\int_{v_t(\vtheta_0^v)}^{\infty}\int_{c_t(\vtheta_0^c)}^{c_t(\vtheta^c)}f_t(x,y)\D y \D x\\
		&\hspace{1cm}- \1_{\{c_t(\vtheta^c)\leq c_t(\vtheta_0^c)\}} \big[c_t(\vtheta^c)-c_t(\vtheta_0^c)\big]\int_{v_t(\vtheta_0^v)}^{\infty}\int_{c_t(\vtheta^c)}^{c_t(\vtheta_0^c)}f_t(x,y)\D y \D x\\
		&\hspace{1cm} - \E_{t-1}\Big[\big\{Y_t - c_t(\vtheta_0^c)\big\} \1_{\{X_t>v_t(\vtheta_0^v),\ c_t(\vtheta_0^c)<Y_t\leq c_t(\vtheta^c)\}}\Big]\\
		&\hspace{1cm}+ \E_{t-1}\Big[\big\{Y_t - c_t(\vtheta_0^c)\big\} \1_{\{X_t>v_t(\vtheta_0^v),\ c_t(\vtheta^c)<Y_t\leq c_t(\vtheta_0^c)\}}\Big]\\
		&= \1_{\{c_t(\vtheta^c)>c_t(\vtheta_0^c)\}} \big[c_t(\vtheta^c)-c_t(\vtheta_0^c)\big]\int_{v_t(\vtheta_0^v)}^{\infty}\int_{c_t(\vtheta_0^c)}^{c_t(\vtheta^c)}f_t(x,y)\D y \D x\\
		&\hspace{1cm}- \1_{\{c_t(\vtheta^c)\leq c_t(\vtheta_0^c)\}} \big[c_t(\vtheta^c)-c_t(\vtheta_0^c)\big]\int_{v_t(\vtheta_0^v)}^{\infty}\int_{c_t(\vtheta^c)}^{c_t(\vtheta_0^c)}f_t(x,y)\D y \D x\\
		&\hspace{1cm} - \1_{\{c_t(\vtheta_0^c) < c_t(\vtheta^c)\}}\int_{v_t(\vtheta_0^v)}^{\infty}\int_{c_t(\vtheta_0^c)}^{c_t(\vtheta^c)}\big[y-c_t(\vtheta_0^c)\big]f_t(x,y)\D y \D x\\
		&\hspace{1cm} + \1_{\{c_t(\vtheta_0^c) > c_t(\vtheta^c)\}}\int_{v_t(\vtheta_0^v)}^{\infty}\int_{c_t(\vtheta^c)}^{c_t(\vtheta_0^c)}\big[y-c_t(\vtheta_0^c)\big]f_t(x,y)\D y \D x\\
		&= \1_{\{c_t(\vtheta_0^c) < c_t(\vtheta^c)\}}\int_{v_t(\vtheta_0^v)}^{\infty}\int_{c_t(\vtheta_0^c)}^{c_t(\vtheta^c)}\big[c_t(\vtheta^c)-y\big]f_t(x,y)\D y \D x\\
		&\hspace{1cm} + \1_{\{c_t(\vtheta_0^c) > c_t(\vtheta^c)\}}\int_{v_t(\vtheta_0^v)}^{\infty}\int_{c_t(\vtheta^c)}^{c_t(\vtheta_0^c)}\big[y-c_t(\vtheta^c)\big]f_t(x,y)\D y \D x.
	\end{align*}
	By similar steps as in the derivation leading to \eqref{eq:plug in}, we have that for any $\tau>0$,
	\begin{align}
		\E_{t-1}\big[\Delta S^{\CoVaR}(\vtheta^c)\big] &\geq \1_{\{c_t(\vtheta_0^c)+\tau < c_t(\vtheta^c)\}}\int_{v_t(\vtheta_0^v)}^{\infty}\int_{c_t(\vtheta_0^c)}^{c_t(\vtheta_0^c)+\tau}\big[\tau + \{c_t(\vtheta_0^c)-y\}\big]f_t(x,y)\D y \D x\notag\\
		&\hspace{1cm} + \1_{\{c_t(\vtheta_0^c) > c_t(\vtheta^c)+\tau\}}\int_{v_t(\vtheta_0^v)}^{\infty}\int_{c_t(\vtheta_0^c)-\tau}^{c_t(\vtheta_0^c)}\big[\tau - \{c_t(\vtheta_0^c)-y\}\big]f_t(x,y)\D y \D x\notag\\
		&=\1_{\{c_t(\vtheta_0^c)+\tau < c_t(\vtheta^c)\}}\int_{c_t(\vtheta_0^c)}^{c_t(\vtheta_0^c)+\tau}\big[\tau + \{c_t(\vtheta_0^c)-y\}\big]\bigg(\int_{v_t(\vtheta_0^v)}^{\infty} f_t(x,y)\D x\bigg)\D y \notag\\
		&\hspace{1cm} + \1_{\{c_t(\vtheta_0^c) > c_t(\vtheta^c)+\tau\}}\int_{c_t(\vtheta_0^c)-\tau}^{c_t(\vtheta_0^c)}\big[\tau - \{c_t(\vtheta_0^c)-y\}\big]\bigg(\int_{v_t(\vtheta_0^v)}^{\infty}f_t(x,y)\D x \bigg)\D y.\label{eq:(p.2.2)}
	\end{align}
	Then, for any $\lambda\in\mathbb{R}$,
	\begin{align*}
		\int_{v_t(\vtheta_0^v)}^{\infty}&f_t\big(x,c_t(\vtheta_0^c)+\lambda\big)\D x \\
		&= \int_{v_t(\vtheta_0^v)}^{\infty}f_t\big(x,c_t(\vtheta_0^c)\big)\D x´+ \bigg[\int_{v_t(\vtheta_0^v)}^{\infty}f_t\big(x,c_t(\vtheta_0^c)+\lambda\big)\D x - \int_{v_t(\vtheta_0^v)}^{\infty}f_t\big(x,c_t(\vtheta_0^c)\big)\D x\bigg] \\
		&\geq \int_{v_t(\vtheta_0^v)}^{\infty}f_t\big(x,c_t(\vtheta_0^c)\big)\D x-\bigg|\int_{v_t(\vtheta_0^v)}^{\infty}f_t\big(x,c_t(\vtheta_0^c)+\lambda\big)\D x - \int_{v_t(\vtheta_0^v)}^{\infty}f_t\big(x,c_t(\vtheta_0^c)\big)\D x\bigg|.
	\end{align*}
	Now,
	\begin{align*}
		\bigg|\int_{v_t(\vtheta_0^v)}^{\infty}&f_t\big(x,c_t(\vtheta_0^c)+\lambda\big)\D x - \int_{v_t(\vtheta_0^v)}^{\infty}f_t\big(x,c_t(\vtheta_0^c)\big)\D x\bigg|\\
		&=\bigg|\int_{-\infty}^{\infty}f_t\big(x,c_t(\vtheta_0^c)+\lambda\big)\D x - \int_{-\infty}^{v_t(\vtheta_0^v)}f_t\big(x,c_t(\vtheta_0^c)+\lambda\big)\D x\\
		&\hspace{2cm}- \int_{-\infty}^{\infty}f_t\big(x,c_t(\vtheta_0^c)\big)\D x + \int_{-\infty}^{v_t(\vtheta_0^v)}f_t\big(x,c_t(\vtheta_0^c)\big)\D x\bigg|\\
		&\leq \bigg|\int_{-\infty}^{\infty}f_t\big(x,c_t(\vtheta_0^c)+\lambda\big)\D x - \int_{-\infty}^{\infty}f_t\big(x,c_t(\vtheta_0^c)\big)\D x\bigg|\\
		&\hspace{2cm} + \bigg|\int_{-\infty}^{v_t(\vtheta_0^v)}f_t\big(x,c_t(\vtheta_0^c)+\lambda\big)\D x - \int_{-\infty}^{v_t(\vtheta_0^v)}f_t\big(x,c_t(\vtheta_0^c)\big)\D x\bigg|\\
		&= \Big|f_t^{Y}\big(c_t(\vtheta_0^c)+\lambda\big) - f_t^{Y}\big(c_t(\vtheta_0^c)\big)\Big|\\
		&\hspace{2cm} + \Big|\partial_2 F_t\big(v_t(\vtheta_0^v), c_t(\vtheta_0^c)+\lambda\big) - \partial_2 F_t\big(v_t(\vtheta_0^v), c_t(\vtheta_0^c)\big)\Big|\\
		&\leq K|\lambda| + K|\lambda|=2K|\lambda|,
	\end{align*}
	where the last inequality follows from Assumption~\ref{ass:cons}~\ref{it:Lipschitz cons}.
	Using this and Assumption~\ref{ass:cons}~\ref{it:cond dens}, it holds on the set $\big\{\int_{v_t(\vtheta_0^v)}^{\infty}f_t\big(x,c_t(\vtheta_0^c)\big)\D x>f_2\big\}$ that there exists some sufficiently small $\overline{\lambda}>0$, such that
	\[
	\int_{v_t(\vtheta_0^v)}^{\infty}f_t\big(x,c_t(\vtheta_0^c)+\lambda\big)\D x\geq \frac{f_2}{2}>0
	\]
	for all $|\lambda|\leq\overline{\lambda}$ and all $t\in\mathbb{N}$.
	Plugging this into \eqref{eq:(p.2.2)}, we get for $\tau\in(0,\overline{\lambda})$ that
	\begin{align*}
		\E_{t-1}\big[\Delta S^{\CoVaR}(\vtheta^c)\big] & =\E_{t-1}\Big[\Delta S^{\CoVaR}(\vtheta^c)\big(\1_{\{\int_{v_t(\vtheta_0^v)}^{\infty}f_t(x,c_t(\vtheta_0^c))\D x>f_2\}} + \1_{\{\int_{v_t(\vtheta_0^v)}^{\infty}f_t(x,c_t(\vtheta_0^c))\D x\leq f_2\}}\big)\Big]\\	
		&\geq \E_{t-1}\big[\Delta S^{\CoVaR}(\vtheta^c)\big]\1_{\{\int_{v_t(\vtheta_0^v)}^{\infty}f_t(x,c_t(\vtheta_0^c))\D x>f_2\}} \\
		&\geq \bigg(\1_{\{c_t(\vtheta_0^c)+\tau < c_t(\vtheta^c)\}}\int_{c_t(\vtheta_0^c)}^{c_t(\vtheta_0^c)+\tau}\big[\tau + \{c_t(\vtheta_0^c)-y\}\big]\frac{f_2}{2}\D y \\
		&\hspace{1cm} + \1_{\{c_t(\vtheta_0^c) > c_t(\vtheta^c)+\tau\}}\int_{c_t(\vtheta_0^c)-\tau}^{c_t(\vtheta_0^c)}\big[\tau - \{c_t(\vtheta_0^c)-y\}\big]\frac{f_2}{2}\D y\bigg)\times\\
		&\hspace{8cm} \times\1_{\{\int_{v_t(\vtheta_0^v)}^{\infty}f_t(x,c_t(\vtheta_0^c))\D x>f_2\}} \\
		&= \bigg(\1_{\{c_t(\vtheta_0^c)+\tau < c_t(\vtheta^c)\}}\frac{1}{4}\tau^2 f_2 + \1_{\{c_t(\vtheta_0^c)>c_t(\vtheta^c)+\tau\}}\frac{1}{4}\tau^2 f_2\bigg) \1_{\{\int_{v_t(\vtheta_0^v)}^{\infty}f_t(x,c_t(\vtheta_0^c))\D x>f_2\}}\\
		&= \1_{\{|c_t(\vtheta^c)-c_t(\vtheta_0^c)|>\tau,\ \int_{v_t(\vtheta_0^v)}^{\infty}f_t(x,c_t(\vtheta_0^c))\D x>f_2\}}\frac{1}{4}\tau^2 f_2.
	\end{align*}
	Taking unconditional expectations gives 
	\begin{align*}
		&\E\big[\Delta S^{\CoVaR}(\vtheta^c)\big] \\
		&\geq \frac{1}{4}\tau^2 f_2\P\big\{|c_t(\vtheta^c) - c_t(\vtheta_0^c)|>\tau,\ \int_{v_t(\vtheta_0^v)}^{\infty}f_t(x,c_t(\vtheta_0^c))\D x>f_2\big\}\\
		&=\frac{1}{4}\tau^2 f_2 \P\Big\{\int_{v_t(\vtheta_0^v)}^{\infty}f_t\big(x,c_t(\vtheta_0^c)\big)\D x>f_2\Big\} \P\Big\{\big|c_t(\vtheta^c) - c_t(\vtheta_0^c)\big|>\tau\,\Big\vert\, \int_{v_t(\vtheta_0^v)}^{\infty}f_t\big(x,c_t(\vtheta_0^c)\big)\D x>f_2\Big\}\\
		&> \frac{1}{4}\tau^2 f_2 p_2 \P\Big\{\big|c_t(\vtheta^c) - c_t(\vtheta_0^c)\big|>\tau\,\Big\vert\, \int_{v_t(\vtheta_0^v)}^{\infty}f_t\big(x,c_t(\vtheta_0^c)\big)\D x>f_2\Big\},
	\end{align*}
	where the final step follows from Assumption~\ref{ass:cons}~\ref{it:cond dens}.
\end{proof}

\section{Misspecified Model Initialization}
\label{sec:StartingValues}

Theorem~\ref{thm:cons} in the main paper shows consistency of the two-step M-estimator under Assumption~\ref{ass:cons}.
Consistent with the related literature proposing estimators based on non-smooth loss functions for dynamic models \citep{EM04,WhiteKimManganelli2015,PZC19,Catania2022}, this consistency result is achieved under correct model specification in \eqref{eqn:TrueModelParameters}.
This also implies correct specification at time $t=1$, i.e., correct use of the initialization information $\init$.
However, Theorem~\ref{thm:cons} does not explicitly treat the case of misspecified initial values.
Such a case may arise naturally (but not exclusively) when the initialization information $\init$ includes the infinite past, because then correctly specified starting values are in general impossible to compute in finite time.
Hence, one may have to work with incorrect starting values.

To make this more concrete, suppose that the information set $\mathcal{F}_{t-1}$ (relevant to the forecasting problem at hand) includes the initialization information $\init =(X_0,Y_0,\mZ_0,X_{-1},Y_{-1},\mZ_{-1},\ldots)^\prime$ in \eqref{eqn:GeneralModel}.
Then, interest centers on VaR and CoVaR conditioned on the infinite past with $\mathcal{F}_{t-1}=\mathcal{F}_{t-1}^{\full} = \sigma\big((X_{t-1},Y_{t-1})^\prime,\mZ_{t-1}, (X_{t-2},Y_{t-2})^\prime,\mZ_{t-2}, \ldots \big)$, and the relevant quantities are $\VaR_{t,\b}=\VaR_{\b}(F_{X_t\mid\mathcal{F}_{t-1}^{\full}})$ and $\CoVaR_{t,\a\mid\b}=\CoVaR_{\a\mid\b}(F_{X_t,Y_t\mid\mathcal{F}_{t-1}^{\full}})$.
%
%
Of course, with the above infinite sequence of initialization information $\init$, it becomes infeasible to practically compute the model values $\big((v_1(\vtheta^v), c_1(\vtheta^c)\big)^\prime$,$\ldots,$ $\big(v_n(\vtheta^v), c_n(\vtheta^c)\big)^\prime$ required for the M-estimators in \eqref{eqn:MestVaR} and \eqref{eqn:MestCoVaR} when only truncated observations, e.g., starting at time $t=1$, are available.

However, we show in Theorem \ref{prop:StartingValues} below that for models satisfying a \emph{contraction condition}, consistent parameter estimation still remains feasible with arbitrary (and, hence, possibly misspecified) initial values $\tilde v_1(\vtheta^v)$ and  $\tilde c_1(\vtheta^c)$ for time $t=1$ that may or may not depend on $\vtheta^v$ and $\vtheta^c$.
Hence, estimation must be carried out based on the truncated information set $\widetilde{\mathcal{F}}_{t-1}=\sigma\big((X_{t-1},Y_{t-1})^\prime, \mZ_{t-1}, \ldots,(X_{1},Y_{1})^\prime,  \mZ_{1}, \tilde v_1(\vtheta^v), \tilde c_1(\vtheta^c) \big)$.
Corresponding to the generally incorrect initial values, we denote by $\tilde v_t(\vtheta^v)$ and $\tilde c_t(\vtheta^c)$  for $t \ge 2$ models that are initialized with $\tilde v_{1}(\vtheta^v)$ and $\tilde c_{1}(\vtheta^c)$.
The respective M-estimators are denoted by $\widetilde{\vtheta}_{n}^{v}$ and $\widetilde{\vtheta}_{n}^{c}$, which are as in \eqref{eqn:MestVaR} and \eqref{eqn:MestCoVaR}, but are based on the model values $\tilde v_{t}(\vtheta^v)$ and $\tilde c_{t}(\vtheta^c)$ and, thus, on the initializations $\tilde v_{1}(\vtheta^v)$ and $\tilde c_{1}(\vtheta^c)$.


\setcounter{thm}{1}

\begin{thm}
	\label{prop:StartingValues}
	Let Assumption~\ref{ass:cons} hold and suppose that $\E\big[ \sup_{\vtheta^v \in \mTheta^v} | v_1(\vtheta^v)- \tilde v_1(\vtheta^v)| \big] < \infty$ and $\E\big[\sup_{\vtheta^c \in \mTheta^c}|  c_{1}(\vtheta^c) - \tilde c_{1}(\vtheta^c) |\big] < \infty$.
	Moreover, assume that the \emph{contraction condition} holds, i.e., for some $\rho < 1$, we have 
	$|v_t(\vtheta^v)- \tilde v_{t}(\vtheta^v)| < \rho |v_{t-1}(\vtheta^v)- \tilde v_{t-1}(\vtheta^v)|$ and $|c_t(\vtheta^c)- \tilde c_{t}(\vtheta^c)| < \rho |c_{t-1}(\vtheta^c)- \tilde c_{t-1}(\vtheta^c)|$ for all $(\vtheta^{v\prime},\vtheta^{c\prime})^\prime\in\mTheta$, and all $t \ge 2$.
	Then, as $n\to\infty$, $\widetilde{\vtheta}_{n}^{v}\overset{\P}{\longrightarrow}\vtheta_0^{v}$ and $\widetilde{\vtheta}_{n}^{c}\overset{\P}{\longrightarrow}\vtheta_0^{c}$.
\end{thm}

\begin{proof}
	See Appendix~\ref{sec:prop1}.
\end{proof}

The proof of Theorem~\ref{prop:StartingValues} draws on the basic consistency proof of Theorem~\ref{thm:cons} and additionally bounds the difference between the objective functions evaluated at the estimators $\widetilde{\vtheta}_{n}^{v}$ and  $\widetilde{\vtheta}_{n}^{c}$ based on misspecified starting information and the practically infeasible estimators $\widehat{\vtheta}_{n}^{v}$ and  $\widehat{\vtheta}_{n}^{c}$.
The proof is complicated by the \emph{discontinuity} of the CoVaR scoring function in \eqref{eqn:MestVaR}, which contrasts with the GARCH literature in \citet{Straumann2006} and \citet{Blasques2018}.

\setcounter{rem}{2}
\begin{rem}
	In the GARCH literature, a model initialization in the infinite past is often imposed, such that the dynamic iterations as in \eqref{eqn:SAVCoCAViaRModelClass} hold for all $t \in \mathbb{Z}$ \citep{Francq2004,FZ10,Blasques2018,Hog25}.
	On a theoretical level, this often guarantees stationarity of the model and sidesteps the explicit choice of model initialization.
	This practice---while perhaps useful as an approximation---is empirically questionable since all economic and financial time series evidently started in the finite past.
\end{rem}

As Proposition~\ref{prop:ARMAclass} below shows, the additional \emph{contraction condition} in Theorem~\ref{prop:StartingValues} is for example satisfied by the following ARMA-type model class for the pair (VaR, CoVaR):
\begin{align}
	\label{eqn:LinearModels}
	\begin{pmatrix} v_t(\vtheta^v) \\ c_t(\vtheta^c) \end{pmatrix}
	= \vomega + 
	{\boldsymbol{A}} \mathbf{W}_{t-1} +
	\mB \begin{pmatrix} v_{t-1}(\vtheta^v) \\ c_{t-1}(\vtheta^c) \end{pmatrix}, \qquad t \in \mathbb{Z},
\end{align}
for some $2 \times k$ parameter matrix $\boldsymbol{A}$, and a $k \times 2$ matrix $\mathbf{W}_{t-1}$ of $\mathcal{F}_{t-1}^{\full}$-measurable observables. 
The parameters $\vomega \in \mathbb{R}^2$, $\mB \in  \mathbb{R}^{2 \times 2}$ and ${\boldsymbol{A}}$ are collected in $\vtheta^{v}$ and $\vtheta^{c}$.
As under our modeling framework the VaR (CoVaR) equation does not depend on $\vtheta^c$ ($\vtheta^v$), the matrix $\mB$ must be diagonal.
These models nest the ones in \eqref{eqn:SAVCoCAViaRModelClass} and \eqref{eqn:ASCoCAViaRModelClass} below, such that they encompass all models used in our simulations and empirical application.

\begin{prop}
	\label{prop:ARMAclass}
	Consider a model of the form \eqref{eqn:LinearModels} with spectral radius of the diagonal matrix $\mB$ strictly smaller than one. 
	Then, the contraction condition of Theorem~\ref{prop:StartingValues} holds.
\end{prop}

\begin{proof}
	See Appendix~\ref{sec:prop1}.
\end{proof}

Together with Theorem~\ref{prop:StartingValues}, this proposition shows that the standard assumption of a spectral radius smaller one of the coefficient matrix for the lagged model values---well known, e.g., from vector autoregressive conditional mean models---is sufficient for consistency of the arbitrarily initialized M-estimators for the ARMA-type model class given in  \eqref{eqn:LinearModels}.
Models that are non-linear in their lagged model values often require a more detailed treatment as in, e.g., \citet{Straumann2006} and \citet{Blasques2018}. In our case, such an analysis would be further complicated by the discontinuity of the loss function.
Therefore, the derivation of asymptotic results for $\widetilde{\vtheta}_{n}^{v}$ and $\widetilde{\vtheta}_{n}^{c}$ in non-linear models is beyond the scope of this paper.

We emphasize that the result of Theorem \ref{prop:StartingValues} is mainly of interest from a theoretical perspective. 
While it is interesting that the true parameters of model \eqref{eqn:LinearModels} can be recovered from a finite stretch of data, for practical forecasting purposes we still have to plug in the estimates $\widetilde{\vtheta}_n^v$ and $\widetilde{\vtheta}_n^c$ in the ``false'' truncated recursions for computing the practically relevant out-of-sample forecasts $\tilde{v}_{n+1}(\cdot)$ and $\tilde{c}_{n+1}(\cdot)$.
Also, models initialized in the infinite past may not be seen as realistic, since the economic and financial time series we have in mind for applications all originated in the finite past.

From our experience gained in the simulations, the empirical application and from results comparing two different initialization methods in Section~\ref{sec:Initializations}, we find that the assumed initialization only has a negligible effect on the model parameter estimates and the CoVaR forecasts of interest in sufficiently large samples.
So for practical purposes, one may assume the simple and practically convenient initialization $\big(v_1(\vtheta^v), c_1(\vtheta^c)\big)^\prime=\vomega$ without any loss in forecast accuracy.

\section{Asymptotic Normality}
\label{Asymptotic Normality}

To show asymptotic normality of our estimators $\widehat{\vtheta}_n^v$ and $\widehat{\vtheta}_n^c$  in \eqref{eqn:MestVaR} and \eqref{eqn:MestCoVaR}, we have to impose additional regularity conditions. 
We denote the partial derivative of $F_t(\cdot,\cdot)$ with respect to the $i$-th argument by $\partial_i F_t(\cdot,\cdot)$ for $i=1,2$. 
Furthermore, define the derivative of the VaR score with respect to $\vtheta^v$ as 
\[
\Vv_t(\vtheta^v) = \nabla v_t(\vtheta^v)\big[\1_{\{X_t\leq v_t(\vtheta^v)\}}-\beta\big]\overset{X_t\neq v_t(\vtheta^v)}{=} \frac{\partial}{\partial \vtheta^v}S^{\VaR}(v_t(\vtheta^v),X_t))
\]
and the derivative of the CoVaR score with respect to $\vtheta^c$ as
\[
\Cc_t(\vtheta^c,\vtheta^v) = \1_{\{X_t> v_t(\vtheta^v)\}}\nabla c_t(\vtheta^c)\big[\1_{\{Y_t\leq c_t(\vtheta^c)\}}-\alpha\big] \overset{Y_t\neq c_t(\vtheta^c)}{=} \frac{\partial}{\partial \vtheta^c}S^{\CoVaR}\big((v_t(\vtheta^v),c_t(\vtheta^c))^\prime, (X_t,Y_t)^\prime\big);
\]
for details see \eqref{eq:gt} and \eqref{eq:(12+)} in Section~\ref{sec:thm2}.

\setcounter{assumption}{1}

\begin{assumption}\label{ass:an}
	$ $ \vspace{-0.3cm}
	\renewcommand{\theenumi}{(\roman{enumi})}
	\begin{enumerate}
		\item\label{it:int} $\vtheta_0^v\in\inter(\mTheta^v)$ and $\vtheta_0^c\in\inter(\mTheta^c)$, where $\inter(\cdot)$ denotes the interior of a set.
		
		\item\label{it:diff} For all $t\in\mathbb{N}$, $v_t(\cdot)$ and $c_t(\cdot)$ are a.s.\ twice continuously differentiable on $\inter(\mTheta^v)$ and $\inter(\mTheta^c)$, respectively. Moreover, $\nabla v_t(\cdot)\neq\vzero$ and $\nabla c_t(\cdot)\neq\vzero$ a.s.
		
		\item\label{it:bound2.1} There exists a neighborhood of $\vtheta_0^v$, such that $\big\Vert\nabla^2v_t(\vtheta^v)\big\Vert\leq V_2(\mathcal{F}_{t-1})$, $\big\Vert\nabla^2 v_t(\vtheta^v) - \nabla^2 v_t(\vtau^v)\big\Vert\leq V_3(\mathcal{F}_{t-1})\big\Vert\vtheta^v-\vtau^v\big\Vert$ for all elements $\vtheta^v$, $\vtau^v$ of that neighborhood and all $t\in\mathbb{N}$.
		
		\item\label{it:bound2.2} There exists a neighborhood of $\vtheta_0^c$, such that $\big\Vert\nabla c_t(\vtheta^c)\big\Vert\leq C_1(\mathcal{F}_{t-1})$, $\big\Vert\nabla^2 c_t(\vtheta^c)\big\Vert\leq C_2(\mathcal{F}_{t-1})$, $\big\Vert\nabla^2 c_t(\vtheta^c) - \nabla^2 c_t(\vtau^c)\big\Vert\leq C_3(\mathcal{F}_{t-1})\big\Vert\vtheta^c-\vtau^c\big\Vert$ for all elements $\vtheta^c$, $\vtau^c$ of that neighborhood and all $t\in\mathbb{N}$.
		
		\item\label{it:mom bounds cons2} For all $t\in\mathbb{N}$ it holds that $\E\big[V_1^{3+\iota}(\mathcal{F}_{t-1})\big]\leq K$, $\E\big[V_2^{(3+\iota)/2}(\mathcal{F}_{t-1})\big]\leq K$, $\E\big[V_3(\mathcal{F}_{t-1})\big]\leq K$, $\E\big[C_1^{3+\iota}(\mathcal{F}_{t-1})\big]\leq K$, $\E\big[C_2^{(3+\iota)/2}(\mathcal{F}_{t-1})\big]\leq K$, $\E\big[C_{3}(\mathcal{F}_{t-1})\big]\leq K$ for some $\iota>0$.
		
		\item\label{it:Lipschitz an} For all $t\in\mathbb{N}$, $\big|\partial_1 F_t(x,y)-\partial_1 F_t(x^\prime,y)\big|\leq K|x-x^\prime|$ for all $x,x^\prime,y,y^\prime\in\mathbb{R}$.
		
		\item\label{it:bound cdf} For all $t\in\mathbb{N}$, $f_t^{Y}(\cdot)\leq K$, $f_t(\cdot,\cdot)\leq K$, $\big|\partial_i F_t(\cdot,\cdot) \big|\leq K$ for $i=1,2$.
		
		\item\label{it:pd} The matrices
		\begin{align*}
			\mV_{n}&=\frac{1}{n}\sum_{t=1}^{n}\beta(1-\beta)\E\big[\nabla v_t(\vtheta_0^v)\nabla^\prime v_t(\vtheta_0^v)\big]\in\mathbb{R}^{p\times p},\\
			\mC_{n}&=\frac{1}{n}\sum_{t=1}^{n}\alpha(1-\alpha)(1-\beta)\E\big[\nabla c_t(\vtheta_0^c)\nabla^\prime c_t(\vtheta_0^c)\big]\in\mathbb{R}^{q\times q},\\
			\mLambda_n & =\frac{1}{n}\sum_{t=1}^{n}\E\Big[f_t^{X}\big(v_t(\vtheta_0^v)\big)\nabla v_t(\vtheta_0^v)\nabla^\prime v_t(\vtheta_0^v)\Big]\in\mathbb{R}^{p\times p},\\
			\mLambda_{n,(1)} &= \frac{1}{n}\sum_{t=1}^{n}\E\bigg\{ \nabla c_t(\vtheta_0^c)\nabla^\prime c_t(\vtheta_0^c) \Big[  f_t^{Y}\big(c_t(\vtheta_0^c)\big) - \partial_2 F_{t}\big(v_t(\vtheta_0^v), c_t(\vtheta_0^c)\big)\Big] \bigg\}\in\mathbb{R}^{q\times q}
			\intertext{are uniformly positive definite, and}
			\mLambda_{n,(2)} &= \frac{1}{n}\sum_{t=1}^{n}\E\bigg\{\nabla c_t(\vtheta_0^c)\nabla^\prime v_t(\vtheta_0^v)\Big[\alpha f_t^X\big(v_t(\vtheta_0^v)\big) - \partial_1 F_t\big(v_t(\vtheta_0^v), c_t(\vtheta_0^c)\big)\Big]\bigg\}\in\mathbb{R}^{q\times p}
		\end{align*}
		has uniformly full rank.
		
		\item\label{it:lambda func} For the function $\vlambda_n(\vtheta^c,\vtheta^v)=\frac{1}{n}\sum_{t=1}^{n}\E\big[\Cc_t(\vtheta^c,\vtheta^v)\big]$ it holds for sufficiently small $\varepsilon>0$ that if $\big\Vert\vtheta^v-\vtheta_0^v\big\Vert\leq\varepsilon$ and $\vtheta^v$, $\vtheta^c=\vtheta^c(\vtheta^v)$ satisfy $\vlambda_n(\vtheta^c,\vtheta^v)=\vzero$, then $\big\Vert\vtheta^c-\vtheta_0^c\big\Vert\leq K\big\Vert\vtheta^v-\vtheta_0^v\big\Vert$.
		
		\item\label{it:eq bound} $\sup_{\vtheta^v\in\mTheta^v}\sum_{t=1}^{n}\1_{\{X_t=v_t(\vtheta^v)\}}=O(1)$ and $\sup_{\vtheta^c\in\mTheta^c}\sum_{t=1}^{n}\1_{\{Y_t=c_t(\vtheta^c)\}}=O(1)$ a.s., as $n\to\infty$.
		
		\item\label{it:mixing} $\Big\{\big(X_t,\ Y_t,\ \mZ_t^\prime,\ v_t(\vtheta_0^v),\ \nabla^\prime v_t(\vtheta_0^v),\ c_t(\vtheta_0^c),\ \nabla^\prime c_t(\vtheta_0^c)\big)^\prime\Big\}_{t\in\mathbb{N}}$ is $\alpha$-mixing with mixing coefficients $\alpha(\cdot)$ satisfying $\sum_{m=1}^{\infty}\alpha^{(\tilde{q}-2)/\tilde{q}}(m)<\infty$ for some $\tilde{q}>2$.
		
	\end{enumerate}
\end{assumption}

Item~\ref{it:int} of Assumption~\ref{ass:an} is essential for the asymptotic normality of extremum estimators. In fact, examples of non-normal estimators can easily be constructed when the true parameter is on the boundary \citep[p.~2144]{NeweyMcFadden1994}. Items~\ref{it:diff}--\ref{it:mom bounds cons2} are again smoothness conditions on the model and moment conditions that are similar in spirit to those of \citet{EM04} and \citet{PZC19}. 
Items~\ref{it:Lipschitz an}--\ref{it:bound cdf} are similar to Assumption~\ref{ass:cons}~\ref{it:Lipschitz cons} and \ref{it:mom bounds cons}.
The assumptions on the matrices in item~\ref{it:pd} are standard in regression contexts with heterogeneously distributed observations; see \citet{Whi01} and in particular his Definition 2.17.
Item~\ref{it:lambda func} is a smoothness assumption on the expectations of the CoVaR score derivatives.
It is very similar in spirit to Assumption~3~(a) enforced by \citet{Zin02} for her two-step smoothed least median of squares estimator.
Item~\ref{it:eq bound} bounds the number of exact equalities of financial losses and all possible (VaR, CoVaR) model values; see also \citet[Assumption~2~(G)]{PZC19}. 
Finally, item~\ref{it:mixing} is a standard mixing condition that ensures suitable central limit theorems hold; for instance those in Lemmas~V.\ref{lem:7} and C.\ref{lem:7 tilde}. We refer to, e.g., \citet[Definition~3.42]{Whi01} for a definition of mixing that also covers possibly non-stationary sequences.

\begin{thm}\label{thm:an}
	Suppose Assumptions~\ref{ass:cons} and \ref{ass:an} hold. Then, as $n\to\infty$, 
	\begin{align*}
		\sqrt{n}\mM_{n}^{-1/2}\overline{\mGamma}_n\begin{pmatrix} \widehat{\vtheta}_n^v - \vtheta_0^v\\  \widehat{\vtheta}_n^c - \vtheta_0^c\end{pmatrix}\overset{d}{\longrightarrow}N(\vzero,\mI),
	\end{align*}
	where 
	\begin{equation*}
		\overline{\mGamma}_n  =\begin{pmatrix}\mLambda_n &\vzeros \\ \mLambda_{n,(2)} & \mLambda_{n,(1)} \end{pmatrix}\in\mathbb{R}^{(p+q)\times(p+q)},\qquad	\mM_{n}	=\begin{pmatrix}\mV_n &\vzeros \\ \vzeros & \mC_{n}\end{pmatrix}\in\mathbb{R}^{(p+q)\times(p+q)}.
	\end{equation*}
\end{thm}

\begin{proof}
	See Appendix~\ref{sec:thm2}.
\end{proof}

The joint asymptotic normality of Theorem~\ref{thm:an} may be useful when testing joint restrictions on the parameters, such as the significance of some (macroeconomic or financial) variables in explaining VaR \textit{and} CoVaR dynamics. Note that $\mM_{n}^{-1}$ exists and is positive definite (because $\mM_n$ is positive definite by Assumption~\ref{ass:an}~\ref{it:pd} and \citet[Section A.8.3]{Luetkepohl:05}), such that $\mM_{n}^{-1/2}$ exists \citep[Section A.9.2]{Luetkepohl:05}.

Similarly as in \citet{EM04} and \citet{PZC19}, the key step in the proof of Theorem~\ref{thm:an} is to apply Lemma~A.1 of \citet{Wei91}. However, since we consider a two-step estimator, our arguments necessarily extend those of the aforementioned authors, who only consider one-step estimators.
Specifically, in showing asymptotic normality of $\widehat{\vtheta}_{n}^{c}$ one has to take into account the fact that it depends on the first-step estimate $\widehat{\vtheta}_{n}^{v}$, which complicates the technical treatment.

One can also show that the first-step estimator $\widehat{\vtheta}_{n}^{v}$ increases the asymptotic variance of $\widehat{\vtheta}_{n}^{c}$ relative to the case where $\vtheta_0^v$ is known.
To see this, define $\mGamma_n^{-1} =\Big(\begin{array}{c|c}
	-\mLambda_{n,(1)}^{-1}\mLambda_{n,(2)}\mLambda_n^{-1} & \mLambda_{n,(1)}^{-1}	\end{array}\Big)\in\mathbb{R}^{q\times (p+q)}$ to be the lower $q$ rows of $\overline{\mGamma}_n^{-1}$; see \eqref{eq:(C.6p)}. 
Then, our proof shows that if the true value of $\vtheta_0^v$ were known in the second step, the asymptotic variance of $\widehat{\vtheta}_n^c$ would be given by $\mLambda_{n,(1)}^{-1}\mC_n \mLambda_{n,(1)}^{-1}$. Comparing this with the asymptotic variance $\mGamma_n^{-1} \mM_n (\mGamma_n^{-1})^\prime$ from Theorem~\ref{thm:an} demonstrates that, as expected, the first-stage estimation has an asymptotic effect on $\widehat{\vtheta}_n^c$. 
More precisely, the first-step estimator increases the asymptotic variance of $\widehat{\vtheta}_{n}^{c}$ (in terms of the Loewner order), because
\[
\mGamma_n^{-1} \mM_n (\mGamma_n^{-1})^\prime - \mLambda_{n,(1)}^{-1}\mC_n \mLambda_{n,(1)}^{-1} = (\mLambda_{n,(1)}^{-1}\mLambda_{n,(2)}\mLambda_n^{-1}) \mV_n\big(\mLambda_{n,(1)}^{-1}\mLambda_{n,(2)}\mLambda_n^{-1}\big)^\prime
\]
is positive semi-definite. 

This is akin to two-step GMM estimation, where uncorrelated first- and second-step moment conditions imply an increased variance of the second-step estimator; see \citet[Eq.~(6.9)]{NeweyMcFadden1994}. 
In our context of two-step M-estimation, the derivatives (with respect to the parameters) of the two components of the scoring functions in \eqref{eq:loss}, i.e., $\Vv_t(\vtheta_0^v)$ and $\Cc_t(\vtheta_0^c,\vtheta_0^v)$, play the role of these moment conditions.
The uncorrelatedness of $\Vv_t(\vtheta_0^v)$ and $\Cc_t(\vtheta_0^c,\vtheta_0^v)$ is the reason for the off-diagonal zero blocks in the matrix $\mM_n$; see in particular \eqref{eq:(C.4+)}.

\setcounter{example}{1}
\begin{example}\label{ex:2}
	Although our models will mainly be used for financial data, there is also a growing interest in macroeconomic tail risks, as demonstrated by the rapidly growing literature on \citeauthor{ABG19}'s \citeyearpar{ABG19} Growth-at-Risk. The standard models in macroeconometrics are VAR($m$) processes
	\begin{equation*}
		\begin{pmatrix}
			X_t\\ Y_t
		\end{pmatrix} =
		\vphi_0 + \mPhi_1\begin{pmatrix}
			X_{t-1}\\ Y_{t-1}
		\end{pmatrix}
		+\ldots+
		\mPhi_m\begin{pmatrix}
			X_{t-m}\\ Y_{t-m}
		\end{pmatrix}
		+\vvarepsilon_t,	
	\end{equation*}
	where $\vphi_0=(\phi_{X,0}, \phi_{Y,0})^\prime$, $\mPhi_j\in\mathbb{R}^{2\times 2}$ ($j=1,\ldots,m$) and the $\vvarepsilon_t=(\varepsilon_{X,t}, \varepsilon_{Y,t})^\prime$ are i.i.d.~and independent of $\mathcal{F}_{t-1}=\sigma\big((X_{t-1},Y_{t-1})^\prime,\ldots,(X_{1},Y_{1})^\prime,\init\big)$, where $\init=(X_0,Y_0,\ldots,X_{-m+1},Y_{-m+1})^\prime$. 
	It is not difficult to show that VAR($m$) models imply the (VaR, CoVaR) dynamics
	\begin{equation}\label{eq:VARp}
		\begin{pmatrix}
			v_t(\vtheta^v)\\ c_t(\vtheta^c)
		\end{pmatrix}
		=\vphi_0^\ast + \mPhi_1\begin{pmatrix}
			X_{t-1}\\ Y_{t-1}
		\end{pmatrix}
		+\ldots+
		\mPhi_m\begin{pmatrix}
			X_{t-m}\\ Y_{t-m}
		\end{pmatrix},
	\end{equation}
	where $\vphi_0^\ast=\big(\phi_{X,0}+\VaR_\b(F_{\varepsilon_X}),\ \phi_{Y,0} + \CoVaR_{\a\mid\b}(F_{\varepsilon_X,\varepsilon_Y})\big)^\prime$; see Section~\ref{PR VAR} for details.
	For $m=1$, this reduces to the model in \eqref{eqn:SAVCoCAViaRModelClass} with $\mB=\vzero$. 
	Note that in the (VaR, CoVaR) model \eqref{eq:VARp} only the observations $(X_0,Y_0)^\prime,\ldots,(X_{-m+1},Y_{-m+1})^\prime$ are required for initialization even though the VAR($m$) process may have originated in the infinite past. 
	Therefore, infinite starting values play no role in linear models such as \eqref{eq:VARp}.
	Indeed, even if the VAR($m$) process originated in the infinite past, and the relevant information set for the forecaster is $\mathcal{F}_{t-1}^{\full} = \sigma\big((X_{t-1},Y_{t-1})^\prime, (X_{t-2},Y_{t-2})^\prime, \ldots \big)$, then it still holds that $v_t(\vtheta_0^v)=\VaR_\b(F_{X_t\mid\mathcal{F}_{t-1}})=\VaR_\b(F_{X_t\mid\mathcal{F}_{t-1}^{\full}})$ and $c_t(\vtheta_0^v)=\CoVaR_{\a|\b}(F_{X_t,Y_t\mid\mathcal{F}_{t-1}})=\CoVaR_{\a|\b}(F_{X_t,Y_t\mid\mathcal{F}_{t-1}^{\full}})$.
	Therefore, the VaR and CoVaR dynamics of a VAR process started infinitely long ago can be recovered from a truncated information set (if $m$ initial values are available).
\end{example}

\begin{rem}
	\label{rem:CondVerification}
	\begin{enumerate}
		\item[(a)] Assumptions~\ref{ass:cons} and \ref{ass:an} contain primitive conditions (such as the differentiability of $v_t(\cdot)$ and $c_t(\cdot)$) and high-level conditions (such as the ULLN for the scores). We verify all items of Assumptions~\ref{ass:cons} and \ref{ass:an} for a model similar to a CCC--GARCH of \citet{Bol90} in Section~\ref{sec:verification}.
		In fact, we consider model \eqref{eqn:ECCCmodel} with diagonal $\widetilde{\mA}$ and $\widetilde{\mB}$ in Section~\ref{sec:verification}.
		
		\item[(b)] Next to non-linear CCC--GARCH-type models, we also verify our Assumptions~\ref{ass:cons} and \ref{ass:an} for linear VAR models in Section~\ref{sec:verificationVAR}. VAR models are widely used in macroeconomics, where \citet{ABG19,Aea21} have recently drawn attention to tail risks by introducing the Growth-at-Risk (i.e., the VaR of GDP growth). While Growth-at-Risk focuses on the VaR as a univariate measure of risk, applying our (VaR, \textit{CoVaR}) models may shed light on the \textit{interconnectedness} of macroeconomic tail risks. Therefore, verifying our assumptions for VAR models, which are well-known to adequately capture the dynamics of many macroeconomic time series, is of particular interest, as it paves the way for future applications of our (VaR, CoVaR) models in macroeconomics.
	\end{enumerate}
\end{rem}

\section{Consistent Estimation of the Asymptotic Variance}\label{avar}

To draw inference on the model parameters, we require consistent estimates of the matrices $\mM_n$ and $\overline{\mGamma}_n$ in Theorem~\ref{thm:an}. For the matrices pertaining to the VaR parameters, these are well-explored. For instance, \citet{EM04} and \citet{PZC19} estimate $\mV_n$ and $\mLambda_n$ via
\begin{align*}
	\widehat{\mV}_n &= \frac{1}{n}\sum_{t=1}^{n}\beta(1-\beta)\nabla v_{t}(\widehat{\vtheta}_n^v)\nabla v_{t}^\prime(\widehat{\vtheta}_n^v),\\
	\widehat{\mLambda}_n &= \frac{1}{n}\sum_{t=1}^{n}(2\widehat{b}_{n,x})^{-1}\1_{\big\{|X_t-v_t(\widehat{\vtheta}_n^v)|<\widehat{b}_{n,x}\big\}}\nabla v_t(\widehat{\vtheta}_n^v)\nabla^\prime v_t(\widehat{\vtheta}_n^v),
\end{align*}
respectively, where $\widehat{b}_{n,x}=o_{\P}(1)$ is a (possibly stochastic) bandwidth. As we are mainly interested in the CoVaR parameters, we still have to estimate $\mLambda_{n,(1)}$, $\mLambda_{n,(2)}$ and $\mC_n$. For the latter matrix, we employ the estimator
\[
\widehat{\mC}_n = \frac{1}{n}\sum_{t=1}^{n}\alpha(1-\alpha)(1-\beta)\nabla c_{t}(\widehat{\vtheta}_n^c)\nabla^\prime c_{t}(\widehat{\vtheta}_n^c).
\]
The estimators we propose for $\mLambda_{n,(1)}$ and $\mLambda_{n,(2)}$ are more involved and are given by
\begin{align*}
	\widehat{\mLambda}_{n,(1)} & = \frac{1}{n}\sum_{t=1}^{n} \nabla c_t(\widehat{\vtheta}_n^c)\nabla^\prime c_t(\widehat{\vtheta}_n^c)  (2\widehat{b}_{n,y})^{-1}\Big[ \1_{\big\{|Y_t-c_t(\widehat{\vtheta}_n^c)|<\widehat{b}_{n,y}\big\}} -\1_{\big\{X_t\leq v_t(\widehat{\vtheta}_n^v),\ |Y_t-c_t(\widehat{\vtheta}_n^c)|<\widehat{b}_{n,y}\big\}} \Big],\\
	\widehat{\mLambda}_{n,(2)} & = \frac{1}{n}\sum_{t=1}^{n} \nabla c_t(\widehat{\vtheta}_n^c)\nabla^\prime v_t(\widehat{\vtheta}_n^v) (2\widehat{b}_{n,x})^{-1}\Big[ \alpha \1_{\big\{|X_t-v_t(\widehat{\vtheta}_n^v)|<\widehat{b}_{n,x}\big\}} - \1_{\big\{|X_t-v_t(\widehat{\vtheta}_n^v)|<\widehat{b}_{n,x},\ Y_t\leq c_t(\widehat{\vtheta}_n^c)\big\}} \Big],
\end{align*}
where $\widehat{b}_{n,y}=o_{\P}(1)$ is another (possibly stochastic) bandwidth.
Remark~\ref{rem:asvar} below reports the data-dependent bandwidth choices we use in the empirical parts of the paper.
These estimators rely on kernel density estimates of the derivatives of the conditional c.d.f., and use a rectangular kernel. However, at the expense of some additional technicality other kernels may also be considered.

To prove consistency of the estimators, we require some further conditions.

\begin{assumption}\label{ass:avar}
	$ $ \vspace{-0.2cm}
	\renewcommand{\theenumi}{(\roman{enumi})}
	\begin{enumerate}
		\item\label{it:bw} It holds for $\widehat{b}_{n,w}$ ($w\in\{x,y\}$) that $\widehat{b}_{n,w}/b_{n,w}\overset{\P}{\longrightarrow}1$, where the non-stochastic $b_{n,w}>0$ satisfies $b_{n,w}=o(1)$ and $b_{n,w}^{-1}=o\big(n^{1/2}\big)$, as $n\to\infty$.
		
		\item\label{it:Lambda} $\frac{1}{n}\sum_{t=1}^{n}f_t^{X}\big(v_t(\vtheta_0^v)\big)\nabla v_t(\vtheta_0^v)\nabla^\prime v_t(\vtheta_0^v)-\mLambda_n \overset{\P}{\longrightarrow}\vzero$.
		
		\item\label{it:Lambda1} $\frac{1}{n}\sum_{t=1}^{n}\nabla c_t(\vtheta_0^c)\nabla^\prime c_t(\vtheta_0^c) \Big[  f_t^{Y}\big(c_t(\vtheta_0^c)\big) - \partial_2 F_{t}\big(v_t(\vtheta_0^v), c_t(\vtheta_0^c)\big)\Big] - \mLambda_{n,(1)}\overset{\P}{\longrightarrow}\vzero$.
		
		\item\label{it:Lambda2} $\frac{1}{n}\sum_{t=1}^{n}\nabla c_t(\vtheta_0^c)\nabla^\prime v_t(\vtheta_0^v)\Big[\alpha f_t^X\big(v_t(\vtheta_0^v)\big) - \partial_1 F_t\big(v_t(\vtheta_0^v), c_t(\vtheta_0^c)\big)\Big] - \mLambda_{n,(2)}\overset{\P}{\longrightarrow}\vzero$.
		
		\item\label{it:mom bounds cons3} For all $t\in\mathbb{N}$ it holds that $\E\big[C_1^4(\mathcal{F}_{t-1})\big]\leq K$, $\E\big[V_1^4(\mathcal{F}_{t-1})\big]\leq K$.

	\end{enumerate}
\end{assumption}

Item~\ref{it:bw} of Assumption~\ref{ass:avar} ensures that the bandwidth vanishes, but not too fast.
Items~\ref{it:Lambda}--\ref{it:Lambda2} provide standard laws of large numbers for $\mLambda_n$, $\mLambda_{n,(1)}$ and $\mLambda_{n,(2)}$.
(Note that laws of large numbers automatically apply for the remaining matrices $\mV_n$ and $\mC_n$ thanks to the mixing properties of $\nabla v_t(\vtheta_0^v)$ and $\nabla c_t(\vtheta_0^c)$ from Assumption~\ref{ass:an}~\ref{it:mixing}.)
Finally, item~\ref{it:mom bounds cons3} strengthens the moment bounds for the gradients of the VaR and CoVaR models.

The next theorem shows consistency of our asymptotic variance estimators:

\begin{thm}\label{thm:avar}
	Suppose Assumptions~\ref{ass:cons}--\ref{ass:avar} hold. Then, as $n\to\infty$, $\widehat{\mV}_n -\mV_n\overset{\P}{\longrightarrow}\vzero$, $\widehat{\mC}_n - \mC_{n} \overset{\P}{\longrightarrow}\vzero$, $\widehat{\mLambda}_n - \mLambda_n \overset{\P}{\longrightarrow}\vzero$, $\widehat{\mLambda}_{n,(1)} - \mLambda_{n,(1)}\overset{\P}{\longrightarrow}\vzero$ and $\widehat{\mLambda}_{n,(2)} - \mLambda_{n,(2)} \overset{\P}{\longrightarrow}\vzero$.
\end{thm}

\begin{proof}
	See Section~\ref{sec:thm3}.
\end{proof}

\begin{rem}\label{rem:asvar}
	The bandwidths are only specified asymptotically in Assumption~\ref{ass:avar}, such that little theoretical guidance is available in finite samples.
	In practice, we recommend to follow the (stochastic) bandwidth choice of \citet{Koenker2005book} and \citet{MachadoSantosSilva2013}.
	Formally, the bandwidths are chosen as
	\begin{align*}
		&\widehat{b}_{n,x} = \text{MAD} \Big[\big\{X_t-v_t(\widehat{\vtheta}_n^v) \big\}_{t=1,\ldots,n}\Big] \left[ \Phi^{-1} \big( \beta + m(n,\beta) \big) - \Phi^{-1} \big( \beta - m(n,\beta) \big) \right], \\
		&\widehat{b}_{n,y} = \text{MAD} \Big[\big\{Y_t-c_t(\widehat{\vtheta}_n^c) \big\}_{t=1,\ldots,n}\Big] \left[ \Phi^{-1} \big( \alpha + m((1-\beta)n, \alpha) \big) - \Phi^{-1} \big(\alpha - m((1-\beta)n,\alpha)  \big) \right], \\
		&m(n, \tau) =  n^{-1/3} \left( \Phi^{-1} (0.975)\right)^{2/3} \left( \frac{1.5 (\phi(\Phi^{-1}(\tau)))^2}{2(\Phi^{-1}(\tau))^2 +1} \right)^{1/3},
	\end{align*}
	where $\text{MAD}(\cdot)$ refers to the sample median absolute deviation, and $\phi(\cdot)$ and $\Phi^{-1}(\cdot)$ denote the density and quantile functions of the standard normal distribution.
	The choice of $\widehat{b}_{n,y}$ is a straightforward adaptation for the density at the CoVaR, which is simply estimated as the VaR at level $\alpha$ of only $(1-\beta) n$ observations.
\end{rem}

\section{Proofs of Theorem~\ref{prop:StartingValues} and Proposition~\ref{prop:ARMAclass}}
\label{sec:prop1}

\begin{proof}[{\textbf{Proof of Theorem~\ref{prop:StartingValues}:}}]
	As in the proof of Theorem~\ref{thm:cons}, we show consistency of $\widetilde{\vtheta}_{n}^{v}$ by checking the conditions of Lemma~2.2 in \citet{Whi80} and by explicitly treating the difference between $\widetilde{\vtheta}_{n}^{v}$ and $\widehat{\vtheta}_{n}^{v}$.

	We define $\overline{Q}_n^v(\vtheta^v)$ as in the proof of Theorem~\ref{thm:cons}, but now use $\frac{1}{n} \sum_{t=1}^n  S^{\VaR}(\tilde v_t(\vtheta^v), X_t)$ (based on $\tilde v_t$ instead of $v_t$) as its empirical counterpart.
	We first show the ULLN condition of Lemma~2.2 in \citet{Whi80}.
	For this, write
	\begin{align*}
		\sup_{\vtheta^v \in \mTheta^v}& \bigg| \frac{1}{n} \sum_{t=1}^n   S^{\VaR}(\tilde v_t(\vtheta^v), X_t) - \mathbb{E} \big[ S^{\VaR}(v_t(\vtheta^v), X_t) \big] \bigg| \\
		&\le \sup_{\vtheta^v \in \mTheta^v} \bigg| \frac{1}{n} \sum_{t=1}^n  S^{\VaR}(\tilde v_t(\vtheta^v), X_t) - S^{\VaR}(v_t(\vtheta^v), X_t) \bigg| \\
		&\qquad + \sup_{\vtheta^v \in \mTheta^v} \bigg| \frac{1}{n} \sum_{t=1}^n   S^{\VaR}(v_t(\vtheta^v), X_t) - \mathbb{E} \big[ S^{\VaR}(v_t(\vtheta^v), X_t) \big] \bigg|.
	\end{align*}
	For the second right-hand side term, a ULLN applies by Assumption~\ref{ass:cons}~\ref{it:ULLN}.
	Due to the Lipschitz-continuity of the tick loss function $S^{\VaR}(\cdot, \cdot)$ in its first argument, we get for the first term that
	\begin{align}
		\sup_{\vtheta^v \in \mTheta^v} \bigg|\frac{1}{n} \sum_{t=1}^n &   S^{\VaR}(\tilde v_t(\vtheta^v), X_t) - S^{\VaR}(v_t(\vtheta^v), X_t) \bigg| \notag\\
		&\le \frac{1}{n} \sum_{t=1}^n \sup_{\vtheta^v \in \mTheta^v}  \max(\beta, 1-\beta)  \cdot \left|  v_{t}(\vtheta^v) - \tilde v_{t}(\vtheta^v) \right| \label{eq:LipschitzTickLoss} \\
		&\le  \max(\beta, 1-\beta)  \frac{1}{n} \sum_{t=1}^n \sup_{\vtheta^v \in \mTheta^v} \rho^{t-1}  \cdot \left|  v_{1}(\vtheta^v) - \tilde v_{1}(\vtheta^v) \right| \notag\\
		&\le  \max(\beta, 1-\beta) \sup_{\vtheta^v \in \mTheta^v}\left|  v_{1}(\vtheta^v) - \tilde v_{1}(\vtheta^v) \right| \frac{1}{n} \sum_{t=1}^n  \rho^{t-1} \overset{\P}{\longrightarrow} 0,\notag
	\end{align}
	where the second inequality follows by the contraction condition in Theorem~\ref{prop:StartingValues}.
	The convergence of the above term follows due to the convergence of the geometric sum,
	and $\sup_{\vtheta^v \in \mTheta^v}|  v_{1}(\vtheta^v) - \tilde v_{1}(\vtheta^v) |=O_{\P}(1)$ since $\E\big[\sup_{\vtheta^v \in \mTheta^v}|  v_{1}(\vtheta^v) - \tilde v_{1}(\vtheta^v) |\big]<\infty$ by assumption.
	As the remaining conditions of Lemma~2.2 in \citet{Whi80} are not affected by the starting values and can be shown as in Theorem~\ref{thm:cons}, we obtain consistency of $\widetilde{\vtheta}^v_n$.

	We now turn to the M-estimator of the CoVaR parameters, where we apply the same proof strategy as above and only derive the ULLN condition of Lemma~2.2 in \citet{Whi80}.
	For this, $\overline{Q}_n^c(\vtheta^c)$ is defined as in the proof of Theorem~\ref{thm:cons} and we use 
	\[
	\frac{1}{n} \sum_{t=1}^n S^{\CoVaR}\big((\tilde v_t(\widetilde{\vtheta}^v_n), \tilde c_t(\vtheta^c))^\prime, (X_t, Y_t)^\prime\big)
	\]
	as its empirical counterpart, which is based on $\tilde v_t(\cdot)$ and $\tilde c_t(\cdot)$, and relies on the consistent estimator $\widetilde{\vtheta}^v_n$.	
	To derive the ULLN, we use the bound
	\begin{align*}
		&\sup_{\vtheta^c \in \mTheta^c} \bigg| \frac{1}{n} \sum_{t=1}^n   S^{\CoVaR}\big((\tilde v_t(\widetilde{\vtheta}^v_n), \tilde c_t(\vtheta^c))^\prime, (X_t, Y_t)^\prime\big) - \mathbb{E} \big[ S^{\CoVaR}\big((v_t(\vtheta^v_0), c_t(\vtheta^c))^\prime, (X_t, Y_t)^\prime\big) \big] \bigg| \\
		\le \, &\sup_{\vtheta^c \in \mTheta^c} \bigg| \frac{1}{n} \sum_{t=1}^n  S^{\CoVaR}\big((\tilde v_t(\widetilde{\vtheta}^v_n), \tilde c_t(\vtheta^c))^\prime, (X_t, Y_t)^\prime\big) - S^{\CoVaR}\big((
		v_t(\widetilde{\vtheta}^v_n),  c_t(\vtheta^c))^\prime, (X_t, Y_t)^\prime\big)   \bigg|\\
		+ \, &\sup_{\vtheta^c \in \mTheta^c} \bigg| \frac{1}{n} \sum_{t=1}^n  S^{\CoVaR}\big((
		v_t(\widetilde{\vtheta}^v_n),  c_t(\vtheta^c))^\prime, (X_t, Y_t)^\prime\big) - \mathbb{E} \big[ S^{\CoVaR}\big((v_t(\vtheta^v_0), c_t(\vtheta^c))^\prime, (X_t, Y_t)^\prime\big) \big] \bigg|.
	\end{align*}
	To show that the second term on the right-hand side above is $o_\mathbb{P}(1)$, we can apply the same steps as in the proof of Theorem~\ref{thm:cons} (where we show that  $\sup_{\vtheta^c \in \mTheta^c}  \big| Q_n^c(\vtheta^c) - \overline{Q}_n^c(\vtheta^c) \big| =  o_\mathbb{P}(1)$) and replace $\widehat{\vtheta}^v_n$ by $\widetilde{\vtheta}^v_n$.
	Notice that this proof only uses consistency of $\widehat{\vtheta}_n^v$, which we have established for $\widetilde{\vtheta}^v_n$ above.
	
	For the first term on the right-hand side, we write
	\begin{align}
		&\sup_{\vtheta^c \in \mTheta^c} \bigg| \frac{1}{n} \sum_{t=1}^n  S^{\CoVaR}\big((\tilde v_t(\widetilde{\vtheta}^v_n), \tilde c_t(\vtheta^c))^\prime, (X_t, Y_t)^\prime\big) - S^{\CoVaR}\big((
		v_t(\widetilde{\vtheta}^v_n),  c_t(\vtheta^c))^\prime, (X_t, Y_t)^\prime\big)  \bigg| \notag\\
		\le \, &\sup_{\vtheta^c \in \mTheta^c} \bigg| \frac{1}{n} \sum_{t=1}^n  S^{\CoVaR}\big((\tilde v_t(\widetilde{\vtheta}^v_n), \tilde c_t(\vtheta^c))^\prime, (X_t, Y_t)^\prime\big) - S^{\CoVaR}\big((
		\tilde v_t(\widetilde{\vtheta}^v_n),  c_t(\vtheta^c))^\prime, (X_t, Y_t)^\prime\big)  \bigg| \label{eq:ProofStartVal1}\\
		+ \, &\sup_{\vtheta^c \in \mTheta^c} \bigg| \frac{1}{n} \sum_{t=1}^n  S^{\CoVaR}\big((\tilde v_t(\widetilde{\vtheta}^v_n), c_t(\vtheta^c))^\prime, (X_t, Y_t)^\prime\big) - S^{\CoVaR}\big((
		v_t(\widetilde{\vtheta}^v_n),  c_t(\vtheta^c))^\prime, (X_t, Y_t)^\prime\big) \bigg|.\label{eq:ProofStartVal11}
	\end{align}
	For the term in \eqref{eq:ProofStartVal1}, we obtain that
	\begin{align*}
		&\sup_{\vtheta^c \in \mTheta^c} \bigg| \frac{1}{n} \sum_{t=1}^n  S^{\CoVaR}\big((\tilde v_t(\widetilde{\vtheta}^v_n), \tilde c_t(\vtheta^c))^\prime, (X_t, Y_t)^\prime\big) - S^{\CoVaR}\big((
		\tilde v_t(\widetilde{\vtheta}^v_n),  c_t(\vtheta^c))^\prime, (X_t, Y_t)^\prime\big)  \bigg| \\
		\le \, &\frac{1}{n} \sum_{t=1}^n\mathds{1}_{\{X_t > \tilde v_t(\widetilde{\vtheta}^v_n)\}} \sup_{\vtheta^c \in \mTheta^c} \left|  \big( \mathds{1}_{\{ Y_t \le  \tilde c_t(\vtheta^c)\}} - \alpha \big) \big( \tilde c_t(\vtheta^c) -  Y_t \big) -  \big( \mathds{1}_{\{ Y_t \le  c_t(\vtheta^c)\}} - \alpha \big) \big( c_t(\vtheta^c) -  Y_t \big) \right| \\
		\le \, &\frac{1}{n} \sum_{t=1}^n  \sup_{\vtheta^c \in \mTheta^c} \left|  \big( \mathds{1}_{\{ Y_t \le  \tilde c_t(\vtheta^c)\}} - \alpha \big) \big( \tilde c_t(\vtheta^c) -  Y_t \big) -  \big( \mathds{1}_{\{ Y_t \le  c_t(\vtheta^c)\}} - \alpha \big) \big( c_t(\vtheta^c) -  Y_t \big) \right| \\
		\le \, &\frac{1}{n} \sum_{t=1}^n  \sup_{\vtheta^c \in \mTheta^c} \max(\alpha, 1-\alpha)  \left| \tilde c_t(\vtheta^c) -  c_t(\vtheta^c) \right| \\
		\le \, &\max(\alpha, 1-\alpha)  \frac{1}{n} \sum_{t=1}^n \sup_{\vtheta^c \in \mTheta^c} \rho^{t-1}  \cdot \left|  c_{1}(\vtheta^c) - \tilde c_{1}(\vtheta^c) \right|\\
		\le \, &\max(\alpha, 1-\alpha) \sup_{\vtheta^c \in \mTheta^c}\left|  c_{1}(\vtheta^c) - \tilde c_{1}(\vtheta^c) \right| \frac{1}{n} \sum_{t=1}^n \rho^{t-1}  \overset{\P}{\longrightarrow} 0,
	\end{align*}
	where, for the third inequality, we have used the Lipschitz-continuity of the tick loss function in its first argument as in \eqref{eq:LipschitzTickLoss}, and the penultimate inequality uses the imposed contraction condition for the CoVaR model.
	
	For the term in \eqref{eq:ProofStartVal11}, we get that for any $\varepsilon > 0$ and $\delta > 0$ small enough,
	\begin{align*}
		&\mathbb{P} \Bigg\{ \sup_{\vtheta^c \in \mTheta^c} \bigg|  \frac{1}{n} \sum_{t=1}^n  S^{\CoVaR}\big((\tilde v_t(\widetilde{\vtheta}^v_n), c_t(\vtheta^c))^\prime, (X_t, Y_t)^\prime\big) - S^{\CoVaR}\big((
		v_t(\widetilde{\vtheta}^v_n),  c_t(\vtheta^c))^\prime, (X_t, Y_t)^\prime\big)  \bigg|  > \varepsilon \Bigg\} \notag\\
		&\le  \mathbb{P} \Bigg\{ \frac{1}{n} \sum_{t=1}^n  \big| \mathds{1}_{\{X_t > \tilde v_t(\widetilde{\vtheta}^v_n)\}}  - \mathds{1}_{\{X_t > v_t(\widetilde{\vtheta}^v_n)\}}  \big|  \sup_{\vtheta^c \in \mTheta^c}  \left| \big( \mathds{1}_{\{ Y_t \le  c_t(\vtheta^c)\}} - \alpha \big) \big( c_t(\vtheta^c) -  Y_t \big)  \right|  > \varepsilon \Bigg\} \notag \\
		&\le \mathbb{P} \Bigg\{ \frac{1}{n} \sum_{t=1}^n  \big| \mathds{1}_{\{X_t > \tilde v_t(\widetilde{\vtheta}^v_n)\}}  - \mathds{1}_{\{X_t > v_t(\widetilde{\vtheta}^v_n)\}}  \big| \sup_{\vtheta^c \in \mTheta^c}  \left| \big( \mathds{1}_{\{ Y_t \le  c_t(\vtheta^c)\}} - \alpha \big) \big( c_t(\vtheta^c) -  Y_t \big)  \right|  > \varepsilon, \\
		&\hspace{2cm}\quad \big\Vert  \widetilde{\vtheta}^v_n - \vtheta_0^v \big\Vert \le \delta \Bigg\} + \mathbb{P} \left\{ \big\Vert  \widetilde{\vtheta}^v_n - \vtheta_0^v \big\Vert > \delta \right\}.
	\end{align*} 
	The second term after the last inequality is $o(1)$ by consistency of $\widetilde{\vtheta}^v_n$, and the first term is smaller or equal to
	\begin{align}		
		&\mathbb{P} \left\{ \frac{1}{n} \sum_{t=1}^n  \sup_{\Vert \vtheta^v - \vtheta_0^v \Vert \le \delta} \big| \mathds{1}_{\{X_t > \tilde v_t(\vtheta^v)\}}  - \mathds{1}_{\{X_t > v_t(\vtheta^v)\}}  \big| \sup_{\vtheta^c \in \mTheta^c}  \left| \big( \mathds{1}_{\{ Y_t \le  c_t(\vtheta^c)\}} - \alpha \big) \big( c_t(\vtheta^c) -  Y_t \big)  \right|  > \varepsilon \right\} \notag \\
		&\quad \le \mathbb{P} \left\{ \frac{1}{n} \sum_{t=1}^n   \sup_{\Vert \vtheta^v - \vtheta_0^v \Vert \le \delta} \big| \mathds{1}_{\{X_t > \tilde v_t(\vtheta^v)\}}  - \mathds{1}_{\{X_t > v_t(\vtheta^v)\}}  \big|   \Big( \sup_{\vtheta^c \in \mTheta^c}  \big| c_t(\vtheta^c) \big| + \big|Y_t \big|  \Big)   > \varepsilon \right\} \notag \\
		&\quad \le \mathbb{P} \left\{ \frac{1}{n} \sum_{t=1}^n   \sup_{\Vert \vtheta^v - \vtheta_0^v \Vert \le \delta} \big| \mathds{1}_{\{X_t > \tilde v_t(\vtheta^v)\}}  - \mathds{1}_{\{X_t > v_t(\vtheta^v)\}}  \big|  C(\mathcal{F}_{t-1})  > \varepsilon/2 \right\} \label{eq:TermUntreated} \\
		&\qquad + \mathbb{P} \left\{ \frac{1}{n} \sum_{t=1}^n   \sup_{\Vert \vtheta^v - \vtheta_0^v \Vert \le \delta} \big| \mathds{1}_{\{X_t > \tilde v_t(\vtheta^v)\}}  - \mathds{1}_{\{X_t > v_t(\vtheta^v)\}}  \big| \cdot \big|Y_t\big|  > \varepsilon/2 \right\}. \label{eq:TermTreated}
	\end{align}
	
	We continue to show how the term in \eqref{eq:TermTreated} can be bounded and already notice that the  term in \eqref{eq:TermUntreated} can be treated in a simplified fashion (without the need to apply H\"{o}lder's inequality below).
	
	For the term in \eqref{eq:TermTreated},	Markov's inequality, H\"{o}lder's inequality (with $r = (1 + \iota)/\iota$ and $s=1+\iota$ for $\iota > 0$ from Assumption~\ref{ass:cons}~\ref{it:mom bounds cons}) and the law of total expectation yield that
	\begin{align}
		\P &\left\{ \frac{1}{n} \sum_{t=1}^n  \sup_{\Vert \vtheta^v - \vtheta_0^v \Vert \le \delta} \big| \mathds{1}_{\{X_t > \tilde v_t(\vtheta^v)\}}  - \mathds{1}_{\{X_t > v_t(\vtheta^v)\}}  \big|  \cdot  \big| Y_t \big|  > \varepsilon/2 \right\} \notag \\
		&\le\frac{2}{\varepsilon} \E \bigg[ \frac{1}{n} \sum_{t=1}^n \big| Y_t \big|  \sup_{\Vert \vtheta^v - \vtheta_0^v \Vert \le \delta} \big| \mathds{1}_{\{X_t > \tilde v_t(\vtheta^v)\}}  - \mathds{1}_{\{X_t > v_t(\vtheta^v)\}}  \big|   \bigg] \notag \\
		&\le\frac{2}{\varepsilon} \frac{1}{n} \sum_{t=1}^n   \big(\mathbb{E}|Y_t|^s  \big)^{1/s} \left(\mathbb{E} \left[ \left( \sup_{\Vert \vtheta^v - \vtheta_0^v \Vert \le \delta} \big| \mathds{1}_{\{X_t > \tilde v_t(\vtheta^v)\}}  - \mathds{1}_{\{X_t > v_t(\vtheta^v)\}}  \big| \right)^r  \right] \right)^{1/r} \notag \\
		&\le\frac{2}{\varepsilon} K^{1/s} \frac{1}{n} \sum_{t=1}^n  \left( \mathbb{E} \left[ \sup_{\Vert \vtheta^v - \vtheta_0^v \Vert \le \delta} \big| \mathds{1}_{\{X_t > \tilde v_t(\vtheta^v)\}}  - \mathds{1}_{\{X_t > v_t(\vtheta^v)\}}  \big| \right]  \right)^{1/r} \notag \\
		&\le \frac{2}{\varepsilon} K^{1/s}  \frac{1}{n} \sum_{t=1}^n \left( \sum_{k=1}^4  \mathbb{E} \left[ \sup_{\Vert \vtheta^v - \vtheta_0^v \Vert \le \delta} \big| \mathds{1}_{\{X_t > \tilde v_t(\vtheta^v)\}}  - \mathds{1}_{\{X_t > v_t(\vtheta^v)\}}  \big|   \;\mathds{1}_{\{\mathfrak{A}_{k,t}\}} \right]  \right)^{1/r},
		\label{eqn:SupInfNew}
	\end{align}
	where we define the four sets
	\begin{align*}
		&\mathfrak{A}_{1,t} :=  \left\{ \inf_{\Vert \vtheta^v - \vtheta_0^v \Vert \le \delta} v_t(\vtheta^v) \le \inf_{\Vert \vtheta^v - \vtheta_0^v \Vert \le \delta} \tilde v_t(\vtheta^v),
		\quad 
		\sup_{\Vert \vtheta^v - \vtheta_0^v \Vert \le \delta} v_t(\vtheta^v) \le \sup_{\Vert \vtheta^v - \vtheta_0^v \Vert \le \delta} \tilde v_t(\vtheta^v)  \right\}, 
		\\
		&\mathfrak{A}_{2,t}  :=  \left\{ \inf_{\Vert \vtheta^v - \vtheta_0^v \Vert \le \delta} v_t(\vtheta^v) \ge \inf_{\Vert \vtheta^v - \vtheta_0^v \Vert \le \delta} \tilde v_t(\vtheta^v),
		\quad 
		\sup_{\Vert \vtheta^v - \vtheta_0^v \Vert \le \delta} v_t(\vtheta^v) \ge \sup_{\Vert \vtheta^v - \vtheta_0^v \Vert \le \delta} \tilde v_t(\vtheta^v)  \right\}, 
		\\
		&\mathfrak{A}_{3,t}  :=  \left\{ \inf_{\Vert \vtheta^v - \vtheta_0^v \Vert \le \delta} v_t(\vtheta^v) \ge \inf_{\Vert \vtheta^v - \vtheta_0^v \Vert \le \delta} \tilde v_t(\vtheta^v),
		\quad 
		\sup_{\Vert \vtheta^v - \vtheta_0^v \Vert \le \delta} v_t(\vtheta^v) \le \sup_{\Vert \vtheta^v - \vtheta_0^v \Vert \le \delta} \tilde v_t(\vtheta^v)  \right\}, 
		\\
		&\mathfrak{A}_{4,t}  :=  \left\{ \inf_{\Vert \vtheta^v - \vtheta_0^v \Vert \le \delta} v_t(\vtheta^v) \le \inf_{\Vert \vtheta^v - \vtheta_0^v \Vert \le \delta} \tilde v_t(\vtheta^v),
		\quad 
		\sup_{\Vert \vtheta^v - \vtheta_0^v \Vert \le \delta} v_t(\vtheta^v) \ge \sup_{\Vert \vtheta^v - \vtheta_0^v \Vert \le \delta} \tilde v_t(\vtheta^v)  \right\},
	\end{align*}
	such that $\Omega = \bigcup_{k=1}^4 \mathfrak{A}_{k,t}$ for all $t \in \mathbb{N}$.

	We now separately treat the four ($k=1,\dots,4$) conditional expectations in \eqref{eqn:SupInfNew}.
	For the expectation involving $\mathfrak{A}_{1,t}$, we define the $\mathcal{F}_{t-1}$-measurable quantities
	\begin{align}
		\label{eq:DefArgMinMax}
		\vtheta^v_\text{min} := \underset{\Vert \vtheta^v - \vtheta_0^v \Vert \le \delta}{\arg\min} \, v_t(\vtheta^v)
		\qquad \text{ and } \qquad
		\vtheta^v_\text{max} := \underset{\Vert \vtheta^v - \vtheta_0^v \Vert \le \delta}{\arg\max} \, \tilde v_t(\vtheta^v).
	\end{align}
	Hence, on $\mathfrak{A}_{1,t}$, we have that
	\begin{align*}
		\sup_{\Vert \vtheta^v - \vtheta_0^v \Vert \le \delta}  \big| \mathds{1}_{\{X_t > \tilde v_t(\vtheta^v)\}}  - \mathds{1}_{\{X_t > v_t(\vtheta^v)\}} \big| 
		&\le \mathds{1}_{\{ v_t(\vtheta^v_\text{min} ) \le X_t \le \tilde v_t(\vtheta^v_\text{max}  ) \}}  \\
		&\le \mathds{1}_{\{ v_t(\vtheta^v_\text{min} ) \le X_t \le v_t(\vtheta^v_\text{max}  ) \}}  
		+ \mathds{1}_{\{ v_t(\vtheta^v_\text{max} ) \le X_t \le \tilde v_t(\vtheta^v_\text{max}  ) \}}.
	\end{align*}
	Thus, for the expectation in \eqref{eqn:SupInfNew} involving $\mathfrak{A}_{1,t}$, we get that
	\begin{multline}\label{eq:ProofDiffMinMax}
		\mathbb{E} \bigg[ \sup_{\Vert \vtheta^v - \vtheta_0^v \Vert \le \delta} \big| \mathds{1}_{\{X_t > \tilde v_t(\vtheta^v)\}}  - \mathds{1}_{\{X_t > v_t(\vtheta^v)\}}  \big|  \;\mathds{1}_{\{\mathfrak{A}_{1,t}\}} \bigg]  \\
		\le	\mathbb{E} \Big[ \mathds{1}_{\{ v_t(\vtheta^v_\text{min} ) \le X_t \le v_t(\vtheta^v_\text{max}  ) \}}   \;\mathds{1}_{\{\mathfrak{A}_{1,t}\}} \Big]
		+ \mathbb{E} \Big[ \mathds{1}_{\{ v_t(\vtheta^v_\text{max} ) \le X_t \le \tilde v_t(\vtheta^v_\text{max}  ) \}} \;\mathds{1}_{\{\mathfrak{A}_{1,t}\}} \Big]. 	
	\end{multline}
	For the first right-hand side term in \eqref{eq:ProofDiffMinMax}, we apply the same arguments as applied for the term $B_{11t}$ in \eqref{eq:(B.11t)} to get that
	\begin{align}
		\mathbb{E} \left[ \mathds{1}_{\{ v_t(\vtheta^v_\text{min} ) \le X_t \le v_t(\vtheta^v_\text{max}  ) \}}   \;\mathds{1}_{\{\mathfrak{A}_{1,t}\}} \right] 
		&= \mathbb{E} \left[ \mathbb{E}_{t-1} \left[ \mathds{1}_{\{ v_t(\vtheta^v_\text{min} ) \le X_t \le v_t(\vtheta^v_\text{max}  ) \}}  \right]  \mathds{1}_{\{\mathfrak{A}_{1,t}\}} \right] \notag\\
		&= \mathbb{E} \left[ \int_{v_t(\vtheta^v_\text{min})}^{v_t(\vtheta^v_\text{max})} f_t^X(x) \mathrm{d}x \,  \mathds{1}_{\{\mathfrak{A}_{1,t}\}} \right] \notag\\
		&\le \mathbb{E} \left[ K \left| v_t(\vtheta^v_\text{min}) - v_t(\vtheta^v_\text{max}) \right|  \;\mathds{1}_{\{\mathfrak{A}_{1,t}\}} \right] \notag\\
		&\le K \, \mathbb{E} \left[ \big\Vert\nabla v_{t}(\vtheta^{v\ast}) \big\Vert \cdot  \Vert \vtheta^v_\text{min} - \vtheta^v_\text{max} \Vert  \;\mathds{1}_{\{\mathfrak{A}_{1,t}\}} \right] \notag\\
		&\le K \, \mathbb{E} \left[ V_1(\mathcal{F}_{t-1} ) \right] \cdot \delta \notag\\
		&\le C \delta,\label{eq:ExpectIndicator1}
	\end{align}
	where $\vtheta^{v\ast}$ is some mean value between $\vtheta^v_\text{min}$ and $\vtheta^v_\text{max}$.
	
	For the second term in \eqref{eq:ProofDiffMinMax}, we get using similar arguments that 
	\begin{align}
		\mathbb{E} \left[ \mathds{1}_{\{ v_t(\vtheta^v_\text{max} ) \le X_t \le \tilde v_t(\vtheta^v_\text{max}  ) \}}   \mathds{1}_{\{\mathfrak{A}_{1,t}\}} \right] 
		&\le \mathbb{E} \left[ K \cdot \big| v_t(\vtheta^v_\text{max}) -  \tilde v_t(\vtheta^v_\text{max})  \big|   \;\mathds{1}_{\{\mathfrak{A}_{1,t}\}} \right] \notag\\
		&\le K \rho^{t-1}    \cdot \mathbb{E} \Big[  \big| v_1(\vtheta^v_\text{max}) -  \tilde v_1(\vtheta^v_\text{max}) \big|  \;\mathds{1}_{\{\mathfrak{A}_{1,t}\}} \Big] \notag\\
		&\le C \rho^{t-1}. \label{eq:ExpectIndicator2}
	\end{align}	
	In the last two inequalities leading to \eqref{eq:ExpectIndicator2}, we used the contraction condition and the assumption that $\mathbb{E} \big[ \sup_{\vtheta^v \in \mTheta^v} | v_1(\vtheta^v)- \tilde v_1(\vtheta^v)| \big] < \infty$.
	
	We now show that the terms involving $\mathfrak{A}_{2,t}$,  $\mathfrak{A}_{3,t}$ and $\mathfrak{A}_{4,t}$ in \eqref{eqn:SupInfNew} result in terms that are of the same order as \eqref{eq:ExpectIndicator1} and \eqref{eq:ExpectIndicator2}.
	
	For the expectation in \eqref{eqn:SupInfNew} pertaining to $\mathfrak{A}_{2,t}$, we essentially interchange $\tilde v_t$ and $v_t$ and redefine
	\begin{align*}
		\vtheta^v_\text{min} := \underset{\Vert \vtheta^v - \vtheta_0^v \Vert \le \delta}{\arg\min} \, \tilde v_t(\vtheta^v)
		\qquad \text{ and } \qquad
		\vtheta^v_\text{max} := \underset{\Vert \vtheta^v - \vtheta_0^v \Vert \le \delta}{\arg\max} \, v_t(\vtheta^v),
	\end{align*}
	such that on the set $\mathfrak{A}_{2,t}$,
	\begin{align*}
		\sup_{\Vert \vtheta^v - \vtheta_0^v \Vert \le \delta}  \big| \mathds{1}_{\{X_t > \tilde v_t(\vtheta^v)\}}  - \mathds{1}_{\{X_t > v_t(\vtheta^v)\}} \big| 
		&\le \mathds{1}_{\{ \tilde v_t(\vtheta^v_\text{min} ) \le X_t \le v_t(\vtheta^v_\text{max}  ) \}}  \\
		&\le \mathds{1}_{\{ \tilde v_t(\vtheta^v_\text{min} ) \le X_t \le v_t(\vtheta^v_\text{min}  ) \}}  
		+ \mathds{1}_{\{ v_t(\vtheta^v_\text{min} ) \le X_t \le  v_t(\vtheta^v_\text{max}  ) \}}.
	\end{align*}
	The remaining steps of the proof apply just as in the first case treated in detail above.
	
	For the third expectation in \eqref{eqn:SupInfNew} that involves $\mathfrak{A}_{3,t}$, we define the $\mathcal{F}_{t-1}$-measurable quantities
	\begin{align*}
		\vtheta^v_\text{min} := \underset{\Vert \vtheta^v - \vtheta_0^v \Vert \le \delta}{\arg\min} \, \tilde v_t(\vtheta^v)
		\qquad \text{ and } \qquad
		\vtheta^v_\text{max} := \underset{\Vert \vtheta^v - \vtheta_0^v \Vert \le \delta}{\arg\max} \, \tilde v_t(\vtheta^v).
	\end{align*}
	Then, on the set $\mathfrak{A}_{3,t}$,
	\begin{align*}
		&\sup_{\Vert \vtheta^v - \vtheta_0^v \Vert \le \delta}  \big| \mathds{1}_{\{X_t > \tilde v_t(\vtheta^v)\}}  - \mathds{1}_{\{X_t > v_t(\vtheta^v)\}} \big| 
		\le \mathds{1}_{\{ \tilde v_t(\vtheta^v_\text{min} ) \le X_t \le \tilde v_t(\vtheta^v_\text{max}  ) \}}  \\
		&\qquad \qquad \le \mathds{1}_{\{ \tilde v_t(\vtheta^v_\text{min} ) \le X_t \le v_t(\vtheta^v_\text{min}  ) \}}  
		+ \mathds{1}_{\{ v_t(\vtheta^v_\text{min} ) \le X_t \le v_t(\vtheta^v_\text{max}  ) \}}
		+ \mathds{1}_{\{ v_t(\vtheta^v_\text{max} ) \le X_t \le \tilde v_t(\vtheta^v_\text{max}  ) \}},
	\end{align*}
	and the proof is as in the two previous cases, noting that the first and third indicator function above can both be treated as the term in \eqref{eq:ExpectIndicator2}, and the second indicator as the term in \eqref{eq:ExpectIndicator1}.
	
	Finally, for the fourth expectation in \eqref{eqn:SupInfNew} pertaining to $\mathfrak{A}_{4,t}$, we define the $\mathcal{F}_{t-1}$-measurable quantities
	\begin{align*}
		\vtheta^v_\text{min} := \underset{\Vert \vtheta^v - \vtheta_0^v \Vert \le \delta}{\arg\min} \,  v_t(\vtheta^v)
		\qquad \text{ and } \qquad
		\vtheta^v_\text{max} := \underset{\Vert \vtheta^v - \vtheta_0^v \Vert \le \delta}{\arg\max} \, v_t(\vtheta^v).
	\end{align*}
	This simplifies the proof as then, on the set $\mathfrak{A}_{4,t}$,
	\begin{align*}
		&\sup_{\Vert \vtheta^v - \vtheta_0^v \Vert \le \delta}  \big| \mathds{1}_{\{X_t > \tilde v_t(\vtheta^v)\}}  - \mathds{1}_{\{X_t > v_t(\vtheta^v)\}} \big| 
		\le \mathds{1}_{\{ v_t(\vtheta^v_\text{min} ) \le X_t \le v_t(\vtheta^v_\text{max}  ) \}}.
	\end{align*}
	Therefore, the same arguments as above can be applied.
	
	Hence, by using \eqref{eq:ExpectIndicator1} and \eqref{eq:ExpectIndicator2} (together with the equivalent terms arising for $\mathfrak{A}_{2,t}$, $\mathfrak{A}_{3,t}$ and $\mathfrak{A}_{4,t}$) for the term in \eqref{eqn:SupInfNew}, we get that
	\begin{align*}
		&\frac{2}{\varepsilon} K^{1/s}  \frac{1}{n} \sum_{t=1}^n \Bigg( \sum_{k=1}^4  \mathbb{E} \bigg[ \sup_{\Vert \vtheta^v - \vtheta_0^v \Vert \le \delta} \big| \mathds{1}_{\{X_t > \tilde v_t(\vtheta^v)\}}  - \mathds{1}_{\{X_t > v_t(\vtheta^v)\}}  \big|  \;\mathds{1}_{\{\mathfrak{A}_{k,t}\}} \bigg]  \Bigg)^{1/r} \\
		&\quad\le  C \frac{1}{n} \sum_{t=1}^n \Big( C \delta +  C \rho^{t-1}  \Big)^{1/r} \\
		&\quad\le  C \frac{1}{n} \sum_{t=1}^n  \Big(C^{1/r} \delta^{1/r} +  C^{1/r} \rho^{(t-1)/r} \Big) \\
		&\quad\le  C  \delta^{1/r} +  C  \frac{1}{n} \sum_{t=1}^n \rho^{(t-1)/r} .
	\end{align*}
	
	The first term is bounded by (a function of) the arbitrarily small $\delta$ and the second term above converges to zero as the geometric sum $\lim_{n\to \infty} \sum_{t=1}^n \rho^{(t-1)/r}  = \lim_{n\to \infty} \sum_{t=1}^n \big( \rho^{1/r} \big)^{t-1}  < \infty$ is bounded, since $\rho^{1/r} < 1$ for $1/r >0$. 
	This implies that \eqref{eqn:SupInfNew} is $o(1)$. 
	A simplified line of arguments shows that the term \eqref{eq:TermUntreated} is also is $o(1)$, such that the term in  \eqref{eq:ProofStartVal11} is $o_\mathbb{P}(1)$.
	This verifies the ULLN assumption of Lemma~2.2 in \citet{Whi80}. 
	Because the remaining assumptions of that lemma follow as in Theorem~\ref{thm:cons}, the proof is concluded.
\end{proof}

\begin{proof}[{\textbf{Proof of Proposition~\ref{prop:ARMAclass}:}}]
	For models of the form \eqref{eqn:LinearModels}, the matrix\
	\[
	\mB=\begin{pmatrix}B_{11} & B_{12}\\ B_{21} & B_{22}\end{pmatrix}=\begin{pmatrix}B_{11} & 0\\0 & B_{22}\end{pmatrix}
	\]
	is diagonal, because in \eqref{eqn:GeneralModel} $v_t(\cdot)$ is a function of $\vtheta^v$ only and $c_t(\cdot)$ is a function of $\vtheta^c$ only.
	Since $\mB$ is diagonal, the eigenvalues are simply the entries on the main diagonal.
	Therefore, the spectral radius $\rho(\mB)$ (i.e., the largest eigenvalue of $\mB$ in absolute terms) is given by $\rho(\mB)=\max\big\{|B_{11}|,|B_{22}|\big\}$, which is smaller than one by assumption.
	
	Direct calculations for models of the form \eqref{eqn:LinearModels} yield that
	\begin{align*}
		\begin{pmatrix} v_t(\vtheta^v) \\ c_t(\vtheta^c) \end{pmatrix} - \begin{pmatrix} \tilde v_t(\vtheta^v) \\ \tilde c_t(\vtheta^c) \end{pmatrix} 
		= \mB \left[ \begin{pmatrix} v_{t-1}(\vtheta^v) \\ c_{t-1}(\vtheta^c) \end{pmatrix} -  \begin{pmatrix} \tilde v_{t-1}(\vtheta^v) \\ \tilde c_{t-1}(\vtheta^c) \end{pmatrix} \right], 
	\end{align*} 	
	such that the contraction condition for the VaR and CoVaR model holds with $\rho = \max\big\{|B_{11}|,|B_{22}|\big\}$.
\end{proof}

\section{Proof of Theorem~\ref{thm:an}}
\label{sec:thm2}

We split the proof of Theorem~\ref{thm:an} into two parts. In Section~\ref{Asymptotic Normality of the VaR Parameter Estimator}, we prove the asymptotic normality of $\widehat{\vtheta}_n^v$, and in Section~\ref{Asymptotic Normality of the CoVaR Parameter Estimator} the joint convergence of $\widehat{\vtheta}_n^{v}$ and $\widehat{\vtheta}_n^{c}$. 

\subsection{Asymptotic Normality of the VaR Parameter Estimator}\label{Asymptotic Normality of the VaR Parameter Estimator}

The proof of the asymptotic normality of $\widehat{\vtheta}_n^v$ follows closely that of Theorem~2 in \citet{EM04}, although some of our assumptions differ from theirs. Our main motivation for detailing the proof is that some of the subsequent results are needed to prove the asymptotic normality of $\widehat{\vtheta}_n^c$ in Section~\ref{Asymptotic Normality of the CoVaR Parameter Estimator}.

Before giving the formal proof, define 
\begin{equation}\label{eq:gt prev}
	\Vv_t(\vtheta^v) := \nabla v_t(\vtheta^v)\big[\1_{\{X_t\leq v_t(\vtheta^v)\}}-\beta\big].
\end{equation}
This quantity may be interpreted as the derivative of the VaR score with respect to the model parameters, because for $v\neq x$,
\begin{equation*}
	\frac{\partial}{\partial v}S^{\VaR}(v,x) = \big[\1_{\{x\leq v\}}-\beta\big],
\end{equation*}
such that by the chain rule $\Vv_t(\vtheta^v)=\frac{\partial}{\partial \vtheta^v}S^{\VaR}(v_t(\vtheta^v),X_t)$ for $v_t(\vtheta^v)\neq X_t$.

Similarly as in \citet{EM04} and \citet{PZC19}, the key step in the proof is to apply a mean value expansion and Lemma~A.1 of \citet{Wei91} to prove the following result.

\begin{lemV}\label{lem:1}
	Suppose Assumptions~\ref{ass:cons} and \ref{ass:an} hold. Then, as $n\to\infty$,
	\[
	\sqrt{n}\big(\widehat{\vtheta}_n^{v}-\vtheta_0^v\big) = \Big[\mLambda_n^{-1}+o_{\P}(1)\Big]\bigg[-\frac{1}{\sqrt{n}}\sum_{t=1}^{n}\Vv_{t}(\vtheta_0^v) + o_{\P}(1)\bigg].
	\]
\end{lemV}

\begin{proof}
	See Section~\ref{add results VaR norm}.
\end{proof}

Note that $\mLambda_n^{-1}$ in Lemma~V.\ref{lem:1} exists due to Assumption~\ref{ass:an}~\ref{it:pd}.

\begin{lemV}\label{lem:7}
	Suppose Assumptions~\ref{ass:cons} and \ref{ass:an} hold. Then, as $n\to\infty$,
	\[
	n^{-1/2}\mV_n^{-1/2}\sum_{t=1}^{n}\Vv_{t}(\vtheta_0^v)\overset{d}{\longrightarrow}N(\vzeros,\mI),
	\]
	where $\mV_n=\frac{1}{n}\sum_{t=1}^{n}\E\big[ \Vv_{t}(\vtheta_0^v)\Vv_{t}^\prime(\vtheta_0^v)\big]=\beta(1-\beta)\frac{1}{n}\sum_{t=1}^{n}\E\big[\nabla v_t(\vtheta_0^v)\nabla^\prime v_t(\vtheta_0^v)\big]$.
\end{lemV}

\begin{proof}
	See Section~\ref{add results VaR norm}.
\end{proof}

\begin{proof}[{\textbf{Proof of Theorem~\ref{thm:an}:}}]
	In this first part of the proof, we show asymptotic normality of $\widehat{\vtheta}_n^v$. From Lemma~V.\ref{lem:1}, we have the expansion
	\[
	\sqrt{n}\big(\widehat{\vtheta}_n^{v}-\vtheta_0^v\big) = \Big[\mLambda_n^{-1}+o_{\P}(1)\Big]\bigg[-\frac{1}{\sqrt{n}}\sum_{t=1}^{n}\Vv_{t}(\vtheta_0^v) + o_{\P}(1)\bigg].
	\]
	From this, Slutzky's theorem and Lemma~V.\ref{lem:7}, we then obtain that, as $n\to\infty$,
	\[
	\sqrt{n}\mV_n^{-1/2}\mLambda_n\big(\widehat{\vtheta}_n^{v}-\vtheta_0^v\big)\overset{d}{\longrightarrow}N(\vzeros,\mI),
	\]
	which is the claimed result for the VaR parameter estimator.
\end{proof}

\subsection{Joint Asymptotic Normality of the Parameter Estimators}\label{Asymptotic Normality of the CoVaR Parameter Estimator}

Showing asymptotic normality of $\widehat{\vtheta}_n^c$ requires some further preliminary notation and lemmas. To see the analogy to the proof of the asymptotic normality of $\widehat{\vtheta}_n^v$ clearer, we label the lemmas as Lemma~C.\ref{lem:1 tilde}, C.\ref{lem:7 tilde}, etc.

Define
\begin{equation}\label{eq:(12+) prev}
	\Cc_{t}(\vtheta^c, \vtheta^v) := \1_{\{X_t> v_t(\vtheta^v)\}}\nabla c_t(\vtheta^c)\big[\1_{\{Y_t\leq c_t(\vtheta^c)\}}-\alpha\big].
\end{equation}
Similarly as above, this quantity may be viewed as the derivative of the CoVaR score with respect to the CoVaR parameters, since for $c\neq y$,
\begin{equation*}
	\frac{\partial}{\partial c}S^{\CoVaR}\big((v,c)^\prime,(x,y)^\prime\big) = \1_{\{x> v\}}\big[\1_{\{y\leq c\}}-\alpha\big]
\end{equation*}
and, therefore, $\Cc_{t}(\vtheta^c, \vtheta^v)=\frac{\partial}{\partial \vtheta^c}S^{\CoVaR}\big((v_t(\vtheta^v), c_t(\vtheta^c))^\prime,(X_t,Y_t)^\prime\big)$ for $c_t(\vtheta^c)\neq Y_t$ by the chain rule.

The crucial step in the proof is once again to apply a mean value expansion and Lemma~A.1 of \citet{Wei91} to prove the following lemma:

\begin{lemC}\label{lem:1 tilde}
	Suppose Assumptions~\ref{ass:cons} and \ref{ass:an} hold. Then, as $n\to\infty$,
	\begin{multline*}
		\sqrt{n}\big(\widehat{\vtheta}_n^{c}-\vtheta_0^c\big) = \big[\mLambda_{n,(1)}^{-1}\mLambda_{n,(2)}\mLambda_n^{-1}+o_{\P}(1)\big]\bigg[\frac{1}{\sqrt{n}}\sum_{t=1}^{n}\Vv_{t}(\vtheta_0^v) + o_{\P}(1)\bigg]\\ - \big[\mLambda_{n,(1)}^{-1}+o_{\P}(1)\big]\bigg[\frac{1}{\sqrt{n}}\sum_{t=1}^{n}\Cc_{t}(\vtheta_{0,n}^c, \widehat{\vtheta}_n^v) + o_{\P}(1)\bigg],
	\end{multline*}
	where $\vtheta_{0,n}^c$ is defined in the proof of Lemma~C.\ref{lem:1 tilde}.
\end{lemC}

\begin{proof}
	See Section~\ref{add results CoVaR norm}.
\end{proof}

\begin{lemC}\label{lem:7 tilde}
	Suppose Assumptions~\ref{ass:cons} and \ref{ass:an} hold. Then, as $n\to\infty$,
	\[
	n^{-1/2}\mC_{n}^{-1/2}\sum_{t=1}^{n}\Cc_{t}(\vtheta_{0,n}^c, \widehat{\vtheta}_n^v)\overset{d}{\longrightarrow}N(\vzeros,\mI),
	\]
	where $\mC_{n}=\frac{1}{n}\sum_{t=1}^{n}\E\big[ \Cc_{t}(\vtheta_0^c,\vtheta_0^v)\Cc_{t}^\prime(\vtheta_0^c, \vtheta_0^v)\big]=\alpha(1-\alpha)(1-\beta)\frac{1}{n}\sum_{t=1}^{n}\E\big[\nabla c_t(\vtheta_0^c)\nabla^\prime c_t(\vtheta_0^c)\big]$.
\end{lemC}

\begin{proof}
	See Section~\ref{add results CoVaR norm}.
\end{proof}

\begin{proof}[{\textbf{Proof of Theorem~\ref{thm:an} (continued):}}]
	We can now show joint asymptotic normality of $\widehat{\vtheta}_n^{v}$ and $\widehat{\vtheta}_n^{c}$. From Lemma~C.\ref{lem:1 tilde}, we have the decomposition
	\begin{align}
		\sqrt{n}\big(\widehat{\vtheta}_n^{c}-\vtheta_0^c\big) &= \big[\mLambda_{n,(1)}^{-1}\mLambda_{n,(2)}\mLambda_n^{-1}+o_{\P}(1)\big]\bigg[\frac{1}{\sqrt{n}}\sum_{t=1}^{n}\Vv_{t}(\vtheta_0^v) + o_{\P}(1)\bigg]\notag\\
		&\hspace{6cm}- \big[\mLambda_{n,(1)}^{-1}+ o_{\P}(1)\big]\bigg[\frac{1}{\sqrt{n}}\sum_{t=1}^{n}\Cc_{t}(\vtheta_{0,n}^c, \widehat{\vtheta}_n^v) + o_{\P}(1)\bigg]\notag\\
		&= \Big(\begin{array}{c|c}
			\mLambda_{n,(1)}^{-1}\mLambda_{n,(2)}\mLambda_n^{-1}+o_{\P}(1) & - \mLambda_{n,(1)}^{-1}+o_{\P}(1)\end{array}\Big)
		\begin{pmatrix}
			\frac{1}{\sqrt{n}}\sum_{t=1}^{n}\Vv_{t}(\vtheta_0^v) + o_{\P}(1)\\
			\frac{1}{\sqrt{n}}\sum_{t=1}^{n}\Cc_{t}(\vtheta_{0,n}^c, \widehat{\vtheta}_n^v) + o_{\P}(1)
		\end{pmatrix}.\label{eq:(C.2p)}
	\end{align}
	From simple adaptations of the proofs of Lemmas~V.\ref{lem:7} and C.\ref{lem:7 tilde}, we have that, as $n\to\infty$,
	\begin{equation}\label{eq:(C.2pp)}
		n^{-1/2}\mM_{n}^{-1/2}\sum_{t=1}^{n}\begin{pmatrix}\Vv_{t}(\vtheta_0^v)\\ \Cc_{t}(\vtheta_{0,n}^c,\widehat{\vtheta}_n^v)\end{pmatrix}\overset{d}{\longrightarrow}N(\vzeros, \mI),
	\end{equation}
	where 
	\begin{equation}\label{eq:(C.4+)}
		\mM_{n}  =\frac{1}{n}\sum_{t=1}^{n}\begin{pmatrix}
			\E\big[\Vv_{t}(\vtheta_0^v)\Vv_{t}^\prime(\vtheta_0^v)\big] & \E\big[\Vv_{t}(\vtheta_0^v)\Cc_{t}^\prime(\vtheta_{0}^c,\vtheta_0^v)\big]\\
			\E\big[\Cc_{t}(\vtheta_{0}^c,\vtheta_0^v)\Vv_{t}^\prime(\vtheta_0^v)\big] & \E\big[\Cc_{t}(\vtheta_{0}^c,\vtheta_0^v)\Cc_{t}^\prime(\vtheta_{0}^c,\vtheta_0^v)\big]
		\end{pmatrix}= \begin{pmatrix}
			\mV_n & \vzeros\\
			\vzeros & \mC_{n}
		\end{pmatrix}
	\end{equation}
	is positive definite by Assumption~\ref{ass:an} \ref{it:pd}. 
	For the above display, we refer to \eqref{eq:V simpl} and \eqref{eq:B simpl} (in Sections~\ref{add results VaR norm} and \ref{add results CoVaR norm}, respectively) for the upper left and lower right block matrices.
	The zero matrices in \eqref{eq:(C.4+)} arise because
	\begin{align*}
		\E\big[\Cc_{t}(\vtheta_{0}^c,\vtheta_0^v)\Vv_{t}^\prime(\vtheta_0^v)\big] &= \E\Big[ \nabla c_t(\vtheta_0^c)\nabla^\prime v_t(\vtheta_0^v)\1_{\{X_t>v_t(\vtheta_0^v)\}} \big(\1_{\{X_t\leq v_t(\vtheta_0^v)\}} - \beta\big)\big(\1_{\{Y_t\leq c_t(\vtheta_0^v)\}} - \alpha\big)\Big]\\
		&= -\beta\E\bigg\{ \nabla c_t(\vtheta_0^c)\nabla^\prime v_t(\vtheta_0^v) \E_{t-1}\Big[\1_{\{X_t>v_t(\vtheta_0^v)\}}\big(\1_{\{Y_t\leq c_t(\vtheta_0^v)\}} - \alpha\big)\Big]\bigg\}\\
		&= -\beta\E\bigg\{ \nabla c_t(\vtheta_0^c)\nabla^\prime v_t(\vtheta_0^v) \Big[ \P_{t-1}\big\{X_t>v_t(\vtheta_0^v),\ Y_t\leq c_t(\vtheta_0^v)\big\} -\alpha(1-\beta) \Big]\bigg\}\\
		&=  -\beta\E\bigg\{ \nabla c_t(\vtheta_0^c)\nabla^\prime v_t(\vtheta_0^v) \Big[ \P_{t-1}\big\{Y_t\leq c_t(\vtheta_0^v)\mid X_t>v_t(\vtheta_0^v)\big\}(1-\beta) -\alpha(1-\beta) \Big]\bigg\}\\
		&=  -\beta\E\bigg\{ \nabla c_t(\vtheta_0^c)\nabla^\prime v_t(\vtheta_0^v) \Big[ \alpha(1-\beta) -\alpha(1-\beta) \Big]\bigg\}=\vzeros
	\end{align*}
	and, similarly, $\E\big[\Vv_{t}(\vtheta_0^v)\Cc_{t}^\prime(\vtheta_{0}^c,\vtheta_0^v)\big]=\vzeros$. 
	We obtain from Lemma~V.\ref{lem:1} and \eqref{eq:(C.2p)} that
	\begin{align}
		\sqrt{n}\begin{pmatrix}
			\widehat{\vtheta}_n^v - \vtheta_0^v\\
			\widehat{\vtheta}_n^c - \vtheta_0^c
		\end{pmatrix} &= \Bigg[\begin{pmatrix}
			-\mLambda_{n}^{-1} & \vzero\\
			\mLambda_{n,(1)}^{-1}\mLambda_{n,(2)}\mLambda_{n}^{-1} & -\mLambda_{n,(1)}^{-1}
		\end{pmatrix} + o_{\P}(1)\Bigg]\notag\\
		&\hspace{4cm}\times\Bigg[\mM_n^{1/2}n^{-1/2}\mM_n^{-1/2}\sum_{t=1}^{n}\begin{pmatrix}
			\Vv_{t}(\vtheta_0^v)\\
			\Cc_{t}(\vtheta_{0,n}^c,\widehat{\vtheta}_n^v)
		\end{pmatrix} + o_{\P}(1)\Bigg]\notag\\
		&=-\overline{\mGamma}_n^{-1}\mM_n^{1/2}n^{-1/2}\mM_n^{-1/2}\sum_{t=1}^{n}\begin{pmatrix}
			\Vv_{t}(\vtheta_0^v)\\
			\Cc_{t}(\vtheta_{0,n}^c,\widehat{\vtheta}_n^v)
		\end{pmatrix} + o_{\P}(1),\label{eq:tbmt}
	\end{align}
	where we exploited that (by definition of $\overline{\mGamma}_n$)
	\begin{equation}\label{eq:(C.6p)}
		\overline{\mGamma}_n^{-1}=\begin{pmatrix}
			\mLambda_{n}^{-1} & \vzero\\
			-\mLambda_{n,(1)}^{-1}\mLambda_{n,(2)}\mLambda_{n}^{-1} & \mLambda_{n,(1)}^{-1}
		\end{pmatrix}.
	\end{equation}
	By Assumption~\ref{ass:an}~\ref{it:pd}, we may premultiply \eqref{eq:tbmt} through with $\mM_n^{-1/2}\overline{\mGamma}_n$ to obtain that, as $n\to\infty$,
	\begin{align*}
		\sqrt{n}\mM_n^{-1/2}\overline{\mGamma}_n\begin{pmatrix}
			\widehat{\vtheta}_n^v - \vtheta_0^v\\
			\widehat{\vtheta}_n^c - \vtheta_0^c
		\end{pmatrix}  &= -n^{-1/2}\mM_n^{-1/2}\sum_{t=1}^{n}\begin{pmatrix}
			\Vv_{t}(\vtheta_0^v)\\
			\Cc_{t}(\vtheta_{0,n}^c,\widehat{\vtheta}_n^v)
		\end{pmatrix} + o_{\P}(1)\\
		&\overset{d}{\longrightarrow}N(\vzero,\mI),
	\end{align*}
	where the convergence follows from \eqref{eq:(C.2pp)}.
	This, however, is just the conclusion.
\end{proof}

\section{Proofs of Lemmas~V.\ref{lem:1} and V.\ref{lem:7}}\label{add results VaR norm}

Before proving Lemmas~V.\ref{lem:1} and V.\ref{lem:7}, we first need to introduce some notation.
Recall from \eqref{eq:gt prev} that
\begin{equation}\label{eq:gt}
	\Vv_t(\vtheta^v)= \nabla v_t(\vtheta^v)\big[\1_{\{X_t\leq v_t(\vtheta^v)\}}-\beta\big].
\end{equation}
Thus, by the LIE and Assumption~\ref{ass:cons} \ref{it:smooth},
\begin{align*}
	\E\big[\Vv_{t}(\vtheta^v)\big] &= \E\Big\{ \nabla v_t(\vtheta^v)\E_{t-1}\big[\1_{\{X_t\leq v_t(\vtheta^v)\}}-\beta\big] \Big\}\\	
	&= \E\Big\{ \nabla v_t(\vtheta^v)\big[\P_{t-1}\{X_t\leq v_t(\vtheta^v)\}-\beta\big] \Big\}\\	
	&=\E\Big\{\nabla v_t(\vtheta^v)\big[F_t^{X}\big(v_t(\vtheta^v)\big)-\beta\big]\Big\}.
\end{align*}
Assumption~\ref{ass:an} and the dominated convergence theorem (DCT) allow us to interchange differentiation and expectation to yield that
\begin{equation}
	\frac{\partial}{\partial \vtheta^v}\E\big[\Vv_{t}(\vtheta^v)\big] =\E\Big\{\nabla^2 v_t(\vtheta^v)\big[F_t^{X}\big(v_t(\vtheta^v)\big)-\beta\big] + \nabla v_t(\vtheta^v)\nabla^\prime v_t(\vtheta^v)f_t^{X}\big(v_t(\vtheta^v)\big)\Big\}.\label{eq:Lambda v}
\end{equation}
Evaluating this quantity at the true parameter gives
\begin{equation*}
	\frac{\partial}{\partial \vtheta^v}\E\big[\Vv_{t}(\vtheta^v)\big]\Big\vert_{\vtheta^v=\vtheta_0^v} =\E\Big[\nabla v_t(\vtheta_0^v)\nabla^\prime v_t(\vtheta_0^v)f_t^{X}\big(v_t(\vtheta_0^v)\big)\Big].
\end{equation*}
We define
\begin{align}
	\vlambda_n(\vtheta^v) &= \frac{1}{n}\sum_{t=1}^{n}\E\big[\Vv_{t}(\vtheta^v)\big],\notag\\
	\mLambda_n(\vtheta^\ast) &= \frac{1}{n}\sum_{t=1}^{n}\frac{\partial}{\partial\vtheta^v}\E\big[\Vv_{t}(\vtheta^v)\big]\Big\vert_{\vtheta^v=\vtheta^\ast}.\label{eq:(C.41)}
\end{align}
With $\mLambda_{n}$ as defined in Assumption~\ref{ass:an}~\ref{it:pd}, it then holds that
\[
\mLambda_n = \mLambda_n(\vtheta_0^v) = \frac{1}{n}\sum_{t=1}^{n}\E\Big[\nabla v_t(\vtheta_0^v)\nabla^\prime v_t(\vtheta_0^v)f_t^{X}\big(v_t(\vtheta_0^v)\big)\Big].
\]
We are now in a position to prove Lemma~V.\ref{lem:1}.

\begin{proof}[{\textbf{Proof of Lemma~V.\ref{lem:1}:}}]
	The mean value theorem (MVT) implies that for all $i=1,\ldots,p$,
	\[
	\lambda_n^{(i)}(\widehat{\vtheta}_n^{v}) = \lambda_n^{(i)}(\vtheta_0^{v}) + \mLambda_n^{(i,\cdot)}(\vtheta^\ast_{i}) (\widehat{\vtheta}_n^v- \vtheta_0^v),
	\]
	where $\vlambda_n(\cdot)=\big(\lambda_n^{(1)}(\cdot),\ldots,\lambda_n^{(p)}(\cdot)\big)^\prime$, $\mLambda_n^{(i,\cdot)}(\cdot)$ denotes the $i$-th row of $\mLambda_n(\cdot)$ and $\vtheta^\ast_{i}$ lies on the line connecting $\vtheta_0^v$ and $\widehat{\vtheta}_n^v$. To economize on notation, we shall slightly abuse notation (here and elsewhere) by writing this as
	\begin{equation}\label{eq:NE MVT}
		\vlambda_n(\widehat{\vtheta}_n^{v}) = \vlambda_n(\vtheta_0^{v}) + \mLambda_n(\vtheta^\ast) (\widehat{\vtheta}_n^v- \vtheta_0^v)
	\end{equation}
	for some $\vtheta^\ast$ between $\widehat{\vtheta}_n^{v}$ and $\vtheta_0^{v}$; keeping in mind that the value of $\vtheta^\ast$ is in fact different from row to row in $\mLambda_n(\vtheta^\ast)$. However, this does not change any of the subsequent arguments. Interpreted verbatim, \eqref{eq:NE MVT} would be an instance of what \citet{Fea13} call the \textit{non-existent mean value theorem}, which is widely applied in statistics; e.g., by \citet{EM04} and \citet{PZC19}.
	
	We have that
	\begin{equation}\label{eq:lambda nought}
		\vlambda_n(\vtheta_0^{v}) = \frac{1}{n}\sum_{t=1}^{n}\E\Big\{\nabla v_t(\vtheta_0^{v})\big[F_t^{X}\big(v_t(\vtheta_0^{v})\big)-\beta\big]\Big\} =\vzeros,
	\end{equation}
	since $F_t^{X}\big(v_t(\vtheta_0^{v})\big)=\beta$ by correct specification. Plugging this into \eqref{eq:NE MVT} gives
	\[
	\vlambda_n(\widehat{\vtheta}_n^{v}) = \mLambda_n(\vtheta^\ast) (\widehat{\vtheta}_n^v- \vtheta_0^v).
	\]
	To establish the claim of the lemma, we therefore only have to show that
	\begin{enumerate}
		\item[(i)] $\mLambda_n^{-1}(\vtheta^\ast)-\mLambda_n^{-1} = o_{\P}(1)$;
		\item[(ii)] $\sqrt{n}\vlambda_n(\widehat{\vtheta}_n^{v})=-\frac{1}{\sqrt{n}}\sum_{t=1}^{n}\Vv_{t}(\vtheta_0^v) + o_{\P}(1)$.
	\end{enumerate}

	Claim (i) is verified in Lemma~V.\ref{lem:1+}. For this, note that since $\vtheta^\ast$ is a mean value between $\widehat{\vtheta}_n^v$ and $\vtheta_0^v$, and $\widehat{\vtheta}_n^v\overset{\P}{\longrightarrow}\vtheta_0^v$ (from Theorem~\ref{thm:cons}), it also follows that $\vtheta^\ast\overset{\P}{\longrightarrow}\vtheta_0^v$.

	To prove (ii), by Lemma~A.1 in \citet{Wei91}, it suffices to show that
	\begin{enumerate}
		\item[(ii.a)] conditions (N1)--(N5) in the notation of \citet{Wei91} hold;
		\item[(ii.b)] $\frac{1}{\sqrt{n}}\sum_{t=1}^{n}\Vv_{t}(\widehat{\vtheta}_n^{v})=o_{\P}(1)$;
		\item[(ii.c)] $\widehat{\vtheta}_n^{v}\overset{\P}{\longrightarrow}\vtheta_0^v$.
	\end{enumerate}
	For (ii.a), note that (N1) is shown in Lemma~V.\ref{lem:0} and (N2) follows from \eqref{eq:lambda nought}. The mixing condition (N5) is implied by our Assumption~\ref{ass:an} \ref{it:mixing}. Condition (N3) is verified in Lemmas~V.\ref{lem:3}--V.\ref{lem:5} below and the remaining condition (N4) in Lemma~V.\ref{lem:6}. The result in (ii.b) follows from Lemma~V.\ref{lem:2} and (ii.c) follows from Theorem~\ref{thm:cons}. In sum, the desired result follows.
\end{proof}

\setcounter{lemV}{2}

\begin{lemV}\label{lem:1+}
	Suppose Assumptions~\ref{ass:cons} and \ref{ass:an} hold. Then, as $n\to\infty$, $\mLambda_n^{-1}(\vtheta^\ast)-\mLambda_n^{-1}\overset{\P}{\longrightarrow}\vzero$ for any $\vtheta^\ast$ with $\vtheta^\ast\overset{\P}{\longrightarrow}\vtheta_0^v$.
\end{lemV}

\begin{proof}
	We first show that $\big\Vert\mLambda_{n}(\vtau) - \mLambda_{n}(\vtheta)\big\Vert \leq C \big\Vert\vtau - \vtheta\big\Vert$ for all $\vtau,\vtheta\in\mathcal{N}(\vtheta_0^v)$, where $\mathcal{N}(\vtheta_0^v)$ is some neighborhood of $\vtheta_0^v$ such that Assumptions~\ref{ass:cons} \ref{it:bound} and \ref{ass:an} \ref{it:bound2.1} hold.
	
	%
	
	Use \eqref{eq:Lambda v} and \eqref{eq:(C.41)} to write
	\begin{align}
		\big\Vert\mLambda_{n}(\vtau) - \mLambda_{n}(\vtheta)\big\Vert &= \bigg\Vert \frac{1}{n}\sum_{t=1}^{n} \E\Big\{\nabla^2 v_t(\vtau)\big[F_t^{X}\big(v_t(\vtau)\big)-\beta\big] - \nabla^2 v_t(\vtheta)\big[F_t^{X}\big(v_t(\vtheta)\big)-\beta\big]\notag\\
		&	\hspace{2cm} + \nabla v_t(\vtau)\nabla^\prime v_t(\vtau)f_t^{X}\big(v_t(\vtau)\big) - \nabla v_t(\vtheta)\nabla^\prime v_t(\vtheta)f_t^{X}\big(v_t(\vtheta)\big)\Big\}\bigg\Vert\notag\\
		&\leq \frac{1}{n}\sum_{t=1}^{n} \E\Big\Vert\nabla^2 v_t(\vtau)\Big[F_t^{X}\big(v_t(\vtau)\big)-F_t^{X}\big(v_t(\vtheta)\big)\Big]\Big\Vert\notag\\
		& \hspace{2cm} + \E\Big\Vert F_t^{X}\big(v_t(\vtheta)\big)\big[\nabla^2 v_t(\vtau) - \nabla^2 v_t(\vtheta)\big]\Big\Vert\notag\\
		& \hspace{2cm} +  \beta \E\Big\Vert \nabla^2 v_t(\vtau) - \nabla^2 v_t(\vtheta)\Big\Vert\notag\\
		&	\hspace{2cm} + \E\Big\Vert\nabla v_t(\vtau)\nabla^\prime v_t(\vtau)f_t^{X}\big(v_t(\vtau)\big)- \nabla v_t(\vtheta)\nabla^\prime v_t(\vtheta)f_t^{X}\big(v_t(\vtheta)\big)\Big\Vert.\label{eq:decomp Lambda}
	\end{align}
	By Assumption~\ref{ass:an} \ref{it:bound2.1}, we have that
	\begin{align}
		\E\Big\Vert F_t^{X}\big(v_t(\vtheta)\big)\big[\nabla^2 v_t(\vtau) - \nabla^2 v_t(\vtheta)\big]\Big\Vert
		&\leq \E\Big[\big\Vert F_t^{X}\big(v_t(\vtheta)\big)\big\Vert\cdot \big\Vert\nabla^2 v_t(\vtau) - \nabla^2 v_t(\vtheta)\big\Vert\Big]\notag\\	
		&\leq \E\big[V_3(\mathcal{F}_{t-1})\big] \big\Vert\vtau-\vtheta\big\Vert\leq C\big\Vert\vtau-\vtheta\big\Vert,\label{eq:(D.7m)}\\
		\beta \E\big\Vert \nabla^2 v_t(\vtau) - \nabla^2 v_t(\vtheta)\big\Vert & \leq \beta \E\big[V_3(\mathcal{F}_{t-1})\big]\big\Vert\vtau - \vtheta\big\Vert\leq C \Vert\vtau - \vtheta\Vert.\label{eq:nought term}
	\end{align}
	The other two terms in \eqref{eq:decomp Lambda} can be dealt with as follows. First, the MVT and Assumption~\ref{ass:an} imply that
	\begin{align}
		\E\Big\Vert\nabla^2 v_t(\vtau)\big[ F_t^{X}\big(v_t(\vtau)\big)-F_t^{X}\big(v_t(\vtheta)\big)\big]\Big\Vert & \leq \E\Big[\big\Vert\nabla^2 v_t(\vtau)\big\Vert \cdot \big\Vert F_t^{X}\big(v_t(\vtau)\big)-F_t^{X}\big(v_t(\vtheta)\big)\big\Vert\Big]\notag\\	
		&\leq \E\Big[ V_2(\mathcal{F}_{t-1})f_t^{X}\big(v_t(\vtheta^{\ast})\big)\big|\nabla^\prime v_t(\vtheta^{\ast})(\vtau-\vtheta)\big|\Big]\notag\\
		&\leq K \E\big[V_1(\mathcal{F}_{t-1})V_2(\mathcal{F}_{t-1})\big]\Vert\vtau-\vtheta\Vert\notag\\
		&\leq K \Big\{\E\big[V_1^3(\mathcal{F}_{t-1})\big]\Big\}^{1/3}\Big\{\E\big[V_2^{3/2}(\mathcal{F}_{t-1})\big]\Big\}^{2/3}\Vert\vtau-\vtheta\Vert\notag\\
		&\leq C \Vert\vtau-\vtheta\Vert,\label{eq:first term}
	\end{align}
	where $\vtheta^{\ast}$ is some mean value on the line connecting $\vtau$ and $\vtheta$, and the penultimate step uses H\"{o}lder's inequality.
	
	Second, we again use the MVT to obtain for some $\vtheta^{\ast}$ between $\vtau$ and $\vtheta$ (where $\vtheta^{\ast}$ may vary from line to line) that
	\begin{align}
		&\E\Big\Vert\nabla v_t(\vtau)\nabla^\prime v_t(\vtau)f_t^{X}\big(v_t(\vtau)\big) - \nabla v_t(\vtheta)\nabla^\prime v_t(\vtheta)f_t^{X}\big(v_t(\vtheta)\big)\Big\Vert\notag\\
		&=\bigg\Vert\E\Big[\nabla v_t(\vtau)\nabla^\prime v_t(\vtau)f_t^{X}\big(v_t(\vtau)\big) - \nabla v_t(\vtheta)\nabla^\prime v_t(\vtau)f_t^{X}\big(v_t(\vtau)\big)\notag\\
		&\hspace{1.8cm} + \nabla v_t(\vtheta)\nabla^\prime v_t(\vtau)f_t^{X}\big(v_t(\vtau)\big) - \nabla v_t(\vtheta)\nabla^\prime v_t(\vtheta)f_t^{X}\big(v_t(\vtau)\big)\notag\\
		&\hspace{1.8cm} + \nabla v_t(\vtheta)\nabla^\prime v_t(\vtheta)f_t^{X}\big(v_t(\vtau)\big) - \nabla v_t(\vtheta)\nabla^\prime v_t(\vtheta)f_t^{X}\big(v_t(\vtheta)\big)\Big] \bigg\Vert\notag\\
		&=\bigg\Vert \E\Big[\nabla^2 v_t(\vtheta^{\ast})(\vtau - \vtheta)\nabla^\prime v_t(\vtau)f_t^{X}\big(v_t(\vtau)\big)\notag\\
		&\hspace{1.8cm} + \nabla v_t(\vtheta)(\vtau - \vtheta)^\prime [\nabla^2 v_t(\vtheta^{\ast})]^\prime f_t^{X}\big(v_t(\vtau)\big) \notag\\
		&\hspace{1.8cm} + \nabla v_t(\vtheta)\nabla^\prime v_t(\vtheta)\big\{f_t^{X}\big(v_t(\vtau)\big) - f_t^{X}\big(v_t(\vtheta)\big)\big\}\Big] \bigg\Vert\notag\\
		&\leq \E\big[KV_2(\mathcal{F}_{t-1})V_1(\mathcal{F}_{t-1}) + KV_1(\mathcal{F}_{t-1})V_2(\mathcal{F}_{t-1}) + K V_1^3(\mathcal{F}_{t-1}) \big]\Vert\vtau - \vtheta\Vert\notag\\
		&\leq K\bigg(2\Big\{\E\big[V_1^3(\mathcal{F}_{t-1})\big]\Big\}^{1/3}\Big\{\E\big[V_2^{3/2}(\mathcal{F}_{t-1})\big]\Big\}^{2/3}  + \Big\{\E\big[V_1^3(\mathcal{F}_{t-1})\big]\Big\}\bigg)\Vert\vtau - \vtheta\Vert\notag\\
		&\leq C\Vert\vtau - \vtheta\Vert,\label{eq:second term}
	\end{align}
	where we used Assumption~\ref{ass:cons} to conclude that
	\[
	\big|f_t^{X}\big(v_t(\vtau)\big) - f_t^{X}\big(v_t(\vtheta)\big)\big| \leq K \big|v_t(\vtau) - v_t(\vtheta)\big|\leq K\big|\nabla^\prime v_t(\vtheta^\ast)(\vtau - \vtheta)\big|\leq K V_1(\mathcal{F}_{t-1})\big\Vert\vtau - \vtheta\big\Vert
	\]
	for the first inequality.

	Plugging \eqref{eq:(D.7m)}--\eqref{eq:second term} into \eqref{eq:decomp Lambda}, we get that
	\begin{equation}\label{eq:cty Lambda_n}
		\big\Vert\mLambda_{n}(\vtau) - \mLambda_{n}(\vtheta)\big\Vert \leq C \Vert\vtau - \vtheta\Vert.
	\end{equation}
	
	Using this, we obtain for $\delta>0$ with $\big\{\vtheta\in\mathbb{R}^{p}\colon\Vert\vtheta-\vtheta_0^v\Vert\leq\delta\big\}\subset\mathcal{N}(\vtheta_0^v)$ that
	\begin{align*}
		\P\Big\{\big\Vert\mLambda_n(\vtheta^\ast) -\mLambda_n \big\Vert >\varepsilon\Big\}&=\P\Big\{\big\Vert\mLambda_n(\vtheta^\ast) -\mLambda_n(\vtheta_0^v)\big\Vert >\varepsilon\Big\}\\
		&\leq \P\Big\{\big\Vert\mLambda_n(\vtheta^\ast) -\mLambda_n(\vtheta_0^v)\big\Vert >\varepsilon,\ \big\Vert\vtheta^\ast - \vtheta_0^v\big\Vert\leq\delta\Big\} + \P\Big\{\big\Vert\vtheta^\ast - \vtheta_0^v\big\Vert>\delta\Big\}\\
		&\leq \P\bigg\{\sup_{\Vert\vtheta-\vtheta_0^v\Vert\leq\delta}\big\Vert\mLambda_n(\vtheta) -\mLambda_n(\vtheta_0^v)\big\Vert >\varepsilon\bigg\} + o(1)\\
		&\leq \P\bigg\{\sup_{\Vert\vtheta-\vtheta_0^v\Vert\leq\delta}C\big\Vert\vtheta -\vtheta_0^v\big\Vert >\varepsilon\bigg\} + o(1)\\
		&=0+o(1),
	\end{align*}
	where the final line additionally requires $\delta<\varepsilon/C$. 
	This shows that $\mLambda_n(\vtheta^\ast)-\mLambda_n\overset{\P}{\longrightarrow}\vzero$.

	By Assumption~\ref{ass:an} \ref{it:pd}, $\mLambda_{n}(\vtheta_0^v)=\mLambda_n$ is non-singular.
	Continuity of $\mLambda_n(\cdot)$ (recall \eqref{eq:cty Lambda_n}) therefore implies that $\mLambda_n(\vtheta^v)$ is non-singular in a neighborhood of $\vtheta_0^v$. 
	Hence, continuity of the inverse combined with \citet[Lemma~4.29]{Whi01} and $\mLambda_n(\vtheta^\ast)-\mLambda_n\overset{\P}{\longrightarrow}\vzero$ implies that $\mLambda_n^{-1}(\vtheta^\ast)-\mLambda_n^{-1}\overset{\P}{\longrightarrow}\vzero$, as desired.
\end{proof}

\begin{lemV}\label{lem:0}
	Suppose Assumptions~\ref{ass:cons} and \ref{ass:an} hold. Then, condition (N1) of \citet{Wei91} holds, i.e., for all $t\in\mathbb{N}$, the stochastic process $\Omega\times\mTheta^v\ni(\omega,\vtheta^v)\mapsto \Vv_t(\vtheta^v)$ is separable in the sense of \citet[pp.~51--52]{Doo53}, and $\Vv_t(\vtheta^v)$ is measurable for each $\vtheta^v\in\mTheta^v$.
\end{lemV}

\begin{proof}
	From \eqref{eq:gt} it follows that $\Vv_t(\vtheta^v)$ is a function of measurable random variables (by Assumption~\ref{ass:cons}~\ref{it:smooth}) and, therefore, is itself measurable.
	
	It remains to establish separability.
	By definition \citep[cf.~also ][Chapter~III.2]{GikhmanSkorokhod2004}, we have to prove that for any open set $G\subset\mTheta^v$ and any closed set $F\subset\mathbb{R}^p$, the sets
	\begin{align*}
		&\big\{\omega\in\Omega\colon \Vv_t(\vtheta^v)\in F\ \text{for all}\ \vtheta^v\in G\big\}\quad\text{and}\\
		&\big\{\omega\in\Omega\colon \Vv_t(\vtheta^v)\in F\ \text{for all}\ \vtheta^v\in G\cap\mathbb{Q}^{p}\big\}
	\end{align*}
	differ from each other by at most a subset of some null set $N\subset\Omega$. 
	To do so, we show that the two sets coincide outside of the null set where $\nabla v_t(\cdot)=\vzero$ (see Assumption~\ref{ass:an}~\ref{it:diff}).
	The inclusion ``$\subset$'' is immediate, so we only prove ``$\supset$''. 
	
	To that end, fix some closed $F\subset\mathbb{R}^p$ and some open $G\subset\mTheta^v$, and let $\omega\in\Omega$ be such, that $\Vv_t(\vtheta^v)\in F$ for all $\vtheta^v\in G\cap\mathbb{Q}^{p}$ and $\nabla v_t(\vtheta^v)\neq\vzero$ for all $\vtheta^v\in G\cap\mathbb{Q}^{p}$ (such that we are outside of the null set where $\nabla v_t(\cdot)=\vzero$). 
	Then, we have to show that 
	\begin{equation*}
		\Vv_t(\vtheta^v)\in F\qquad\text{for all }\vtheta^v\in G\setminus\mathbb{Q}^{p}.
	\end{equation*}
	Fix some $\vtheta^v\in G\setminus\mathbb{Q}^{p}$. 
	
	If $\vtheta^v$ is a continuity point of $\Vv_t(\cdot)$, then $\Vv_t(\vtheta^v)\in F$ follows from arguments used to prove Proposition~C.2 in \citet{DimiBayer2019}.
	
	It remains to consider the case where $\vtheta^v$ is a point of discontinuity of $\Vv_t(\cdot)$. 
	From \eqref{eq:gt} and the continuous differentiability of $v_t(\cdot)$ (see Assumption~\ref{ass:an}~\ref{it:diff}) there is a discontinuity of $\Vv_t(\cdot)$ in $\vtheta^v$ only if for every neighborhood $\mathcal{N}(\vtheta^v)$ of $\vtheta^v$ there exist $\underline{\vtheta}^v,\overline{\vtheta}^v\in \mathcal{N}(\vtheta^v)$, such that
	\begin{equation}\label{eq:ineq sep}
		v_t(\underline{\vtheta}^v)<X_t=v_t(\vtheta^v)\leq v_t(\overline{\vtheta}^v)
	\end{equation}
	(because in that case there is a discontinuity in the indicator in \eqref{eq:gt}).
	In principle, two cases can occur: The first where there is a saddle point of $v_t(\cdot)$ in $\vtheta^v$ (such that \eqref{eq:ineq sep} only holds for $\overline{\vtheta}^v=\vtheta^v$ if the neighborhood is chosen sufficiently small), and the second where this is not the case. 
	However, the first case cannot occur, because if there were a saddle point of $v_t(\cdot)$ in $\vtheta^v$ then necessarily $\nabla v_t(\vtheta^v)=\vzero$, which is ruled out by our choice of $\omega$.


	Therefore, we only consider the second case, where there necessarily exists some $\overline{\vtheta}^v\in\mathcal{N}(\vtheta^v)\setminus\{\vtheta^v\}$, such that the right-hand side inequality in \eqref{eq:ineq sep} is strict.
	Since $G$ is open, the neighborhood can be chosen to lie in $G$, i.e., $\mathcal{N}(\vtheta^v)\subset G$.
	By continuity of $v_t(\cdot)$ (see Assumption~\ref{ass:cons}~\ref{it:diff c}) and \eqref{eq:ineq sep}, there exists in $\mathcal{N}(\vtheta^v)$ a connected set of points $\mathcal{L}$ connecting $\overline{\vtheta}^v$ and $\vtheta^v$, such that
	\[
	X_t=v_t(\vtheta^v)< v_t(\overline{\vtheta}^v_{\mathcal{L}})\qquad\text{for all}\quad \overline{\vtheta}^v_{\mathcal{L}}\in\mathcal{L}.
	\]
	Since $\mathbb{Q}^{p}$ is dense in $\mathbb{R}^p$, there exists a sequence $\vtheta_k^v$ in $\mathcal{L}\cap\mathbb{Q}^{p}\subset G\cap\mathbb{Q}^p$, such that $\vtheta_k^v\longrightarrow\vtheta^v$, as $k\to\infty$. 
	By continuity of $\Vv_t(\cdot)$ on $\mathcal{L}$, it holds that $\Vv_t(\vtheta_k^v)\longrightarrow\Vv_t(\vtheta^v)$, as $k\to\infty$. 
	From $\vtheta_k^v\in G\cap\mathbb{Q}^p$, it follows by assumption that $\Vv_t(\vtheta_k^v)\in F$. Because $F$ is a closed set, the limit of the convergent series $\Vv_t(\vtheta_k^v)$ necessarily lies in $F$, whence $\Vv_t(\vtheta^v)\in F$.
\end{proof}

\begin{lemV}\label{lem:3}
	Suppose Assumptions~\ref{ass:cons} and \ref{ass:an} hold. Then, condition (N3) (i) of \citet{Wei91} holds, i.e.,
	\[
	\big\Vert\vlambda_n(\vtheta)\big\Vert \geq a\big\Vert\vtheta-\vtheta_0^v\big\Vert\qquad\text{for}\ \big\Vert\vtheta-\vtheta_0^v\big\Vert\leq d_0
	\]
	for sufficiently large $n$ and some $a>0$ and $d_0>0$.
\end{lemV}

\begin{proof}
	Choose $d_0>0$ sufficiently small, such that $\big\{\vtheta\in\mTheta^v:\ \big\Vert\vtheta-\vtheta_0^v\big\Vert<d_0\big\}$ is a subset of the neighborhoods of Assumptions~\ref{ass:cons} \ref{it:bound} and \ref{ass:an} \ref{it:bound2.1}. 
	Similarly as in the proof of Lemma~V.\ref{lem:1}, the MVT and \eqref{eq:lambda nought} imply that
	\begin{align}
		\vlambda_n(\vtheta) &= \vlambda_n(\vtheta_0^{v}) + \mLambda_n(\vtheta^\ast)(\vtheta-\vtheta_0^v)\notag\\
		&= \mLambda_n(\vtheta^\ast)(\vtheta-\vtheta_0^v)\label{eq:(C.11+)}
	\end{align}
	for some $\vtheta^\ast$ between $\vtheta$ and $\vtheta_0^v$. (Again, to be precise we should allow for the mean value $\vtheta^\ast$ to vary across rows of $\mLambda_n(\vtheta^\ast)$; however, for the subsequent argument this does not matter.)
	Since $\mLambda_n=\mLambda_{n}(\vtheta_0^v)$, \eqref{eq:cty Lambda_n} implies that 
	\begin{equation}\label{eq:(D.10m)}
		\big\Vert\mLambda_n(\vtheta^{\ast}) - \mLambda_n\big\Vert\leq C \big\Vert\vtheta^\ast - \vtheta_0^v\big\Vert
	\end{equation}
	for some large enough $C>0$. 
	By Assumption~\ref{ass:an} \ref{it:pd}, $\mLambda_n$ has eigenvalues bounded below uniformly by some constant $a>0$. 
	Thus, since $\mLambda_n$ is symmetric, the eigenvalues of $\mLambda_n^\prime\mLambda_n=\mLambda_n\mLambda_n$ are bounded below by $a^2$.
	Theorem~1.13 and in particular Exercise 1.14.1 in \citet{MN19} then imply
	\begin{align}
		\vx^\prime \mLambda_n^\prime\mLambda_n \vx &\geq a^2 \vx^\prime\vx\qquad\text{for all }\vx\in\mathbb{R}^p\notag\\
		\intertext{or, equivalently,}
		\big\Vert\mLambda_n\vx \big\Vert^2 &\geq a^2\Vert\vx\Vert^2\qquad\text{for all }\vx\in\mathbb{R}^p.\label{eq:(D.15p)}
	\end{align}
	Therefore,
	\begin{align*}
		\big\Vert\vlambda_n(\vtheta)\big\Vert &= \big\Vert\mLambda_n(\vtheta^\ast)(\vtheta - \vtheta_0^v)\big\Vert\\
		&=\Big\Vert\mLambda_n(\vtheta - \vtheta_0^v) - \big[\mLambda_n - \mLambda_n(\vtheta^\ast)\big](\vtheta - \vtheta_0^v)\Big\Vert\\
		&\geq \big\Vert\mLambda_n(\vtheta - \vtheta_0^v)\big\Vert - \Big\Vert\big[\mLambda_n - \mLambda_n(\vtheta^\ast)\big](\vtheta - \vtheta_0^v)\Big\Vert\\
		&\geq a\big\Vert\vtheta-\vtheta_0^v\big\Vert - \big\Vert\mLambda_n - \mLambda_n(\vtheta^\ast)\big\Vert \cdot\big\Vert\vtheta - \vtheta_0^v\big\Vert\\
		&\geq \Big(a - C\big\Vert\vtheta - \vtheta_0^v\big\Vert\Big)\big\Vert\vtheta - \vtheta_0^v\big\Vert,
	\end{align*}
	where the first step follows from \eqref{eq:(C.11+)}, the third step from the triangle inequality, the fourth step from \eqref{eq:(D.15p)} and submultiplicativity of the Frobenius norm, and the final step from \eqref{eq:(D.10m)} combined with the fact that $\Vert\vtheta^\ast-\vtheta_0^v\Vert\leq\Vert\vtheta-\vtheta_0^v\Vert$ (since $\vtheta^\ast$ is a mean value between $\vtheta$ and $\vtheta_0^v$).
	The conclusion follows upon choosing $d_0$ small enough, such that $a - C\big\Vert\vtheta - \vtheta_0^v\big\Vert>0$.
\end{proof}

\begin{lemV}\label{lem:4}
	Suppose Assumptions~\ref{ass:cons} and \ref{ass:an} hold, and define
	\[
	\mu_t(\vtheta, d) =\sup_{\Vert\vtau-\vtheta\Vert\leq d}\big\Vert\Vv_{t}(\vtau) - \Vv_{t}(\vtheta)\big\Vert.
	\]
	Then, condition (N3) (ii) of \citet{Wei91} holds, i.e.,
	\[
	\E\big[\mu_t(\vtheta, d)\big] \leq bd\qquad\text{for}\ \big\Vert\vtheta-\vtheta_0^v\big\Vert + d\leq d_0
	\]
	for sufficiently large $n$ and some strictly positive $b$, $d$, $d_0$.
\end{lemV}

\begin{proof}
	Choose $d_0>0$ sufficiently small, such that $\big\{\vtheta\in\mTheta^v:\ \big\Vert\vtheta-\vtheta_0^v\big\Vert<d_0\big\}$ is a subset of the neighborhoods of Assumptions~\ref{ass:cons} \ref{it:bound} and \ref{ass:an} \ref{it:bound2.1}. Recalling the definition of $\Vv_{t}(\vtheta^v)$ from \eqref{eq:gt}, we decompose
	\begin{align*}
		\mu_t(\vtheta, d) &= \sup_{\Vert\vtau-\vtheta\Vert\leq d}\Big\Vert \nabla v_t(\vtau)\1_{\{X_t\leq v_t(\vtau)\}} - \nabla v_t(\vtheta)\1_{\{X_t\leq v_t(\vtheta)\}} + \beta\big[\nabla v_t(\vtheta)- \nabla v_t(\vtau)\big]\Big\Vert\\
		&\leq  \sup_{\Vert\vtau-\vtheta\Vert\leq d}\Big\Vert \nabla v_t(\vtau)\1_{\{X_t\leq v_t(\vtau)\}} - \nabla v_t(\vtheta)\1_{\{X_t\leq v_t(\vtheta)\}}\Big\Vert + \beta\sup_{\Vert\vtau-\vtheta\Vert\leq d}\big\Vert\nabla v_t(\vtau)-\nabla v_t(\vtheta)\big\Vert\\
		&=:\mu_{1t}(\vtheta, d) + \mu_{2t}(\vtheta, d).
	\end{align*}
	Similarly as in the proof of Theorem~\ref{thm:cons}, we define the $\mathcal{F}_{t-1}$-measurable quantities
	\begin{align*}
		\underline{\vtau} &:= \argmin_{\norm{\vtau-\vtheta}\leq d}v_t(\vtau),\\
		\overline{\vtau}  &:= \argmax_{\norm{\vtau-\vtheta}\leq d}v_t(\vtau),
	\end{align*}
	which exist by continuity of $v_t(\cdot)$.
	
	We first consider $\mu_{1t}(\vtheta, d)$. To be able to take the indicators out of the supremum in $\mu_{1t}(\vtheta, d)$, we distinguish two cases:
	
	\textbf{Case 1: $X_t\leq v_t(\vtheta)$}
	
	We further distinguish two cases (a)--(b):
	
	(a) If $X_t<v_t(\underline{\vtau})$, then both indicators in $\mu_{1t}(\vtheta,d)$ are equal to one, such that
	\[
	\mu_{1t}(\vtheta, d) = \sup_{\norm{\vtau-\vtheta}\leq d}\norm{\nabla v_t(\vtau) - \nabla v_t(\vtheta)}.
	\]
	
	(b) If $v_t(\underline{\vtau})\leq X_t \leq v_t(\overline{\vtau})$, then
	\begin{align}
		\mu_{1t}(\vtheta, d) &=\max\bigg\{\sup_{\substack{\norm{\vtau-\vtheta}\leq d\\ X_t\leq v_t(\vtau)}}\norm{\nabla v_t(\vtau) - \nabla v_t(\vtheta)},\ \norm{\nabla v_t(\vtheta)}\bigg\}\notag\\
		&\leq \sup_{\norm{\vtau-\vtheta}\leq d}\norm{\nabla v_t(\vtau) - \nabla v_t(\vtheta)} + \norm{\nabla v_t(\vtheta)}.\notag
	\end{align}
	
	(Notice that the third case that $X_t>v_t(\overline{\vtau})$ cannot occur, because already $X_t\leq v_t(\vtheta)$.)
	
	Together, (a) and (b) yield that
	\begin{align}
		\mu_{1t}(\vtheta, d) &\leq \sup_{\norm{\vtau-\vtheta}\leq d}\norm{\nabla v_t(\vtau) - \nabla v_t(\vtheta)} + \1_{\{v_t(\underline{\vtau})\leq X_t \leq v_t(\vtheta)\}}\norm{\nabla v_t(\vtheta)}\notag\\
		&\leq \sup_{\norm{\vtau-\vtheta}\leq d}\norm{\nabla v_t(\vtau) - \nabla v_t(\vtheta)} + \1_{\{v_t(\underline{\vtau})\leq X_t \leq v_t(\vtheta)\}}\sup_{\Vert\vtau-\vtheta_0^v\Vert\leq d_0}\norm{\nabla v_t(\vtau)},\label{eq:(p.9.1)}
	\end{align}
	where we used that $\big\Vert\vtheta-\vtheta_0^v\big\Vert\leq d_0-d\leq d_0$ to bound $\norm{\nabla v_t(\vtheta)}$ in the second line.
	
	\textbf{Case 2: $X_t> v_t(\vtheta)$}
	
	In this case,
	\begin{align}
		\mu_{1t}(\vtheta, d) &= \1_{\{X_t\leq v_t(\overline{\vtau})\}}\sup_{\substack{\Vert\vtau-\vtheta\Vert\leq d\\ X_t\leq v_t(\vtau)}}\norm{\nabla v_t(\vtau)}\notag\\
		&\leq \1_{\{X_t\leq v_t(\overline{\vtau})\}}\sup_{\Vert\vtau-\vtheta\Vert\leq d}\norm{\nabla v_t(\vtau)}.\label{eq:(N.12.2)}
	\end{align}
	Note that $\norm{\vtau - \vtheta}\leq d$ and our assumption $\norm{\vtheta-\vtheta_0^v}+d\leq d_0$ together imply that $\vtau$ is in a $d_0$-neighborhood of $\vtheta_0^v$, because $\norm{\vtau-\vtheta_0^v}=\norm{\vtau-\vtheta + \vtheta-\vtheta_0^v}\leq \norm{\vtau-\vtheta} + \norm{\vtheta-\vtheta_0^v}\leq d + (d_0-d)=d_0$.
	Hence, from \eqref{eq:(N.12.2)},
	\begin{equation}\label{eq:(p.9.2)}
		\mu_{1t}(\vtheta, d) \leq \1_{\{v_t(\vtheta)<X_t\leq v_t(\overline{\vtau})\}}\sup_{\Vert\vtau-\vtheta_0^v\Vert\leq d_0}\norm{\nabla v_t(\vtau)}.
	\end{equation}
	
	Combining the results from Cases 1 and 2 (i.e., \eqref{eq:(p.9.1)} and \eqref{eq:(p.9.2)}), we deduce that
	\begin{multline}
		\mu_{1t}(\vtheta, d) \leq \Big[ \1_{\{v_t(\underline{\vtau})\leq X_t\leq v_t(\vtheta)\}} + \1_{\{v_t(\vtheta)<X_t\leq v_t(\overline{\vtau})\}}\Big] \sup_{\Vert\vtau - \vtheta_0^v\Vert\leq d_0}\norm{\nabla v_t(\vtau)}\\
		+ \sup_{\Vert\vtau-\vtheta\Vert\leq d}\norm{\nabla v_t(\vtau) - \nabla v_t(\vtheta)}.\label{eq:mu t 1 decomp}
	\end{multline}
	By Assumptions~\ref{ass:cons}~\ref{it:cond dens} and \ref{it:bound} we have
	\begin{align}
		\E_{t-1}\Big[\1_{\{v_t(\underline{\vtau})\leq X_t\leq v_t(\vtheta)\}}\Big] &=\int_{v_t(\underline{\vtau})}^{v_t(\vtheta)}f_t^{X}(x)\D x\notag\\
		&\leq K\big|v_t(\vtheta) - v_t(\underline{\vtau})\big|\notag\\
		&= K\big|\nabla v_t(\vtheta^\ast)(\vtheta-\underline{\vtau})\big|\notag\\
		&\leq K V_1(\mathcal{F}_{t-1})\big\Vert\vtheta - \underline{\vtau}\big\Vert \notag\\
		&\leq K V_1(\mathcal{F}_{t-1}) d,\label{eq:bound 1}
	\end{align}
	and, similarly,
	\begin{equation}
		\E_{t-1}\Big[\1_{\{v_t(\vtheta)<X_t\leq v_t(\overline{\vtau})\}}\Big]  \leq K V_1(\mathcal{F}_{t-1}) d.\label{eq:bound 3}
	\end{equation}
	Moreover, the MVT implies that for some $\vtheta^\ast$ on the line connecting $\vtau$ and $\vtheta$,
	\begin{align}
		\sup_{\norm{\vtau-\vtheta}\leq d}\big\Vert\nabla v_t(\vtau) - \nabla v_t(\vtheta)\big\Vert &= \sup_{\norm{\vtau-\vtheta}\leq d}\big\Vert\nabla^2 v_t(\vtheta^\ast)(\vtau-\vtheta)\big\Vert\notag\\
		&\leq V_2(\mathcal{F}_{t-1})\norm{\vtau - \vtheta}\notag\\
		&\leq V_2(\mathcal{F}_{t-1}) d.\label{eq:(N.14.1)}
	\end{align}
	Therefore, utilizing \eqref{eq:mu t 1 decomp}--\eqref{eq:(N.14.1)},
	\[
	\E\big[\mu_{1t}(\vtheta,d)\big] \leq 2\E\big[K V_1^2(\mathcal{F}_{t-1})\big]d + \E\big[V_2(\mathcal{F}_{t-1})\big]d\leq Cd.
	\]
	
	By arguments leading to \eqref{eq:(N.14.1)}, we also have that
	\[
	\E\big[\mu_{2t}(\vtheta,d)\big]\leq\beta\E\big[V_2(\mathcal{F}_{t-1})\big]d\leq Cd.
	\]
	
	Overall, 
	\[
	\E\big[\mu_t(\vtheta,d)\big]\leq \E\big[\mu_{1t}(\vtheta,d)\big] + \E\big[\mu_{2t}(\vtheta,d)\big]\leq bd
	\]
	for some suitable $b>0$, as desired.
\end{proof}

\begin{lemV}\label{lem:5}
	Suppose Assumptions~\ref{ass:cons} and \ref{ass:an} hold. Then, condition (N3) (iii) of \citet{Wei91} holds, i.e.,
	\[
	\E\big[\mu_t^r(\vtheta, d)\big] \leq cd\qquad\text{for}\ \norm{\vtheta-\vtheta_0^v} + d\leq d_0
	\]
	for sufficiently large $n$ and some $c>0$, $d\geq0$, $d_0>0$ and $r>2$.
\end{lemV}

\begin{proof}
	The arguments resemble those in the proof of Lemma~V.\ref{lem:4}. We again pick $d_0>0$ sufficiently small, such that $\big\{\vtheta\in\mTheta^v:\ \norm{\vtheta-\vtheta_0^v}<d_0\big\}$ is a subset of the neighborhoods of Assumptions~\ref{ass:cons} \ref{it:bound} and \ref{ass:an} \ref{it:bound2.1}. We also work with $\mu_{it}(\vtheta, d)$ ($i=1,2$), $\underline{\vtau}$ and $\overline{\vtau}$ as defined in the proof of Lemma~V.\ref{lem:4}.
	
	By the $c_r$-inequality \citep[e.g.,][Theorem~9.28]{Dav94} and the fact that $\mu_t(\vtheta, d)\leq \mu_{1t}(\vtheta, d) + \mu_{2t}(\vtheta, d)$, we get that 
	\[
	\E\big[\mu_t^r(\vtheta, d)\big] \leq 2^{r-1}\sum_{i=1}^{2}\E\big[\mu_{it}^{r}(\vtheta, d)\big].
	\]
	Hence, to prove the claim, it suffices to show that $\E\big[\mu_{it}^{2+\iota}(\vtheta, d)\big]\leq cd$ for $\iota>0$ (from Assumption~\ref{ass:an} \ref{it:mom bounds cons2}) and some $c>0$.
	
	\textbf{First term ($\E\big[\mu_{1t}^{2+\iota}(\vtheta, d)\big]$):} Following the same arguments that led to \eqref{eq:mu t 1 decomp}, we obtain that
	\begin{multline*}
		\mu_{1t}^{2+\iota}(\vtheta, d) \leq \Big[\1_{\{v_t(\underline{\vtau})\leq X_t\leq v_t(\vtheta)\}} + \1_{\{v_t(\vtheta)<X_t\leq v_t(\overline{\vtau})\}}\Big]\sup_{\Vert\vtau - \vtheta_0^v\Vert\leq d_0}\norm{\nabla v_t(\vtau)}^{2+\iota}\\
		+ \sup_{\Vert\vtau-\vtheta\Vert\leq d}\norm{\nabla v_t(\vtau) - \nabla v_t(\vtheta)}^{2+\iota}.
	\end{multline*}
	By the LIE,
	\begin{multline}
		\E\big[\mu_{1t}^{2+\iota}(\vtheta, d)\big] \leq \E\bigg\{ \E_{t-1}\Big[\1_{\{v_t(\underline{\vtau})\leq X_t\leq v_t(\vtheta)\}} + \1_{\{v_t(\vtheta)<X_t\leq v_t(\overline{\vtau})\}}\Big]\sup_{\Vert\vtau - \vtheta_0^v\Vert\leq d_0}\norm{\nabla v_t(\vtau)}^{2+\iota}\bigg\}\\
		+ \E\bigg[\sup_{\Vert\vtau-\vtheta\Vert\leq d}\norm{\nabla v_t(\vtau) - \nabla v_t(\vtheta)}^{2+\iota}\bigg].\label{eq:mu t d bound}
	\end{multline}
	From \eqref{eq:bound 1}--\eqref{eq:bound 3},
	\[
	\E_{t-1}\Big[\1_{\{v_t(\underline{\vtau})\leq X_t\leq v_t(\vtheta)\}} + \1_{\{v_t(\vtheta)<X_t\leq v_t(\overline{\vtau})\}}\Big]\leq 2K V_1(\mathcal{F}_{t-1})d.
	\]
	From Assumption~\ref{ass:cons} \ref{it:bound},
	\[
	\sup_{\Vert\vtau - \vtheta_0^v\Vert\leq d_0}\norm{\nabla v_t(\vtau)}^{2+\iota}\leq V_1^{2+\iota}(\mathcal{F}_{t-1}),
	\]
	and from \eqref{eq:(N.14.1)}
	\begin{align}
		\sup_{\norm{\vtau-\vtheta}\leq d}\norm{\nabla v_t(\vtau) - \nabla v_t(\vtheta)}^{2+\iota}&= \sup_{\norm{\vtau-\vtheta}\leq d}\norm{\nabla v_t(\vtau) - \nabla v_t(\vtheta)}\times\sup_{\norm{\vtau-\vtheta}\leq d}\norm{\nabla v_t(\vtheta) - \nabla v_t(\vtau)}^{1+\iota}\notag\\
		&\leq V_2(\mathcal{F}_{t-1})d \times 2^\iota\bigg(\big\Vert \nabla v_t(\vtheta)\big\Vert^{1+\iota} + \sup_{\norm{\vtau-\vtheta}\leq d}\big\Vert\nabla v_t(\vtau) \big\Vert^{1+\iota}\bigg)\notag\\
		&\leq V_2(\mathcal{F}_{t-1})d \times 2^{1+\iota} V_1^{1+\iota}(\mathcal{F}_{t-1})\notag\\
		&=	2^{1+\iota}V_1^{1+\iota}(\mathcal{F}_{t-1})V_2(\mathcal{F}_{t-1})d\label{eq:(N.14.X)},
	\end{align}
	where we used the $c_r$-inequality in the second line.
	Inserting the above three relations into \eqref{eq:mu t d bound}, we get (using Assumption~\ref{ass:an} \ref{it:mom bounds cons2}) that
	\begin{align*}
		\E\big[\mu_{1t}^{2+\iota}(\vtheta, d)\big] &\leq \E\big[ 2K V_1(\mathcal{F}_{t-1}) V_1^{2+\iota}(\mathcal{F}_{t-1})d\big] + \E\big[2^{1+\iota}V_1^{1+\iota}(\mathcal{F}_{t-1})V_2(\mathcal{F}_{t-1})d\big]\\
		&\leq 2K \E\big[  V_1^{3+\iota}(\mathcal{F}_{t-1})\big] d  + 2^{1+\iota} \Big\{\E\big[V_1^{3+\iota}(\mathcal{F}_{t-1})\big]\Big\}^{(1+\iota)/(3+\iota)}\Big\{\E\big[V_2^{(3+\iota)/2}(\mathcal{F}_{t-1})\big]\Big\}^{2/(3+\iota)} d\\
		&\leq Cd,
	\end{align*}
	where the penultimate step follows from H\"{o}lder's inequality.
	
	\textbf{Second term ($\E\big[\mu_{2t}^{2+\iota}(\vtheta, d)\big]$):} By definition of $\mu_{2t}(\vtheta, d)$ and exploiting \eqref{eq:(N.14.X)}, we have that
	\begin{align*}
		\E\big[\mu_{2t}^{2+\iota}(\vtheta, d)\big]&= \beta^{2+\iota} \E\bigg[\sup_{\norm{\vtau-\vtheta}\leq d}\norm{\nabla v_t(\vtau) - \nabla v_t(\vtheta)}^{2+\iota}\bigg]\\
		&\leq \beta^{2+\iota} \E\big[2^{1+\iota}V_1^{1+\iota}(\mathcal{F}_{t-1})V_2(\mathcal{F}_{t-1})d\big]\\
		&\leq 2^{1+\iota}\beta^{2+\iota}\Big\{\E\big[V_1^{3+\iota}(\mathcal{F}_{t-1})\big]\Big\}^{(1+\iota)/(3+\iota)}\Big\{\E\big[V_2^{(3+\iota)/2}(\mathcal{F}_{t-1})\big]\Big\}^{2/(3+\iota)} d\\
		&\leq Cd.
	\end{align*}
	
	Combining the results for the first and second term, the conclusion follows.
\end{proof}

\begin{lemV}\label{lem:6}
	Suppose Assumptions~\ref{ass:cons} and \ref{ass:an} hold. Then,
	\[
	\E\norm{\Vv_{t}(\vtheta_0^v)}^{2+\iota}\leq C\qquad\text{for all }t\in\mathbb{N}.
	\]
	In particular, condition (N4) of \citet{Wei91} holds.
\end{lemV}

\begin{proof}
	The claim follows easily from the definition of $\Vv_{t}(\cdot)$ in \eqref{eq:gt}, which implies that
	\begin{align*}
		\E\norm{\Vv_{t}(\vtheta_0^v)}^{2+\iota} &\leq \E\norm{\nabla v_t(\vtheta_0^v)}^{2+\iota}\\
		&\leq \E\big[V_1^{2+\iota}(\mathcal{F}_{t-1})\big]\leq K
	\end{align*}
	by Assumption~\ref{ass:cons} \ref{it:bound} and Assumption~\ref{ass:an} \ref{it:mom bounds cons2}.
\end{proof}

\begin{lemV}\label{lem:2}
	Suppose Assumptions~\ref{ass:cons} and \ref{ass:an} hold. Then, as $n\to\infty$,
	\[
	\frac{1}{\sqrt{n}}\sum_{t=1}^{n}\Vv_{t}(\widehat{\vtheta}_n^{v})=o_{\P}(1).
	\]
\end{lemV}

\begin{proof}
	Recall from Assumption~\ref{ass:cons} \ref{it:compact} that $\mTheta^{v}\subset\mathbb{R}^{p}$, such that $\vtheta^v=(\theta_1^v,\ldots,\theta_p^{v})^\prime$. Let $\ve_1,\ldots,\ve_p$ denote the standard basis of $\mathbb{R}^{p}$. Define
	\[
	S_{j,n}^{\VaR}(a) := \frac{1}{\sqrt{n}}\sum_{t=1}^{n}S^{\VaR}\big(v_t(\widehat{\vtheta}_n^{v}+a \ve_j), X_t\big),\qquad j=1,\ldots,p,
	\]
	where $a\in\mathbb{R}$. Let $G_{j,n}(a)$ be the right partial derivative of $S_{j,n}^{\VaR}(a)$, such that (see \eqref{eq:gt})
	\[
	G_{j,n}(a)=\frac{1}{\sqrt{n}}\sum_{t=1}^{n}\nabla_j v_t(\widehat{\vtheta}_n^{v}+a \ve_j)\big[\1_{\{X_t\leq v_t(\widehat{\vtheta}_n^{v}+a \ve_j)\}}-\beta\big],
	\]
	where $\nabla_j v_t(\cdot)$ is the $j$-th component of $\nabla v_t(\cdot)$. Then, $G_{j,n}(0)=\lim_{\xi\downarrow0}G_{j,n}(\xi)$ is the right partial derivative of
	\[
	S_n^{\VaR}(\vtheta^{v}) := \frac{1}{\sqrt{n}}\sum_{t=1}^{n}S^{\VaR}\big(v_t(\vtheta^v), X_t\big)
	\]
	at $\widehat{\vtheta}_n^{v}$ in the direction $\theta_j^{v}$.
	Correspondingly, $\lim_{\xi\downarrow0}G_{j,n}(-\xi)$ is the left partial derivative. Because $S_n^{\VaR}(\cdot)$ achieves its minimum at $\widehat{\vtheta}_n^{v}$, the left derivative must be non-positive and the right derivative must be non-negative. Thus, for sufficiently small $\xi>0$,
	\begin{align*}
		\big|G_{j,n}(0)\big| &\leq G_{j,n}(\xi) - G_{j,n}(-\xi)\\
		&= \frac{1}{\sqrt{n}}\sum_{t=1}^{n}\nabla_j v_t(\widehat{\vtheta}_n^{v}+\xi \ve_j)\big[\1_{\{X_t\leq v_t(\widehat{\vtheta}_n^{v}+\xi \ve_j)\}}-\beta\big]\\
		&\hspace{2cm} - \frac{1}{\sqrt{n}}\sum_{t=1}^{n}\nabla_j v_t(\widehat{\vtheta}_n^{v}-\xi \ve_j)\big[\1_{\{X_t\leq v_t(\widehat{\vtheta}_n^{v}-\xi \ve_j)\}}-\beta\big].
	\end{align*}
	By continuity of $\nabla v_t(\cdot)$ (see Assumption~\ref{ass:an} \ref{it:diff}) it follows upon letting $\xi\to0$ that
	\begin{align}
		\big|G_{j,n}(0)\big| &\leq \frac{1}{\sqrt{n}}\sum_{t=1}^{n}\big|\nabla_j v_t(\widehat{\vtheta}_n^{v})\big|\1_{\{X_t= v_t(\widehat{\vtheta}_n^{v})\}}\notag\\
		&\leq \frac{1}{\sqrt{n}}\Big[\max_{t=1,\ldots,n}V_1(\mathcal{F}_{t-1})\Big]\sum_{t=1}^{n}\1_{\{X_t= v_t(\widehat{\vtheta}_n^{v})\}}.\label{eq:(N.7.1)}
	\end{align}
	Subadditivity, Markov's inequality and Assumption~\ref{ass:an} \ref{it:mom bounds cons2} imply
	\begin{align*}
		\P\Big\{n^{-1/2}\max_{t=1,\ldots,n}V_1(\mathcal{F}_{t-1})>\varepsilon\Big\} &\leq \sum_{t=1}^{n}\P\Big\{V_1(\mathcal{F}_{t-1})>\varepsilon n^{1/2}\Big\}\\
		&\leq \sum_{t=1}^{n}\varepsilon^{-3}n^{-3/2}\E\big[V_1^3(\mathcal{F}_{t-1})\big]\\
		& =O(n^{-1/2})=o(1).
	\end{align*}
	Combining this with Assumption~\ref{ass:an} \ref{it:eq bound}, we obtain from \eqref{eq:(N.7.1)} that
	\[
	\big|G_{j,n}(0)\big| \overset{\text{a.s.}}{=} o_{\P}(1)O(1)=o_{\P}(1).
	\]
	As this holds for every $j=1,\ldots,p$, we get that 
	\[
	\frac{1}{\sqrt{n}}\sum_{t=1}^{n}\Vv_{t}(\widehat{\vtheta}_n^{v})=\frac{1}{\sqrt{n}}\sum_{t=1}^{n}\nabla v_t(\widehat{\vtheta}_n^{v})\big[\1_{\{X_t\leq v_t(\widehat{\vtheta}_n^{v})\}} - \beta\big]=o_{\P}(1),
	\]
	which is just the conclusion.
\end{proof}

We have now established all lemmas that were used in the proof of Lemma~V.\ref{lem:1}.
Therefore, we may now turn to verifying Lemma~V.\ref{lem:7}:

\begin{proof}[{\textbf{Proof of Lemma~V.\ref{lem:7}:}}]
	By a Cram\'{e}r--Wold device, it suffices to show the univariate convergence
	\[
	n^{-1/2}\sum_{t=1}^{n}\vlambda^\prime\mV_n^{-1/2} \Vv_{t}(\vtheta_0^v)\overset{d}{\longrightarrow}N(0, 1),\qquad n\to\infty,
	\]
	for any $\vlambda\in\mathbb{R}^p$ with $\vlambda^\prime\vlambda=1$. This, however, follows easily from a standard central limit theorem for mixing random variables \citep[e.g.,][Theorem~5.20]{Whi01}, since $\E\big[\vlambda^\prime\mV_n^{-1/2} \Vv_{t}(\vtheta_0^v)\big]=0$, Lemma~V.\ref{lem:6} holds, the array $\big\{\vlambda^\prime\mV_n^{-1/2} \Vv_{t}(\vtheta_0^v)\big\}_{n\in\mathbb{N},\, t=1,\ldots,n}$ is (row-wise) $\alpha$-mixing of size $-\tilde q/( \tilde q-2)$ with $\tilde q>2$ from Assumption~\ref{ass:an}~\ref{it:mixing} and
	\begin{align*}
		\Var\bigg(n^{-1/2}\sum_{t=1}^{n}\vlambda^\prime\mV_n^{-1/2} \Vv_{t}(\vtheta_0^v)\bigg) &= \vlambda^\prime\mV_n^{-1/2}\bigg\{\frac{1}{n}\sum_{t=1}^{n}\Var\big( \Vv_{t}(\vtheta_0^v)\big)\bigg\}\mV_n^{-1/2}\vlambda\\
		&= \vlambda^\prime\mV_n^{-1/2}\bigg\{\frac{1}{n}\sum_{t=1}^{n}\E\big[ \Vv_{t}(\vtheta_0^v)\Vv_{t}^\prime(\vtheta_0^v)\big]\bigg\}\mV_n^{-1/2}\vlambda\\
		&= \vlambda^\prime\mV_n^{-1/2}\mV_n\mV_n^{-1/2}\vlambda\\
		&=\vlambda^\prime\vlambda=1>0.
	\end{align*}
	Here, the first equality exploits the fact that $\big\{\Vv_{t}(\vtheta_0^v)\big\}_{t\in\mathbb{N}}$ is a martingale difference sequence (because $\E_{t-1}\big[\Vv_{t}(\vtheta_0^v)\big]=\nabla v_t(\vtheta_0^v)\E_{t-1}\big[\1_{\{X_t\leq v_t(\vtheta_0^v)\}}-\beta\big]=\vzero$) and, hence, uncorrelated.
	For the third equality, note that
	\begin{align}
		\E\big[ \Vv_{t}(\vtheta_0^v)\Vv_{t}^\prime(\vtheta_0^v)\big]	&= \E\Big[ \nabla v_{t}(\vtheta_0^v)\nabla^\prime v_{t}(\vtheta_0^v)\big(\1_{\{X_t\leq v_t(\vtheta_0^v)\}} - \beta\big)^2\Big]\notag\\
		&= \E\bigg\{ \nabla v_{t}(\vtheta_0^v)\nabla^\prime v_{t}(\vtheta_0^v)\E_{t-1}\Big[\big(\1_{\{X_t\leq v_t(\vtheta_0^v)\}} - \beta\big)^2\Big]\bigg\}\notag\\
		&= \E\bigg\{ \nabla v_{t}(\vtheta_0^v)\nabla^\prime v_{t}(\vtheta_0^v)\E_{t-1}\Big[\big(\1_{\{X_t\leq v_t(\vtheta_0^v)\}} -  2\beta\1_{\{X_t\leq v_t(\vtheta_0^v)\}} +  \beta^2\big)\Big]\bigg\}\notag\\
		&= \E\bigg\{ \nabla v_{t}(\vtheta_0^v)\nabla^\prime v_{t}(\vtheta_0^v)\Big[(\beta - 2\beta^2 + \beta^2)\Big]\bigg\}\notag\\
		&=\beta(1-\beta)\E\big[\nabla v_t(\vtheta_0^v)\nabla^\prime v_t(\vtheta_0^v)\big].\label{eq:V simpl}
	\end{align}
	Theorem~5.20 of \citet{Whi01} now implies the desired univariate convergence, which concludes the proof.
\end{proof}

\section{Proofs of Lemmas~C.\ref{lem:1 tilde} and C.\ref{lem:7 tilde}}\label{add results CoVaR norm}

Before proving Lemmas~C.\ref{lem:1 tilde} and C.\ref{lem:7 tilde}, we have to introduce some additional notation.
Recall from \eqref{eq:(12+) prev} that
\begin{equation}\label{eq:(12+)}
	\Cc_{t}(\vtheta^c, \vtheta^v)= \1_{\{X_t> v_t(\vtheta^v)\}}\nabla c_t(\vtheta^c)\big[\1_{\{Y_t\leq c_t(\vtheta^c)\}}-\alpha\big].
\end{equation}
The LIE and Assumption~\ref{ass:cons} \ref{it:smooth} imply that
\begin{align*}
	\E\big[\Cc_{t}(\vtheta^c, \vtheta^v)\big] &= \E\bigg\{\nabla c_t(\vtheta^c)\E_{t-1}\big[\1_{\{X_t>v_t(\vtheta^v),\, Y_t\leq c_t(\vtheta^c)\}}-\alpha \1_{\{X_t>v_t(\vtheta^v)\}}\big]\bigg\}\\	
	&= \E\bigg\{\nabla c_t(\vtheta^c) \Big[\P_{t-1}\big\{X_t> v_t(\vtheta^v), Y_t\leq c_t(\vtheta^c)\big\}-\alpha \P_{t-1}\big\{X_t> v_t(\vtheta^v)\big\}\Big]\bigg\}\\
	&= \E\bigg\{\nabla c_t(\vtheta^c) \Big[ F_t^{Y}\big(c_t(\vtheta^c)\big) - F_{t}\big(v_t(\vtheta^v), c_t(\vtheta^c)\big)-\alpha \big\{ 1-F_{t}^{X}\big(v_t(\vtheta^v)\big)\big\}\Big]\bigg\}.
\end{align*}
Assumptions~\ref{ass:cons}--\ref{ass:an} and the DCT allow us to interchange differentiation and expectation to yield that
\begin{align}
	\frac{\partial}{\partial \vtheta^c}\E\big[\Cc_{t}(\vtheta^c, \vtheta^v)\big] &= \E\bigg\{\nabla^2 c_t(\vtheta^c) \Big[ F_t^{Y}\big(c_t(\vtheta^c)\big) - F_{t}\big(v_t(\vtheta^v), c_t(\vtheta^c)\big)-\alpha \big\{ 1-F_{t}^{X}\big(v_t(\vtheta^v)\big)\big\}\Big]\bigg\}\notag\\
	&\hspace{2cm} + \E\bigg\{\nabla c_t(\vtheta^c)\nabla^\prime c_t(\vtheta^c) \Big[  f_t^{Y}\big(c_t(\vtheta^c)\big) - \partial_2 F_{t}\big(v_t(\vtheta^v), c_t(\vtheta^c)\big)\Big]\bigg\},\label{eq:(p.15)}\\
	\frac{\partial}{\partial \vtheta^v}\E\big[\Cc_{t}(\vtheta^c, \vtheta^v)\big] &= \E\bigg\{\nabla c_t(\vtheta^c)\nabla^\prime v_t(\vtheta^v)\Big[\alpha f_t^X\big(v_t(\vtheta^v)\big) - \partial_1 F_t\big(v_t(\vtheta^v), c_t(\vtheta^c)\big)\Big]\bigg\},\label{eq:(p.16)}
\end{align}
where $\partial_i F_t(\cdot,\cdot)$ denotes the partial derivative with respect to the $i$-th component ($i=1,2$). Evaluating these quantities at the true parameters gives
\begin{align}
	\frac{\partial}{\partial \vtheta^c}\E\big[\Cc_{t}(\vtheta^c, \vtheta^v)\big]\Big\vert_{\substack{\vtheta^c=\vtheta_0^c\\ \vtheta^v=\vtheta_0^v}} &= \E\bigg\{\nabla c_t(\vtheta_0^c)\nabla^\prime c_t(\vtheta_0^c) \Big[  f_t^{Y}\big(c_t(\vtheta_0^c)\big) - \partial_2 F_{t}\big(v_t(\vtheta_0^v), c_t(\vtheta_0^c)\big)\Big]\bigg\},\notag\\
	\frac{\partial}{\partial \vtheta^v}\E\big[\Cc_{t}(\vtheta^c, \vtheta^v)\big]\Big\vert_{\substack{\vtheta^c=\vtheta_0^c\\ \vtheta^v=\vtheta_0^v}} &= \E\bigg\{\nabla c_t(\vtheta_0^c)\nabla^\prime v_t(\vtheta_0^v)\Big[\alpha f_t^X\big(v_t(\vtheta_0^v)\big) - \partial_1 F_t\big(v_t(\vtheta_0^v), c_t(\vtheta_0^c)\big)\Big]\bigg\}.\notag
\end{align}
Define
\begin{align*}
	\vlambda_n(\vtheta^c,\vtheta^v) &= \frac{1}{n}\sum_{t=1}^{n}\E\big[\Cc_{t}(\vtheta^c,\vtheta^v)\big],\\
	\mLambda_{n,(1)}(\vtheta^{c\ast},\vtheta^{v\ast}) &= \frac{1}{n}\sum_{t=1}^{n}\frac{\partial}{\partial\vtheta^c}\E\big[\Cc_{t}(\vtheta^c,\vtheta^v)\big]\Big\vert_{\substack{\vtheta^c=\vtheta^{c\ast}\\ \vtheta^v=\vtheta^{v\ast}}},\\
	\mLambda_{n,(2)}(\vtheta^{c\ast},\vtheta^{v\ast}) &= \frac{1}{n}\sum_{t=1}^{n}\frac{\partial}{\partial\vtheta^v}\E\big[\Cc_{t}(\vtheta^c,\vtheta^v)\big]\Big\vert_{\substack{\vtheta^c=\vtheta^{c\ast}\\  \vtheta^v=\vtheta^{v\ast}}}.
\end{align*}
With $\mLambda_{n,(1)}$ and $\mLambda_{n,(2)}$ as defined in Assumption~\ref{ass:an}~\ref{it:pd}, it then holds that
\begin{align*}
	\mLambda_{n,(1)}&=\mLambda_{n,(1)}(\vtheta_0^{c},\vtheta_0^{v}) = \frac{1}{n}\sum_{t=1}^{n}\E\bigg\{\nabla c_t(\vtheta_0^c)\nabla^\prime c_t(\vtheta_0^c) \Big[  f_t^{Y}\big(c_t(\vtheta_0^c)\big) - \partial_2 F_{t}\big(v_t(\vtheta_0^v), c_t(\vtheta_0^c)\big)\Big]\bigg\},\\
	\mLambda_{n,(2)}&=\mLambda_{n,(2)}(\vtheta_0^{c},\vtheta_0^{v}) = \frac{1}{n}\sum_{t=1}^{n}\E\bigg\{\nabla c_t(\vtheta_0^c)\nabla^\prime v_t(\vtheta_0^v)\Big[\alpha f_t^X\big(v_t(\vtheta_0^v)\big) - \partial_1 F_t\big(v_t(\vtheta_0^v), c_t(\vtheta_0^c)\big)\Big]\bigg\}.
\end{align*}

\begin{proof}[{\textbf{Proof of Lemma~C.\ref{lem:1 tilde}:}}]
	The MVT (applied twice) and $\vlambda_n(\vtheta_0^{c}, \vtheta_0^{v})=\vzero$ imply that
	\begin{align*}
		\vlambda_n(\widehat{\vtheta}_n^{c}, \widehat{\vtheta}_n^{v}) &= \vlambda_n(\vtheta_0^{c}, \widehat{\vtheta}_n^{v}) + \mLambda_{n,(1)}(\vtheta^{c\ast}, \widehat{\vtheta}_n^{v}) (\widehat{\vtheta}_n^c- \vtheta_0^c)\\
		&=\vlambda_n(\vtheta_0^{c}, \vtheta_0^{v}) + \mLambda_{n,(2)}(\vtheta_0^c,\vtheta^{v\ast})(\widehat{\vtheta}_n^v- \vtheta_0^v) + \mLambda_{n,(1)}(\vtheta^{c\ast}, \widehat{\vtheta}_n^{v}) (\widehat{\vtheta}_n^c- \vtheta_0^c)\\
		&=\mLambda_{n,(2)}(\vtheta_0^c,\vtheta^{v\ast})(\widehat{\vtheta}_n^v- \vtheta_0^v) + \mLambda_{n,(1)}(\vtheta^{c\ast}, \widehat{\vtheta}_n^{v}) (\widehat{\vtheta}_n^c- \vtheta_0^c)
	\end{align*}
	for some $\vtheta^{c\ast}$ ($\vtheta^{v\ast}$) on the line connecting $\widehat{\vtheta}_n^{c}$ and $\vtheta_0^{c}$ ($\widehat{\vtheta}_n^{v}$ and $\vtheta_0^{v}$). (Recall our convention to abuse notation slightly when applying the mean value theorem to multivariate functions. Again, the argument goes through componentwise, yet the notation would be more involved.) Thus, 
	\begin{align*}
		\mLambda_{n,(1)}(\vtheta^{c\ast}, \widehat{\vtheta}_n^{v})\sqrt{n}(\widehat{\vtheta}_n^{c} -  \vtheta_0^c) &=   \sqrt{n}\vlambda_n(\widehat{\vtheta}_n^{c}, \widehat{\vtheta}_n^{v}) - \mLambda_{n,(2)}(\vtheta_0^c,\vtheta^{v\ast})\sqrt{n}(\widehat{\vtheta}_n^v- \vtheta_0^v)\\
		&=  \sqrt{n}\vlambda_n(\widehat{\vtheta}_n^{c}, \widehat{\vtheta}_n^{v}) + \mLambda_{n,(2)}(\vtheta_0^c,\vtheta^{v\ast})\times\\
		&\hspace{4cm}\times\big[\mLambda_n^{-1}+o_{\P}(1)\big]\bigg[\frac{1}{\sqrt{n}}\sum_{t=1}^{n}\Vv_{t}(\vtheta_0^v) + o_{\P}(1)\bigg],
	\end{align*}
	where we used Lemma~V.\ref{lem:1} in the second step. In light of this, we only have to show that
	\begin{enumerate}
		\item[(i)] $\mLambda_{n,(1)}^{-1}(\vtheta^{c\ast}, \widehat{\vtheta}_n^{v})=\mLambda_{n,(1)}^{-1} + o_{\P}(1)$;
		\item[(ii)] $\mLambda_{n,(2)}(\vtheta_0^c,\vtheta^{v\ast})=\mLambda_{n,(2)} + o_{\P}(1)$;
		\item[(iii)] $\sqrt{n}\vlambda_n(\widehat{\vtheta}_n^{c},\widehat{\vtheta}_n^{v})=-\frac{1}{\sqrt{n}}\sum_{t=1}^{n}\Cc_{t}(\vtheta_{0,n}^c,\widehat{\vtheta}_n^{v}) + o_{\P}(1)$;
	\end{enumerate}
	where $\vtheta_{0,n}^c$ is defined below in this proof. 
	
	Items (i) and (ii) are established in Lemma~C.\ref{lem:1+ tilde}. 
	For this, note that $\vtheta^{c\ast}\overset{\P}{\longrightarrow}\vtheta_0^c$ ($\vtheta^{v\ast}\overset{\P}{\longrightarrow}\vtheta_0^v$), because $\vtheta^{c\ast}$ ($\vtheta^{v\ast}$) lies on the line connecting $\widehat{\vtheta}_n^{c}$ and $\vtheta_0^{c}$ ($\widehat{\vtheta}_n^{v}$ and $\vtheta_0^{v}$), and $\widehat{\vtheta}_n^{c}\overset{\P}{\longrightarrow}\vtheta_0^c$ ($\widehat{\vtheta}_n^{v}\overset{\P}{\longrightarrow}\vtheta_0^v$) from Theorem~\ref{thm:cons}.
	
	To prove (iii), by Lemma~A.1 in \citet{Wei91}, it suffices to show that
	\begin{enumerate}
		\item[(iii.a)] conditions (N1)--(N5) in the notation of \citet{Wei91} hold;
		\item[(iii.b)] $\frac{1}{\sqrt{n}}\sum_{t=1}^{n}\Cc_{t}(\widehat{\vtheta}_n^{c},\widehat{\vtheta}_n^{v})=o_{\P}(1)$;
		\item[(iii.c)] $\widehat{\vtheta}_n^{c}-\vtheta_{0,n}^c\overset{\P}{\longrightarrow}\vzero$.
	\end{enumerate}
	For (iii.a), (N1) is shown in Lemma~C.\ref{lem:0 tilde}, and the mixing condition (N5) follows from Assumption~\ref{ass:an}~\ref{it:mixing}. 
	Establishing (N2) requires some more work than in the proof of Lemma~V.\ref{lem:1}. 
	To do so, we have to show that for each $n$ there exists some $\vtheta_{0,n}^{c}\in\mTheta^{c}$, such that $\vlambda_n(\vtheta_{0,n}^{c},\widehat{\vtheta}_n^v)=\vzero$. 
	Fix $n\in\mathbb{N}$ and recall from the LIE that $\vlambda_n(\vtheta_0^c,\vtheta_0^v)=\vzero$.
	Next, we apply the implicit function theorem \citep[e.g.,][Theorem~9.2]{Mun91}. This is possible because of the invertibility of 
	\[
	\mLambda_{n,(1)}=\mLambda_{n,(1)}(\vtheta_0^c, \vtheta_0^v) = \frac{1}{n}\sum_{t=1}^{n}  \frac{\partial}{\partial \vtheta^c}\E\big[\Cc_{t}(\vtheta^c, \vtheta^v)\big]\Big\vert_{\substack{\vtheta^c=\vtheta_0^c\\ \vtheta^v=\vtheta_0^v}}
	\]
	(by Assumption~\ref{ass:an} \ref{it:pd}) and the continuous differentiability of 
	\[
	\vlambda_n(\cdot,\cdot)=\frac{1}{n}\sum_{t=1}^{n}\E\big[\Cc_{t}(\cdot,\cdot)\big]
	\]
	from \eqref{eq:(p.15)}, \eqref{eq:(p.16)}, Assumption~\ref{ass:cons}~\ref{it:Lipschitz cons} and Assumption~\ref{ass:an}~\ref{it:diff}.
	An application of the implicit function theorem now implies that there exist a neighborhood $\mathcal{N}(\vtheta_0^v)\subset\mTheta^v$ of the interior point $\vtheta_0^v$ and a (unique) continuously differentiable function $\vtheta_n^c(\cdot)$ satisfying $\vtheta_n^c(\vtheta_0^v)=\vtheta_0^c$, such that
	\[
	\vlambda_n\big(\vtheta_n^c(\vtheta^v), \vtheta^v\big)=\vzero
	\]
	holds for all $\vtheta^v\in\mathcal{N}(\vtheta_0^v)$. 
	The continuity of $\vtheta^c_n(\cdot)$ ensures that we can choose $\mathcal{N}(\vtheta_0^v)$ such that the image of the map $\mathcal{N}(\vtheta_0^v)\ni\vtheta\mapsto\vtheta_n^c(\vtheta)$ is in $\mTheta^c$ (using also the fact that $\vtheta_0^c$ is in the interior of $\mTheta^c$). We conclude that on the set $\big\{\omega\in\Omega\colon\Vert\widehat{\vtheta}_n^v-\vtheta_0^v\Vert\leq\varepsilon_0\big\}$ with $\varepsilon_0>0$ chosen such that $\widehat{\vtheta}_n^v\in \mathcal{N}(\vtheta_0^v)$, there exists $\vtheta_{0,n}^c\in\mTheta^c$ (viz.~$\vtheta_{0,n}^{c}=\vtheta_n^c(\widehat{\vtheta}_n^v)$) with 
	\begin{equation}\label{eq:(2.1)}
		\vlambda_n(\vtheta_{0,n}^c, \widehat{\vtheta}_n^v)=\vzero.
	\end{equation}
	Hence, (N2) holds.
	
	Note that $\varepsilon_0$ can be chosen independently of $n$. To see this, observe that in the notation of the proof of the implicit function theorem in \citet[Theorem~9.2]{Mun91},
	\[
	F(\vx,\vy)=\begin{pmatrix}
		\vx\\
		f(\vx,\vy)
	\end{pmatrix} = \begin{pmatrix}
		\vtheta^v\\
		\vlambda_n(\vtheta^c,\vtheta^v)
	\end{pmatrix},
	\]
	where the last equality is taken to imply that $\vx$, $\vy$ and $f(\vx,\vy)$ in \citeauthor{Mun91}' \citeyearpar{Mun91} notation play the role of $\vtheta^v$, $\vtheta^c$ and $\vlambda_n(\vtheta^c,\vtheta^v)$, respectively.
	The Jacobian of $F(\vx,\vy)$ is given by
	\[
	D\,F = \begin{pmatrix}
		\mI & \vzero\\
		\partial f/\partial \vx & \partial f/\partial \vy
	\end{pmatrix}=\begin{pmatrix}
		\mI & \vzero\\
		\partial\vlambda_n/\partial\vtheta^v & \partial\vlambda_n/\partial\vtheta^c 
	\end{pmatrix}.
	\]
	In the proof of the implicit function theorem in \citet{Mun91}, the neighborhood $\mathcal{N}(\vtheta_0^v)$ is chosen such that
	\begin{equation}\label{eq:det}
		\det(D\,F) = \det\begin{pmatrix}
			\mI & \vzero\\
			\partial\vlambda_n/\partial\vtheta^v & \partial\vlambda_n/\partial\vtheta^c 
		\end{pmatrix}\neq0
	\end{equation}
	\citep[cf.~the proof of Theorem~8.3 in][]{Mun91} and, for a certain small constant $c>0$,
	\begin{equation}\label{eq:ct DF}
		\Bigg\Vert \begin{pmatrix}
			\mI & \vzero\\
			\partial\vlambda_n(\vtheta^c,\vtheta^v)/\partial\vtheta^v & \partial\vlambda_n(\vtheta^c,\vtheta^v)/\partial\vtheta^c 
		\end{pmatrix} - \begin{pmatrix}
			\mI & \vzero\\
			\partial\vlambda_n(\vtheta_0^c,\vtheta_0^v)/\partial\vtheta^v & \partial\vlambda_n(\vtheta_0^c,\vtheta_0^v)/\partial\vtheta^c 
		\end{pmatrix} \Bigg\Vert<c
	\end{equation}
	\citep[cf.~the proof of Lemma~8.1 in][]{Mun91} for all $(\vtheta^{c}, \vtheta^{v})\in\mathcal{N}(\vtheta_0^c)\times \mathcal{N}(\vtheta_0^v)$, where $\mathcal{N}(\vtheta_0^c)$ is some neighborhood of $\vtheta_0^c$.
	
	Regarding \eqref{eq:det}, note that it is equivalent to
	\begin{align}
		0 &\neq \det(\mI)\det\big(\partial\vlambda_n/\partial\vtheta^c\big)\notag\\
		&=\det\big(\partial\vlambda_n(\vtheta^c,\vtheta^v)/\partial\vtheta^c\big)\notag\\
		&=\det\big(\mLambda_{n,(1)}(\vtheta^c,\vtheta^v)\big)\notag\\
		&=\det\big(\mLambda_{n,(1)}(\vtheta^c,\vtheta^v)-\mLambda_{n,(1)}(\vtheta_0^c,\vtheta_0^v) + \mLambda_{n,(1)}(\vtheta_0^c,\vtheta_0^v)\big).\label{eq:det input}
	\end{align}
	As shown in \eqref{eq:cty lambda1} below,
	\[
	\big\Vert\mLambda_{n,(1)}(\vtheta^c,\vtheta^v)-\mLambda_{n,(1)}(\vtheta_0^c,\vtheta_0^v)\big\Vert\leq C \big\Vert\vtheta^c-\vtheta_0^c\big\Vert + C \big\Vert\vtheta^v-\vtheta_0^v\big\Vert
	\]
	for some $C>0$ independent of $n$. 
	Moreover, by Assumption~\ref{ass:an}~\ref{it:pd}, $\det\big(\mLambda_{n,(1)}(\vtheta_0^c,\vtheta_0^v)\big)\geq a>0$ for some $a>0$ not depending on $n$.
	The continuity of the determinant in \eqref{eq:det input} therefore implies that the neighborhoods of $\vtheta_0^c$ and $\vtheta_0^v$ where \eqref{eq:det} holds can be chosen independently of $n$.
	
	Regarding \eqref{eq:ct DF}, observe that by the triangle inequality, \eqref{eq:cty lambda2} and \eqref{eq:cty lambda1},
	\begin{align*}
		&\Bigg\Vert \begin{pmatrix}
			\mI & \vzero\\
			\partial\vlambda_n(\vtheta^c,\vtheta^v)/\partial\vtheta^v & \partial\vlambda_n(\vtheta^c,\vtheta^v)/\partial\vtheta^c 
		\end{pmatrix} - \begin{pmatrix}
			\mI & \vzero\\
			\partial\vlambda_n(\vtheta_0^c,\vtheta_0^v)/\partial\vtheta^v & \partial\vlambda_n(\vtheta_0^c,\vtheta_0^v)/\partial\vtheta^c 
		\end{pmatrix} \Bigg\Vert\\
		&\leq \big\Vert \mLambda_{n,(1)}(\vtheta^c,\vtheta^v) - \mLambda_{n,(1)}(\vtheta_0^c,\vtheta_0^v)\big\Vert + \big\Vert \mLambda_{n,(2)}(\vtheta^c,\vtheta^v) - \mLambda_{n,(2)}(\vtheta_0^c,\vtheta_0^v)\big\Vert\\
		&\leq C \big\Vert \vtheta^c - \vtheta_0^c \big\Vert + C \big\Vert \vtheta^v - \vtheta_0^v\big\Vert,
	\end{align*}
	where $C$ is independent of $n$.
	Therefore, the neighborhoods of $\vtheta_0^c$ and $\vtheta_0^v$ where \eqref{eq:ct DF} holds can be chosen independently of $n$.
	
	Overall, it follows that it is possible to choose a universal $\varepsilon_0>0$ that works for all $n$.

	Without further mention, we work on the set $\big\{\omega\in\Omega\colon\Vert\widehat{\vtheta}_n^v-\vtheta_0^v\Vert\leq\varepsilon_0\big\}$ to verify the remaining conditions used to establish (iii) in the following, such that the existence of $\vtheta_{0,n}^c$ satisfying \eqref{eq:(2.1)} is guaranteed. 
	This is possible, because
	\begin{align*}
		&\P\bigg\{\Big\Vert \sqrt{n}\vlambda_n(\widehat{\vtheta}_n^{c},\widehat{\vtheta}_n^{v})+\frac{1}{\sqrt{n}}\sum_{t=1}^{n}\Cc_{t}(\vtheta_{0,n}^c,\widehat{\vtheta}_n^{v}) \Big\Vert>\varepsilon \bigg\}\\
		&\leq \P\bigg\{\Big\Vert \sqrt{n}\vlambda_n(\widehat{\vtheta}_n^{c},\widehat{\vtheta}_n^{v})+\frac{1}{\sqrt{n}}\sum_{t=1}^{n}\Cc_{t}(\vtheta_{0,n}^c,\widehat{\vtheta}_n^{v}) \Big\Vert>\varepsilon,\  \Vert\widehat{\vtheta}_n^v-\vtheta_0^v\Vert\leq\varepsilon_0\bigg\}
		+  \P\big\{\Vert\widehat{\vtheta}_n^v-\vtheta_0^v\Vert>\varepsilon_0\big\}\\
		&= \P\bigg\{\Big\Vert \sqrt{n}\vlambda_n(\widehat{\vtheta}_n^{c},\widehat{\vtheta}_n^{v})+\frac{1}{\sqrt{n}}\sum_{t=1}^{n}\Cc_{t}(\vtheta_{0,n}^c,\widehat{\vtheta}_n^{v}) \Big\Vert>\varepsilon,\  \Vert\widehat{\vtheta}_n^v-\vtheta_0^v\Vert\leq\varepsilon_0\bigg\}
		+  o(1),
	\end{align*}
	where we used Theorem~\ref{thm:cons} in the final step. 
	Thus, for the purposes of verifying (iii), we may assume that $\Vert\widehat{\vtheta}_n^v-\vtheta_0^v\Vert\leq\varepsilon_0$.
	Hence, it is for $\widehat{\vtheta}_n^v$ with $\Vert\widehat{\vtheta}_n^v-\vtheta_0^v\Vert\leq\varepsilon_0$ that we establish (N3)--(N4).
	Specifically, condition (N3) is verified in Lemmas~C.\ref{lem:3 tilde}--C.\ref{lem:5 tilde} below, and the remaining condition (N4) in Lemma~C.\ref{lem:6 tilde}. This shows (iii.a). 
	The result in (iii.b) follows from Lemma~C.\ref{lem:2 tilde}.
	Finally, (iii.c) follows from Theorem~\ref{thm:cons} in combination with $\vtheta_{0,n}^c\overset{\P}{\longrightarrow}\vtheta_0^c$, which follows from Assumption~\ref{ass:an}~\ref{it:lambda func} via
	\[
	\big\Vert \vtheta_{0,n}^c - \vtheta_0^c\big\Vert = \big\Vert \vtheta_{n}^c(\widehat{\vtheta}_n^v) - \vtheta_n^c(\vtheta_0^v)\big\Vert \leq K\big\Vert \widehat{\vtheta}_n^v - \vtheta_0^v\big\Vert=o_{\P}(1).
	\]
	In sum, we obtain the desired result.
\end{proof}

\setcounter{lemC}{2}

\begin{lemC}\label{lem:1+ tilde}
	Suppose Assumptions~\ref{ass:cons} and \ref{ass:an} hold. Then, as $n\to\infty$, 
	\begin{align*}
		\mLambda_{n,(1)}^{-1}(\vtheta^{c\ast},\widehat{\vtheta}_n^v)-\mLambda_{n,(1)}^{-1}	&\overset{\P}{\longrightarrow}\vzero\qquad\text{for any}\ \vtheta^{c\ast}\overset{\P}{\longrightarrow}\vtheta_0^c\quad \text{and}\quad \widehat{\vtheta}_n^v\overset{\P}{\longrightarrow}\vtheta_0^v, \\
		\mLambda_{n,(2)}(\vtheta_0^c,\vtheta^{v\ast})-\mLambda_{n,(2)}	&\overset{\P}{\longrightarrow}\vzero\qquad\text{for any}\ \vtheta^{v\ast}\overset{\P}{\longrightarrow}\vtheta_0^v.
	\end{align*}
\end{lemC}

\begin{proof}
	We begin with the second statement. Our first goal is to show that
	\begin{equation}\label{eq:cty lambda2}
		\big\Vert\mLambda_{n,(2)}(\vtau^c,\vtau^v) - \mLambda_{n,(2)}\big\Vert \leq C \big\Vert\vtau^c - \vtheta_0^c\big\Vert + C \big\Vert\vtau^v - \vtheta_0^v\big\Vert
	\end{equation}
	for any $\vtau^v$ in some neighborhood $\mathcal{N}(\vtheta_0^v)$ ensuring Assumption~\ref{ass:cons}~\ref{it:bound} and Assumption~\ref{ass:an}~\ref{it:bound2.1} hold, and any $\vtau^c$ in some neighborhood $\mathcal{N}(\vtheta_0^c)$ ensuring that Assumption~\ref{ass:an}~\ref{it:bound2.2} holds.
	By the triangle inequality,
	\[
	\big\Vert\mLambda_{n,(2)}(\vtau^c,\vtau^v) - \mLambda_{n,(2)}\big\Vert\leq \big\Vert\mLambda_{n,(2)}(\vtau^c,\vtau^v) - \mLambda_{n,(2)}(\vtheta_0^c, \vtau^v)\big\Vert + \big\Vert\mLambda_{n,(2)}(\vtheta_0^c, \vtau^v) - \mLambda_{n,(2)}\big\Vert.
	\]
	To establish \eqref{eq:cty lambda2}, it therefore suffices to show that
	\begin{align}
		\big\Vert \mLambda_{n,(2)}(\vtau^c,\vtau^v) - \mLambda_{n,(2)}(\vtheta_0^c, \vtau^v)\big\Vert & \leq C \big\Vert \vtau^c - \vtheta_0^c \big\Vert,\label{eq:(E.9)}\\
		\big\Vert \mLambda_{n,(2)}(\vtheta_0^c,\vtau^v) - \mLambda_{n,(2)}\big\Vert & \leq C \big\Vert \vtau^v - \vtheta_0^v \big\Vert.\label{eq:(E.10)}
	\end{align}
	First, we prove \eqref{eq:(E.10)}. 
	Exploit \eqref{eq:(p.16)} to write 
	\begin{align*}
		&\big\Vert\mLambda_{n,(2)}(\vtheta_0^c,\vtau^v) - \mLambda_{n,(2)}\big\Vert =\big\Vert\mLambda_{n,(2)}(\vtheta_0^c,\vtau^v) - \mLambda_{n,(2)}(\vtheta_0^c, \vtheta_0^v)\big\Vert\\
		&= \bigg\Vert\frac{1}{n}\sum_{t=1}^{n} \E\bigg\{\nabla c_t(\vtheta_0^c)\nabla^\prime v_t(\vtau^v)\Big[\alpha f_t^X\big(v_t(\vtau^v)\big) - \partial_1 F_t\big(v_t(\vtau^v), c_t(\vtheta_0^c)\big)\Big]\bigg\}\\
		&\hspace{1.27cm} - \E\bigg\{\nabla c_t(\vtheta_0^c)\nabla^\prime v_t(\vtheta_0^v)\Big[\alpha f_t^X\big(v_t(\vtheta_0^v)\big) - \partial_1 F_t\big(v_t(\vtheta_0^v), c_t(\vtheta_0^c)\big)\Big]\bigg\} \bigg\Vert\\
		&\leq \frac{\alpha}{n}\sum_{t=1}^{n} \E\Big\Vert\nabla c_t(\vtheta_0^c)\nabla^\prime v_t(\vtau^v)f_t^X\big(v_t(\vtau^v)\big) - \nabla c_t(\vtheta_0^c)\nabla^\prime v_t(\vtheta_0^v)f_t^X\big(v_t(\vtheta_0^v)\big)\Big\Vert\\
		&\hspace{1.27cm}+\frac{1}{n}\sum_{t=1}^{n}\E\Big\Vert\nabla c_t(\vtheta_0^c)\nabla^\prime v_t(\vtau^v)\partial_1 F_t\big(v_t(\vtau^v), c_t(\vtheta_0^c)\big) - \nabla c_t(\vtheta_0^c)\nabla^\prime v_t(\vtheta_0^v)\partial_1 F_t\big(v_t(\vtheta_0^v), c_t(\vtheta_0^c)\big)\Big\Vert\\
		&=:A_{2n} + B_{2n}.
	\end{align*}
	We only consider $A_{2n}$, because $B_{2n}$ can be dealt with similarly. A mean value expansion around $\vtheta_0^v$ implies
	\begin{align*}
		A_{2n} &=  \frac{\alpha}{n}\sum_{t=1}^{n} \E\bigg\Vert\nabla c_t(\vtheta_0^c)\big[\nabla^\prime v_t(\vtau^v) - \nabla^\prime v_t(\vtheta_0^v)\big] f_t^X\big(v_t(\vtau^v)\big) \\
		& \hspace{3cm} + \nabla c_t(\vtheta_0^c)\nabla^\prime v_t(\vtheta_0^v)\Big[f_t^X\big(v_t(\vtau^v) - f_t^X\big(v_t(\vtheta_0^v)\big)\Big]\bigg\Vert\\
		&\leq \frac{\alpha}{n}\sum_{t=1}^{n} \E\Big\Vert\nabla c_t(\vtheta_0^c)\big[\vtau^v - \vtheta_0^v\big]^\prime \big[\nabla^2 v_t(\vtau^{\ast})\big]^\prime f_t^X\big(v_t(\vtau^v)\big)\Big\Vert \\
		& \hspace{3cm} + K\E\Big\Vert\nabla c_t(\vtheta_0^c)\nabla^\prime v_t(\vtheta_0^v)\big|v_t(\vtau^v) - v_t(\vtheta_0^v)\big|\Big\Vert\\
		&\leq \frac{\alpha}{n}\sum_{t=1}^{n} K\E\big[C_1(\mathcal{F}_{t-1}) V_2(\mathcal{F}_{t-1})\big] \big\Vert\vtau^v - \vtheta_0^v\big\Vert \\
		& \hspace{3cm} + K\E\big[C_1(\mathcal{F}_{t-1}) V_1^2(\mathcal{F}_{t-1})\big] \big\Vert\vtau^v - \vtheta_0^v\big\Vert\\
		&\leq \frac{C}{n}\sum_{t=1}^{n} \Big\{\E\big[C_1^3(\mathcal{F}_{t-1})\big]\Big\}^{1/3} \Big\{\E\big[V_2^{3/2}(\mathcal{F}_{t-1})\big]\Big\}^{2/3} \big\Vert\vtau^v - \vtheta_0^v\big\Vert \\
		&\hspace{3cm} + \Big\{\E\big[C_1^3(\mathcal{F}_{t-1})\big]\Big\}^{1/3} \Big\{\E\big[V_1^{3}(\mathcal{F}_{t-1})\big]\Big\}^{2/3} \big\Vert\vtau^v - \vtheta_0^v\big\Vert \\
		&\leq C \norm{\vtau^v - \vtheta_0^v},
	\end{align*}
	where $\vtau^{\ast}$ is some mean value between $\vtau^v$ and $\vtheta_0^v$. Using identical arguments, we may show that 
	\[
	B_{2n} \leq C \norm{\vtau^v - \vtheta_0^v}.
	\]
	Hence, \eqref{eq:(E.10)} follows.
	
	To prove \eqref{eq:(E.9)}, we exploit \eqref{eq:(p.16)} to obtain
	\begin{align*}
		&\big\Vert \mLambda_{n,(2)}(\vtau^c,\vtau^v)-\mLambda_{n,(2)}(\vtheta_0^c,\vtau^v)  \big\Vert\\
		&=\bigg\Vert \frac{1}{n}\sum_{t=1}^{n} \E\bigg\{\nabla c_t(\vtau^c)\nabla^\prime v_t(\vtau^v)\Big[\alpha f_t^X\big(v_t(\vtau^v)\big) - \partial_1 F_t\big(v_t(\vtau^v), c_t(\vtau^c)\big)\Big]\bigg\}\\
		&\hspace{2cm} - \E\bigg\{\nabla c_t(\vtheta_0^c)\nabla^\prime v_t(\vtau^v)\Big[\alpha f_t^X\big(v_t(\vtau^v)\big) - \partial_1 F_t\big(v_t(\vtau^v), c_t(\vtheta_0^c)\big)\Big]\bigg\}  \bigg\Vert\\
		&=\bigg\Vert \frac{1}{n}\sum_{t=1}^{n} \E\bigg\{\big[\nabla c_t(\vtau^c)-\nabla c_t(\vtheta_0^c)\big]\nabla^\prime v_t(\vtau^v)\alpha f_t^X\big(v_t(\vtau^v)\big)\bigg\}\\
		&\hspace{2cm} - \E\bigg\{\big[\nabla c_t(\vtau^c)-\nabla c_t(\vtheta_0^c)\big]\nabla^\prime v_t(\vtau^v)\partial_1 F_t\big(v_t(\vtau^v), c_t(\vtau^c)\big)\bigg\}\\
		&\hspace{2cm} - \E\bigg\{\nabla c_t(\vtheta_0^c)\nabla^\prime v_t(\vtau^v)\Big[\partial_1 F_t\big(v_t(\vtau^v), c_t(\vtau^c)\big) - \partial_1 F_t\big(v_t(\vtau^v), c_t(\vtheta_0^c)\big)\Big]\bigg\}  \bigg\Vert\\
		&\leq \frac{1}{n}\sum_{t=1}^{n} \E\Big\Vert\big[\nabla c_t(\vtau^c)-\nabla c_t(\vtheta_0^c)\big]\nabla^\prime v_t(\vtau^v)\alpha f_t^X\big(v_t(\vtau^v)\big)\Big\Vert\\
		&\hspace{2cm} + \frac{1}{n}\sum_{t=1}^{n} \E\Big\Vert\big[\nabla c_t(\vtau^c)-\nabla c_t(\vtheta_0^c)\big]\nabla^\prime v_t(\vtau^v)\partial_1 F_t\big(v_t(\vtau^v), c_t(\vtau^c)\big)\Big\Vert\\
		&\hspace{2cm} + \frac{1}{n}\sum_{t=1}^{n} \E\bigg\Vert\nabla c_t(\vtheta_0^c)\nabla^\prime v_t(\vtau^v)\Big[\partial_1 F_t\big(v_t(\vtau^v), c_t(\vtau^c)\big) - \partial_1 F_t\big(v_t(\vtau^v), c_t(\vtheta_0^c)\big)\Big]\bigg\Vert.
	\end{align*}
	The first term can easily be bounded by
	\begin{align*}
		\frac{1}{n}\sum_{t=1}^{n} \E\Big[\big\Vert\nabla c_t(\vtau^c)-\nabla c_t(\vtheta_0^c)\big\Vert V_1(\mathcal{F}_{t-1})\alpha K\Big]&=\frac{\alpha K}{n}\sum_{t=1}^{n} \E\Big[\big\Vert \nabla^2 c_t(\vtau^{c\ast}) (\vtau^c - \vtheta_0^c)\big\Vert V_1(\mathcal{F}_{t-1})\Big]\\
		&\leq \frac{\alpha K}{n}\sum_{t=1}^{n} \E\Big[C_2(\mathcal{F}_{t-1}) \big\Vert \vtau^c - \vtheta_0^c\big\Vert  V_1(\mathcal{F}_{t-1})\Big]\\
		&\leq C \big\Vert \vtau^c - \vtheta_0^c\big\Vert,
	\end{align*}
	the second term can be bounded similarly, and the third term is bounded by
	\begin{align*}
		&\frac{1}{n}\sum_{t=1}^{n} \E\bigg[\big\Vert\nabla c_t(\vtheta_0^c)\big\Vert \cdot \big\Vert\nabla^\prime v_t(\vtau^v)\big\Vert \cdot \Big|\partial_1\partial_2 F_t\big(v_t(\vtau^v), c_t(\vtau^{c\ast})\big) \big\{c_t(\vtau^{c}) - c_t(\vtheta_0^{c})\big\}\Big|\bigg] \\
		&\leq \frac{1}{n}\sum_{t=1}^{n} \E\bigg[C_1(\mathcal{F}_{t-1}) V_1(\mathcal{F}_{t-1}) f_t\big(v_t(\vtau^v), c_t(\vtau^{c\ast})\big) \Big| \nabla^\prime c_t(\vtau^{c\ast})\big\{\vtau^{c} - \vtheta_0^{c}\big\}\Big|\bigg] \\
		&\leq \frac{1}{n}\sum_{t=1}^{n} \E\big[C_1^2(\mathcal{F}_{t-1}) V_1(\mathcal{F}_{t-1}) K \big] \big\Vert \vtau^{c} - \vtheta_0^{c}\big\Vert\\
		&\leq \frac{K}{n}\sum_{t=1}^{n} \Big\{\E\big[C_1^3(\mathcal{F}_{t-1})\big]\Big\}^{2/3} \Big\{\E\big[V_1^3(\mathcal{F}_{t-1}) \big]\Big\}^{1/3} \big\Vert \vtau^{c} - \vtheta_0^{c}\big\Vert\\
		&\leq C \big\Vert \vtau^{c} - \vtheta_0^{c}\big\Vert.
	\end{align*}
	Overall, \eqref{eq:(E.9)} follows.
	
	Since $\vtheta^{v\ast} \overset{\P}{\longrightarrow} \vtheta_0^v$, the second statement of the lemma easily follows from \eqref{eq:(E.10)}.
	
	To prove the first claim of the lemma, decompose
	\begin{align}
		\big\Vert\mLambda_{n,(1)} - \mLambda_{n,(1)}(\vtau^c,\vtau^v)\big\Vert &\leq \big\Vert\mLambda_{n,(1)} - \mLambda_{n,(1)}(\vtheta_0^c,\vtau^v)\big\Vert + \big\Vert\mLambda_{n,(1)}(\vtheta_0^c,\vtau^v) - \mLambda_{n,(1)}(\vtau^c,\vtau^v)\big\Vert\notag\\
		&=:A_{3n} + B_{3n}. \label{eq:(3.1)}
	\end{align}
	First, consider $A_{3n}$. Use \eqref{eq:(p.15)} to write
	\begin{align*}
		A_{3n} &= \big\Vert \mLambda_{n,(1)} - \mLambda_{n,(1)}(\vtheta_{0}^c,\vtau^v)\big\Vert=\big\Vert \mLambda_{n,(1)}(\vtheta_{0}^c,\vtheta_0^v) - \mLambda_{n,(1)}(\vtheta_{0}^c,\vtau^v)\big\Vert\\
		&= \bigg\Vert \frac{1}{n}\sum_{t=1}^{n} \E\bigg\{\nabla c_t(\vtheta_0^c)\nabla^\prime c_t(\vtheta_0^c) \Big[  f_t^{Y}\big(c_t(\vtheta_0^c)\big) - \partial_2 F_{t}\big(v_t(\vtheta_0^v), c_t(\vtheta_0^c)\big)\Big]\bigg\}\\
		&\hspace{1.3cm} -  \E\bigg\{\nabla c_t(\vtheta_0^c)\nabla^\prime c_t(\vtheta_0^c) \Big[  f_t^{Y}\big(c_t(\vtheta_0^c)\big) - \partial_2 F_{t}\big(v_t(\vtau^v), c_t(\vtheta_0^c)\big)\Big]\bigg\}\\
		&\hspace{1.3cm} + \E\bigg\{\nabla^2 c_t(\vtheta_0^c) \Big[ F_t^{Y}\big(c_t(\vtheta_0^c)\big) - F_{t}\big(v_t(\vtheta_0^v), c_t(\vtheta_0^c)\big)-\alpha \big\{ 1-F_{t}^{X}\big(v_t(\vtheta_0^v)\big)\big\}\Big]\bigg\}\\
		&\hspace{1.3cm} - \E\bigg\{\nabla^2 c_t(\vtheta_0^c) \Big[ F_t^{Y}\big(c_t(\vtheta_0^c)\big) - F_{t}\big(v_t(\vtau^v), c_t(\vtheta_0^c)\big)-\alpha \big\{ 1-F_{t}^{X}\big(v_t(\vtau^v)\big)\big\}\Big]\bigg\} \bigg\Vert\\
		&= \bigg\Vert \frac{1}{n}\sum_{t=1}^{n} \E\bigg\{\nabla c_t(\vtheta_0^c)\nabla^\prime c_t(\vtheta_0^c) \Big[  \partial_2 F_{t}\big(v_t(\vtau^v), c_t(\vtheta_0^c)\big) - \partial_2 F_{t}\big(v_t(\vtheta_0^v), c_t(\vtheta_0^c)\big)\Big]\bigg\}\\
		&\hspace{1.3cm} + \E\bigg\{\nabla^2 c_t(\vtheta_0^c) \Big[ F_{t}\big(v_t(\vtau^v), c_t(\vtheta_0^c)\big) - F_{t}\big(v_t(\vtheta_0^v), c_t(\vtheta_0^c)\big)\Big]\bigg\}\\
		&\hspace{1.3cm}-\alpha\E\bigg\{\nabla^2 c_t(\vtheta_0^c)  \Big[ F_{t}^{X}\big(v_t(\vtau^v)\big) - F_{t}^{X}\big(v_t(\vtheta_0^v)\big) \Big]\bigg\} \bigg\Vert.
	\end{align*}
	
	\textbf{First term ($A_{3n}$):} A mean value expansion around $\vtheta_0^v$ gives
	\begin{align*}
		&\norm{\E\bigg\{\nabla c_t(\vtheta_0^c)\nabla^\prime c_t(\vtheta_0^c) \Big[  \partial_2 F_{t}\big(v_t(\vtau^v), c_t(\vtheta_0^c)\big) - \partial_2 F_{t}\big(v_t(\vtheta_0^v), c_t(\vtheta_0^c)\big)\Big]\bigg\}}\\
		&\leq \E\norm{\nabla c_t(\vtheta_0^c)\nabla^\prime c_t(\vtheta_0^c)  \partial_1\partial_2 F_{t}\big(v_t(\vtau^{v\ast}), c_t(\vtheta_0^c)\big) \nabla^\prime v_t(\vtau^{v\ast}) \big(\vtau^v - \vtheta_0^v\big) }\\
		&\leq  \E\Big[C_1^2(\mathcal{F}_{t-1}) f_t\big(v_t(\vtau^{v\ast}), c_t(\vtheta_0^c)\big) V_1(\mathcal{F}_{t-1})\Big] \big\Vert\vtau^v - \vtheta_0^v\big\Vert\\
		&\leq K \E\Big[C_1^2(\mathcal{F}_{t-1}) V_1(\mathcal{F}_{t-1})\Big] \big\Vert\vtau^v - \vtheta_0^v\big\Vert\\
		&\leq K \Big\{\E\big[C_1^3(\mathcal{F}_{t-1})\big]\Big\}^{2/3} \Big\{\E\big[V_1^3(\mathcal{F}_{t-1})\big]\Big\}^{1/3} \big\Vert\vtau^v - \vtheta_0^v\big\Vert\\
		&\leq C \big\Vert\vtau^v - \vtheta_0^v\big\Vert,
	\end{align*}
	where $\vtau^{v\ast}$ is some value on the line connecting $\vtheta_0^v$ and $\vtau^v$, and the penultimate step follows from H\"{o}lder's inequality.

	\textbf{Second term ($A_{3n}$):} A mean value expansion around $\vtheta_0^v$ gives
	\begin{align*}
		&\norm{\E\bigg\{\nabla^2 c_t(\vtheta_0^c) \Big[ F_{t}\big(v_t(\vtau^v), c_t(\vtheta_0^c)\big) - F_{t}\big(v_t(\vtheta_0^v), c_t(\vtheta_0^c)\big)\Big]\bigg\}}\\
		&\leq \E\norm{\nabla^2 c_t(\vtheta_0^c) \partial_1 F_{t}\big(v_t(\vtau^{v\ast}), c_t(\vtheta_0^c)\big) \nabla^\prime v_t(\vtau^{v\ast}) \big(\vtau^v - \vtheta_0^v\big)}\\
		&\leq K\E\Big[C_2(\mathcal{F}_{t-1}) V_1(\mathcal{F}_{t-1})\Big] \big\Vert\vtau^v - \vtheta_0^v\big\Vert\\
		&\leq K\Big\{\E\big[C_2^{3/2}(\mathcal{F}_{t-1})\big]\Big\}^{2/3}\Big\{\E\big[V_1^3(\mathcal{F}_{t-1})\big]\Big\}^{1/3} \big\Vert\vtau^v - \vtheta_0^v\big\Vert\\
		&\leq C \big\Vert\vtau^v - \vtheta_0^v\big\Vert,
	\end{align*}
	where $\vtau^{v\ast}$ is some value on the line connecting $\vtheta_0^v$ and $\vtau^{v}$, and the penultimate step follows from H\"{o}lder's inequality.

	\textbf{Third term ($A_{3n}$):} A mean value expansion around $\vtheta_0^v$ gives
	\begin{align*}
		&\norm{\E\bigg\{\nabla^2 c_t(\vtheta_0^c) \Big[ F_{t}^{X}\big(v_t(\vtau^{v})\big) - F_{t}^{X}\big(v_t(\vtheta_0^v)\big)\Big]\bigg\}}\\
		&\leq \E\norm{\nabla^2 c_t(\vtheta_0^c) f_{t}^{X}\big(v_t(\vtau^{v\ast})\big) \nabla^\prime v_t(\vtau^{v\ast}) \big(\vtau^{v} - \vtheta_0^v\big)}\\
		&\leq K\E\Big[C_2(\mathcal{F}_{t-1}) V_1(\mathcal{F}_{t-1})\Big] \big\Vert\vtau^{v} - \vtheta_0^v\big\Vert\\
		&\leq K\Big\{\E\big[C_2^{3/2}(\mathcal{F}_{t-1})\big]\Big\}^{2/3}\Big\{\E\big[V_1^{3}(\mathcal{F}_{t-1})\big]\Big\}^{1/3} \big\Vert\vtau^{v} - \vtheta_0^v\big\Vert\\
		&\leq C \big\Vert\vtau^{v} - \vtheta_0^v\big\Vert,
	\end{align*}
	where $\vtau^{v\ast}$ is some value on the line connecting $\vtheta_0^v$ and $\vtau^{v}$, and the penultimate step follows from H\"{o}lder's inequality. 
	
	Combining the results for the different terms of $A_{3n}$, we conclude that
	\begin{equation}\label{eq:A34}
		A_{3n} \leq C \big\Vert\vtau^{v} - \vtheta_0^v\big\Vert.
	\end{equation}
	
	Now, consider $B_{3n}$. For this, use \eqref{eq:(p.15)} to write
	\begin{align*}
		& \mLambda_{n,(1)}(\vtau^c,\vtau^v) - \mLambda_{n,(1)}(\vtheta_{0}^c,\vtau^v) \\
		&= \frac{1}{n}\sum_{t=1}^{n}\E\bigg\{\nabla^2 c_t(\vtau^c) \Big[ F_t^{Y}\big(c_t(\vtau^c)\big) - F_{t}\big(v_t(\vtau^v), c_t(\vtau^c)\big)-\alpha \big\{ 1-F_{t}^{X}\big(v_t(\vtau^v)\big)\big\}\Big]\bigg\}\\
		& \ -\frac{1}{n}\sum_{t=1}^{n}\E\bigg\{\nabla^2 c_t(\vtheta_{0}^c) \Big[ F_t^{Y}\big(c_t(\vtheta_{0}^c)\big) - F_{t}\big(v_t(\vtau^v), c_t(\vtheta_{0}^c)\big)-\alpha \big\{ 1-F_{t}^{X}\big(v_t(\vtau^v)\big)\big\}\Big]\bigg\}\\
		&\hspace{2cm} + \frac{1}{n}\sum_{t=1}^{n}\E\bigg\{\nabla c_t(\vtau^c)\nabla^\prime c_t(\vtau^c) f_t^{Y}\big(c_t(\vtau^c)\big) \bigg\}\\
		&\hspace{2cm} - \frac{1}{n}\sum_{t=1}^{n}\E\bigg\{\nabla c_t(\vtheta_{0}^c)\nabla^\prime c_t(\vtheta_{0}^c) f_t^{Y}\big(c_t(\vtheta_{0}^c)\big)\bigg\}\\
		&\hspace{4cm} - \frac{1}{n}\sum_{t=1}^{n}\E\bigg\{\nabla c_t(\vtau^c)\nabla^\prime c_t(\vtau^c) \partial_2 F_{t}\big(v_t(\vtau^v), c_t(\vtau^c)\big)\bigg\}\\
		&\hspace{4cm} + \frac{1}{n}\sum_{t=1}^{n}\E\bigg\{\nabla c_t(\vtheta_{0}^c)\nabla^\prime c_t(\vtheta_{0}^c) \partial_2 F_{t}\big(v_t(\vtau^v), c_t(\vtheta_{0}^c)\big)\bigg\}.
	\end{align*}

	\textbf{First term ($B_{3n}$):} Using a mean value expansion around $\vtheta_{0}^c$ and Assumption~\ref{ass:an}, we obtain that
	\begin{align*}
		&\Bigg\Vert\frac{1}{n}\sum_{t=1}^{n}\E\bigg\{\nabla^2 c_t(\vtau^c) \Big[ F_t^{Y}\big(c_t(\vtau^c)\big) - F_{t}\big(v_t(\vtau^v), c_t(\vtau^c)\big)-\alpha \big\{ 1-F_{t}^{X}\big(v_t(\vtau^v)\big)\big\}\Big]\bigg\}\\
		& -\frac{1}{n}\sum_{t=1}^{n}\E\bigg\{\nabla^2 c_t(\vtheta_{0}^c) \Big[ F_t^{Y}\big(c_t(\vtheta_{0}^c)\big) - F_{t}\big(v_t(\vtau^v), c_t(\vtheta_{0}^c)\big)-\alpha \big\{ 1-F_{t}^{X}\big(v_t(\vtau^v)\big)\big\}\Big]\bigg\}\Bigg\Vert \\
		&\leq \frac{1}{n}\sum_{t=1}^{n}\E\norm{\nabla^2 c_t(\vtau^c)F_t^{Y}\big(c_t(\vtau^c)\big) - \nabla^2 c_t(\vtheta_{0}^c)F_t^{Y}\big(c_t(\vtheta_{0}^c)\big)}\\
		& \hspace{2cm} + \frac{1}{n}\sum_{t=1}^{n}\E\norm{\nabla^2 c_t(\vtau^c)F_{t}\big(v_t(\vtau^v), c_t(\vtau^c)\big) - \nabla^2 c_t(\vtheta_{0}^c)F_{t}\big(v_t(\vtau^v), c_t(\vtheta_{0}^c)\big)}\\
		& \hspace{2cm} + \frac{\alpha}{n}\sum_{t=1}^{n}\E\norm{\big\{\nabla^2 c_t(\vtau^c) - \nabla^2 c_t(\vtheta_{0}^c)\big\} \big\{ 1-F_{t}^{X}\big(v_t(\vtau^v)\big)\big\}}\\
		&\leq \frac{1}{n}\sum_{t=1}^{n}\E\norm{\big[\nabla^2 c_t(\vtau^c) -\nabla^2 c_t(\vtheta_{0}^c)\big]  F_t^{Y}\big(c_t(\vtau^c)\big)}  \\
		& \hspace{2cm}+ \frac{1}{n}\sum_{t=1}^{n}\E\norm{\nabla^2 c_t(\vtheta_{0}^c)\big[F_t^{Y}\big(c_t(\vtau^c)\big) - F_t^{Y}\big(c_t(\vtheta_{0}^c)\big)\big]}\\
		& \hspace{2cm} + \frac{1}{n}\sum_{t=1}^{n}\E\norm{\big[\nabla^2 c_t(\vtau^c) - \nabla^2 c_t(\vtheta_{0}^c)\big]F_{t}\big(v_t(\vtau^v), c_t(\vtau^c)\big)}\\
		& \hspace{2cm} + \frac{1}{n}\sum_{t=1}^{n}\E\norm{ \nabla^2 c_t(\vtheta_{0}^c)\big[F_{t}\big(v_t(\vtau^v), c_t(\vtau^c)\big) - F_{t}\big(v_t(\vtau^v), c_t(\vtheta_{0}^c)\big)\big]}\\
		& \hspace{2cm} + \frac{\alpha}{n}\sum_{t=1}^{n}\E\big[C_3(\mathcal{F}_{t-1})\big]\norm{\vtau^c - \vtheta_{0}^c}\\
		&\leq \frac{1}{n}\sum_{t=1}^{n}\E\big[C_3(\mathcal{F}_{t-1})\big]\norm{\vtau^c - \vtheta_{0}^c}  \\
		& \hspace{2cm}+ \frac{1}{n}\sum_{t=1}^{n}\E\norm{\nabla^2 c_t(\vtheta_{0}^c) \nabla^\prime c_t(\vtau^{c\ast}) f_t^{Y}\big(c_t(\vtau^{c\ast})\big) \big(\vtau^c - \vtheta_{0}^c\big)}\\
		& \hspace{2cm} + \frac{1}{n}\sum_{t=1}^{n}\E\big[C_3(\mathcal{F}_{t-1})\big]\norm{\vtau^c - \vtheta_{0}^c}\\
		& \hspace{2cm} + \frac{1}{n}\sum_{t=1}^{n}\E\norm{ \nabla^2 c_t(\vtheta_{0}^c) \nabla^\prime c_t(\vtau^{c\ast}) \partial_2 F_{t}\big(v_t(\vtau^v), c_t(\vtau^{c\ast})\big) \big(\vtau^c - \vtheta_{0}^c\big)}\\
		& \hspace{2cm} + \frac{\alpha}{n}\sum_{t=1}^{n}\E\big[C_3(\mathcal{F}_{t-1})\big]\norm{\vtau^c - \vtheta_{0}^c}\\
		&\leq \frac{1}{n}\sum_{t=1}^{n}\E\big[C_3(\mathcal{F}_{t-1})\big]\norm{\vtau^c - \vtheta_{0}^c}  \\
		& \hspace{2cm}+ \frac{K}{n}\sum_{t=1}^{n}\E\big[C_2(\mathcal{F}_{t-1}) C_1(\mathcal{F}_{t-1})\big]\norm{\vtau^c - \vtheta_{0}^c}\\
		& \hspace{2cm} + \frac{1}{n}\sum_{t=1}^{n}\E\big[C_3(\mathcal{F}_{t-1})\big]\norm{\vtau^c - \vtheta_{0}^c}\\
		& \hspace{2cm} + \frac{K}{n}\sum_{t=1}^{n}\E\big[C_2(\mathcal{F}_{t-1}) C_1(\mathcal{F}_{t-1})\big]\norm{\vtau^c - \vtheta_{0}^c}\\
		& \hspace{2cm} + \frac{\alpha}{n}\sum_{t=1}^{n}\E\big[C_3(\mathcal{F}_{t-1})\big]\norm{\vtau^c - \vtheta_{0}^c}\\
		&\leq \frac{2+\alpha}{n}\sum_{t=1}^{n}\E\big[C_3(\mathcal{F}_{t-1})\big]\norm{\vtau^c - \vtheta_{0}^c} + \frac{2K}{n}\sum_{t=1}^{n}\Big\{\E\big[C_2^{3/2}(\mathcal{F}_{t-1})\big]\Big\}^{2/3} \Big\{\E\big[C_1^{3}(\mathcal{F}_{t-1})\big]\Big\}^{1/3}\norm{\vtau^c - \vtheta_{0}^c}\\
		&\leq C \big\Vert\vtau^c - \vtheta_{0}^c\big\Vert,
	\end{align*}
	where $\vtau^{c\ast}$ is some mean value between $\vtau^c$ and $\vtheta_{0}^c$ that may change from line to line.

	\textbf{Second term ($B_{3n}$):} This term is dealt with similarly as the term in \eqref{eq:second term}. Expand
	\begin{align*}
		&\bigg\Vert\frac{1}{n}\sum_{t=1}^{n}\E\Big[\nabla c_t(\vtau^c)\nabla^\prime c_t(\vtau^c)f_t^{Y}\big(c_t(\vtau^c)\big) - \nabla c_t(\vtheta_{0}^{c})\nabla^\prime c_t(\vtheta_{0}^{c})f_t^{Y}\big(c_t(\vtheta_{0}^c)\big)\Big]\bigg\Vert\\
		&=\bigg\Vert \frac{1}{n}\sum_{t=1}^{n}\E\Big[\nabla c_t(\vtau^c)\nabla^\prime c_t(\vtau^c)f_t^{Y}\big(c_t(\vtau^c)\big) - \nabla c_t(\vtheta_{0}^c)\nabla^\prime c_t(\vtau^c)f_t^{Y}\big(c_t(\vtau^c)\big)\\
		&\hspace{1.8cm} + \nabla c_t(\vtheta_{0}^c)\nabla^\prime c_t(\vtau^c)f_t^{Y}\big(c_t(\vtau^c)\big) - \nabla c_t(\vtheta_{0}^c)\nabla^\prime c_t(\vtheta_{0}^c)f_t^{Y}\big(c_t(\vtau^c)\big)\\
		&\hspace{1.8cm} + \nabla c_t(\vtheta_{0}^c)\nabla^\prime c_t(\vtheta_{0}^c)f_t^{Y}\big(c_t(\vtau^c)\big) - \nabla c_t(\vtheta_{0}^c)\nabla^\prime c_t(\vtheta_{0}^c)f_t^{Y}\big(c_t(\vtheta_{0}^c)\big)\Big] \bigg\Vert\\
		&=\bigg\Vert \frac{1}{n}\sum_{t=1}^{n}\E\Big[\nabla^2 c_t(\vtau^{c\ast})(\vtau^c - \vtheta_{0}^c)\nabla^\prime c_t(\vtau^c)f_t^{Y}\big(c_t(\vtau^c)\big)\\
		&\hspace{1.8cm} + \nabla c_t(\vtheta_{0}^c)(\vtau^c - \vtheta_{0}^c)^\prime [\nabla^2 c_t(\vtau^{c\ast})]^\prime f_t^{Y}\big(c_t(\vtau^c)\big) \\
		&\hspace{1.8cm} + \nabla c_t(\vtheta_{0}^c)\nabla^\prime c_t(\vtheta_{0}^c)\big\{f_t^{Y}\big(c_t(\vtau^c)\big) - f_t^{Y}\big(c_t(\vtheta_{0}^c)\big)\big\}\Big] \bigg\Vert\\
		&\leq \frac{1}{n}\sum_{t=1}^{n}\E\big[KC_2(\mathcal{F}_{t-1})C_1(\mathcal{F}_{t-1}) + KC_1(\mathcal{F}_{t-1})C_2(\mathcal{F}_{t-1}) + K C_1^3(\mathcal{F}_{t-1}) \big]\norm{\vtau^c - \vtheta_{0}^c}\\
		&\leq \frac{1}{n}\sum_{t=1}^{n}\bigg(2K\Big\{\E\big[C_1^{3}(\mathcal{F}_{t-1})\big]\Big\}^{1/3}\Big\{\E\big[C_2^{3/2}(\mathcal{F}_{t-1})\big]\Big\}^{2/3} + K \E\big[C_1^3(\mathcal{F}_{t-1}) \big]\bigg)\norm{\vtau^c - \vtheta_{0}^c}\\
		&\leq C\big\Vert\vtau^c - \vtheta_{0}^c\big\Vert,
	\end{align*}
	where $\vtau^{c\ast}$ is some mean value between $\vtau^c$ and $\vtheta_{0}^{c}$ that may be different in different places.
	
	\textbf{Third term ($B_{3n}$):} The third term can be dealt with similarly as the second term to show that 
	\begin{multline*}
		\bigg\Vert\frac{1}{n}\sum_{t=1}^{n}\E\Big[\nabla c_t(\vtau^c)\nabla^\prime c_t(\vtau^c)\partial_2 F_{t}\big(v_t(\vtau^v), c_t(\vtau^c)\big) - \nabla c_t(\vtheta_{0}^{c})\nabla^\prime c_t(\vtheta_{0}^{c})\partial_2 F_{t}\big(v_t(\vtau^v), c_t(\vtheta_{0}^c)\big)\Big]\bigg\Vert\\
		\leq C\big\Vert\vtau^c - \vtheta_{0}^c\big\Vert.
	\end{multline*}
	Thus, we have shown that
	\begin{equation}\label{eq:B3}
		B_{3n}\leq C \big\Vert\vtau^c - \vtheta_{0}^c\big\Vert.
	\end{equation}
	
	Plugging \eqref{eq:A34} and \eqref{eq:B3} into \eqref{eq:(3.1)} gives
	\begin{equation}\label{eq:cty lambda1}
		\big\Vert\mLambda_{n,(1)} - \mLambda_{n,(1)}(\vtau^c,\vtau^v)\big\Vert\leq C \big\Vert\vtau^v - \vtheta_0^v\big\Vert + C \big\Vert\vtau^c - \vtheta_{0}^c\big\Vert.
	\end{equation}
	Therefore, the assumed consistency of the parameter estimators implies that $\mLambda_{n,(1)}(\vtheta^{c\ast},\widehat{\vtheta}_n^v)-\mLambda_{n,(1)}\overset{\P}{\longrightarrow}\vzero$.
	By positive definiteness of $\mLambda_{n,(1)}=\mLambda_{n,(1)}(\vtheta_0^c, \vtheta_0^v)$ and continuity of $\mLambda_{n,(1)}(\cdot,\cdot)$, $\mLambda_{n,(1)}^{-1}(\vtheta^{c\ast},\widehat{\vtheta}_n^v)-\mLambda_{n,(1)}^{-1}\overset{\P}{\longrightarrow}\vzero$ follows from this by Lemma~4.29 of \citet{Whi01}.
\end{proof}

\begin{lemC}\label{lem:0 tilde}
	Suppose Assumptions~\ref{ass:cons} and \ref{ass:an} hold. Then, condition (N1) of \citet{Wei91} holds, i.e., for all $t\in\mathbb{N}$, the stochastic process $\Omega\times\mTheta\ni\big(\omega,\vtheta=(\vtheta^{v},\vtheta^{c})\big)\mapsto \Cc_t(\vtheta^c,\vtheta^v)$ is separable in the sense of \citet[pp.~51--52]{Doo53}, and $\Cc_t(\vtheta^c,\vtheta^v)$ is measurable for all $(\vtheta^c,\vtheta^v)\in\mTheta$.
\end{lemC}

\begin{proof}
	From \eqref{eq:(12+)} it follows that $\Cc_t(\vtheta^c, \vtheta^v)$ is a function of measurable random variables (by Assumption~\ref{ass:cons}~\ref{it:smooth}) and, therefore, is itself measurable.
	
	Similarly as in the proof of Lemma~V.\ref{lem:0}, we show that for any open set $G\subset\mTheta$ and any closed set $F\subset\mathbb{R}^q$, the sets
	\begin{align*}
		&\big\{\omega\in\Omega\colon \Cc_t(\vtheta^c,\vtheta^v)\in F\ \text{for all}\ \vtheta=(\vtheta^{v},\vtheta^{c})\in G\big\}\quad\text{and}\\
		&\big\{\omega\in\Omega\colon \Cc_t(\vtheta^c,\vtheta^v)\in F\ \text{for all}\ \vtheta=(\vtheta^{v},\vtheta^{c})\in G\cap\mathbb{Q}^{p+q}\big\}
	\end{align*}
	coincide outside of the null set where $\nabla v_t(\cdot)=\vzero$ or $\nabla c_t(\cdot)=\vzero$.
	The inclusion ``$\subset$'' is immediate, so we only prove ``$\supset$''. 
	
	To that end, fix some closed $F\subset\mathbb{R}^{q}$ and some open $G\subset\mTheta$, and let $\omega\in\Omega$ be such, that $\Cc_t(\vtheta^c,\vtheta^v)\in F$ for all $\vtheta\in G\cap\mathbb{Q}^{p+q}$, and $\nabla v_t(\vtheta^v)\neq\vzero$ and $\nabla c_t(\vtheta^c)\neq\vzero$ for all $\vtheta\in G\cap\mathbb{Q}^{p+q}$ (such that we are outside of the null set where $\nabla v_t(\cdot)=\vzero$ or $\nabla c_t(\cdot)=\vzero$).
	Then, we have to show that 
	\begin{equation*}
		\Cc_t(\vtheta^v,\vtheta^c)\in F\qquad\text{for all }\vtheta\in G\setminus\mathbb{Q}^{p+q}.
	\end{equation*}
	Fix some $\vtheta=(\vtheta^v,\vtheta^c)\in G\setminus\mathbb{Q}^{p+q}$. 
	
	If $\vtheta$ is a continuity point of $\Cc_t(\cdot,\cdot)$, then $\Cc_t(\vtheta^c,\vtheta^v)\in F$ follows from arguments used to prove Proposition~C.2 in \citet{DimiBayer2019}.
	
	It remains to consider the case where $\vtheta$ is a point of discontinuity of $\Cc_t(\cdot,\cdot)$. 
	From \eqref{eq:(12+)} and the continuous differentiability of $v_t(\cdot)$ and $c_t(\cdot)$ (see Assumption~\ref{ass:an}~\ref{it:diff}) there is a discontinuity of $\Cc_t(\cdot,\cdot)$ in $\vtheta$ only if for every neighborhood $\mathcal{N}(\vtheta)$ of $\vtheta$ there exist $\underline{\vtheta}=(\underline{\vtheta}^{v},\underline{\vtheta}^{c}),\overline{\vtheta}=(\overline{\vtheta}^{v},\overline{\vtheta}^{c})\in \mathcal{N}(\vtheta)$, such that
	\begin{align}
		v_t(\underline{\vtheta}^v)<X_t&=v_t(\vtheta^v)\leq v_t(\overline{\vtheta}^v)\qquad\text{or}\label{eq:(2a)}\\
		c_t(\underline{\vtheta}^c)<Y_t&=c_t(\vtheta^c)\leq c_t(\overline{\vtheta}^c)\label{eq:(2b)}
	\end{align}
	(because in that case there is a discontinuity in one of the indicators in \eqref{eq:(12+)}).
	Three cases can be distinguished: There is a discontinuity of $\Cc_t(\cdot,\cdot)$ only in $\vtheta^v$ (i.e., \eqref{eq:(2a)} holds but not \eqref{eq:(2b)}), there is a discontinuity only in $\vtheta^c$ (i.e., \eqref{eq:(2b)} holds but not \eqref{eq:(2a)}) and there is a discontinuity in both $\vtheta^v$ and $\vtheta^c$ (i.e., \eqref{eq:(2a)} and \eqref{eq:(2b)} hold).
	
	We only consider the case where \eqref{eq:(2a)} and \eqref{eq:(2b)} hold, as the others can be treated similarly. As in the proof of Lemma~V.\ref{lem:0}, there cannot be saddle points in $v_t(\cdot)$ and $c_t(\cdot)$ by Assumption~\ref{ass:an}~\ref{it:diff}. 
	Therefore, there exists some $\widetilde{\vtheta}=(\widetilde{\vtheta}^{v},\widetilde{\vtheta}^{c})\in\mathcal{N}(\vtheta)\setminus\{\vtheta\}$ with $\widetilde{\vtheta}^v\neq\vtheta^v$ and $\widetilde{\vtheta}^c\neq\vtheta^c$, such that the right-hand side inequality in \eqref{eq:(2a)} is strict for $\overline{\vtheta}^v=\widetilde{\vtheta}^v$ and the right-hand side inequality in \eqref{eq:(2b)} is strict for $\overline{\vtheta}^c=\widetilde{\vtheta}^c$.
	Here, the neighborhood $\mathcal{N}(\vtheta)$ is chosen sufficiently small to lie in $G$, which is possible since $G$ is open.
	By \eqref{eq:(12+)} and the continuity of $v_t(\cdot)$ and $c_t(\cdot)$, there exists in $\mathcal{N}(\vtheta)$ a connected set of points $\mathcal{L}$ connecting $\overline{\vtheta}$ and $\vtheta$, such that
	\[
	X_t=v_t(\vtheta^v)< v_t(\overline{\vtheta}^v_{\mathcal{L}})\qquad\text{and}\qquad Y_t=c_t(\vtheta^v)< c_t(\overline{\vtheta}^c_{\mathcal{L}})\quad\text{for all }\overline{\vtheta}_{\mathcal{L}} = (\overline{\vtheta}_{\mathcal{L}}^{v}, \overline{\vtheta}_{\mathcal{L}}^{c})\in\mathcal{L}.
	\]
	Since $\mathbb{Q}^{p+q}$ is dense in $\mathbb{R}^{p+q}$, there exists a sequence $\vtheta_k=(\vtheta_k^{v}, \vtheta_k^{c})$ in $\mathcal{L}\cap\mathbb{Q}^{p+q}\subset G\cap\mathbb{Q}^{p+q}$, such that $\vtheta_k\longrightarrow\vtheta$, as $k\to\infty$.
	By continuity of $\Cc_t(\cdot,\cdot)$ on $\mathcal{L}$ it holds that $\Cc_t(\vtheta_k^c,\vtheta_k^v)\longrightarrow \Cc_t(\vtheta^c,\vtheta^v)$, as $k\to\infty$. 
	Since $\vtheta_k\in G\cap\mathbb{Q}^{p+q}$, $\Cc_t(\vtheta_k^c, \vtheta_k^v)\in F$ by assumption.
	Since $F$ is closed, the limit of the convergent series $\Cc_t(\vtheta_k^c, \vtheta_k^v)$ necessarily lies in $F$, whence $\Cc_t(\vtheta^c, \vtheta^v)\in F$.
\end{proof}

\begin{lemC}\label{lem:3 tilde}
	Suppose Assumptions~\ref{ass:cons} and \ref{ass:an} hold. Then, condition (N3) (i) of \citet{Wei91} holds, i.e.,
	\[
	\big\Vert\vlambda_n(\vtheta,\widehat{\vtheta}_n^{v})\big\Vert \geq a\big\Vert\vtheta-\vtheta_{0,n}^c\big\Vert\qquad\text{for}\ \big\Vert\vtheta-\vtheta_{0,n}^c\big\Vert\leq d_0
	\]
	for sufficiently large $n$ and some $a>0$ and $d_0>0$.
\end{lemC}

\begin{proof}
	Since we implicitly work with $\widehat{\vtheta}_n^v$ satisfying $\big\Vert\widehat{\vtheta}_n^v-\vtheta_0^v\big\Vert\leq\varepsilon_0$ (see the proof of Lemma~C.\ref{lem:1 tilde}), $\vtheta_{0,n}^c$ satisfying \eqref{eq:(2.1)} exists.
	A mean value expansion around $\vtheta_{0,n}^c$ and \eqref{eq:(2.1)} imply
	\begin{align*}
		\vlambda_n(\vtheta,\widehat{\vtheta}_n^v) &= \vlambda_n(\vtheta_{0,n}^c,\widehat{\vtheta}_n^v) + \mLambda_{n,(1)}(\vtheta^\ast,\widehat{\vtheta}_n^v)(\vtheta-\vtheta_{0,n}^c)\\
		&= \mLambda_{n,(1)}(\vtheta^\ast,\widehat{\vtheta}_n^v)(\vtheta-\vtheta_{0,n}^c)
	\end{align*}
	for some $\vtheta^\ast$ on the line connecting $\vtheta_{0,n}^{c}$ and $\vtheta$.
	By Assumption~\ref{ass:an} \ref{it:pd}, $\mLambda_{n,(1)}$ has eigenvalues bounded below by some $a>0$, such that (similarly as in the proof of Lemma~V.\ref{lem:3})
	\begin{align}
		\big\Vert\vlambda_n(\vtheta,\widehat{\vtheta}_n^{v})\big\Vert &= \big\Vert\mLambda_{n,(1)}(\vtheta^\ast,\widehat{\vtheta}_n^v)(\vtheta-\vtheta_{0,n}^c)\big\Vert\notag\\
		&= \big\Vert\mLambda_{n,(1)}(\vtheta-\vtheta_{0,n}^c) - \big[\mLambda_{n,(1)} - \mLambda_{n,(1)}(\vtheta^\ast,\widehat{\vtheta}_n^v)\big](\vtheta-\vtheta_{0,n}^c)\big\Vert\notag\\
		&\geq \big\Vert\mLambda_{n,(1)}(\vtheta-\vtheta_{0,n}^c)\big\Vert - \big\Vert\big[\mLambda_{n,(1)} - \mLambda_{n,(1)}(\vtheta^\ast,\widehat{\vtheta}_n^v)\big](\vtheta-\vtheta_{0,n}^c)\big\Vert\notag\\
		&\geq a\big\Vert\vtheta-\vtheta_{0,n}^c\big\Vert - \big\Vert\mLambda_{n,(1)} - \mLambda_{n,(1)}(\vtheta^\ast,\widehat{\vtheta}_n^v)\big\Vert\cdot \big\Vert\vtheta-\vtheta_{0,n}^c\big\Vert.\label{eq:(31-)}
	\end{align}
	Our goal in the following is to show that  for sufficiently large $n$,
	\begin{equation}\label{eq:to show an}
		\big\Vert\mLambda_{n,(1)} - \mLambda_{n,(1)}(\vtheta^\ast,\widehat{\vtheta}_n^v)\big\Vert \leq c
	\end{equation}
	for some $c\in(0,a)$, because then the lemma is established. We show \eqref{eq:to show an} by decomposing the left-hand side into the three terms
	\begin{align}
		&\big\Vert\mLambda_{n,(1)} - \mLambda_{n,(1)}(\vtheta^\ast,\widehat{\vtheta}_n^v)\big\Vert \notag\\
		&= \Big\Vert \big[\mLambda_{n,(1)} - \mLambda_{n,(1)}(\vtheta_{0}^c,\widehat{\vtheta}_n^v)\big] + \big[ \mLambda_{n,(1)}(\vtheta_{0}^{c},\widehat{\vtheta}_n^v) - \mLambda_{n,(1)}(\vtheta_{0,n}^{c},\widehat{\vtheta}_n^v)\big] \notag\\
		&\hspace{8cm}+ \big[ \mLambda_{n,(1)}(\vtheta_{0,n}^{c},\widehat{\vtheta}_n^v) - \mLambda_{n,(1)}(\vtheta^{\ast},\widehat{\vtheta}_n^v)\big] \Big\Vert\notag\\ 
		&\leq \big\Vert \mLambda_{n,(1)} - \mLambda_{n,(1)}(\vtheta_{0}^c,\widehat{\vtheta}_n^v)\big\Vert + \big\Vert\mLambda_{n,(1)}(\vtheta_{0}^{c},\widehat{\vtheta}_n^v) - \mLambda_{n,(1)}(\vtheta_{0,n}^{c},\widehat{\vtheta}_n^v)\big\Vert \notag\\
		&\hspace{8cm} + \big\Vert\mLambda_{n,(1)}(\vtheta_{0,n}^{c},\widehat{\vtheta}_n^v) - \mLambda_n^{(1)}(\vtheta^{\ast},\widehat{\vtheta}_n^v)\big\Vert\notag\\
		&=:A_{4n} + B_{4n} + C_{4n}.\label{eq:(32)}
	\end{align}
	
	The term $A_{4n}$ is similar to $A_{3n}$ from the proof of Lemma~C.\ref{lem:1+ tilde}.
	We deduce from \eqref{eq:A34} and $\big\Vert\widehat{\vtheta}_n^v - \vtheta_0^v\big\Vert\leq\varepsilon_0$ that
	\[
	A_{4n}\leq C\big\Vert\widehat{\vtheta}_n^v - \vtheta_0^v\big\Vert\leq C\varepsilon_0.
	\]
	By a suitable choice of $\varepsilon_0>0$, $A_{4n}$ can be made arbitrarily small.

	Now, consider $C_{4n}$. Reasoning similarly as for \eqref{eq:B3} in the proof of Lemma~C.\ref{lem:1+ tilde}, we obtain that 
	\begin{equation*}
		C_{4n}=\big\Vert\mLambda_{n,(1)}(\vtheta^\ast,\widehat{\vtheta}_n^v) - \mLambda_{n,(1)}(\vtheta_{0,n}^c,\widehat{\vtheta}_n^v)\big\Vert\leq C \big\Vert\vtheta^\ast - \vtheta_{0,n}^c\big\Vert\leq C\big\Vert \vtheta - \vtheta_{0,n}^c\big\Vert\leq C d_0,
	\end{equation*}
	where the second inequality follows from the fact that $\vtheta^\ast$ lies on the line connecting $\vtheta$ and $\vtheta_{0,n}^c$.
	By a suitable choice of $d_0>0$, $C_{4n}$ can be made arbitrarily small.
	
	Similarly, we may show that
	\begin{equation*}
		B_{4n}=\big\Vert\mLambda_n^{(1)}(\vtheta_{0}^{c},\widehat{\vtheta}_n^v) - \mLambda_n^{(1)}(\vtheta_{0,n}^{c},\widehat{\vtheta}_n^v)\big\Vert\leq C \big\Vert\vtheta_{0}^{c} - \vtheta_{0,n}^{c}\big\Vert,
	\end{equation*}
	which again may be made arbitrarily small because---by definition of $\vtheta_{0,n}^{c}$---Assumption~\ref{ass:an}~\ref{it:lambda func} implies
	\begin{align}
		\big\Vert\vtheta_{0}^{c} - \vtheta_{0,n}^{c}\big\Vert &= \big\Vert\vtheta_{n}^{c}(\vtheta_{0}^{v}) - \vtheta_{n}^{c}(\widehat{\vtheta}_{n}^{v})\big\Vert\notag\\
		&\leq K \big\Vert \widehat{\vtheta}_{n}^{v} - \vtheta_{0}^{v}\big\Vert\label{eq:(E.21m)}\\
		&\leq K\varepsilon_0.\label{eq:bound theta0nc}
	\end{align}
	Therefore, $B_{4n}$ can be made arbitrarily small by a suitable choice of $\varepsilon_0>0$.

	In total, we have shown that $A_{4n}+B_{4n}+C_{4n}$ can be made arbitrarily small, which---by \eqref{eq:(32)}---establishes \eqref{eq:to show an}. Combining this with \eqref{eq:(31-)}, the conclusion follows.
\end{proof}

\begin{lemC}\label{lem:4 tilde}
	Suppose Assumptions~\ref{ass:cons} and \ref{ass:an} hold, and define
	\[
	\mu_t(\vtheta, d) =\sup_{\norm{\vtau-\vtheta}\leq d}\big\Vert\Cc_{t}(\vtau,\widehat{\vtheta}_n^v) - \Cc_{t}(\vtheta,\widehat{\vtheta}_n^v)\big\Vert.
	\]
	Then, condition (N3) (ii) of \citet{Wei91} holds, i.e.,
	\[
	\E\big[\mu_t(\vtheta, d)\big] \leq bd\qquad\text{for}\ \big\Vert\vtheta-\vtheta_{0,n}^c\big\Vert + d\leq d_0
	\]
	for sufficiently large $n$ and some strictly positive $b$, $d$, $d_0$.
\end{lemC}

\begin{proof}
	Choose $d_0>0$ and $d>0$ such that $\mathcal{N}(\vtheta_{0,n}^c)=\big\{\vtheta\in\mTheta^c:\ \big\Vert\vtheta-\vtheta_{0,n}^c\big\Vert+d\leq d_0\big\}$ is a subset of the neighborhood of Assumption~\ref{ass:an} \ref{it:bound2.2}. This is possible because for all $\vtheta\in\mathcal{N}(\vtheta_{0,n}^c)$, we have that
	\begin{align*}
		\big\Vert\vtheta - \vtheta_{0}^c\big\Vert&= \big\Vert\vtheta - \vtheta_{0,n}^c + \vtheta_{0,n}^c - \vtheta_{0}^c\big\Vert\\
		&\leq  \big\Vert\vtheta - \vtheta_{0,n}^c\big\Vert + \big\Vert\vtheta_{0,n}^c - \vtheta_{0}^c\big\Vert\\
		&\leq d_0 + K\varepsilon_0,
	\end{align*}
	where we used \eqref{eq:bound theta0nc}.
	Therefore, by a suitable choice of $d_0$ and $\varepsilon_0>0$, we can ensure that any $\vtheta\in\mathcal{N}(\vtheta_{0,n}^c)$ also lies in the neighborhood of $\vtheta_0^c$ of Assumption~\ref{ass:an} \ref{it:bound2.2}.
	
	Recalling the definition of $\Cc_{t}(\vtheta^c, \vtheta^v)$ from \eqref{eq:(12+)}, we decompose
	\begin{align*}
		\mu_t(\vtheta, d) &= \sup_{\norm{\vtau-\vtheta}\leq d}\norm{ \1_{\{X_t> v_t(\widehat{\vtheta}_n^v)\}}\nabla c_t(\vtau) \big[\1_{\{Y_t\leq c_t(\vtau)\}} - \alpha\big] - \1_{\{X_t> v_t(\widehat{\vtheta}_n^v)\}}\nabla c_t(\vtheta) \big[\1_{\{Y_t\leq c_t(\vtheta)\}} - \alpha\big]}\\
		&=  \sup_{\norm{\vtau-\vtheta}\leq d}\Big\Vert \1_{\{X_t> v_t(\widehat{\vtheta}_n^v)\}}\big[\nabla c_t(\vtau) \1_{\{Y_t\leq c_t(\vtau)\}} - \nabla c_t(\vtheta) \1_{\{Y_t\leq c_t(\vtheta)\}}\big] \\
		&\hspace{8cm} + \alpha \1_{\{X_t> v_t(\widehat{\vtheta}_n^v)\}}\big[\nabla c_t(\vtheta) - \nabla c_t(\vtau)\big]\Big\Vert\\
		&\leq \sup_{\norm{\vtau-\vtheta}\leq d}\big\Vert \nabla c_t(\vtau) \1_{\{Y_t\leq c_t(\vtau)\}} - \nabla c_t(\vtheta) \1_{\{Y_t\leq c_t(\vtheta)\}}\big\Vert \\
		&\hspace{8cm} + \alpha \sup_{\norm{\vtau-\vtheta}\leq d}\big\Vert\nabla c_t(\vtau) - \nabla c_t(\vtheta)\big\Vert\\
		&=:\mu_{1t}(\vtheta, d) + \mu_{2t}(\vtheta, d).
	\end{align*}
	The remainder of the proof is almost identical to that of Lemma~V.\ref{lem:4}. One merely has to replace $\vtheta_0^v$ by $\vtheta_{0,n}^{c}$, $v_t(\cdot)$ by $c_t(\cdot)$, and $X_t$ by $Y_t$. We nonetheless carry out the steps here to demonstrate where in the proof the various conditions of Assumption~\ref{ass:an} come into play.
	
	Define the $\mathcal{F}_{t-1}$-measurable quantities
	\begin{align*}
		\underline{\vtau} &:= \argmin_{\norm{\vtau-\vtheta}\leq d}c_t(\vtau),\\
		\overline{\vtau} &:= \argmax_{\norm{\vtau-\vtheta}\leq d}c_t(\vtau),
	\end{align*}
	which exist by continuity of $c_t(\cdot)$.
	
	We first consider $\mu_{1t}(\vtheta, d)$. To take the indicators out of the supremum in $\mu_{1t}(\vtheta, d)$, we distinguish two cases:
	
	\textbf{Case 1: $Y_t\leq c_t(\vtheta)$}
	
	We further distinguish two cases (a)--(b).
	
	(a) If $Y_t<c_t(\underline{\vtau})$, then both indicators in $\mu_{1t}(\vtheta, d)$ equal one, such that
	\[
	\mu_{1t}(\vtheta, d) = \sup_{\norm{\vtau-\vtheta}\leq d}\norm{\nabla c_t(\vtau) - \nabla c_t(\vtheta)}.
	\]
	
	(b) If $c_t(\underline{\vtau})\leq Y_t \leq c_t(\overline{\vtau})$, then
	\begin{align}
		\mu_{1t}(\vtheta, d) &=\max\bigg\{\sup_{\substack{\norm{\vtau-\vtheta}\leq d\\ Y_t\leq c_t(\vtau)}}\norm{\nabla c_t(\vtau) - \nabla c_t(\vtheta)},\ \norm{\nabla c_t(\vtheta)}\bigg\}\notag\\
		&\leq \sup_{\norm{\vtau-\vtheta}\leq d}\norm{\nabla c_t(\vtau) - \nabla c_t(\vtheta)} + \norm{\nabla c_t(\vtheta)}.\notag
	\end{align}
	
	(Notice that the third case that $Y_t>c_t(\overline{\vtau})$ cannot occur, because already $Y_t\leq c_t(\vtheta)$.)
	
	Together, (a) and (b) yield that
	\begin{align}
		\mu_{1t}(\vtheta,d) &\leq \sup_{\norm{\vtau-\vtheta}\leq d}\norm{\nabla c_t(\vtau) - \nabla c_t(\vtheta)} + \1_{\{c_t(\underline{\vtau})\leq Y_t\leq c_t(\vtheta)\}}\norm{\nabla c_t(\vtheta)}\notag\\
		&\leq \sup_{\norm{\vtau-\vtheta}\leq d}\norm{\nabla c_t(\vtau) - \nabla c_t(\vtheta)} + \1_{\{c_t(\underline{\vtau})\leq Y_t\leq c_t(\vtheta)\}}\sup_{\Vert\vtau-\vtheta_{0,n}^c\Vert\leq d_0}\norm{\nabla c_t(\vtau)}, \label{eq:(p.26.1)}
	\end{align}
	where we used that $\vtheta$ is an a $d_0$-neighborhood of $\vtheta_{0,n}^c$ by assumption.
	
	\textbf{Case 2: $Y_t> c_t(\vtheta)$}
	
	In this case,
	\begin{align}
		\mu_{1t}(\vtheta, d) &= \1_{\{Y_t\leq c_t(\overline{\vtau})\}}\sup_{\substack{\norm{\vtau-\vtheta}\leq d\\ Y_t\leq c_t(\vtau)}}\norm{\nabla c_t(\vtau)}\notag\\
		&\leq \1_{\{Y_t\leq c_t(\overline{\vtau})\}}\sup_{\norm{\vtau-\vtheta}\leq d}\norm{\nabla c_t(\vtau)}.\label{eq:(N.12.2) tilde}
	\end{align}
	Note that $\norm{\vtau - \vtheta}\leq d$ and our assumption $\big\Vert\vtheta-\vtheta_{0,n}^{c}\big\Vert+d\leq d_0$ together imply that $\vtau$ is in a $d_0$-neighborhood of $\vtheta_{0,n}^c$, because $\big\Vert\vtau-\vtheta_{0,n}^c\big\Vert=\big\Vert\vtau-\vtheta + \vtheta-\vtheta_{0,n}^c\big\Vert\leq \norm{\vtau-\vtheta} + \big\Vert\vtheta-\vtheta_{0,n}^c\big\Vert\leq d + (d_0-d)=d_0$.)
	Hence, from \eqref{eq:(N.12.2) tilde},
	\begin{equation}\label{eq:(p.26.2)}
		\mu_{1t}(\vtheta, d) \leq \1_{\{c_t(\vtheta)<Y_t\leq c_t(\overline{\vtau})\}}\sup_{\Vert\vtau-\vtheta_{0,n}^c\Vert\leq d_0}\norm{\nabla c_t(\vtau)}.
	\end{equation}
	
	Combining the results from Cases 1 and 2 (i.e., \eqref{eq:(p.26.1)} and \eqref{eq:(p.26.2)}), we obtain that
	\begin{multline}\label{eq:(E.24p)}
		\mu_{1t}(\vtheta, d) \leq \Big[ \1_{\{c_t(\underline{\vtau})\leq Y_t\leq c_t(\vtheta)\}} + \1_{\{c_t(\vtheta)<Y_t\leq c_t(\overline{\vtau})\}}\Big] \sup_{\Vert\vtau - \vtheta_{0,n}^c\Vert\leq d_0}\norm{\nabla c_t(\vtau)}\\
		+ \sup_{\norm{\vtau-\vtheta}\leq d}\norm{\nabla c_t(\vtau) - \nabla c_t(\vtheta)}.
	\end{multline}
	By Assumption~\ref{ass:an} \ref{it:bound cdf} we have
	\begin{align*}
		\E_{t-1}\Big[\1_{\{c_t(\underline{\vtau})\leq Y_t\leq c_t(\vtheta)\}}\Big] &=\int_{c_t(\underline{\vtau})}^{c_t(\vtheta)}f_t^{Y}(y)\D y\notag\\
		&\leq K\big|c_t(\vtheta) - c_t(\underline{\vtau})\big|\\
		&= K\big|\nabla^\prime c_t(\vtheta^\ast)(\vtheta - \underline{\vtau})\big|\notag\\
		&\leq K C_1(\mathcal{F}_{t-1})\norm{\vtheta - \underline{\vtau}} \\
		&\leq K C_1(\mathcal{F}_{t-1}) d,
	\end{align*}
	and, similarly,
	\begin{equation*}
		\E_{t-1}\Big[\1_{\{c_t(\vtheta)<Y_t\leq c_t(\overline{\vtau})\}}\Big]  \leq K C_1(\mathcal{F}_{t-1}) d.
	\end{equation*}
	Moreover, we have by the MVT that for some $\vtheta^\ast$ on the line connecting $\vtau$ and $\vtheta$,
	\begin{align}
		\sup_{\norm{\vtau-\vtheta}\leq d}\big\Vert\nabla c_t(\vtau) - \nabla c_t(\vtheta)\big\Vert &= \sup_{\norm{\vtau-\vtheta}\leq d}\big\Vert\nabla^2 c_t(\vtheta^\ast)(\vtau-\vtheta)\big\Vert\notag\\
		&\leq C_2(\mathcal{F}_{t-1})\big\Vert\vtau - \vtheta\big\Vert\\
		&\leq C_2(\mathcal{F}_{t-1}) d.\label{eq:(N.14.1) tilde}
	\end{align}
	Therefore, using \eqref{eq:(E.24p)},
	\[
	\E\big[\mu_{1t}(\vtheta,d)\big] \leq \E\big[K C_1^2(\mathcal{F}_{t-1})\big]d + \E\big[C_2(\mathcal{F}_{t-1})\big]d\leq Cd.
	\]

	By arguments leading to \eqref{eq:(N.14.1) tilde}, we also have that
	\[
	\E\big[\mu_{2t}(\vtheta,d)\big]\leq\alpha\E\big[C_2(\mathcal{F}_{t-1})\big]d\leq Cd.
	\]
	
	Overall, 
	\[
	\E\big[\mu_t(\vtheta,d)\big]\leq \E\big[\mu_{1t}(\vtheta,d)\big] + \E\big[\mu_{2t}(\vtheta,d)\big]\leq bd
	\]
	for some suitable $b>0$, as desired.
\end{proof}

\begin{lemC}\label{lem:5 tilde}
	Suppose Assumptions~\ref{ass:cons} and \ref{ass:an} hold. Then, condition (N3) (iii) of \citet{Wei91} holds, i.e.,
	\[
	\E\big[\mu_t^r(\vtheta, d)\big] \leq cd\qquad\text{for}\ \big\Vert\vtheta-\vtheta_{0,n}^c\big\Vert + d\leq d_0
	\]
	for sufficiently large $n$ and some $c>0$, $d\geq0$, $d_0>0$ and $r>2$.
\end{lemC}

\begin{proof}
	\textit{Mutatis mutandis} the proof follows the same lines as that of Lemma~V.\ref{lem:5} and, hence, is omitted.
\end{proof}

\begin{lemC}\label{lem:6 tilde}
	Suppose Assumptions~\ref{ass:cons} and \ref{ass:an} hold. Then,
	\[
	\E\big\Vert\Cc_{t}(\vtheta_{0,n}^c, \widehat{\vtheta}_n^v)\big\Vert^{2+\iota}\leq C\qquad\text{for all }t\in\mathbb{N}.
	\]
	In particular, condition (N4) of \citet{Wei91} holds.
\end{lemC}

\begin{proof}
	Note that from \eqref{eq:(E.21m)}, 
	\begin{equation*}
		\big\Vert\vtheta_{0}^{c} - \vtheta_{0,n}^{c}\big\Vert \leq K \big\Vert \widehat{\vtheta}_{n}^{v}- \vtheta_{0}^{v}\big\Vert,
	\end{equation*}
	such that, for $\widehat{\vtheta}_{n}^{v}$ satisfying$\big\Vert\widehat{\vtheta}_n^v - \vtheta_0^v\big\Vert\leq\varepsilon_0$, $\vtheta_{0,n}^{c}$ is in the neighborhood of $\vtheta_0^c$ of Assumption~\ref{ass:an}~\ref{it:bound2.2} for a suitable small $\varepsilon_0>0$. Therefore, from Assumption~\ref{ass:an}~\ref{it:bound2.2}--\ref{it:mom bounds cons2} and \eqref{eq:(12+)}, we deduce that
	\begin{equation*}
		\E\big\Vert\Cc_{t}(\vtheta_{0,n}^c, \widehat{\vtheta}_n^v)\big\Vert^{2+\iota} \leq \E\big\Vert\nabla c_t(\vtheta_{0,n}^c)\big\Vert^{2+\iota} \leq \E\big[C_1^{2+\iota}(\mathcal{F}_{t-1})\big]\leq C.
	\end{equation*}
	This ends the proof.
\end{proof}

\begin{lemC}\label{lem:2 tilde}
	Suppose Assumptions~\ref{ass:cons} and \ref{ass:an} hold. Then, as $n\to\infty$,
	\[
	\frac{1}{\sqrt{n}}\sum_{t=1}^{n}\Cc_{t}(\widehat{\vtheta}_n^{c}, \widehat{\vtheta}_n^{v})=o_{\P}(1).
	\]
\end{lemC}

\begin{proof}
	Recall from Assumption~\ref{ass:cons} \ref{it:compact} that $\mTheta^{c}\subset\mathbb{R}^{q}$, such that $\vtheta^c=(\theta_1^c,\ldots,\theta_q^c)^\prime$. Let $\ve_1,\ldots,\ve_q$ denote the standard basis of $\mathbb{R}^{q}$. Define
	\[
	S_{j,n}^{\CoVaR}(a) := \frac{1}{\sqrt{n}}\sum_{t=1}^{n}S^{\CoVaR}\Big(\big(v_t(\widehat{\vtheta}_n^{v}), c_t(\widehat{\vtheta}_n^{c}+a \ve_j)\big)^\prime, (X_t,Y_t)^\prime\Big),\qquad j=1,\ldots,q,
	\]
	where $a\in\mathbb{R}$. Let $G_{j,n}(a)$ be the right partial derivative of $S_{j,n}^{\CoVaR}(a)$, such that (see \eqref{eq:(12+)})
	\[
	G_{j,n}(a)=\frac{1}{\sqrt{n}}\sum_{t=1}^{n}\1_{\{X_t> v_t(\widehat{\vtheta}_n^{v})\}}\nabla_{j} c_t(\widehat{\vtheta}_n^{c}+a \ve_j)\big[\1_{\{Y_t\leq c_t(\widehat{\vtheta}_n^{c}+a \ve_j)\}}-\alpha\big],
	\]
	where $\nabla_j c_t(\cdot)$ is the $j$-th component of $\nabla c_t(\cdot)$. Then, $G_{j,n}(0)=\lim_{\xi\downarrow0}G_{j,n}(\xi)$ is the right partial derivative of
	\[
	S_n^{\CoVaR}(\vtheta^{c}) = \frac{1}{\sqrt{n}}\sum_{t=1}^{n}S^{\CoVaR}\Big(\big(v_t(\widehat{\vtheta}_n^v), c_t(\vtheta^c)\big)^\prime, (X_t,Y_t)^\prime\Big)
	\]
	at $\widehat{\vtheta}_n^{c}$ in the direction $\theta_j^{c}$.
	Correspondingly, $\lim_{\xi\downarrow0}G_{j,n}(-\xi)$ is the left partial derivative. Because $S_n^{\CoVaR}(\cdot)$ achieves its minimum at $\widehat{\vtheta}_n^{c}$, the left derivative must be non-positive and the right derivative must be non-negative. Thus, for sufficiently small $\xi>0$,
	\begin{align*}
		\big|G_{j,n}(0)\big| &\leq G_{j,n}(\xi) - G_{j,n}(-\xi)\\
		&= \frac{1}{\sqrt{n}}\sum_{t=1}^{n}\1_{\{X_t> v_t(\widehat{\vtheta}_n^{v})\}}\nabla_{j} c_t(\widehat{\vtheta}_n^{c}+\xi \ve_j)\big[\1_{\{Y_t\leq c_t(\widehat{\vtheta}_n^{c}+\xi \ve_j)\}}-\alpha\big]\\
		&\hspace{2cm} - \frac{1}{\sqrt{n}}\sum_{t=1}^{n}\1_{\{X_t> v_t(\widehat{\vtheta}_n^{v})\}}\nabla_{j} c_t(\widehat{\vtheta}_n^{c}-\xi \ve_j)\big[\1_{\{Y_t\leq c_t(\widehat{\vtheta}_n^{c}-\xi \ve_j)\}}-\alpha\big].
	\end{align*}
	By continuity of $\nabla c_t(\cdot)$ (see Assumption~\ref{ass:an} \ref{it:diff}) it follows upon letting $\xi\to0$ that
	\begin{align}
		\big|G_{j,n}(0)\big| &\leq \frac{1}{\sqrt{n}}\sum_{t=1}^{n}\1_{\{X_t> v_t(\widehat{\vtheta}_n^{v})\}}\big|\nabla_{j} c_t(\widehat{\vtheta}_n^{c})\big|\1_{\{Y_t= c_t(\widehat{\vtheta}_n^{c})\}}\notag\\
		&\leq \frac{1}{\sqrt{n}}\sum_{t=1}^{n}\big|\nabla_{j} c_t(\widehat{\vtheta}_n^{c})\big|\1_{\{Y_t= c_t(\widehat{\vtheta}_n^{c})\}}\notag\\	
		&\leq \frac{1}{\sqrt{n}}\Big[\max_{t=1,\ldots,n}C_1(\mathcal{F}_{t-1})\Big]\sum_{t=1}^{n}\1_{\{Y_t= c_t(\widehat{\vtheta}_n^{c})\}}.\label{eq:(N.7.1) tilde}
	\end{align}
	Subadditivity, Markov's inequality and Assumption~\ref{ass:an}~\ref{it:mom bounds cons2} imply
	\begin{align*}
		\P\Big\{n^{-1/2}\max_{t=1,\ldots,n}C_1(\mathcal{F}_{t-1})>\varepsilon\Big\} &\leq \sum_{t=1}^{n}\P\Big\{C_1(\mathcal{F}_{t-1})>\varepsilon n^{1/2}\Big\}\\
		&\leq \sum_{t=1}^{n}\varepsilon^{-3}n^{-3/2}\E\big[C_1^3(\mathcal{F}_{t-1})\big]\\
		& =O(n^{-1/2})=o(1).
	\end{align*}
	Combining this with Assumption~\ref{ass:an} \ref{it:eq bound}, we obtain from \eqref{eq:(N.7.1) tilde} that
	\[
	\big|G_{j,n}(0)\big| \overset{\text{a.s.}}{=} o_{\P}(1)O(1)=o_{\P}(1).
	\]
	As this holds for every $j=1,\ldots,q$, we get that 
	\[
	\frac{1}{\sqrt{n}}\sum_{t=1}^{n}\Cc_{t}(\widehat{\vtheta}_n^{c}, \widehat{\vtheta}_n^{v})\overset{\eqref{eq:(12+)}}{=}\frac{1}{\sqrt{n}}\sum_{t=1}^{n}\1_{\{X_t> v_t(\widehat{\vtheta}_n^{v})\}}\nabla c_t(\widehat{\vtheta}_n^{c})\big[\1_{\{Y_t\leq c_t(\widehat{\vtheta}_n^{c})\}}-\alpha\big]=o_{\P}(1),
	\]
	which is just the conclusion.
\end{proof}

We have now established all lemmas used in the proof of Lemma~C.\ref{lem:1 tilde}.
It only remains to prove Lemma~C.\ref{lem:7 tilde}.

\begin{proof}[{\textbf{Proof of Lemma~C.\ref{lem:7 tilde}:}}]
	Similarly as in Lemma~C.\ref{lem:6 tilde}, it may be shown that
	\begin{equation}\label{eq:(p.36)}
		\E\norm{\Cc_{t}(\vtheta_{0}^c, \vtheta_0^v)}^{2+\iota}\leq C\qquad\text{for all }t\in\mathbb{N}.
	\end{equation}
	Hence, following similar steps used in the proof of Lemma~V.\ref{lem:7}, we can show that, as $n\to\infty$,
	\[
	n^{-1/2}\mC_{n}^{-1/2}\sum_{t=1}^{n}\Cc_{t}(\vtheta_{0}^c, \vtheta_0^v)\overset{d}{\longrightarrow}N(\vzeros,\mI),
	\]
	where, by the LIE and the definition of the CoVaR $c_t(\vtheta_0^c)$,
	\begin{align}
		\E&\big[ \Cc_{t}(\vtheta_0^c,\vtheta_0^v)\Cc_{t}^\prime(\vtheta_0^c, \vtheta_0^v)\big]\notag\\
		&= \E\Big[ \nabla c_t(\vtheta_0^c)\nabla^\prime c_t(\vtheta_0^c) \1_{\{X_t>v_t(\vtheta_0^v)\}}\big(\1_{\{Y_t\leq c_t(\vtheta_0^c)\}} - \alpha\big)^2\Big]\notag\\
		&= \E\Big\{ \nabla c_t(\vtheta_0^c)\nabla^\prime c_t(\vtheta_0^c) \E_{t-1}\big[\1_{\{X_t>v_t(\vtheta_0^v),\ Y_t\leq c_t(\vtheta_0^v)\}} - 2\alpha \1_{\{X_t>v_t(\vtheta_0^v),\ Y_t\leq c_t(\vtheta_0^v)\}} +\alpha^2\1_{\{X_t>v_t(\vtheta_0^v)\}} \big]\Big\}\notag\\
		&= \E\bigg\{ \nabla c_t(\vtheta_0^c)\nabla^\prime c_t(\vtheta_0^c) \Big[\P_{t-1}\big\{X_t>v_t(\vtheta_0^v),\ Y_t\leq c_t(\vtheta_0^c)\big\} \notag\\
		&\hspace{5cm}- 2\alpha \P_{t-1}\big\{X_t>v_t(\vtheta_0^v),\ Y_t\leq c_t(\vtheta_0^c)\big\} +\alpha^2\P_{t-1}\big\{X_t>v_t(\vtheta_0^v)\big\} \Big]\bigg\}\notag\\
		&= \E\bigg\{ \nabla c_t(\vtheta_0^c)\nabla^\prime c_t(\vtheta_0^c) \Big[\P_{t-1}\big\{Y_t\leq c_t(\vtheta_0^c)\mid X_t>v_t(\vtheta_0^v)\big\}(1-\beta) \notag\\
		&\hspace{5cm}- 2\alpha \P_{t-1}\big\{Y_t\leq c_t(\vtheta_0^c)\mid X_t>v_t(\vtheta_0^v)\big\}(1-\beta)+\alpha^2(1-\beta) \Big]\bigg\}\notag\\
		&= \E\bigg\{ \nabla c_t(\vtheta_0^c)\nabla^\prime c_t(\vtheta_0^c) \Big[\alpha(1-\beta)- 2\alpha^2(1-\beta)+\alpha^2(1-\beta) \Big]\bigg\}\notag\\
		&=\alpha(1-\alpha)(1-\beta)\E\big[\nabla c_t(\vtheta_0^c)\nabla^\prime c_t(\vtheta_0^c)\big].\label{eq:B simpl}
	\end{align}
	Thus, to establish the lemma, we only have to show that
	\[
	\frac{1}{\sqrt{n}}\sum_{t=1}^{n}\big[\Cc_{t}(\vtheta_{0}^c, \vtheta_0^v) - \Cc_{t}(\vtheta_{0,n}^c, \widehat{\vtheta}_n^v)\big] = o_{\P}(1).
	\]
	Decompose
	\begin{align*}
		\frac{1}{\sqrt{n}}&\sum_{t=1}^{n}\big[\Cc_{t}(\vtheta_{0}^c, \vtheta_0^v) - \Cc_{t}(\vtheta_{0,n}^c, \widehat{\vtheta}_n^v)\big]\\
		&= \frac{1}{\sqrt{n}}\sum_{t=1}^{n}\Big\{\Cc_{t}(\vtheta_{0}^c, \vtheta_0^v) - \E\big[\Cc_{t}(\vtheta_{0}^c, \vtheta_0^v)\big]\Big\} - \Big\{\Cc_{t}(\vtheta_{0}^c, \widehat{\vtheta}_n^v) - \E\big[\Cc_{t}(\vtheta_{0}^c, \widehat{\vtheta}_n^v)\big]\Big\}\\
		&\hspace{1cm} + \frac{1}{\sqrt{n}}\sum_{t=1}^{n}\Big\{\Cc_{t}(\vtheta_{0}^c, \widehat{\vtheta}_n^v) - \E\big[\Cc_{t}(\vtheta_{0}^c, \widehat{\vtheta}_n^v)\big]\Big\} - \Big\{\Cc_{t}(\vtheta_{0,n}^c, \widehat{\vtheta}_n^v) - \E\big[\Cc_{t}(\vtheta_{0,n}^c, \widehat{\vtheta}_n^v)\big]\Big\}\\
		&=:A_{5n} + B_{5n},
	\end{align*}
	where we used that $\E\big[\Cc_{t}(\vtheta_{0}^c, \vtheta_0^v)\big]=\vzero$ and $\vzero=\vlambda_n(\vtheta_{0,n}^c, \widehat{\vtheta}_n^v)=(1/n)\sum_{t=1}^{n}\E\big[\Cc_{t}(\vtheta_{0,n}^c, \widehat{\vtheta}_n^v)\big]$. 
	
	Our first goal is to show that $B_{5n}=o_{\P}(1)$. To task this, note from $\sqrt{n}$-consistency of $\widehat{\vtheta}_n^v$ (from Section~\ref{Asymptotic Normality of the VaR Parameter Estimator} in the main paper) and \eqref{eq:(E.21m)} that
	\begin{equation}\label{eq:root n theta0n}
		\sqrt{n}\big\Vert\vtheta_{0,n}^c - \vtheta_0^c\big\Vert=O_{\P}(1).
	\end{equation}
	Now, we can exploit \eqref{eq:root n theta0n} to prove $B_{5n}=o_{\P}(1)$. Recall from the proof of Lemma~C.\ref{lem:1 tilde} that conditions (N1)--(N5) of \citet{Wei91} are satisfied. Therefore, Lemma~A.1 of \citet{Wei91} implies 
	\begin{equation}\label{eq:1st impl Wei91}
		\sup_{\Vert\vtheta-\vtheta_{0,n}^c\Vert\leq d_0}\frac{\norm{\frac{1}{\sqrt{n}}\sum_{t=1}^{n}\Big\{\Cc_{t}(\vtheta, \widehat{\vtheta}_n^v) - \E\big[\Cc_{t}(\vtheta, \widehat{\vtheta}_n^v)\big]\Big\} - \Big\{\Cc_{t}(\vtheta_{0,n}^c, \widehat{\vtheta}_n^v) - \E\big[\Cc_{t}(\vtheta_{0,n}^c, \widehat{\vtheta}_n^v)\big]\Big\}}}{1+\sqrt{n}\norm{\frac{1}{n}\sum_{t=1}^{n}\E\big[\Cc_{t}(\vtheta, \widehat{\vtheta}_n^v)\big] }}=o_{\P}(1).
	\end{equation}
	We can bound the term in the denominator (using a mean value expansion around $\vtheta_{0,n}^{c}$, \eqref{eq:(p.15)} and the fact that, by construction of $\vtheta_{0,n}^c$, we have $\frac{1}{n}\sum_{t=1}^{n}\E\big[\Cc_{t}(\vtheta_{0,n}^c, \widehat{\vtheta}_n^v)\big]=\vzero$) as follows:
	\begin{align}
		\bigg\Vert\frac{1}{n}&\sum_{t=1}^{n}\E\big[\Cc_{t}(\vtheta, \widehat{\vtheta}_n^v)\big]\bigg\Vert = \bigg\Vert\frac{1}{n}\sum_{t=1}^{n}\E\big[\Cc_{t}(\vtheta_{0,n}^c, \widehat{\vtheta}_n^v)\big] + \frac{1}{n}\sum_{t=1}^{n}\frac{\partial}{\partial \vtheta^c}\E\big[ \Cc_{t}(\vtheta^c, \vtheta^v)\big]\bigg\vert_{\substack{\vtheta^c=\vtheta^\ast\\ \vtheta^v = \widehat{\vtheta}_n^v}} \big(\vtheta - \vtheta_{0,n}^c\big)\bigg\Vert\notag\\
		&= \bigg\Vert\frac{1}{n}\sum_{t=1}^{n}\E\bigg\{\nabla^2 c_t(\vtheta^\ast) \Big[ F_t^{Y}\big(c_t(\vtheta^\ast)\big) - F_{t}\big(v_t(\widehat{\vtheta}_n^v), c_t(\vtheta^\ast)\big)-\alpha \big\{ 1-F_{t}^{X}\big(v_t(\widehat{\vtheta}_n^v)\big)\big\}\Big]\bigg\} \big(\vtheta - \vtheta_{0,n}^c\big)\notag\\
		&\hspace{1cm} + \frac{1}{n}\sum_{t=1}^{n}\E\bigg\{\nabla c_t(\vtheta^\ast)\nabla^\prime c_t(\vtheta^\ast) \Big[  f_t^{Y}\big(c_t(\vtheta^\ast)\big) - \partial_2 F_{t}\big(v_t(\widehat{\vtheta}_n^v), c_t(\vtheta^\ast)\big)\Big]\bigg\} \big(\vtheta - \vtheta_{0,n}^c\big)\bigg\Vert\notag\\
		&\leq \frac{1}{n}\sum_{t=1}^{n}3\E\big[C_2(\mathcal{F}_{t-1})\big] \big\Vert\vtheta - \vtheta_{0,n}^c\big\Vert + \frac{1}{n}\sum_{t=1}^{n}2K \E\big[C_1^2(\mathcal{F}_{t-1})\big] \big\Vert\vtheta - \vtheta_{0,n}^c\big\Vert\notag\\
		&\leq C_{0} d_0,\label{eq:bound denom}
	\end{align}
	where we also used Assumption~\ref{ass:an} in the final two steps, and $C_0$ is a large positive constant.
	For any $\varepsilon>0$, write
	\begin{align*}
		\P\Big\{\norm{B_{5n}}>\varepsilon\Big\} &\leq \P\Big\{\norm{B_{5n}}>\varepsilon,\ \big\Vert\vtheta_0^c - \vtheta_{0,n}^c\big\Vert \leq C/\sqrt{n}\Big\} + \P\Big\{\big\Vert\vtheta_0^c - \vtheta_{0,n}^c\big\Vert > C/\sqrt{n}\Big\}
	\end{align*}
	By \eqref{eq:root n theta0n}, for any $\delta>0$ we may choose $C>0$ sufficiently large, such that $\P\Big\{\big\Vert\vtheta_0^c - \vtheta_{0,n}^c\big\Vert > C/\sqrt{n}\Big\}<\delta/2$. Therefore, $B_{5n}=o_{\P}(1)$ follows if we can show that the first term on the right-hand side of the above display is also bounded by $\delta/2$. To task this, write
	\begin{align*}
		&\P\Big\{\norm{B_{5n}}>\varepsilon,\ \big\Vert\vtheta_0^c - \vtheta_{0,n}^c\big\Vert \leq C/\sqrt{n}\Big\}\\
		&\leq \P\Bigg\{\sup_{\Vert\vtheta - \vtheta_{0,n}^c\Vert\leq C/\sqrt{n}}\bigg\Vert\frac{1}{\sqrt{n}}\sum_{t=1}^{n}\Big\{\Cc_{t}(\vtheta, \widehat{\vtheta}_n^v) - \E\big[\Cc_{t}(\vtheta, \widehat{\vtheta}_n^v)\big]\Big\} - \Big\{\Cc_{t}(\vtheta_{0,n}^c, \widehat{\vtheta}_n^v) - \E\big[\Cc_{t}(\vtheta_{0,n}^c, \widehat{\vtheta}_n^v)\big]\Big\}\bigg\Vert>\varepsilon\Bigg\}\\
		&\leq \P\Bigg\{\sup_{\Vert\vtheta - \vtheta_{0,n}^c\Vert\leq C/\sqrt{n}}\frac{\norm{\frac{1}{\sqrt{n}}\sum_{t=1}^{n}\Big\{\Cc_{t}(\vtheta, \widehat{\vtheta}_n^v) - \E\big[\Cc_{t}(\vtheta, \widehat{\vtheta}_n^v)\big]\Big\} - \Big\{\Cc_{t}(\vtheta_{0,n}^c, \widehat{\vtheta}_n^v) - \E\big[\Cc_{t}(\vtheta_{0,n}^c, \widehat{\vtheta}_n^v)\big]\Big\}}}{1+\sqrt{n}\norm{\frac{1}{n}\sum_{t=1}^{n}\E\big[\Cc_{t}(\vtheta, \widehat{\vtheta}_n^v)\big] }}\\
		&\hspace{7.3cm} \times \sup_{\Vert\vtheta - \vtheta_{0,n}^c\Vert\leq C/\sqrt{n}}\bigg[1+\sqrt{n}\Big\Vert\frac{1}{n}\sum_{t=1}^{n}\E\big[\Cc_{t}(\vtheta, \widehat{\vtheta}_n^v)\big] \Big\Vert\bigg] >\varepsilon\Bigg\}\\
		&\leq\P\Big\{o_{\P}(1)\times \big[1+\sqrt{n} C_{0}(C/\sqrt{n})\big]>\varepsilon\Big\}\\
		&<\delta/2
	\end{align*}
	for sufficiently large $n$, where the penultimate step follows from \eqref{eq:1st impl Wei91} and \eqref{eq:bound denom}. This proves $B_{5n}=o_{\P}(1)$.
	
	It remains to show that $A_{5n}=o_{\P}(1)$. To do so, we again verify conditions (N1)--(N5) of \citet{Wei91} for 
	\begin{align*}
		\vlambda_n(\vtheta^v) &:= \frac{1}{n}\sum_{t=1}^{n}\E\big[\Cc_{t}(\vtheta_0^c, \vtheta^v)\big],\\
		\mu_t(\vtheta,d) &:= \sup_{\norm{\vtau-\vtheta}\leq d}\norm{\Cc_{t}(\vtheta_0^c,\vtau) - \Cc_{t}(\vtheta_0^c,\vtheta)}.
	\end{align*}
	To promote flow, we do this in Section~\ref{Supplementary Proofs}. Lemma~A.1 of \citet{Wei91} then implies that
	\begin{equation*}
		\sup_{\Vert\vtheta-\vtheta_{0}^v\Vert\leq d_0}\frac{\norm{\frac{1}{\sqrt{n}}\sum_{t=1}^{n}\Big\{\Cc_{t}(\vtheta_0^c, \vtheta) - \E\big[\Cc_{t}(\vtheta_0^c,\vtheta)\big]\Big\} - \Big\{\Cc_{t}(\vtheta_0^c, \vtheta_0^v) - \E\big[\Cc_{t}(\vtheta_0^c, \vtheta_0^v)\big]\Big\}}}{1+\sqrt{n}\norm{\frac{1}{n}\sum_{t=1}^{n}\E\big[\Cc_{t}(\vtheta_0^c, \vtheta)\big] }}=o_{\P}(1).
	\end{equation*}
	We can bound the term in the denominator (using a mean value expansion around $\vtheta_{0}^{v}$, \eqref{eq:(p.16)} and $\E\big[\Cc_{t}(\vtheta_0^c, \vtheta_0^v)\big]=\vzero$) as follows:
	\begin{align}
		\bigg\Vert\frac{1}{n}&\sum_{t=1}^{n}\E\big[\Cc_{t}(\vtheta_0^c, \vtheta)\big]\bigg\Vert = \bigg\Vert\frac{1}{n}\sum_{t=1}^{n}\E\big[\Cc_{t}(\vtheta_0^c, \vtheta_0^v)\big] + \frac{1}{n}\sum_{t=1}^{n}\frac{\partial}{\partial \vtheta^v}\E\big[ \Cc_{t}(\vtheta^c, \vtheta^v)\big]\bigg\vert_{\substack{\vtheta^c=\vtheta_0^c\\ \vtheta^v = \vtheta^\ast}} \big(\vtheta - \vtheta_{0}^v\big)\bigg\Vert\notag\\
		&= \bigg\Vert\frac{1}{n}\sum_{t=1}^{n}\E\bigg\{\nabla c_t(\vtheta_0^c)\nabla^\prime v_t(\vtheta^\ast)\Big[\alpha f_t^X\big(v_t(\vtheta^\ast)\big) - \partial_1 F_t\big(v_t(\vtheta^\ast), c_t(\vtheta_0^c)\big)\Big]\bigg\}\big(\vtheta - \vtheta_0^v\big)\bigg\Vert\notag\\
		&\leq \frac{1}{n}\sum_{t=1}^{n}(\alpha K + K) \E\big[V_1(\mathcal{F}_{t-1}) C_1(\mathcal{F}_{t-1})\big] \big\Vert\vtheta - \vtheta_0^v\big\Vert\notag\\
		&\leq \frac{1}{n}\sum_{t=1}^{n}(\alpha K + K) \Big\{\E\big[V_1^2(\mathcal{F}_{t-1})\big]\Big\}^{1/2} \Big\{\E\big[C_1^2(\mathcal{F}_{t-1})\big]\Big\}^{1/2} \big\Vert\vtheta - \vtheta_0^v\big\Vert\notag\\
		&\leq C_{0} d_0,\notag
	\end{align}
	where we also used Assumption~\ref{ass:an} in the final two steps, and $C_0$ is a large positive constant. Now, $A_{5n}=o_{\P}(1)$ follows by similar steps used below \eqref{eq:bound denom} (combined with the fact that $\sqrt{n}\big\Vert\widehat{\vtheta}_n^v-\vtheta_0^v\big\Vert=O_{\P}(1)$ holds instead of \eqref{eq:root n theta0n}).
\end{proof}

\subsection{Supplementary Proofs}\label{Supplementary Proofs}

To complete the proof of Lemma~C.\ref{lem:7 tilde}, this section verifies conditions (N1)--(N5) of \citet{Wei91} for 
\begin{align*}
	\vlambda_n(\vtheta^v) &:= \frac{1}{n}\sum_{t=1}^{n}\E\big[\Cc_{t}(\vtheta_0^c, \vtheta^v)\big],\\
	\mu_t(\vtheta,d) &:= \sup_{\norm{\vtau-\vtheta}\leq d}\norm{\Cc_{t}(\vtheta_0^c,\vtau) - \Cc_{t}(\vtheta_0^c,\vtheta)}.
\end{align*}

The mixing condition of (N5) follows from Assumption~\ref{ass:an} \ref{it:mixing}.
Condition (N2) follows from correct specification of the VaR and CoVaR model such that
\begin{equation}\label{eq:0 cond}
	\vlambda_n(\vtheta_0^v)=\frac{1}{n}\sum_{t=1}^{n}\E\big[\Cc_{t}(\vtheta_0^c, \vtheta_0^v)\big]=\vzero.
\end{equation}
The separability condition (N1) follows from the next lemma:

\begin{lem}\label{lem:0 tilde noch}
	Suppose Assumptions~\ref{ass:cons} and \ref{ass:an} hold. Then, condition (N1) of \citet{Wei91} holds, i.e., for all $t\in\mathbb{N}$, the stochastic process $\Omega\times\mTheta^v\ni(\omega,\vtheta^v)\mapsto \Cc_t(\vtheta_0^c,\vtheta^v)$ is separable in the sense of \citet[pp.~51--52]{Doo53}, and $\Cc_t(\vtheta_0^c,\vtheta^v)$ is measurable for all $\vtheta^v\in\mTheta^v$.
\end{lem}

\begin{proof}
	The proof is similar to that of Lemma~C.\ref{lem:0 tilde} but simpler and, hence, is omitted.
\end{proof}

Next, condition (N3) is verified in Lemmas~\ref{lem:1 supp}--\ref{lem:3 supp}, and condition (N4) in Lemma~\ref{lem:4 supp}.

\begin{lem}\label{lem:1 supp}
	Suppose Assumptions~\ref{ass:cons} and \ref{ass:an} hold. Then, condition (N3) (i) of \citet{Wei91} holds, i.e.,
	\[
	\big\Vert\vlambda_n(\vtheta)\big\Vert \geq a\big\Vert\vtheta-\vtheta_{0}^v\big\Vert\qquad\text{for}\ \big\Vert\vtheta-\vtheta_{0}^v\big\Vert\leq d_0
	\]
	for sufficiently large $n$ and some $a>0$ and $d_0>0$.
\end{lem}

\begin{proof}
	The MVT and \eqref{eq:0 cond} imply
	\begin{align}
		\vlambda_n(\vtheta)&= \vlambda_n(\vtheta_0^v) + \mLambda_{n,(2)}(\vtheta_0^c,\vtheta^\ast)(\vtheta- \vtheta_0^v)\notag\\
		& = \mLambda_{n,(2)}(\vtheta_0^c,\vtheta^\ast)(\vtheta- \vtheta_0^v)\label{eq:lambda c decomp}
	\end{align}
	for some $\vtheta^\ast$ on the line connecting $\vtheta_0^v$ and $\vtheta$. (Again, this is an instance of the non-existent mean value theorem.) 
	
	Assumption~\ref{ass:an}~\ref{it:pd} and \citet[Theorem~4.2.1]{GV13} imply that $\mLambda_{n,(2)}^\prime\mLambda_{n,(2)}$ is uniformly positive definite.
	Therefore, the sequence $\mLambda_{n,(2)}^\prime\mLambda_{n,(2)}$ has the smallest eigenvalue bounded uniformly below by some $a>0$. 
	By \citet[Theorem~1.13 and in particular Exercise 1.14.1]{MN19} this yields that
	\begin{align}
		\vx^\prime\mLambda_{n,(2)}^\prime\mLambda_{n,(2)}\vx & \geq a\vx^\prime\vx\qquad\text{for all }\vx\in\mathbb{R}^p\notag
		\intertext{or, equivalently,}
		\big\Vert\mLambda_{n,(2)}\vx\big\Vert^2 & \geq a \big\Vert\vx\big\Vert^2\qquad\text{for all }\vx\in\mathbb{R}^p.\label{eq:eucl bound}
	\end{align}
	
	Exploiting \eqref{eq:lambda c decomp} and \eqref{eq:eucl bound} gives
	\begin{align}
		\big\Vert\vlambda_n(\vtheta)\big\Vert &= \big\Vert\mLambda_{n,(2)}(\vtheta_0^c,\vtheta^\ast)(\vtheta- \vtheta_0^v)\big\Vert\notag\\
		&= \Big\Vert\mLambda_{n,(2)}(\vtheta- \vtheta_0^v) - \big[\mLambda_{n,(2)} - \mLambda_{n,(2)}(\vtheta_0^c,\vtheta^\ast)\big](\vtheta- \vtheta_0^v)\Big\Vert\notag\\
		&\geq \Big\Vert\mLambda_{n,(2)}(\vtheta- \vtheta_0^v)\Big\Vert - \Big\Vert\big[\mLambda_{n,(2)} - \mLambda_{n,(2)}(\vtheta_0^c,\vtheta^\ast)\big](\vtheta- \vtheta_0^v)\Big\Vert\notag\\
		&\geq \sqrt{a}\big\Vert\vtheta- \vtheta_0^v\big\Vert - \big\Vert\mLambda_{n,(2)} - \mLambda_{n,(2)}(\vtheta_0^c,\vtheta^\ast)\big\Vert\cdot \big\Vert\vtheta- \vtheta_0^v\big\Vert\notag\\
		&=\Big(\sqrt{a} - \big\Vert\mLambda_{n,(2)} - \mLambda_{n,(2)}(\vtheta_0^c,\vtheta^\ast)\big\Vert\Big)\big\Vert\vtheta- \vtheta_0^v\big\Vert.\label{eq:(p.32)}
	\end{align}
	By \eqref{eq:(E.10)}, we have that
	\[
	\big\Vert\mLambda_{(2)} - \mLambda_{n,(2)}(\vtheta_0^c,\vtheta^\ast)\big\Vert\leq C\norm{\vtheta^\ast-\vtheta_0^v}\leq C\norm{\vtheta-\vtheta_0^v}\leq C d_0,
	\]
	where the second inequality follows because $\vtheta^\ast$ lies on the line connecting $\vtheta_0^v$ and $\vtheta$.
	In particular, choosing $d_0>0$ such that $C d_0<\sqrt{a}$, the desired conclusion follows from \eqref{eq:(p.32)}.
\end{proof}

\begin{lem}\label{lem:2 supp}
	Suppose Assumptions~\ref{ass:cons} and \ref{ass:an} hold, and define
	\[
	\mu_t(\vtheta, d) =\sup_{\norm{\vtau-\vtheta}\leq d}\norm{\Cc_{t}(\vtheta_0^c,\vtau) - \Cc_{t}(\vtheta_0^c,\vtheta)}.
	\]
	Then, condition (N3) (ii) of \citet{Wei91} holds, i.e.,
	\[
	\E\big[\mu_t(\vtheta, d)\big] \leq bd\qquad\text{for}\ \big\Vert\vtheta-\vtheta_{0}^v\big\Vert + d\leq d_0
	\]
	for sufficiently large $n$ and some strictly positive $b$, $d$, $d_0$.
\end{lem}

\begin{proof}
	Choose $d_0>0$ sufficiently small, such that $\big\{\vtheta\in\mTheta^v:\ \norm{\vtheta-\vtheta_0^v}<d_0\big\}$ is a subset of the neighborhoods of Assumptions~\ref{ass:cons} \ref{it:bound} and \ref{ass:an} \ref{it:bound2.1}.
	Recalling the definition of $\Cc_{t}(\vtheta^c, \vtheta^v)$ from \eqref{eq:(12+)}, we decompose
	\begin{equation*}
		\mu_t(\vtheta, d) = \sup_{\norm{\vtau-\vtheta}\leq d}\norm{ \big[\1_{\{X_t> v_t(\vtau)\}} - \1_{\{X_t> v_t(\vtheta)\}}\big]\nabla c_t(\vtheta_0^c) \big[\1_{\{Y_t\leq c_t(\vtheta_0^c)\}} - \alpha\big]}.
	\end{equation*}
	The remainder of the proof is similar to that of Lemma~V.\ref{lem:4} and Lemma~C.\ref{lem:4 tilde}.
	
	Define the $\mathcal{F}_{t-1}$-measurable quantities
	\begin{align*}
		\underline{\vtau} &:= \argmin_{\norm{\vtau-\vtheta}\leq d}v_t(\vtau),\\
		\overline{\vtau}  &:= \argmax_{\norm{\vtau-\vtheta}\leq d}v_t(\vtau),
	\end{align*}
	which exist by continuity of $v_t(\cdot)$.
	
	To take the indicators out of the supremum in $\mu_t(\vtheta, d)$, we again consider two cases:
	
	\textbf{Case 1: $X_t\leq v_t(\vtheta)$}
	
	In this case,
	\begin{align*}
		\mu_t(\vtheta, d) &= \1_{\{X_t> v_t(\underline{\vtau})\}} \sup_{\norm{\vtau-\vtheta}\leq d}\norm{\nabla c_t(\vtheta_0^c) \big[\1_{\{Y_t\leq c_t(\vtheta_0^c)\}} - \alpha\big]}\\
		&=  \1_{\{X_t> v_t(\underline{\vtau})\}} \norm{\nabla c_t(\vtheta_0^c) \big[\1_{\{Y_t\leq c_t(\vtheta_0^c)\}} - \alpha\big]}.
	\end{align*}

	\textbf{Case 2: $X_t> v_t(\vtheta)$}
	
	In this case,
	\begin{align*}
		\mu_t(\vtheta, d) &= \1_{\{X_t\leq v_t(\overline{\vtau})\}}\sup_{\norm{\vtau-\vtheta}\leq d}\Big\Vert\nabla c_t(\vtheta_0^c) \big[\1_{\{Y_t\leq c_t(\vtheta_0^c)\}} - \alpha\big]\Big\Vert\\
		&= \1_{\{X_t\leq v_t(\overline{\vtau})\}}\Big\Vert\nabla c_t(\vtheta_0^c) \big[\1_{\{Y_t\leq c_t(\vtheta_0^c)\}} - \alpha\big]\Big\Vert.
	\end{align*}
	(Notice that the third case that $X_t<v_t(\underline{\vtau})$ cannot occur, because already $X_t> v_t(\vtheta)$.)

	Combining the results from Cases 1 and 2 gives
	\begin{equation}\label{eq:(5.1)}
		\mu_t(\vtheta, d) \leq \Big[\1_{\{v_t(\underline{\vtau})< X_t\leq v_t(\vtheta)\}} + \1_{\{v_t(\vtheta)<X_t\leq v_t(\overline{\vtau})\}}\Big] \norm{\nabla c_t(\vtheta_0^c) \big[\1_{\{Y_t\leq c_t(\vtheta_0^c)\}} - \alpha\big]}.
	\end{equation}
	Note that---similarly as below \eqref{eq:(N.12.2)}---$\vtheta$ and $\vtau$ are in a $d_0$-neighborhood of $\vtheta_0^v$, such that the bound for $\nabla v_t(\cdot)$ from Assumption~\ref{ass:cons}~\ref{it:bound} applies.
	Then, using the LIE and \eqref{eq:bound 1}--\eqref{eq:bound 3},
	\begin{align*}
		\E\big[\mu_t(\vtheta, d)\big] &\leq \E\bigg\{\E_{t-1}\Big[\1_{\{v_t(\underline{\vtau})< X_t\leq v_t(\vtheta)\}} + \1_{\{v_t(\vtheta)<X_t\leq v_t(\overline{\vtau})\}}\Big]\norm{\nabla c_t(\vtheta_0^c)}\bigg\}\\
		&\leq \E\big[2K V_1(\mathcal{F}_{t-1})C_1(\mathcal{F}_{t-1})d\big]\\
		&\leq 2K\Big\{\E\big[ V_1^2(\mathcal{F}_{t-1})\big]\Big\}^{1/2}\Big\{\E\big[ C_1^2(\mathcal{F}_{t-1})\big]\Big\}^{1/2}d\\
		&\leq bd
	\end{align*}
	for some sufficiently large $b>0$.
\end{proof}

\begin{lem}\label{lem:3 supp}
	Suppose Assumptions~\ref{ass:cons} and \ref{ass:an} hold. Then, condition (N3) (iii) of \citet{Wei91} holds, i.e.,
	\[
	\E\big[\mu_t^q(\vtheta, d)\big] \leq cd\qquad\text{for}\ \big\Vert\vtheta-\vtheta_0^v\big\Vert + d\leq d_0
	\]
	for sufficiently large $n$ and some $c>0$, $d\geq0$, $d_0>0$ and $q>2$.
\end{lem}

\begin{proof}
	As in the proof of Lemma~\ref{lem:2 supp}, we again pick $d_0>0$ sufficiently small, such that $\big\{\vtheta\in\mTheta^v:\ \norm{\vtheta-\vtheta_0^v}<d_0\big\}$ is a subset of the neighborhoods of Assumptions~\ref{ass:cons} \ref{it:bound} and \ref{ass:an} \ref{it:bound2.1}. 
	We also work with $\underline{\vtau}$ and $\overline{\vtau}$ as defined in the proof of Lemma~\ref{lem:2 supp}. For $\iota>0$ from Assumption~\ref{ass:an} \ref{it:mom bounds cons2}, we get from \eqref{eq:(5.1)}, \eqref{eq:bound 1}--\eqref{eq:bound 3} and the LIE that
	\begin{align*}
		\E\big[\mu_t^{2+\iota}(\vtheta, d)\big] &\leq \E\bigg\{\E_{t-1}\Big[\1_{\{v_t(\underline{\vtau})< X_t\leq v_t(\vtheta)\}} + \1_{\{v_t(\vtheta)<X_t\leq v_t(\overline{\vtau})\}}\Big]\norm{\nabla c_t(\vtheta_0^c)}^{2+\iota}\bigg\}\\
		&\leq \E\big[2K V_1(\mathcal{F}_{t-1})C_1^{2+\iota}(\mathcal{F}_{t-1})d\big]\\
		&\leq 2K\Big\{\E\big[V_1^{3+\iota}(\mathcal{F}_{t-1})\big]\Big\}^{1/(3+\iota)} \Big\{\E\big[C_1^{3+\iota}(\mathcal{F}_{t-1})\big]\Big\}^{(2+\iota)/(3+\iota)}d\\
		&\leq cd
	\end{align*}
	for some sufficiently large $c>0$.
\end{proof}

\begin{lem}\label{lem:4 supp}
	Suppose Assumptions~\ref{ass:cons} and \ref{ass:an} hold. Then, condition (N4) of \citet{Wei91} holds, i.e.,
	\[
	\E\norm{\Cc_{t}(\vtheta_{0}^c, \vtheta_0^v)}^{2}\leq C\qquad\text{for all }t\in\mathbb{N}.
	\]
\end{lem}

\begin{proof}
	The bound follows immediately from \eqref{eq:(p.36)}.
\end{proof}

\section{Proof of Theorem~\ref{thm:avar}}
\label{sec:thm3}

\begin{proof}[{\textbf{Proof of Theorem~\ref{thm:avar}:}}]
	The two convergences $\widehat{\mV}_n - \mV_n\overset{\P}{\longrightarrow}\vzero$ and $\widehat{\mLambda}_n - \mLambda_n\overset{\P}{\longrightarrow}\vzero$ are shown similarly as in \citet[Theorem~3]{EM04}. Similarly as the proof of $\widehat{\mV}_n - \mV_n\overset{\P}{\longrightarrow}\vzero$, the proof of $\widehat{\mC}_n - \mC_{n}\overset{\P}{\longrightarrow}\vzero$ is standard and, hence, is also omitted. It remains to show that $\widehat{\mLambda}_{n,(1)} - \mLambda_{n,(1)}\overset{\P}{\longrightarrow}\vzero$ and $\widehat{\mLambda}_{n,(2)} - \mLambda_{n,(2)}\overset{\P}{\longrightarrow}\vzero$. Since the proofs are very similar, we only show the latter convergence, that is,
	\begin{multline*}
		\frac{1}{n}\sum_{t=1}^{n} \nabla c_t(\widehat{\vtheta}_n^c)\nabla^\prime v_t(\widehat{\vtheta}_n^v) (2\widehat{b}_{n,x})^{-1}\Big[ \alpha \1_{\big\{|X_t-v_t(\widehat{\vtheta}_n^v)|<\widehat{b}_{n,x}\big\}} - \1_{\big\{|X_t-v_t(\widehat{\vtheta}_n^v)|<\widehat{b}_{n,x},\ Y_t\leq c_t(\widehat{\vtheta}_n^c)\big\}} \Big]\\
		- \frac{1}{n}\sum_{t=1}^{n}\E\bigg\{\nabla c_t(\vtheta_0^c)\nabla^\prime v_t(\vtheta_0^v)\Big[\alpha f_t^X\big(v_t(\vtheta_0^v)\big) - \partial_1 F_t\big(v_t(\vtheta_0^v), c_t(\vtheta_0^c)\big)\Big]\bigg\} =o_{\P}(1).
	\end{multline*}
	Here, we only prove that
	\begin{multline}\label{eq:(C.1m)}
		\widehat{\mD}_n - \mD_{n,0}:=\frac{1}{n}\sum_{t=1}^{n} \nabla c_t(\widehat{\vtheta}_n^c)\nabla^\prime v_t(\widehat{\vtheta}_n^v) (2\widehat{b}_{n,x})^{-1}\1_{\big\{|X_t-v_t(\widehat{\vtheta}_n^v)|<\widehat{b}_{n,x},\ Y_t\leq c_t(\widehat{\vtheta}_n^c)\big\}} \\
		- \frac{1}{n}\sum_{t=1}^{n}\E\Big[\nabla c_t(\vtheta_0^c)\nabla^\prime v_t(\vtheta_0^v)\partial_1 F_t\big(v_t(\vtheta_0^v), c_t(\vtheta_0^c)\big)\Big] =o_{\P}(1),
	\end{multline}
	because 
	\begin{multline*}
		\frac{1}{n}\sum_{t=1}^{n} \nabla c_t(\widehat{\vtheta}_n^c)\nabla^\prime v_t(\widehat{\vtheta}_n^v) (2\widehat{b}_{n,x})^{-1} \alpha \1_{\big\{|X_t-v_t(\widehat{\vtheta}_n^v)|<\widehat{b}_{n,x}\big\}} \\
		- \frac{1}{n}\sum_{t=1}^{n}\E\Big[\nabla c_t(\vtheta_0^c)\nabla^\prime v_t(\vtheta_0^v)\alpha f_t^X\big(v_t(\vtheta_0^v)\big)\Big] =o_{\P}(1)
	\end{multline*}
	can be shown similarly. Define
	\[
	\widetilde{\mD}_{n} := \frac{1}{n}\sum_{t=1}^{n} \nabla c_t(\vtheta_0^c)\nabla^\prime v_t(\vtheta_0^v) (2b_{n,x})^{-1}\1_{\big\{|X_t-v_t(\vtheta_0^v)|<b_{n,x},\ Y_t\leq c_t(\vtheta_0^c)\big\}}. 
	\]
	Then, to establish \eqref{eq:(C.1m)} it suffices to show that
	\begin{align}
		\widehat{\mD}_{n} - \widetilde{\mD}_{n} & =o_{\P}(1),\label{eq:(1)}\\
		\widetilde{\mD}_{n} - \mD_{n,0} & =o_{\P}(1).\label{eq:(2)}
	\end{align}
	
	We first prove \eqref{eq:(1)}. Observe that
	\begin{align}
		&\big\Vert\widehat{\mD}_{n} - \widetilde{\mD}_{n}\big\Vert \notag\\
		&=\frac{b_{n,x}}{\widehat{b}_{n,x}} \Big\Vert (2b_{n,x} n)^{-1}\sum_{t=1}^{n}\big[\1_{\big\{|X_t-v_t(\widehat{\vtheta}_n^v)|<\widehat{b}_{n,x},\ Y_t\leq c_t(\widehat{\vtheta}_n^c)\big\}} - \1_{\big\{|X_t-v_t(\vtheta_0^v)|<b_{n,x},\ Y_t\leq c_t(\vtheta_0^c)\big\}}\big]\nabla c_t(\widehat{\vtheta}_n^c)\nabla^\prime v_t(\widehat{\vtheta}_n^v) \notag\\
		&\hspace{5cm} + \1_{\big\{|X_t-v_t(\vtheta_0^v)|<b_{n,x},\ Y_t\leq c_t(\vtheta_0^c)\big\}}\big[\nabla c_t(\widehat{\vtheta}_n^c) - \nabla c_t(\vtheta_0^c)\big]\nabla^\prime v_t(\widehat{\vtheta}_n^v)\notag\\
		&\hspace{5cm} + \1_{\big\{|X_t-v_t(\vtheta_0^v)|<b_{n,x},\ Y_t\leq c_t(\vtheta_0^c)\big\}}\nabla c_t(\vtheta_0^c)\big[\nabla^\prime v_t(\widehat{\vtheta}_n^v) - \nabla^\prime v_t(\vtheta_0^v)\big] \notag\\
		&\hspace{5cm} + \frac{b_{n,x}-\widehat{b}_{n,x}}{b_{n,x}}\1_{\big\{|X_t-v_t(\vtheta_0^v)|<b_{n,x},\ Y_t\leq c_t(\vtheta_0^c)\big\}}\nabla c_t(\vtheta_0^c) \nabla^\prime v_t(\vtheta_0^v)\Big\Vert.\label{eq:Dn hat}
	\end{align}
	Write the indicators in the second line of the above display as
	\begin{align}
		&\1_{\big\{|X_t-v_t(\widehat{\vtheta}_n^v)|<\widehat{b}_{n,x},\ Y_t\leq c_t(\widehat{\vtheta}_n^c)\big\}} - \1_{\big\{|X_t-v_t(\vtheta_0^v)|<b_{n,x},\ Y_t\leq c_t(\vtheta_0^c)\big\}}\notag\\
		&\hspace{1cm}= \Big[\1_{\big\{|X_t-v_t(\widehat{\vtheta}_n^v)|<\widehat{b}_{n,x},\ Y_t\leq c_t(\widehat{\vtheta}_n^c)\big\}} - \1_{\big\{|X_t-v_t(\vtheta_0^v)|<b_{n,x},\ Y_t\leq c_t(\widehat{\vtheta}_n^c)\big\}}\Big]\notag\\
		&\hspace{2cm} +\Big[ \1_{\big\{|X_t-v_t(\vtheta_0^v)|<b_{n,x},\ Y_t\leq c_t(\widehat{\vtheta}_n^c)\big\}} - \1_{\big\{|X_t-v_t(\vtheta_0^v)|<b_{n,x},\ Y_t\leq c_t(\vtheta_0^c)\big\}}\Big]\notag\\
		&\hspace{1cm}=:A_{1t} + B_{1t}.\label{eq:A2 B2}
	\end{align}
	For $A_{1t}$, note that for any set $\mathsf{A}\in\mathcal{F}$,
	\begin{multline}
		\big|\1_{\big\{|X_t-v_t(\widehat{\vtheta}_n^v)|<\widehat{b}_{n,x},\ \mathsf{A}\big\}} - \1_{\big\{|X_t-v_t(\vtheta_0^v)|<b_{n,x},\ \mathsf{A}\big\}}\big| \\
		\leq \1_{\big\{|X_t-v_t(\vtheta_0^v) - b_{n,x}|<|v_t(\widehat{\vtheta}_n^v)-v_t(\vtheta_0^v)| + |\widehat{b}_{n,x}-b_{n,x}|,\ \mathsf{A}\big\}} + \1_{\big\{|X_t-v_t(\vtheta_0^v) + b_{n,x}|<|v_t(\widehat{\vtheta}_n^v)-v_t(\vtheta_0^v)| + |\widehat{b}_{n,x}-b_{n,x}|,\ \mathsf{A}\big\}}.\label{eq:(C.2.1)}
	\end{multline}
	To see this, note that the difference of the indicators in the first line of \eqref{eq:(C.2.1)} can only be non-zero if \textit{either} $|X_t-v_t(\widehat{\vtheta}_n^v)|<\widehat{b}_{n,x}$ and $|X_t-v_t(\vtheta_0^v)|\geq b_{n,x}$ \textit{or} $|X_t-v_t(\widehat{\vtheta}_n^v)|\geq \widehat{b}_{n,x}$ and $|X_t-v_t(\vtheta_0^v)|<b_{n,x}$. Consider the former case; the latter case can be dealt with similarly. It holds that
	\begin{align}
		&\big|X_t-v_t(\widehat{\vtheta}_n^v)\big|<\widehat{b}_{n,x}\notag\\
		& \Longleftrightarrow\ -\widehat{b}_{n,x}<X_t-v_t(\widehat{\vtheta}_n^v)<\widehat{b}_{n,x}\notag\\
		& \Longleftrightarrow\ v_t(\widehat{\vtheta}_n^v) - v_t(\vtheta_0^v)-\widehat{b}_{n,x}-b_{n,x}<X_t-v_t(\vtheta_0^v) - b_{n,x}<v_t(\widehat{\vtheta}_n^v) - v_t(\vtheta_0^v) + \widehat{b}_{n,x} - b_{n,x}\label{eq:(1)impl}\\
		& \Longleftrightarrow\ v_t(\widehat{\vtheta}_n^v) - v_t(\vtheta_0^v) - \widehat{b}_{n,x} + b_{n,x} <X_t-v_t(\vtheta_0^v) + b_{n,x}<v_t(\widehat{\vtheta}_n^v) - v_t(\vtheta_0^v) + \widehat{b}_{n,x} + b_{n,x}\label{eq:(2)impl}
	\end{align}
	and
	\begin{equation}\label{eq:(C.8m)}
		\big|X_t-v_t(\vtheta_0^v)\big|\geq b_{n,x}\qquad \Longleftrightarrow\qquad X_t-v_t(\vtheta_0^v)\geq b_{n,x}\quad \text{or}\quad X_t-v_t(\vtheta_0^v)\leq-b_{n,x}.
	\end{equation}
	The inequality $X_t-v_t(\vtheta_0^v)\geq b_{n,x}$ from \eqref{eq:(C.8m)} together with the second inequality of \eqref{eq:(1)impl} yields
	\[
	0\leq X_t-v_t(\vtheta_0^v) - b_{n,x}<v_t(\widehat{\vtheta}_n^v) - v_t(\vtheta_0^v) + \widehat{b}_{n,x} - b_{n,x},
	\]
	which in turn implies that
	\[
	\big|X_t-v_t(\vtheta_0^v) - b_{n,x}\big| < \big|v_t(\widehat{\vtheta}_n^v) - v_t(\vtheta_0^v)\big| + \big| \widehat{b}_{n,x} - b_{n,x}\big|.
	\]
	Therefore, when $X_t-v_t(\vtheta_0^v)\geq b_{n,x}$,
	\begin{equation}\label{eq:C51}
		\big|\1_{\big\{|X_t-v_t(\widehat{\vtheta}_n^v)|<\widehat{b}_{n,x},\ \mathsf{A}\big\}} - \1_{\big\{|X_t-v_t(\vtheta_0^v)|<b_{n,x},\ \mathsf{A}\big\}}\big| \\
		\leq \1_{\big\{|X_t-v_t(\vtheta_0^v) - b_{n,x}|<|v_t(\widehat{\vtheta}_n^v)-v_t(\vtheta_0^v)| + | \widehat{b}_{n,x} - b_{n,x}|,\ \mathsf{A}\big\}}.
	\end{equation}
	But also the inequality $X_t-v_t(\vtheta_0^v)\leq-b_{n,x}$ from \eqref{eq:(C.8m)} together with the first inequality of \eqref{eq:(2)impl} yields
	\[
	v_t(\widehat{\vtheta}_n^v) - v_t(\vtheta_0^v) + b_{n,x} - \widehat{b}_{n,x}<X_t-v_t(\vtheta_0^v) + b_{n,x}\leq 0,
	\]
	which in turn implies that
	\[
	\big|X_t-v_t(\vtheta_0^v) + b_{n,x}\big| < \big|v_t(\vtheta_n^v) - v_t(\vtheta_0^v)\big| + \big| \widehat{b}_{n,x} - b_{n,x}\big|.
	\]
	Therefore, when $X_t-v_t(\vtheta_0^v)\leq-b_{n,x}$,
	\begin{equation}\label{eq:C52}
		\big|\1_{\big\{|X_t-v_t(\widehat{\vtheta}_n^v)|<\widehat{b}_{n,x},\ \mathsf{A}\big\}} - \1_{\big\{|X_t-v_t(\vtheta_0^v)|<b_{n,x},\ \mathsf{A}\big\}}\big| \\
		\leq \1_{\big\{|X_t-v_t(\vtheta_0^v) + b_{n,x}|<|v_t(\widehat{\vtheta}_n^v)-v_t(\vtheta_0^v)| + |\widehat{b}_{n,x} - b_{n,x}|,\ \mathsf{A}\big\}}.
	\end{equation}
	Combining \eqref{eq:C51} and \eqref{eq:C52} gives \eqref{eq:(C.2.1)}. 
	
	For $B_{1t}$, note that for any set $\mathsf{A}\in\mathcal{F}$,
	\begin{equation}
		\big|\1_{\big\{\mathsf{A},\ Y_t\leq c_t(\widehat{\vtheta}_n^c)\big\}} - \1_{\big\{\mathsf{A},\ Y_t\leq c_t(\vtheta_0^c)\big\}}\big|\leq \1_{\big\{\mathsf{A},\ |Y_t-c_t(\vtheta_0^c)|\leq |c_t(\widehat{\vtheta}_n^c) - c_t(\vtheta_0^c)|\big\}}.\label{eq:(C.2.2)}
	\end{equation}
	To see this, note that the difference of the indicators on the left-hand side can only be non-zero when \textit{either} $Y_t\leq c_t(\widehat{\vtheta}_n^c)$ and $Y_t> c_t(\vtheta_0^c)$ \textit{or} $Y_t> c_t(\widehat{\vtheta}_n^c)$ and $Y_t\leq c_t(\vtheta_0^c)$. Again, we only consider the former case, where
	\begin{align*}
		Y_t\leq c_t(\widehat{\vtheta}_n^c)\quad \text{and}\quad Y_t> c_t(\vtheta_0^c)	& \hspace{0.5cm}\Longleftrightarrow\quad c_t(\vtheta_0^c) < Y_t \leq c_t(\widehat{\vtheta}_n^c)\\
		& \hspace{0.5cm}\Longleftrightarrow\quad 0 < Y_t - c_t(\vtheta_0^c) \leq c_t(\widehat{\vtheta}_n^c) - c_t(\vtheta_0^c)\\
		& \hspace{0.5cm}\Longrightarrow\quad |Y_t - c_t(\vtheta_0^c)| \leq |c_t(\widehat{\vtheta}_n^c) - c_t(\vtheta_0^c)|.
	\end{align*}
	This implies \eqref{eq:(C.2.2)}.
	
	Exploiting \eqref{eq:(C.2.1)} and \eqref{eq:(C.2.2)} for $A_{1t}$ and $B_{1t}$, we get from \eqref{eq:A2 B2} that
	\begin{align*}
		&\big|\1_{\big\{|X_t-v_t(\widehat{\vtheta}_n^v)|<\widehat{b}_{n,x},\ Y_t\leq c_t(\widehat{\vtheta}_n^c)\big\}} - \1_{\big\{|X_t-v_t(\vtheta_0^v)|<b_{n,x},\ Y_t\leq c_t(\vtheta_0^c)\big\}}\big|\\
		& \hspace{1cm}\leq \1_{\big\{|X_t-v_t(\vtheta_0^v) - b_{n,x}|<|v_t(\widehat{\vtheta}_n^v)-v_t(\vtheta_0^v)| +| \widehat{b}_{n,x} - b_{n,x}|,\ Y_t\leq c_t(\widehat{\vtheta}_n^c)\big\}} \\
		& \hspace{2cm}+ \1_{\big\{|X_t-v_t(\vtheta_0^v) + b_{n,x}|<|v_t(\widehat{\vtheta}_n^v)-v_t(\vtheta_0^v)| +| \widehat{b}_{n,x} - b_{n,x}|,\ Y_t\leq c_t(\widehat{\vtheta}_n^c)\big\}}\\
		&\hspace{2cm} + \1_{\big\{|X_t-v_t(\vtheta_0^v)|<b_{n,x},\ |Y_t-c_t(\vtheta_0^c)|\leq |c_t(\widehat{\vtheta}_n^c) - c_t(\vtheta_0^c)|\big\}}.
	\end{align*}
	Using this, the MVT and Assumptions~\ref{ass:cons}--\ref{ass:an}, we obtain from \eqref{eq:Dn hat} that
	\begin{align}
		&\big\Vert\widehat{\mD}_{n} - \widetilde{\mD}_{n}\big\Vert \notag\\
		& \hspace{0.5cm}\leq\frac{b_{n,x}}{\widehat{b}_{n,x}}\Big\Vert (2b_{n,x} n)^{-1}\sum_{t=1}^{n}\Big[\1_{\big\{|X_t-v_t(\vtheta_0^v) - b_{n,x}|<|v_t(\widehat{\vtheta}_n^v)-v_t(\vtheta_0^v)|+|\widehat{b}_{n,x} - b_{n,x}|,\ Y_t\leq c_t(\widehat{\vtheta}_n^c)\big\}}\notag\\
		&\hspace{4cm} + \1_{\big\{|X_t-v_t(\vtheta_0^v) + b_{n,x}|<|v_t(\widehat{\vtheta}_n^v)-v_t(\vtheta_0^v)|+|\widehat{b}_{n,x} - b_{n,x}|,\ Y_t\leq c_t(\widehat{\vtheta}_n^c)\big\}} \notag\\
		&\hspace{4cm} + \1_{\big\{|X_t-v_t(\vtheta_0^v)|<b_{n,x},\ |Y_t-c_t(\vtheta_0^c)|\leq |c_t(\widehat{\vtheta}_n^c) - c_t(\vtheta_0^c)|\big\}}\Big]\nabla c_t(\widehat{\vtheta}_n^c)\nabla^\prime v_t(\widehat{\vtheta}_n^v) \notag\\
		&\hspace{4cm} + \1_{\big\{|X_t-v_t(\vtheta_0^v)|<b_{n,x},\ Y_t\leq c_t(\vtheta_0^c)\big\}}\nabla^2 c_t(\vtheta^\ast)\big(\widehat{\vtheta}_n^c - \vtheta_0^c\big)\nabla^\prime v_t(\widehat{\vtheta}_n^v)\notag\\
		&\hspace{4cm} + \1_{\big\{|X_t-v_t(\vtheta_0^v)|<b_{n,x},\ Y_t\leq c_t(\vtheta_0^c)\big\}}\nabla c_t(\vtheta_0^c)\big(\widehat{\vtheta}_n^v - \vtheta_0^v\big)^\prime\big[\nabla^2 v_t(\vtheta^\ast)\big]^\prime \notag\\
		&\hspace{4cm} + \frac{b_{n,x}-\widehat{b}_{n,x}}{b_{n,x}}\1_{\big\{|X_t-v_t(\vtheta_0^v)|<b_{n,x},\ Y_t\leq c_t(\vtheta_0^c)\big\}}\nabla c_t(\vtheta_0^c)\nabla^\prime v_t(\vtheta_0^v)\Big\Vert\notag\\
		& \hspace{0.5cm}\leq \frac{b_{n,x}}{\widehat{b}_{n,x}}(2b_{n,x} n)^{-1}\sum_{t=1}^{n}\Big[\1_{\big\{|X_t-v_t(\vtheta_0^v) - b_{n,x}|<|v_t(\widehat{\vtheta}_n^v)-v_t(\vtheta_0^v)|+|\widehat{b}_{n,x} - b_{n,x}|,\ Y_t\leq c_t(\widehat{\vtheta}_n^c)\big\}}\notag\\
		&\hspace{4cm} + \1_{\big\{|X_t-v_t(\vtheta_0^v) + b_{n,x}|<|v_t(\widehat{\vtheta}_n^v)-v_t(\vtheta_0^v)|+|\widehat{b}_{n,x} - b_{n,x}|,\ Y_t\leq c_t(\widehat{\vtheta}_n^c)\big\}} \notag\\
		&\hspace{4cm} + \1_{\big\{|X_t-v_t(\vtheta_0^v)|<b_{n,x},\ |Y_t-c_t(\vtheta_0^c)|\leq |c_t(\widehat{\vtheta}_n^c) - c_t(\vtheta_0^c)|\big\}}\Big]C_1(\mathcal{F}_{t-1})V_1(\mathcal{F}_{t-1})  \notag\\
		&\hspace{4cm} + \1_{\big\{|X_t-v_t(\vtheta_0^v)|<b_{n,x},\ Y_t\leq c_t(\vtheta_0^c)\big\}}C_2(\mathcal{F}_{t-1})V_1(\mathcal{F}_{t-1})\big\Vert\widehat{\vtheta}_n^c - \vtheta_0^c\big\Vert \notag\\
		&\hspace{4cm} + \1_{\big\{|X_t-v_t(\vtheta_0^v)|<b_{n,x},\ Y_t\leq c_t(\vtheta_0^c)\big\}}C_1(\mathcal{F}_{t-1}) V_2(\mathcal{F}_{t-1})\big\Vert\widehat{\vtheta}_n^v - \vtheta_0^v\big\Vert\notag\\
		&\hspace{4cm} + \frac{|\widehat{b}_{n,x}-b_{n,x}|}{b_{n,x}}\1_{\big\{|X_t-v_t(\vtheta_0^v)|<b_{n,x},\ Y_t\leq c_t(\vtheta_0^c)\big\}}C_1(\mathcal{F}_{t-1})V_1(\mathcal{F}_{t-1})\notag\\
		&=:\frac{b_{n,x}}{\widehat{b}_{n,x}}\big[A_{6n} + B_{6n} + C_{6n} + D_{6n}\big],\label{eq:(C.7m)}
	\end{align}
	where $\vtheta^\ast$ is some mean value that may differ from line to line. We show in turn that $A_{6n}=o_{\P}(1)$, $B_{6n}=o_{\P}(1)$, $C_{6n}=o_{\P}(1)$ and $D_{6n}=o_{\P}(1)$. 
	
	Theorem~\ref{thm:an} and Assumption~\ref{ass:avar} \ref{it:bw} imply that for any $d>0$
	\[
	b_{n,x}^{-1}\big\Vert\widehat{\vtheta}_n^v - \vtheta_0^v\big\Vert \leq d \qquad\text{and}\qquad b_{n,x}^{-1}\big\Vert\widehat{\vtheta}_n^c - \vtheta_0^c\big\Vert \leq d\qquad\text{and}\qquad \frac{\big|\widehat{b}_{n,x} - b_{n,x}\big|}{b_{n,x}}\leq d
	\]
	occur with probability approaching 1, as $n\to\infty$.
	This means that $\P\{\mathsf{E}_n^{C}\}\longrightarrow0$, as $n\to\infty$, where
	\[
	\mathsf{E}_n := \Big\{b_{n,x}^{-1}\big\Vert\widehat{\vtheta}_n^v - \vtheta_0^v\big\Vert<d,\ b_{n,x}^{-1}\big\Vert\widehat{\vtheta}_n^c - \vtheta_0^c\big\Vert<d,\ b_{n,x}^{-1}\big|\widehat{b}_{n,x} - b_{n,x}\big|\leq d\Big\}
	\]
	Therefore, it suffices to consider $A_{6n}$, $B_{6n}$, $C_{6n}$ and $D_{6n}$ on this set, because (e.g.) $\P\{A_{6n}>\varepsilon\}=\P\{A_{6n}>\varepsilon,\ \mathsf{E}_n\} + \P\{A_{6n}>\varepsilon,\ \mathsf{E}_n^{C}\}=\P\{A_{6n}>\varepsilon,\ \mathsf{E}_n\} + o(1)$.
	
	We first consider $A_{6n}$. On the set $\mathsf{E}_n$, we may bound $A_{6n}$ as follows:
	\begin{align*}
		A_{6n} &\leq  (2b_{n,x} n)^{-1}\sum_{t=1}^{n}\Big[\1_{\big\{|X_t-v_t(\vtheta_0^v) - b_{n,x}|<\norm{\nabla v_t(\vtheta^\ast)} \cdot\, \big\Vert\widehat{\vtheta}_n^v - \vtheta_0^v\big\Vert + |\widehat{b}_{n,x}-b_{n,x}|,\ Y_t\leq c_t(\widehat{\vtheta}_n^c)\big\}}\\
		&\hspace{3cm}+ \1_{\big\{|X_t-v_t(\vtheta_0^v) + b_{n,x}|< \norm{\nabla v_t(\vtheta^\ast)}\cdot\,\big\Vert\widehat{\vtheta}_n^v - \vtheta_0^v\big\Vert + |\widehat{b}_{n,x}-b_{n,x}|,\ Y_t\leq c_t(\widehat{\vtheta}_n^c)\big\}} \\
		&\hspace{3cm} + \1_{\big\{|X_t-v_t(\vtheta_0^v)|<b_{n,x},\ |Y_t-c_t(\vtheta_0^c)|\leq \norm{\nabla c_t(\vtheta^\ast)}\cdot\,\big\Vert\widehat{\vtheta}_n^c - \vtheta_0^c\big\Vert\big\}}\Big]C_1(\mathcal{F}_{t-1})V_1(\mathcal{F}_{t-1}) \\
		&\leq (2b_{n,x} n)^{-1}\sum_{t=1}^{n}\Big[\1_{\big\{|X_t-v_t(\vtheta_0^v) - b_{n,x}|<[V_1(\mathcal{F}_{t-1})+1]  d b_{n,x},\ Y_t\leq C(\mathcal{F}_{t-1})\big\}}\\
		&\hspace{3cm}+ \1_{\big\{|X_t-v_t(\vtheta_0^v) + b_{n,x}|< [V_1(\mathcal{F}_{t-1})+1]  d b_{n,x},\ Y_t\leq C(\mathcal{F}_{t-1})\big\}} \\
		&\hspace{3cm} + \1_{\big\{|X_t-v_t(\vtheta_0^v)|<b_{n,x},\ |Y_t-c_t(\vtheta_0^c)|\leq C_1(\mathcal{F}_{t-1})  d b_{n,x}\big\}}\Big]C_1(\mathcal{F}_{t-1})V_1(\mathcal{F}_{t-1})\\
		&=:A_{61n} + A_{62n} + A_{63n}.
	\end{align*}
	We show that each of these terms is $o_{\P}(1)$. Use the LIE to obtain that
	\begin{equation*}
		\E[A_{61n}] \leq (2b_{n,x} n)^{-1}\sum_{t=1}^{n}\E\bigg\{C_1(\mathcal{F}_{t-1})V_1(\mathcal{F}_{t-1})\E_{t-1}\Big[\1_{\big\{|X_t-v_t(\vtheta_0^v) - b_{n,x}|<[V_1(\mathcal{F}_{t-1})+1]  d b_{n,x},\ Y_t\leq C(\mathcal{F}_{t-1})\big\}}\Big]\bigg\}.
	\end{equation*}
	Write
	\begin{align*}
		\E_{t-1}&\Big[\1_{\big\{|X_t-v_t(\vtheta_0^v) - b_{n,x}|<[V_1(\mathcal{F}_{t-1})+1]  d b_{n,x},\ Y_t\leq C(\mathcal{F}_{t-1})\big\}}\Big]\\
		&= \P_{t-1}\Big\{-[V_1(\mathcal{F}_{t-1})+1]  d b_{n,x}<X_t-v_t(\vtheta_0^v) - b_{n,x}<[V_1(\mathcal{F}_{t-1})+1]  d b_{n,x},\ Y_t\leq C(\mathcal{F}_{t-1})\Big\}\\
		&=F_t\big(b_{n,x}+v_t(\vtheta_0^v) + [V_1(\mathcal{F}_{t-1}) +1] d b_{n,x}, C(\mathcal{F}_{t-1})\big) \\
		&\hspace{7cm}- F_t\big(b_{n,x}+v_t(\vtheta_0^v) - [V_1(\mathcal{F}_{t-1})+1]  d b_{n,x}, C(\mathcal{F}_{t-1})\big)\\
		& \leq \sup_{x\in\mathbb{R}} \Big|\partial_1 F_t\big(x, C(\mathcal{F}_{t-1})\big)\Big|2\big[V_1(\mathcal{F}_{t-1})+1\big]  d b_{n,x}\\
		&\leq 2K \big[V_1(\mathcal{F}_{t-1})+1\big]  d b_{n,x}.
	\end{align*}
	Hence, also using the Cauchy--Schwarz inequality and Assumption~\ref{ass:avar} \ref{it:mom bounds cons3},
	\begin{align*}
		0\leq \E[A_{61n}] &\leq (2b_{n,x} n)^{-1}\sum_{t=1}^{n}\E\Big[C_1(\mathcal{F}_{t-1})V_1(\mathcal{F}_{t-1})2K \big\{V_1(\mathcal{F}_{t-1})+1\big\}  d b_{n,x}\Big]\\
		&=Kd \frac{1}{n}\sum_{t=1}^{n}\E\big[C_1(\mathcal{F}_{t-1})V_1^2(\mathcal{F}_{t-1}) \big]\\
		&\leq Kd \frac{1}{n}\sum_{t=1}^{n}\sqrt{\E\big[C_1^2(\mathcal{F}_{t-1})\big]\E\big[V_1^4(\mathcal{F}_{t-1}) \big]}\\
		&\leq C d.
	\end{align*}
	Since $d>0$ can be chosen arbitrarily small, Markov's inequality implies that $A_{61n}=o_{\P}(1)$. 
	
	The proof that $A_{62n}=o_{\P}(1)$ is almost identical and, hence, is omitted. 
	
	It remains to show $A_{63n}=o_{\P}(1)$. Again, we use the LIE to get that
	\begin{equation*}
		\E[A_{63n}] \leq (2b_{n,x} n)^{-1}\sum_{t=1}^{n}\E\bigg\{C_1(\mathcal{F}_{t-1})V_1(\mathcal{F}_{t-1})\E_{t-1}\Big[\1_{\big\{|X_t-v_t(\vtheta_0^v)|<b_{n,x},\ |Y_t-c_t(\vtheta_0^c)|\leq C_1(\mathcal{F}_{t-1})  d b_{n,x}\big\}}\Big]\bigg\}.
	\end{equation*}
	Use Assumption~\ref{ass:an} \ref{it:bound cdf} to write
	\begin{align*}
		\E_{t-1}&\Big[\1_{\big\{|X_t-v_t(\vtheta_0^v)|<b_{n,x},\ |Y_t-c_t(\vtheta_0^c)|\leq C_1(\mathcal{F}_{t-1})  d b_{n,x}\big\}}\Big]\\
		&= \int_{c_t(\vtheta_0^c) - C_1(\mathcal{F}_{t-1})  d b_{n,x}}^{c_t(\vtheta_0^c) + C_1(\mathcal{F}_{t-1})  d b_{n,x}}\int_{v_t(\vtheta_0^v) - b_{n,x}}^{v_t(\vtheta_0^v) + b_{n,x}} f_t(x,y)\D x\D y\\
		&\leq K \big[2b_{n,x}\big] \big[2 C_1(\mathcal{F}_{t-1}) d b_{n,x}\big]\\
		& \leq C b_{n,x}^2 C_1(\mathcal{F}_{t-1})d.
	\end{align*}
	Thus, 
	\begin{align*}
		0\leq \E[A_{63n}] &\leq (2b_{n,x} n)^{-1}\sum_{t=1}^{n}\E\big[C_1(\mathcal{F}_{t-1})V_1(\mathcal{F}_{t-1})C b_{n,x}^2 C_1(\mathcal{F}_{t-1})d\big]\\
		&= \frac{C b_{n,x} d}{n}\sum_{t=1}^{n}\E\big[C_1^2(\mathcal{F}_{t-1})V_1(\mathcal{F}_{t-1}) \big]\\
		&\leq C b_{n,x} d=o(1)d=o(1).
	\end{align*}
	Markov's inequality then implies that $A_{63n}=o_{\P}(1)$. Overall, we obtain that $A_{6n}=o_{\P}(1)$.
	
	On the set $\mathsf{E}_n$, we may bound $B_{6n}$ and $C_{6n}$ as follows:
	\begin{align*}
		0\leq B_{6n}&\leq (2b_{n,x} n)^{-1}\sum_{t=1}^{n}\1_{\big\{|X_t-v_t(\vtheta_0^v)|<b_{n,x},\ Y_t\leq c_t(\vtheta_0^c)\big\}}C_2(\mathcal{F}_{t-1})V_1(\mathcal{F}_{t-1})d b_{n,x},\\
		0\leq C_{6n}&\leq (2b_{n,x} n)^{-1}\sum_{t=1}^{n} \1_{\big\{|X_t-v_t(\vtheta_0^v)|<b_{n,x},\ Y_t\leq c_t(\vtheta_0^c)\big\}}C_1(\mathcal{F}_{t-1}) V_2(\mathcal{F}_{t-1})d b_{n,x},
	\end{align*}
	such that
	\begin{align*}
		0\leq\E\big[B_{6n}\big]&\leq (2b_{n,x} n)^{-1}\sum_{t=1}^{n}\E\big[C_2(\mathcal{F}_{t-1})V_1(\mathcal{F}_{t-1})d b_{n,x}\big]\leq C d,\\
		0\leq\E\big[C_{6n}\big]&\leq (2b_{n,x} n)^{-1}\sum_{t=1}^{n}\E\big[C_1(\mathcal{F}_{t-1}) V_2(\mathcal{F}_{t-1})d b_{n,x}\big]\leq C d.
	\end{align*}
	Again, since $d>0$ can be chosen arbitrarily small, Markov's inequality gives that $B_{6n}=o_{\P}(1)$ and $C_{6n}=o_{\P}(1)$. 
	
	It remains to show that $D_{6n}=o_{\P}(1)$. 
	Use Assumption~\ref{ass:avar}~\ref{it:bw} to bound
	\begin{align*}
		D_{6n}&=(2b_{n,x} n)^{-1}\sum_{t=1}^{n}\frac{|\widehat{b}_{n,x} - b_{n,x}|}{b_{n,x}}\1_{\big\{|X_t-v_t(\vtheta_0^v)|<b_{n,x},\ Y_t\leq c_t(\vtheta_0^c)\big\}}C_1(\mathcal{F}_{t-1})V_1(\mathcal{F}_{t-1}) \\
		&\leq o_{\P}(1)(2b_{n,x} n)^{-1}\sum_{t=1}^{n}\1_{\big\{|X_t-v_t(\vtheta_0^v)|<b_{n,x},\ Y_t\leq C(\mathcal{F}_{t-1})\big\}}C_1(\mathcal{F}_{t-1})V_1(\mathcal{F}_{t-1}).
	\end{align*}
	Since
	\begin{align*}
		\E_{t-1}&\Big[\1_{\big\{|X_t-v_t(\vtheta_0^v)|<b_{n,x},\ Y_t\leq C(\mathcal{F}_{t-1})\big\}}\Big]\\
		&=\P_{t-1}\big\{v_t(\vtheta_0^v)-b_{n,x}<X_t<v_t(\vtheta_0^v)+b_{n,x},\ Y_t\leq C(\mathcal{F}_{t-1})\big\}\\
		&=F_t\big(v_t(\vtheta_0^v)+b_{n,x}, C(\mathcal{F}_{t-1})\big) - F_t\big(v_t(\vtheta_0^v)-b_{n,x}, C(\mathcal{F}_{t-1})\big)\\
		&\leq\sup_{x\in\mathbb{R}}\big|\partial_1 F_t\big(x, C(\mathcal{F}_{t-1})\big)\big|2b_{n,x}\\
		&\leq C b_{n,x},
	\end{align*}
	we obtain from the LIE that
	\begin{align*}
		0&\leq \E\Big[(2b_{n,x} n)^{-1}\sum_{t=1}^{n}\1_{\big\{|X_t-v_t(\vtheta_0^v)|<b_{n,x},\ Y_t\leq C(\mathcal{F}_{t-1})\big\}}C_1(\mathcal{F}_{t-1})V_1(\mathcal{F}_{t-1})\Big]\\
		&= (2b_{n,x} n)^{-1}\sum_{t=1}^{n}\E\bigg\{C_1(\mathcal{F}_{t-1})V_1(\mathcal{F}_{t-1})\E_{t-1}\Big[\1_{\big\{|X_t-v_t(\vtheta_0^v)|<b_{n,x},\ Y_t\leq C(\mathcal{F}_{t-1})\big\}}\Big]\bigg\}\\
		&\leq (2b_{n,x} n)^{-1}\sum_{t=1}^{n}C b_{n,x}\sqrt{\E\big[C_1^2(\mathcal{F}_{t-1})\big]\E\big[V_1^2(\mathcal{F}_{t-1})\big]}\leq C.
	\end{align*}
	Thus, $D_{6n}=o_{\P}(1)O_{\P}(1)=o_{\P}(1)$. We therefore obtain \eqref{eq:(1)} from \eqref{eq:(C.7m)}. 
	
	It remains to show \eqref{eq:(2)}, i.e., $\widetilde{\mD}_n-\mD_{n,0}=o_{\P}(1)$. Write this as
	\begin{align*}
		\widetilde{\mD}_n-\mD_{n,0} &=(2b_{n,x} n)^{-1}\sum_{t=1}^{n}\bigg\{ \nabla c_t(\vtheta_0^c)\nabla^\prime v_t(\vtheta_0^v) \Big[\1_{\big\{|X_t-v_t(\vtheta_0^v)|<b_{n,x},\ Y_t\leq c_t(\vtheta_0^c)\big\}} \\
		&\hspace{7cm} - \E_{t-1}\big[\1_{\big\{|X_t-v_t(\vtheta_0^v)|<b_{n,x},\ Y_t\leq c_t(\vtheta_0^c)\big\}}\big] \Big]\bigg\}\\
		&\hspace{1.3cm} + \frac{1}{n}\sum_{t=1}^{n}\bigg\{\nabla c_t(\vtheta_0^c)\nabla^\prime v_t(\vtheta_0^v) \Big[(2b_{n,x})^{-1}\E_{t-1}\big[\1_{\big\{|X_t-v_t(\vtheta_0^v)|<b_{n,x},\ Y_t\leq c_t(\vtheta_0^c)\big\}}\big]\\
		&\hspace{7cm} - \partial_1 F_t\big(v_t(\vtheta_0^v), c_t(\vtheta_0^c)\big)\Big]\bigg\}\\
		&\hspace{1.3cm} + \frac{1}{n}\sum_{t=1}^{n}\bigg\{\nabla c_t(\vtheta_0^c)\nabla^\prime v_t(\vtheta_0^v)\partial_1 F_t\big(v_t(\vtheta_0^v), c_t(\vtheta_0^c)\big)\\
		&\hspace{7cm} - \E\Big[\nabla c_t(\vtheta_0^c)\nabla^\prime v_t(\vtheta_0^v)\partial_1 F_t\big(v_t(\vtheta_0^v), c_t(\vtheta_0^c)\big)\Big]\bigg\}\\
		&=:A_{7n} + B_{7n} + C_{7n}.
	\end{align*}
	
	Next, we show that each of these terms vanishes in probability. Observe that $\E[A_{7n}]=\vzeros$ and for any $(i,j)$-th element of $A_{7n}$, denoted by $A_{7n,ij}$, we obtain
	\begin{align*}
		\Var(A_{7n,ij}) &=\E\big[A_{7n,ij}^2\big]\\
		&=(2b_{n,x} n)^{-2}\E\bigg[\sum_{t=1}^{n} \nabla_i c_t(\vtheta_0^c)\nabla_j v_t(\vtheta_0^v) \Big\{\1_{\big\{|X_t-v_t(\vtheta_0^v)|<b_{n,x},\ Y_t\leq c_t(\vtheta_0^c)\big\}} \\
		&\hspace{7cm} - \E_{t-1}\big[\1_{\big\{|X_t-v_t(\vtheta_0^v)|<b_{n,x},\ Y_t\leq c_t(\vtheta_0^c)\big\}}\big] \Big\}\bigg]^2\\
		&=(2b_{n,x} n)^{-2}\sum_{t=1}^{n}\E\bigg[ \Big\{\nabla_i c_t(\vtheta_0^c)\nabla_j v_t(\vtheta_0^v) \Big\}^2\Big\{\1_{\big\{|X_t-v_t(\vtheta_0^v)|<b_{n,x},\ Y_t\leq c_t(\vtheta_0^c)\big\}} \\
		&\hspace{7cm} - \E_{t-1}\big[\1_{\big\{|X_t-v_t(\vtheta_0^v)|<b_{n,x},\ Y_t\leq c_t(\vtheta_0^c)\big\}}\big] \Big\}^2\bigg]\\
		&\leq (2b_{n,x} n)^{-2}\sum_{t=1}^{n}\E\big[C_1^2(\mathcal{F}_{t-1})V_1^2(\mathcal{F}_{t-1})\big]\\
		&\leq (2b_{n,x} n)^{-2}\sum_{t=1}^{n}\Big\{\E\big[C_1^4(\mathcal{F}_{t-1})\big]\Big\}^{1/2}\Big\{\E\big[V_1^4(\mathcal{F}_{t-1})\big]\Big\}^{1/2}\\
		&\leq C b_{n,x}^{-2}n^{-1}=o(1),
	\end{align*}
	where we used for the third equality that the expectations of all cross-products are zero by the LIE, and the final inequality exploits Assumption~\ref{ass:avar}~\ref{it:mom bounds cons3}. That $A_{7n}=o_{\P}(1)$ now follows from Chebyshev's inequality. 
	
	For $B_{7n}$, we obtain that
	\begin{align*}
		\E\norm{B_{7n}} &\leq \frac{1}{n}\sum_{t=1}^{n}\E\bigg\Vert\nabla c_t(\vtheta_0^c)\nabla^\prime v_t(\vtheta_0^v) \Big\{(2b_{n,x})^{-1}\E_{t-1}\big[\1_{\big\{|X_t-v_t(\vtheta_0^v)|<b_{n,x},\ Y_t\leq c_t(\vtheta_0^c)\big\}}\big]\\
		&\hspace{10.5cm} - \partial_1 F_t\big(v_t(\vtheta_0^v), c_t(\vtheta_0^c)\big)\Big\}\bigg\Vert\\
		&\leq \frac{1}{n}\sum_{t=1}^{n}\E\bigg[C_1(\mathcal{F}_{t-1})V_1(\mathcal{F}_{t-1}) \Big|(2b_{n,x})^{-1}\E_{t-1}\big[\1_{\big\{|X_t-v_t(\vtheta_0^v)|<b_{n,x},\ Y_t\leq c_t(\vtheta_0^c)\big\}}\big]\\
		&\hspace{10.5cm} - \partial_1 F_t\big(v_t(\vtheta_0^v), c_t(\vtheta_0^c)\big)\Big|\bigg]\\
		&\leq \frac{1}{n}\sum_{t=1}^{n}\E\bigg[C_1(\mathcal{F}_{t-1})V_1(\mathcal{F}_{t-1}) \Big|(2b_{n,x})^{-1}\Big\{F_t\big(v_t(\vtheta_0^v) + b_{n,x}, c_t(\vtheta_0^c)\big) - F_t\big(v_t(\vtheta_0^v) - b_{n,x}, c_t(\vtheta_0^c)\big)\Big\}\\
		&\hspace{10.5cm} - \partial_1 F_t\big(v_t(\vtheta_0^v), c_t(\vtheta_0^c)\big)\Big|\bigg]\\
		&\leq \frac{1}{n}\sum_{t=1}^{n}\E\bigg[C_1(\mathcal{F}_{t-1})V_1(\mathcal{F}_{t-1}) \sup_{x\in[v_t(\vtheta_0^v)-b_{n,x},\; v_t(\vtheta_0^v) + b_{n,x}]}\Big|(2b_{n,x})^{-1}\partial_1 F_t\big(x, c_t(\vtheta_0^c)\big) 2b_{n,x} \\
		&\hspace{10.5cm} - \partial_1 F_t\big(v_t(\vtheta_0^v), c_t(\vtheta_0^c)\big)\Big|\bigg]\\
		&\leq \frac{1}{n}\sum_{t=1}^{n}\E\Big[C_1(\mathcal{F}_{t-1})V_1(\mathcal{F}_{t-1}) \sup_{x\in[v_t(\vtheta_0^v)-b_{n,x},\; v_t(\vtheta_0^v) + b_{n,x}]}K\big|x-v_t(\vtheta_0^v)\big|\Big]\\
		&\leq \frac{1}{n}\sum_{t=1}^{n}Kb_{n,x} \E\big[C_1(\mathcal{F}_{t-1})V_1(\mathcal{F}_{t-1})\big]\\
		&\leq C b_{n,x}=o(1),
	\end{align*}
	such that Markov's inequality implies $B_{7n}=o_{\P}(1)$.
	
	Finally, Assumption~\ref{ass:avar} \ref{it:Lambda2} implies that $C_{7n}=o_{\P}(1)$.
	
	Overall, \eqref{eq:(2)} follows, ending the proof.
\end{proof}

\section{Verifying Assumptions~\ref{ass:cons} and \ref{ass:an} for CCC--GARCH Models}
\label{sec:verification}

In this section, we verify our main Assumptions~\ref{ass:cons} and \ref{ass:an} for a CCC--GARCH process \citep{Bol90}. Subsection~\ref{CCC-GARCH} presents the model and Subsection~\ref{PR} collects some properties of the CCC--GARCH model that will be helpful in verifying Assumptions~\ref{ass:cons} and \ref{ass:an} in Subsections~\ref{Ass1} and \ref{Ass2}, respectively.

\subsection{The CCC--GARCH Model}\label{CCC-GARCH}

The CCC--GARCH model we consider throughout this section is given by the recursion
\begin{equation}\label{eq:CG}
	\begin{pmatrix}
		X_t\\ Y_t
	\end{pmatrix} = 
	\begin{pmatrix}
		\sigma_{X,t}(\vtheta^X)\varepsilon_{X,t}\\ 
		\sigma_{Y,t}(\vtheta^Y)\varepsilon_{Y,t}
	\end{pmatrix},\qquad
	\begin{pmatrix}
		\sigma_{X,t}(\vtheta^X)\\ 
		\sigma_{Y,t}(\vtheta^Y)
	\end{pmatrix}
	=\begin{pmatrix}
		\omega_X + \alpha_X|X_{t-1}| + \beta_X\sigma_{X,t-1}(\vtheta^X)\\ 
		\omega_Y + \alpha_Y|Y_{t-1}| + \beta_Y\sigma_{Y,t-1}(\vtheta^Y)
	\end{pmatrix},\quad t\in\mathbb{Z},
\end{equation}
where we denote the true parameter values by $\vtheta_0^{X}=(\omega_{X,0},\alpha_{X,0},\beta_{X,0})^\prime$ and $\vtheta_0^Y=(\omega_{Y,0},\alpha_{Y,0},\beta_{Y,0})^\prime$.
As in Example~\ref{ex:1}, we do not consider the standard GARCH model that uses squares, but we consider an absolute value specification.
In fact, \eqref{eq:CG} arises as a special case of model \eqref{eqn:ECCCmodel} with diagonal $\widetilde{\mA}$ and $\widetilde{\mB}$.

We make the following assumptions:

\begin{CGassumption}\label{CGass:error}
	The innovations $(\varepsilon_{X,t}, \varepsilon_{Y,t})^\prime$ are i.i.d., independent of $\mathcal{F}_{t-1}=\sigma\big\{(X_{t-1}, Y_{t-1})^\prime,(X_{t-2}, Y_{t-2})^\prime,\ldots \big\}$, and have mean zero and correlation matrix $\mR$. 
	The support of $(\varepsilon_{X,t}, \varepsilon_{Y,t})^\prime$ (defined by its strictly positive Lebesgue density $f_{\varepsilon_X,\varepsilon_Y}(\cdot,\cdot)$) contains an open set and zero. 
	Moreover, the Lebesgue density is positive for all $(x,y)^\prime\in\mathbb{R}^2$ such that $F_{\varepsilon_X,\varepsilon_Y}(x,y)\in(0,1)$. 
	Finally, 
	\begin{align*}
		f_{\varepsilon_X}(\cdot)&\leq K,& f_{\varepsilon_Y}(\cdot)&\leq K,\\
		\big|\partial_i F_{\varepsilon_X,\varepsilon_Y}(\cdot,\cdot)\big|&\leq K\ (i=1,2),& f_{\varepsilon_X,\varepsilon_Y}(\cdot,\cdot)&\leq K,\\
		\big|f_{\varepsilon_X}(x)-f_{\varepsilon_X}(x^\prime)\big|&\leq K |x-x^\prime|,& \big|f_{\varepsilon_Y}(y)-f_{\varepsilon_Y}(y^\prime)\big|&\leq K |y-y^\prime|,\\
		\big|\partial_1F_{\varepsilon_X,\varepsilon_Y}(x,y)-\partial_1 F_{\varepsilon_X,\varepsilon_Y}(x^\prime,y)\big|&\leq K |x-x^\prime|,& \big|\partial_2F_{\varepsilon_X,\varepsilon_Y}(x,y)-\partial_2 F_{\varepsilon_X,\varepsilon_Y}(x,y^\prime)\big|&\leq K |y-y^\prime|.
	\end{align*}
\end{CGassumption}

CCC--GARCH Assumption~\ref{CGass:error} ensures that the $\sigma_{W,t}(\vtheta^W)$ ($W\in\{X,Y\}$) are the conditional standard deviations of the $W_t$. The requirement that the support contains an open set and zero is needed to verify certain mixing properties; see the proof of Lemma~\ref{lem:sd} below.

\begin{CGassumption}\label{CGass:param}
	All parameters from the compact parameter space $\mTheta^{(X,Y)}$ satisfy the following constraints for some $\delta_1\in(0,1/2)$:
	\begin{align}
		&0\leq\delta_1\leq\omega_X,\omega_Y\leq K_{\omega}<\infty,\quad \delta_1\leq\alpha_X, \alpha_Y, \beta_X, \beta_Y \leq 1-\delta_1,\notag\\
		&\E\Big[\max\big\{\alpha_X|\varepsilon_{X,t}|+\beta_X, \alpha_Y|\varepsilon_{Y,t}|+\beta_Y\big\}^2\Big]<1.\label{eq:str stat para}
	\end{align}
	Furthermore, the true parameter vector $\vtheta_0=(\vtheta_0^{X\prime}, \vtheta_0^{Y\prime})^\prime$ lies in the interior of $\mTheta^{(X,Y)}$.
\end{CGassumption}

Among other purposes, CCC--GARCH Assumption~\ref{CGass:param} is required to show that the $(X_t,Y_t)^\prime$ are strictly stationary. 
Use the shorthand $a\vee b :=\max\{a,b\}$, $\alpha_{\ast}:=\alpha_X\vee\alpha_Y$, $\beta_{\ast}:=\beta_X\vee\beta_Y$. 
Then, an easily verifiable sufficient condition for \eqref{eq:str stat para} is $2(\alpha_{\ast}+\beta_{\ast})^2-\beta_{\ast}<1$, because
\begin{align*}
	\E\Big[\max\big\{\alpha_X|\varepsilon_{X,t}|+\beta_X, \alpha_Y|\varepsilon_{Y,t}|+\beta_Y\big\}^2\Big]&\leq \E\Big[\big\{(\alpha_X\vee\alpha_Y)(|\varepsilon_{X,t}|\vee|\varepsilon_{Y,t}|)+(\beta_X\vee\beta_Y)\big\}^2\Big]\\
	&= \E\Big[\alpha_{\ast}^2(|\varepsilon_{X,t}|\vee|\varepsilon_{Y,t}|)^2+2\alpha_{\ast}\beta_{\ast}(|\varepsilon_{X,t}|\vee|\varepsilon_{Y,t}|) + \beta_{\ast}^2\Big]\\
	&\leq \alpha_{\ast}^2\big(\E|\varepsilon_{X,t}|^2 + \E|\varepsilon_{Y,t}|^2\big)+2\alpha_{\ast}\beta_{\ast}\big(\E|\varepsilon_{X,t}| + \E|\varepsilon_{Y,t}|\big) + \beta_{\ast}^2\\
	&\leq 2\alpha_{\ast}^2 + 4\alpha_{\ast}\beta_{\ast} + \beta_{\ast}^2=2(\alpha_{\ast}+\beta_{\ast})^2-\beta_{\ast}^2,
\end{align*}
where we used that $\E|\varepsilon_{W,t}|^2=\Var(\varepsilon_{W,t})=1$ and $\E|\varepsilon_{W,t}|\leq \{\E|\varepsilon_{W,t}|^2\}^{1/2}=1$ ($W\in\{X,Y\}$). For instance, this sufficient condition is satisfied for the parameters $\alpha_{\ast}=0.05$ and $\beta_{\ast}=0.9$, which are common values for financial data.

\begin{CGassumption}\label{CGass:mom}
	$\E|X_t|^{3+\iota}<\infty$ and $\E|Y_t|^{3+\iota}<\infty$ for some $\iota>0$.
\end{CGassumption}

\subsection{Preliminary Results}\label{PR}

Before showing that Assumptions~\ref{ass:cons} and \ref{ass:an} hold under CCC--GARCH Assumptions~\ref{CGass:error}--\ref{CGass:mom} in Subsections~\ref{Ass1} and \ref{Ass2}, respectively, we explore some properties of the model that will be used throughout.

Example~\ref{ex:1} shows that model \eqref{eq:CG} implies the (VaR, CoVaR) dynamics
\begin{align}
	v_t(\vtheta^v) &= \omega_v + \alpha_v|X_{t-1}| + \beta_v v_{t-1}(\vtheta^v),\label{eq:(1.1)}\\
	c_t(\vtheta^c) &= \omega_c + \alpha_c|Y_{t-1}| + \beta_c c_{t-1}(\vtheta^c),\label{eq:(1.2)}
\end{align}
where 
\begin{align}
	\vtheta^v&=(\omega_v, \alpha_v, \beta_v)^\prime = (v\omega_X, v\alpha_X, \beta_X)^\prime,&& v=\VaR_\beta(F_{\varepsilon_{X}}),\label{eq:para trafo 1}\\
	\vtheta^c&=(\omega_c, \alpha_c, \beta_c)^\prime = (c\omega_Y, c\alpha_Y, \beta_Y)^\prime,&& c=\CoVaR_{\alpha\mid\beta}(F_{\varepsilon_{X},\varepsilon_{Y}}).\label{eq:para trafo 2}
\end{align}
Here, we are specifically interested in VaR and CoVaR conditioned on the $\sigma$-field 
\[
\mathcal{F}_{t-1}=\sigma\big((X_{t-1},Y_{t-1})^\prime,\ldots,(X_{1},Y_{1})^\prime,\init\big),\quad\text{where }	\init= (X_0,Y_0,X_{-1},Y_{-1},\ldots)^\prime.
\]
While conditioning on the infinite past may be seen as unrealistic from a modeling perspective (after all no time series extends back indefinitely), it is very convenient for verifying Assumptions~\ref{ass:cons} and \ref{ass:an}, which is already fairly involved. 
One of the advantages of conditioning on the infinite past is that it renders $v_t(\vtheta_0^v)$ and $c_t(\vtheta_0^c)$ stationary, which is required (e.g.) for invoking results of \citet{CC02} below. 
Stationarity of $v_t(\vtheta_0^v)$ and $c_t(\vtheta_0^c)$ is also useful as the matrices from Assumption~\ref{ass:an} \ref{it:pd} no longer depend on $n$, such that we may write
\[
\mLambda=\mLambda_n,\quad \mLambda_{(1)}=\mLambda_{n,(1)},\quad \mV=\mV_n, \ldots.
\]
Hence, it suffices to prove positive definiteness of $\mLambda, \mLambda_{(1)},\mV,\ldots$ instead of \textit{uniform} positive definiteness of $\mLambda_{n}, \mLambda_{n,(1)}, \mV_n,\ldots $.

In view of \eqref{eq:para trafo 1} and \eqref{eq:para trafo 2}, the parameter space of the (VaR, CoVaR) model is
\begin{equation}\label{eq:(D.6p)}
	\mTheta=\mTheta^v\times\mTheta^c = \bigg\{\vtheta=(\vtheta^{v\prime}, \vtheta^{c\prime})^\prime=(\theta_1,\ldots,\theta_6)^\prime\in\mathbb{R}^{6}:\ \Big(\frac{\theta_1}{v},\frac{\theta_2}{v},\theta_3,\frac{\theta_4}{v},\frac{\theta_5}{v},\theta_6\Big)^\prime\in\mTheta^{(X,Y)}\bigg\}.
\end{equation}
In particular, 
\begin{align}
	v_t(\vtheta^v) &= v\sigma_{X,t}(\vtheta^X),\label{eq:(3.11)}\\
	c_t(\vtheta^c) &= c\sigma_{Y,t}(\vtheta^Y),\label{eq:(3.2)}
\end{align}
such that by the chain rule
\begin{align}
	\nabla v_t(\vtheta^v) &= \begin{pmatrix}\frac{\partial v_t(\vtheta^v)}{\partial \omega_v}\\ \frac{\partial v_t(\vtheta^v)}{\partial \alpha_v}\\ \frac{\partial v_t(\vtheta^v)}{\partial \beta_v}\end{pmatrix}
	= \begin{pmatrix}v\frac{\partial \sigma_{X,t}(\vtheta^X=(\omega_v/v, \alpha_v/v,\beta_v)^\prime)}{\partial \omega_v}\\ v\frac{\partial \sigma_{X,t}(\vtheta^X=(\omega_v/v, \alpha_v/v,\beta_v)^\prime)}{\partial \alpha_v}\\ v\frac{\partial \sigma_{X,t}(\vtheta^X=(\omega_v/v, \alpha_v/v,\beta_v)^\prime)}{\partial \beta_v}
	\end{pmatrix}
	=\begin{pmatrix}\frac{\partial \sigma_{X,t}(\vtheta^X)}{\partial \omega_X}\\ \frac{\partial \sigma_{X,t}(\vtheta^X)}{\partial \alpha_X}\\ v\frac{\partial \sigma_{X,t}(\vtheta^X)}{\partial \beta_X}\end{pmatrix}=\nabla \sigma_{X,t}(\vtheta^X)\odot\begin{pmatrix}1\\1\\ v\end{pmatrix},\label{eq:1.1}\\
	\nabla c_t(\vtheta^c) &= \begin{pmatrix}\frac{\partial c_t(\vtheta^c)}{\partial \omega_c}\\ \frac{\partial c_t(\vtheta^c)}{\partial \alpha_c}\\ \frac{\partial c_t(\vtheta^c)}{\partial \beta_c}\end{pmatrix}
	= \begin{pmatrix}c\frac{\partial \sigma_{Y,t}(\vtheta^Y=(\omega_c/c, \alpha_c/c,\beta_c)^\prime)}{\partial \omega_c}\\ c\frac{\partial \sigma_{Y,t}(\vtheta^Y=(\omega_v/v, \alpha_c/c,\beta_c)^\prime)}{\partial \alpha_c}\\ c\frac{\partial \sigma_{Y,t}(\vtheta^Y=(\omega_v/v, \alpha_c/c,\beta_c)^\prime)}{\partial \beta_c}
	\end{pmatrix}
	=\begin{pmatrix}\frac{\partial \sigma_{Y,t}(\vtheta^Y)}{\partial \omega_Y}\\ \frac{\partial \sigma_{Y,t}(\vtheta^Y)}{\partial \alpha_Y}\\ c\frac{\partial \sigma_{Y,t}(\vtheta^Y)}{\partial \beta_Y}\end{pmatrix}=\nabla \sigma_{Y,t}(\vtheta^Y)\odot\begin{pmatrix}1\\1\\ c\end{pmatrix},\label{eq:1.2}
\end{align}
where $\odot$ denotes the (element-by-element) Hadamard product.

\begin{lem}\label{lem:8}
	Under CCC--GARCH Assumptions~\ref{CGass:error}--\ref{CGass:mom}, it holds for $W\in\{X,Y\}$ that
	\begin{align}
		\sigma_{W,t}(\vtheta^W)&=\frac{\omega_W}{1-\beta_W} + \alpha_W W_1(t,\beta_W),\label{eq:lem51}\\
		\big\Vert\nabla\sigma_{W,t}(\vtheta^W)\big\Vert &\leq \frac{\omega_X}{1-\beta_W}\bigg[\frac{1}{\omega_X}+\frac{1}{1-\beta_X}-\frac{1}{\alpha_X}\bigg] + \frac{1}{\alpha_W}\sigma_{W,t}(\vtheta^W) + \alpha_W W_2(t,\beta_W),\notag\\
		\big\Vert\nabla^2\sigma_{W,t}(\vtheta^W)\big\Vert &\leq \frac{2}{(1-\beta_W)^2}\bigg[1+\frac{\omega_W}{1-\beta_W}\bigg] +2W_2(t,\beta_W) + \alpha_W W_3(t,\beta_W),\notag
	\end{align}
	where
	\begin{align*}
		W_1(t,\beta_W) &= \sum_{i=1}^{\infty}\beta_W^{i-1}|W_{t-i}|,\\
		W_2(t,\beta_W) &= \sum_{i=2}^{\infty}(i-1)\beta_W^{i-2}|W_{t-i}|,\\
		W_3(t,\beta_W) &= \sum_{i=3}^{\infty}(i-1)(i-2)\beta_W^{i-3}|W_{t-i}|.
	\end{align*}
\end{lem}

\begin{proof}
	For concreteness, we let $W=X$; the proofs for $W=Y$ are virtually identical. Iterate the volatility recursion to obtain that
	\begin{align}
		\sigma_{X,t}(\vtheta^X) &= \omega_X+\alpha_X|X_{t-1}| + \beta_X \sigma_{X,t-1}(\vtheta^X)\notag\\
		&=\frac{\omega_X}{1-\beta_X} + \alpha_X\sum_{i=1}^{\infty}\beta_X^{i-1}|X_{t-i}|\notag\\
		&=\frac{\omega_X}{1-\beta_X} + \alpha_X X_1(t,\beta_X),\label{eq:(8.1)}
	\end{align}
	from which the first claim in \eqref{eq:lem51} follows.
	
	For the second claim, we use \eqref{eq:(8.1)} to derive that
	\begin{align}
		\nabla \sigma_{X,t}(\vtheta^X)&=\begin{pmatrix}\frac{\partial \sigma_{X,t}(\vtheta^X)}{\partial \omega_X}\\ \frac{\partial \sigma_{X,t}(\vtheta^X)}{\partial \alpha_X}\\ \frac{\partial \sigma_{X,t}(\vtheta^X)}{\partial \beta_X}\end{pmatrix}=\begin{pmatrix}\frac{1}{1-\beta_X}\\ \sum_{i=1}^{\infty}\beta_X^{i-1}|X_{t-i}|\\ \frac{\omega_X}{(1-\beta_X)^2}+\alpha_X\sum_{i=1}^{\infty}(i-1)\beta_X^{i-2}|X_{t-i}|\end{pmatrix}\notag\\
		&=\begin{pmatrix}\frac{1}{1-\beta_X}\\ X_1(t,\beta_X)\\ \frac{\omega_X}{(1-\beta_X)^2}+\alpha_X X_2(t,\beta_X)\end{pmatrix}.\label{eq:(8.2)}
	\end{align}
	Using this, \eqref{eq:(8.1)} and the fact that the Frobenius norm is always less than the sum of the absolute values of the entries,
	\begin{align*}
		\big\Vert\nabla \sigma_{X,t}(\vtheta^X)\big\Vert & \leq \frac{1}{1-\beta_X} + \frac{\omega_X}{(1-\beta_X)^2} + X_1(t,\beta_X) + \alpha_X X_2(t,\beta_X)\\
		&=\frac{1}{1-\beta_X} + \frac{\omega_X}{(1-\beta_X)^2} + \frac{1}{\alpha_X}\bigg[\sigma_{X,t}(\vtheta^X)- \frac{\omega_X}{1-\beta_X}\bigg] + \alpha_X X_2(t,\beta_X)\\
		&=  \frac{\omega_X}{1-\beta_X}\bigg[\frac{1}{\omega_X} + \frac{1}{1-\beta_X} - \frac{1}{\alpha_X}\bigg] + \frac{1}{\alpha_X}\sigma_{X,t}(\vtheta^X) + \alpha_X X_2(t,\beta_X).
	\end{align*}
	
	Finally, it is easy to check from \eqref{eq:(8.2)} that
	\begin{align}
		\nabla^2\sigma_{X,t}(\vtheta^X) &= \begin{pmatrix}\frac{\partial^2 \sigma_{X,t}(\vtheta^X)}{\partial \omega_X^2} & \frac{\partial^2 \sigma_{X,t}(\vtheta^X)}{\partial \omega_X\partial \alpha_X}& \frac{\partial^2 \sigma_{X,t}(\vtheta^X)}{\partial \omega_X\partial \beta_X}\\
			\frac{\partial^2 \sigma_{X,t}(\vtheta^X)}{\partial \alpha_X\partial\omega_X} & \frac{\partial^2 \sigma_{X,t}(\vtheta^X)}{\partial \alpha_X^2}& \frac{\partial^2 \sigma_{X,t}(\vtheta^X)}{\partial \alpha_X\partial \beta_X}\\
			\frac{\partial^2 \sigma_{X,t}(\vtheta^X)}{\partial\beta_X\partial \omega_X} & \frac{\partial^2 \sigma_{X,t}(\vtheta^X)}{\partial \beta_X\partial \alpha_X}& \frac{\partial^2 \sigma_{X,t}(\vtheta^X)}{\partial \beta_X^2}\end{pmatrix}\notag\\
		&=
		\begin{pmatrix}0 & 0& \frac{1}{(1-\beta_X)^2}\\
			0 & 0 & X_2(t,\beta_X)\\
			\frac{1}{(1-\beta_X)^2} & X_2(t,\beta_X) & \frac{2\omega_X}{(1-\beta_X)^3} + \alpha_X X_3(t,\beta_X)\end{pmatrix}.\label{eq:(p.48)}
	\end{align}
	Since the Frobenius norm is always less than the sum of the absolute values of the matrix entries, we have that
	\[
	\big\Vert\nabla^2\sigma_{X,t}(\vtheta^X)\big\Vert \leq \frac{2}{(1-\beta_X)^2} + 2X_2(t,\beta_X) + \frac{2\omega_X}{(1-\beta_X)^3}+\alpha_X X_3(t,\beta_X),
	\]
	which is just the conclusion.
\end{proof}

\begin{lem}\label{lem:9}
	Under CCC--GARCH Assumptions~\ref{CGass:error}--\ref{CGass:mom}, there exists some constant $M>0$, such that for all $\vtheta^v\in\mTheta^v$ and all $\vtheta^c\in\mTheta^c$,
	\begin{align*}
		\big|v_t(\vtheta^v)\big| &\leq V(\mathcal{F}_{t-1})= M\big[K_\omega \delta_1^{-1}+(1-\delta_1)\overline{X}_1(\mathcal{F}_{t-1})\big],\\
		\big|c_t(\vtheta^v)\big| &\leq C(\mathcal{F}_{t-1})= M\big[K_\omega \delta_1^{-1}+(1-\delta_1)\overline{Y}_1(\mathcal{F}_{t-1})\big],\\
		\big\Vert\nabla v_t(\vtheta^v)\big\Vert &\leq V_1(\mathcal{F}_{t-1})=M\big[3K_{\omega}\delta_1^{-2}+\overline{X}_1(\mathcal{F}_{t-1}) + (1-\delta_1)\overline{X}_2(\mathcal{F}_{t-1})\big],\\
		\big\Vert\nabla c_t(\vtheta^c)\big\Vert &\leq C_1(\mathcal{F}_{t-1})=M\big[3K_{\omega}\delta_1^{-2}+\overline{Y}_1(\mathcal{F}_{t-1}) + (1-\delta_1)\overline{Y}_2(\mathcal{F}_{t-1})\big],\\
		\big\Vert\nabla^2 v_t(\vtheta^v)\big\Vert &\leq V_2(\mathcal{F}_{t-1})=M\big[2\delta_1^{-2}(1+\delta_1^{-1}K_{\omega})+ 2 \overline{X}_2(\mathcal{F}_{t-1}) + (1-\delta_1)\overline{X}_3(\mathcal{F}_{t-1})\big],\\
		\big\Vert\nabla^2 c_t(\vtheta^c)\big\Vert &\leq C_2(\mathcal{F}_{t-1})=M\big[2\delta_1^{-2}(1+\delta_1^{-1}K_{\omega})+ 2 \overline{Y}_2(\mathcal{F}_{t-1}) + (1-\delta_1)\overline{Y}_3(\mathcal{F}_{t-1})\big],
	\end{align*}
	where $\overline{W}_i(\mathcal{F}_{t-1}) =W_i(t,1-\delta_1)$ for $W\in\{X,Y\}$ and $i=1,2,3$.
\end{lem}

\begin{proof}
	From \eqref{eq:(8.1)} and CCC--GARCH Assumption~\ref{CGass:param} it is easy to see that 
	\begin{align*}
		\sigma_{X,t}(\vtheta^X) &= \frac{\omega_X}{1-\beta_X} + \alpha_X X_1(t,\beta_X)\\
		&\leq \frac{K_{\omega}}{\delta_1} + (1-\delta_1) X_1(t,1-\delta_1)\\
		&= \frac{K_{\omega}}{\delta_1} + (1-\delta_1) \overline{X}_1(\mathcal{F}_{t-1}),
	\end{align*}
	where we used in the second step that the function $\beta_X\mapsto X_1(t,\beta_X)$ increases in $\beta_X$.
	This and \eqref{eq:(3.11)} imply the first claim, i.e.,
	\begin{equation*}
		\big|v_t(\vtheta^v)\big| = |v|\sigma_{X,t}(\vtheta^X)\leq M \big[K_{\omega}\delta_1^{-1} + (1-\delta_1) \overline{X}_1(\mathcal{F}_{t-1})\big].
	\end{equation*}
	The second claim can be established analogously using \eqref{eq:(3.2)}. 
	
	For the third claim note that from \eqref{eq:1.1}
	\begin{equation}\label{eq:hhh1}
		\big\Vert \nabla v_t(\vtheta^v) \big\Vert = \bigg\Vert \nabla \sigma_{X,t}(\vtheta^X)\odot\begin{pmatrix}1\\1\\ v\end{pmatrix} \bigg\Vert\leq M\big\Vert\nabla \sigma_{X,t}(\vtheta^X) \big\Vert.
	\end{equation}
	The bound now follows from Lemma~\ref{lem:8} because
	\begin{align*}
		\big\Vert \nabla \sigma_{X,t}(\vtheta^X) \big\Vert &\leq \frac{K_{\omega}}{\delta_1}\bigg[\frac{1}{\delta_1} + \frac{1}{\delta_1} -\frac{1}{1-\delta_1}\bigg] + \frac{1}{\alpha_X}\bigg[\frac{\omega_X}{1-\beta_X}+\alpha_X X_1(t,\beta_X)\bigg] + (1-\delta_1)X_2(t,\beta_X)\\
		&\leq \frac{K_{\omega}}{\delta_1}\bigg[\frac{1}{\delta_1} + \frac{1}{\delta_1} -\frac{1}{1-\delta_1}\bigg] + \frac{K_{\omega}}{\delta_1^2} +  X_1(t,1-\delta_1) + (1-\delta_1)X_2(t,1-\delta_1)\\
		&\leq 3 K_{\omega}\delta_1^{-2} +  \overline{X}_1(\mathcal{F}_{t-1})+ (1-\delta_1) \overline{X}_2(\mathcal{F}_{t-1}).
	\end{align*}
	The bound for $\big\Vert \nabla c_t(\vtheta^c) \big\Vert$ follows similarly using \eqref{eq:1.2} instead of \eqref{eq:1.1}.
	
	As in \eqref{eq:hhh1}, it may be shown that
	\[
	\big\Vert \nabla^2 v_t(\vtheta^v) \big\Vert \leq M\big\Vert\nabla^2 \sigma_{X,t}(\vtheta^X) \big\Vert,
	\]
	such that the second-to-last bound again follows from Lemma~\ref{lem:8}. The final bound for $\big\Vert \nabla^2 c_t(\vtheta^c) \big\Vert$ follows similarly.
\end{proof}

In connection with Lemma~\ref{lem:9}, the following lemma---taken from the supplement to \citet{PZC19}---will be useful when verifying the moment conditions in Assumptions~\ref{ass:cons} and \ref{ass:an}.

\begin{lem}[{\citealt[Lemma~10]{PZC19}}]\label{lem:10}
	Let $\{X_i\}_{i\in\mathbb{N}}$ be a sequence of random variables and define $X=\sum_{i=1}^{\infty}a_iX_i$, where $a_i>0$ for all $i\in\mathbb{N}$ and $\sum_{i=1}^{\infty}a_i<\infty$. Then, if $\sup_{i\in\mathbb{N}}\E|X_i|^{p}\leq K$ for some constant $K>0$ and some $p>1$, then $\E|X|^{p}\leq \big(\sum_{i=1}^{\infty}a_i\big)^p K$.
\end{lem}

The next lemma establishes some serial dependence properties of the CCC--GARCH model. 

\begin{lem}\label{lem:sd}
	Under CCC--GARCH Assumptions~\ref{CGass:error}--\ref{CGass:mom}, the process
	\[
	\Big\{\big(X_t,\ Y_t,\ v_t(\vtheta^v),\ \nabla^\prime v_t(\vtheta^v),\ c_t(\vtheta^c),\ \nabla^\prime c_t(\vtheta^c) \big)^\prime\Big\}_{t\in\mathbb{Z}}
	\]
	is strictly stationary and $\beta$-mixing with exponential decay.
\end{lem}

\begin{proof}
	From \eqref{eq:lem51} in Lemma~\ref{lem:8},
	\begin{align*}
		\frac{\partial \sigma_{W,t}(\vtheta^W)}{\partial \omega_W} &= \frac{1}{1-\beta_W},\\
		\frac{\partial \sigma_{W,t}(\vtheta^W)}{\partial \alpha_W} &= W_{1}(t,\beta_W) \overset{\eqref{eq:lem51}}{=}\frac{1}{\alpha_W}\bigg[\sigma_{W,t}(\vtheta^W) - \frac{\omega_W}{1-\beta_W}\bigg].
	\end{align*}
	Therefore,
	\[
	\nabla \sigma_{W,t}(\vtheta^W) = \bigg(\frac{1}{1-\beta_W},\ \frac{1}{\alpha_W}\Big[\sigma_{W,t}(\vtheta^W) - \frac{\omega_W}{1-\beta_W}\Big],\ \frac{\partial \sigma_{W,t}(\vtheta^W)}{\partial \beta_W}\bigg)^\prime,\qquad W\in\{X,Y\}.
	\]
	Combining this with \eqref{eq:(3.11)}--\eqref{eq:(3.2)} and \eqref{eq:1.1}--\eqref{eq:1.2}, it follows for the generated $\sigma$-fields that
	\begin{multline*}
		\sigma\Big\{\big(X_t,\ Y_t,\ v_t(\vtheta^v),\ \nabla^\prime v_t(\vtheta^v),\ c_t(\vtheta^c),\ \nabla^\prime c_t(\vtheta^c) \big)^\prime\Big\}\\
		\subset\sigma\bigg\{\Big(X_t,\ Y_t,\ \sigma_{X,t}(\vtheta^X),\ \frac{\partial \sigma_{X,t}(\vtheta^X)}{\partial \beta_X},\ \sigma_{Y,t}(\vtheta^Y),\ \frac{\partial \sigma_{Y,t}(\vtheta^Y)}{\partial \beta_Y} \Big)^\prime\bigg\}.
	\end{multline*}
	Clearly, 
	\begin{equation}\label{eq:GHM}
		\bigg\{\Big(X_t,\ Y_t,\ \sigma_{X,t}(\vtheta^X),\ \frac{\partial \sigma_{X,t}(\vtheta^X)}{\partial \beta_X},\ \sigma_{Y,t}(\vtheta^Y),\ \frac{\partial \sigma_{Y,t}(\vtheta^Y)}{\partial \beta_Y} \Big)^\prime\bigg\}
	\end{equation}
	is a generalized hidden Markov model in the sense of \citet[Definition~3]{CC02} with hidden chain 
	\begin{equation}\label{eq:HC}
		\bigg\{\Big(\sigma_{X,t}(\vtheta^X),\ \frac{\partial \sigma_{X,t}(\vtheta^X)}{\partial \beta_X},\ \sigma_{Y,t}(\vtheta^Y),\ \frac{\partial \sigma_{Y,t}(\vtheta^Y)}{\partial \beta_Y} \Big)^\prime\bigg\}.
	\end{equation}
	Proposition~4 of \citet{CC02} then implies that if the process in \eqref{eq:HC} is strictly stationary and $\beta$-mixing, then the process in \eqref{eq:GHM} is also strictly stationary and $\beta$-mixing with mixing coefficients bounded by those of \eqref{eq:HC}.
	
	So in a first step we use Proposition~3 of \citet{CC02} to show that \eqref{eq:HC} is strictly stationary and $\beta$-mixing with exponential rate. Write
	\begin{equation}\label{eq:(4.1)}
		\begin{pmatrix}
			\sigma_{X,t}(\vtheta^X)\\ \frac{\partial \sigma_{X,t}(\vtheta^X)}{\partial \beta_X} \\ \sigma_{Y,t}(\vtheta^Y)\\ \frac{\partial \sigma_{Y,t}(\vtheta^Y)}{\partial \beta_Y}
		\end{pmatrix}
		=\begin{pmatrix}
			\omega_X\\ 0\\ \omega_Y\\ 0
		\end{pmatrix}
		+\underbrace{\begin{pmatrix}
				\alpha_X|\varepsilon_{X,t-1}| + \beta_X & 0 &0 &0 \\
				1 & \beta_X & 0& 0\\
				0 &0 & \alpha_Y|\varepsilon_{Y,t-1}| + \beta_Y &0 \\
				0 &0 & 1& \beta_Y
		\end{pmatrix}}_{=:\mA(\varepsilon_{X,t-1}, \varepsilon_{Y,t-1})}
		\begin{pmatrix}
			\sigma_{X,t-1}(\vtheta^X)\\ \frac{\partial \sigma_{X,t-1}(\vtheta^X)}{\partial \beta_X} \\ \sigma_{Y,t-1}(\vtheta^Y)\\ \frac{\partial \sigma_{Y,t-1}(\vtheta^Y)}{\partial \beta_Y}
		\end{pmatrix}.
	\end{equation}
	We now verify the assumptions of Proposition~3 of \citet{CC02}.
	
	\textbf{Condition (e):} Immediate from our CCC--GARCH Assumption~\ref{CGass:error}.
	
	\textbf{Condition ($A_0$):} Immediate from the definition of $\mA(\varepsilon_{X,t-1}, \varepsilon_{Y,t-1})$.
	
	\textbf{Condition ($A_1$):} The spectral radius (i.e., the largest eigenvalue in absolute value) of the matrix
	\[
	\mA(0,0)=\begin{pmatrix}
		\beta_X & 0 &0 &0 \\
		1 & \beta_X & 0& 0\\
		0 &0 & \beta_Y &0 \\
		0 &0 & 1& \beta_Y
	\end{pmatrix}
	\]
	is $\max\{\beta_X, \beta_Y\}$, since the eigenvalues of a triangular matrix equal the diagonal entries. By CCC--GARCH Assumption~\ref{CGass:param}, $\max\{\beta_X, \beta_Y\}<1$, as required.
	
	\textbf{Condition ($A_2^\prime$):} By the previous argument, the spectral radius of the matrix $\mA(\varepsilon_{X,t}, \varepsilon_{Y,t})$ is
	\[
	\rho\big(\mA(\varepsilon_{X,t}, \varepsilon_{Y,t})\big) = \max\big\{\alpha_X|\varepsilon_{X,t}|+\beta_X,\ \alpha_Y|\varepsilon_{Y,t}|+\beta_Y\big\}.
	\]
	Hence, by CCC--GARCH Assumption~\ref{CGass:param}, 
	\[
	\E\Big[\rho^2\big(\mA(\varepsilon_{X,t}, \varepsilon_{Y,t})\big)\Big] = \E\Big[\max\big\{\alpha_X|\varepsilon_{X,t}|+\beta_X,\ \alpha_Y|\varepsilon_{Y,t}|+\beta_Y\big\}^2\Big]<1,
	\]
	as desired.
	
	Proposition~3 now implies that if the recursion in \eqref{eq:(4.1)} is initialized in the infinite past (or, equivalently, initialized from the stationary distribution), then the process in \eqref{eq:HC} is strictly stationary and $\beta$-mixing with exponential decay. In that case, \citet[Proposition~4]{CC02} show that the process in \eqref{eq:GHM} is also strictly stationary and $\beta$-mixing with exponential decay. 
\end{proof}

\subsection{Verification of Assumption~\ref{ass:cons}}\label{Ass1}

We verify items~\ref{it:id}--\ref{it:mom bounds cons} of Assumption~\ref{ass:cons} in this subsection. 
In doing so, we heavily draw on \citet[Appendix~SA.2]{PZC19}.

\textbf{Assumption~\ref{ass:cons}~\ref{it:id}:} Recall from Section~\ref{PR} that 
\begin{equation*}
	\begin{pmatrix}
		\VaR_{t,\beta}\\
		\CoVaR_{t,\alpha\mid\beta}
	\end{pmatrix}
	= \begin{pmatrix}
		v_t(\vtheta_{0}^{v})\\
		c_t(\vtheta_{0}^{c})
	\end{pmatrix}
\end{equation*}
for $\vtheta_0^v=\vtheta_0^{X}\odot(v,v,1)^\prime$ and $\vtheta_0^c=\vtheta_0^{Y}\odot(v,v,1)^\prime$, such that correct specification of the (VaR, CoVaR) model is immediate.

\textbf{Assumption~\ref{ass:cons}~\ref{it:str stat}:} Direct from Lemma~\ref{lem:sd}. (Note that covariates $\mZ_t$ do not appear in our current example.)

\textbf{Assumption~\ref{ass:cons}~\ref{it:cond dist}:} Exploit \eqref{eq:CG} and the fact that $\sigma_{X,t}(\vtheta_0^X)\in\mathcal{F}_{t-1}$ (from \eqref{eq:lem51}) to write 
\begin{align}
	F_t(x,y) & =\P_{t-1}\big\{X_t\leq x,\ Y_t\leq y\big\}\notag\\
	&= \P_{t-1}\big\{\varepsilon_{X,t}\leq x/\sigma_{X,t}(\vtheta_0^X),\ \varepsilon_{Y,t}\leq y/\sigma_{Y,t}(\vtheta_0^Y)\big\}\notag\\
	&= F_{\varepsilon_X,\varepsilon_Y}\big(x/\sigma_{X,t}(\vtheta_0^X),\ y/\sigma_{Y,t}(\vtheta_0^Y)\big).\label{eq:(14.1)}
\end{align}
Since $\sigma_{W,t}(\vtheta_0^W)\geq\omega_{W,0}\geq\delta_1>0$ ($W\in\{X,Y\}$), the fact that $F_t(\cdot,\cdot)$ possesses a positive Lebesgue density follows from CCC--GARCH Assumption~\ref{CGass:error}.

\textbf{Assumption~\ref{ass:cons}~\ref{it:Lipschitz cons}:} Let $\sigma_{X,t}:=\sigma_{X,t}(\vtheta_0^X)$ and $\sigma_{Y,t}:=\sigma_{Y,t}(\vtheta_0^Y)$ for brevity.
From \eqref{eq:(4.11)} below and CCC--GARCH Assumption~\ref{CGass:error} it follows that
\begin{align*}
	\big|f_t^X(x)-f_t^X(x^\prime)\big| &=\frac{1}{\sigma_{X,t}} \big|f_{\varepsilon_X}(x/\sigma_{X,t}) - f_{\varepsilon_X}(x^\prime/\sigma_{X,t})\big|\\
	&\leq \frac{K}{\sigma_{X,t}}\bigg|\frac{x}{\sigma_{X,t}}-\frac{x^\prime}{\sigma_{X,t}}\bigg|=\frac{K}{\sigma_{X,t}^2}|x-x^\prime|\\
	&\leq \frac{K}{\omega_{X,0}^2}|x-x^\prime|\\
	&\leq \frac{K}{\delta_1^2}|x-x^\prime|.
\end{align*}
The Lipschitz continuity of $f_t^{Y}(\cdot)$ follows similarly.

From \eqref{eq:(14.1)}, we have that
\[
\partial_2 F_t(x,y)=\frac{1}{\sigma_{Y,t}}\partial_2 F_{\varepsilon_X,\varepsilon_Y}\Big(\frac{x}{\sigma_{X,t}}, \frac{y}{\sigma_{Y,t}}\Big),
\]
such that, by CCC--GARCH Assumption~\ref{CGass:error},
\begin{align*}
	\big|\partial_2 F_t(x,y) - \partial_2 F_t(x,y^\prime)\big| &=\frac{1}{\sigma_{Y,t}}\bigg|\partial_2 F_{\varepsilon_X,\varepsilon_Y}\Big(\frac{x}{\sigma_{X,t}}, \frac{y}{\sigma_{Y,t}}\Big) - \partial_2 F_{\varepsilon_X,\varepsilon_Y}\Big(\frac{x}{\sigma_{X,t}}, \frac{y^\prime}{\sigma_{Y,t}}\Big)\bigg|\\
	&\leq \frac{K}{\sigma_{Y,t}}\bigg|\frac{y}{\sigma_{Y,t}}-\frac{y^\prime}{\sigma_{Y,t}}\bigg|=\frac{K}{\sigma_{Y,t}^2}|y-y^\prime| \\
	&\leq \frac{K}{\omega_{Y,0}^2}|y-y^\prime|\leq \frac{K}{\delta_1^2}|y-y^\prime|.
\end{align*}

\textbf{Assumption~\ref{ass:cons}~\ref{it:cond dens}:}
Because $F_t^{X}(x)=F_{\varepsilon_X}\big(x/\sigma_{X,t}(\vtheta_0^X)\big)$ from \eqref{eq:(14.1)}, it holds that
\begin{equation}\label{eq:(4.11)}
	f_t^X(x) = \frac{1}{\sigma_{X,t}(\vtheta_0^X)}f_{\varepsilon_X}\big(x/\sigma_{X,t}(\vtheta_0^X)\big).
\end{equation}
By \eqref{eq:CG}, $\sigma_{X,t}(\vtheta_0^X)=\omega_{X,0} + \alpha_{X,0}|X_{t-1}| + \beta_{X,0}\sigma_{X,t-1}(\vtheta_0^X)\geq\omega_{X,0}\geq\delta_1>0$, such that with $f_{\varepsilon_X}(\cdot)\leq K$ from CCC--GARCH Assumption~\ref{CGass:error} we obtain that
\[
f_t^X(x) \leq\frac{K}{\delta_1}.
\]

Moreover, choose $f_1=f_{\varepsilon_X}\big(\VaR_{\beta}(F_{\varepsilon_{X}})\big) / \VaR_{\beta_{\sigma}}\big(\sigma_{X,t}(\vtheta_0^X)\big)>0$ for some $\beta_{\sigma}\in(0,1)$, where $\VaR_{\beta_{\sigma}}\big(\sigma_{X,t}(\vtheta_0^X)\big)$ denotes the $\beta_{\sigma}$-quantile of the stationary distribution of $\sigma_{X,t}(\vtheta_0^X)$ (see the proof of Lemma~\ref{lem:sd}).
Note that because $\sigma_{X,t}(\vtheta_0^X)\geq\omega_{X,0}>0$, the $\beta_{\sigma}$-quantile $\VaR_{\beta_{\sigma}}\big(\sigma_{X,t}(\vtheta_0^X)\big)$ is necessarily positive.
Also, $f_{\varepsilon_X}\big(\VaR_{\beta}(F_{\varepsilon_{X}})\big)>0$ by CCC--GARCH Assumption~\ref{CGass:error}; specifically the assumption that $f_{\varepsilon_X,\varepsilon_Y}(x,y)>0$ for all $(x,y)^\prime\in\mathbb{R}^2$ such that $F_{\varepsilon_X,\varepsilon_Y}(x,y)\in(0,1)$.
Then, we obtain that
\begin{eqnarray*}
	\P\Big\{f_t^{X}\big(v_t(\vtheta_0^v)\big)> f_1\Big\} &\overset{\eqref{eq:(4.11)}}{=} & \P\bigg\{f_{\varepsilon_X}\Big(\frac{v_t(\vtheta_0^v)}{\sigma_{X,t}(\vtheta_0^X)}\Big)> \sigma_{X,t}(\vtheta_0^X)f_1\bigg\}\\
	&\overset{\eqref{eq:(3.11)}}{=}&  \P\Big\{f_{\varepsilon_X}\big(\VaR_{\beta}(F_{\varepsilon_{X}})\big)> \sigma_{X,t}(\vtheta_0^X)f_1\Big\}\\
	&=& \P\Big\{\sigma_{X,t}(\vtheta_0^X)<\VaR_{\beta_{\sigma}}\big(\sigma_{X,t}(\vtheta_0^X)\big)\Big\}\\
	&=& \beta_{\sigma}>0.
\end{eqnarray*}

Similarly, set
\[
f_2=\frac{\int_{\VaR_{\beta}(F_{\varepsilon_{X}})}^{\infty} f_{\varepsilon_X,\varepsilon_{Y}}\big(x, \CoVaR_{\alpha\mid\beta}(F_{\varepsilon_{X}, \varepsilon_{Y}})\big)\D x}{\VaR_{\beta_{\sigma}}\big(\sigma_{Y,t}(\vtheta_0^Y)\big)}>0,
\]
which is positive by CCC--GARCH Assumption~\ref{CGass:error}.
Then, we obtain that
\begin{equation*}
	\P\bigg\{\int_{v_t(\vtheta_0^v)}^{\infty}f_t\big(x, c_t(\vtheta_0^c)\big)\D x> f_2\bigg\}= \beta_{\sigma}>0,
\end{equation*}
since
\begin{align*}
	\int_{v_t(\vtheta_0^v)}^{\infty}f_t\big(x, c_t(\vtheta_0^c)\big)\D x &=\int_{v_t(\vtheta_0^v)}^{\infty}\frac{1}{\sigma_{X,t}(\vtheta_0^X)\sigma_{Y,t}(\vtheta_0^Y)}f_{\varepsilon_X, \varepsilon_Y}\big(x/\sigma_{X,t}(\vtheta_0^X), c_t(\vtheta_0^c)/\sigma_{Y,t}(\vtheta_0^Y)\big)\D x \\
	&= \int_{v_t(\vtheta_0^v)/\sigma_{X,t}(\vtheta_0^X)}^{\infty}\frac{1}{\sigma_{Y,t}(\vtheta_0^Y)}f_{\varepsilon_X, \varepsilon_Y}\big(x, c_t(\vtheta_0^c)/\sigma_{Y,t}(\vtheta_0^Y)\big)\D x\\
	&= \frac{1}{\sigma_{Y,t}(\vtheta_0^Y)} \int_{\VaR_{\beta}(F_{\varepsilon_{X}})}^{\infty}f_{\varepsilon_X, \varepsilon_Y}\big(x, \CoVaR_{\alpha\mid\beta}(F_{\varepsilon_{X}, \varepsilon_{Y}})\big)\D x,
\end{align*}
where the first equality uses \eqref{eq:(14.1)}, the second exploits an obvious substitution, and the last one \eqref{eq:(3.11)} and \eqref{eq:(3.2)}.

\textbf{Assumption~\ref{ass:cons}~\ref{it:compact}:} The compactness of the parameter space directly follows from CCC--GARCH Assumption~\ref{CGass:param} in combination with \eqref{eq:(D.6p)}.

\textbf{Assumption~\ref{ass:cons}~\ref{it:ULLN}:} By Theorem~21.9 in \citet{Dav94} it suffices to show that
\begin{itemize}
	\item[(a)] $\mS\Big(\big(v_t(\vtheta^v), c_t(\vtheta^c)\big)^\prime, (X_t, Y_t)^\prime\Big)$ obeys a pointwise law of large numbers;
	\item[(b)] $\mS\Big(\big(v_t(\vtheta^v), c_t(\vtheta^c)\big)^\prime, (X_t, Y_t)^\prime\Big)$ is stochastically equicontinuous (s.e.).
\end{itemize}

We first prove (a). By Lemma~\ref{lem:sd} and standard mixing inequalities, $\mS\Big(\big(v_t(\vtheta^v), c_t(\vtheta^c)\big)^\prime, (X_t, Y_t)^\prime\Big)$ is strictly stationary and $\alpha$-mixing of any size and, hence, ergodic by \citet[Proposition~3.44]{Whi01}. The ergodic theorem \citep[e.g.,][Theorem~3.34]{Whi01} then implies (a) if the (componentwise) moments of $\mS\Big(\big(v_t(\vtheta^v), c_t(\vtheta^c)\big)^\prime, (X_t, Y_t)^\prime\Big)$ exist. To show this, use Lemma~\ref{lem:9} to obtain that
\begin{align*}
	\E\big|S^{\VaR}\big(v_t(\vtheta^v),X_t\big)\big|&\leq \E|X_t| + \E|v_t(\vtheta^v)|\leq \E|X_t| + \E\big[V(\mathcal{F}_{t-1})\big]\leq C<\infty,\\
	\E\Big|S^{\CoVaR}\Big(\big(v_t(\vtheta^v), c_t(\vtheta^c)\big)^\prime, (X_t, Y_t)^\prime\Big)\Big|&\leq \E|Y_t| + \E|c_t(\vtheta^c)|\leq \E|Y_t| + \E\big[C(\mathcal{F}_{t-1})\big]\leq C<\infty,
\end{align*}
where we used that $\E\big[V(\mathcal{F}_{t-1})\big]\leq K$ and $\E\big[C(\mathcal{F}_{t-1})\big]\leq K$ from the proof of Assumption~\ref{ass:cons}~\ref{it:mom bounds cons} below, and the fact that $\E|X_t|$ and $\E|Y_t|$ are finite by CCC--GARCH Assumption~\ref{CGass:mom}. Thus, the (componentwise) moments of $\mS\Big(\big(v_t(\vtheta^v), c_t(\vtheta^c)\big)^\prime, (X_t, Y_t)^\prime\Big)$ exist and the pointwise law of large numbers follows.

To establish (b) for the first component of $\mS$, it suffices to show that
\[
\bigg|\frac{1}{n}\sum_{t=1}^{n}S^{\VaR}\big(v_t(\vtheta^v),X_t\big) - S^{\VaR}\big(v_t(\widetilde{\vtheta}^v),X_t\big)\bigg|\leq B_n\big\Vert\vtheta^v-\widetilde{\vtheta}^v\big\Vert\qquad\text{for all }\vtheta^v,\widetilde{\vtheta}^v\in\mTheta^v,
\]
where $B_n>0$ is a stochastic sequence not depending on $\vtheta^v$ and $\widetilde{\vtheta}^v$, with $B_n=O_{\P}(1)$ \citep[Theorem~21.10]{Dav94}.
Write
\begin{align}
	& \big|S^{\VaR}\big(v_t(\vtheta^v), X_t\big) - S^{\VaR}\big(v_t(\widetilde{\vtheta}^v), X_t\big)\big|\notag\\
	& =\Big|\big[\1_{\{X_t\leq v_t(\vtheta^v)\}}-\beta\big]\big[v_t(\vtheta^v) - X_t\big] - \big[\1_{\{X_t\leq v_t(\widetilde{\vtheta}^v)\}}-\beta\big]\big[v_t(\widetilde{\vtheta}^v) - X_t\big]\Big|\notag\\
	& = \Big|\1_{\{X_t\leq v_t(\vtheta^v)\}}\big[v_t(\vtheta^v) - X_t\big] - \1_{\{X_t\leq v_t(\widetilde{\vtheta}^v)\}}\big[v_t(\widetilde{\vtheta}^v) - X_t\big] + \beta\big[v_t(\widetilde{\vtheta}^v) - v_t(\vtheta^v)\big]\Big|\notag\\
	& \leq \Big|\1_{\{X_t\leq v_t(\vtheta^v)\}}\big[v_t(\vtheta^v) - X_t\big] - \1_{\{X_t\leq v_t(\widetilde{\vtheta}^v)\}}\big[v_t(\widetilde{\vtheta}^v) - X_t\big]\Big| + \beta\big|v_t(\widetilde{\vtheta}^v) - v_t(\vtheta^v)\big|\notag\\
	&= A_{2t} + B_{2t}.\label{eq:(5.11)}
\end{align}
Consider the two terms separately. First, by Lemma~\ref{lem:9},
\[
B_{2t} \leq \beta\sup_{\vtheta^v\in\mTheta^v}\big\Vert\nabla v_t(\vtheta^v)\big\Vert\cdot \big\Vert\vtheta^v-\widetilde{\vtheta}^v\big\Vert\leq \beta V_1(\mathcal{F}_{t-1}) \big\Vert\vtheta^v-\widetilde{\vtheta}^v\big\Vert.
\]
Second,
\begin{align*}
	A_{2t} &= \frac{1}{2}\Big|v_t(\vtheta^v) - X_t + \big|v_t(\vtheta^v) - X_t\big| - \big\{v_t(\widetilde{\vtheta}^v) - X_t + \big|v_t(\widetilde{\vtheta}^v) - X_t\big|\big\}\Big|\\
	& \leq \frac{1}{2}\Big|v_t(\vtheta^v) - X_t - \big\{v_t(\widetilde{\vtheta}^v) - X_t\big\}\Big| + \frac{1}{2}\Big| \big|v_t(\vtheta^v) - X_t\big| - \big|v_t(\widetilde{\vtheta}^v) - X_t\big|\Big|\\
	& \leq \big|v_t(\vtheta^v) - v_t(\widetilde{\vtheta}^v)\big|\\
	&\leq \sup_{\vtheta^v\in\mTheta^v}\big\Vert\nabla v_t(\vtheta^v)\big\Vert\cdot \big\Vert\vtheta^v-\widetilde{\vtheta}^v\big\Vert\\
	&\leq V_1(\mathcal{F}_{t-1})\big\Vert\vtheta^v-\widetilde{\vtheta}^v\big\Vert,
\end{align*}
where we used Lemma~\ref{lem:9} again in the last step.

Combining the bounds for $A_{2t}$ and $B_{2t}$, we obtain from \eqref{eq:(5.11)} that
\begin{align*}
	\bigg|\frac{1}{n}\sum_{t=1}^{n} S^{\VaR}\big(v_t(\vtheta^v), X_t\big) - S^{\VaR}\big(v_t(\widetilde{\vtheta}^v), X_t\big)\bigg| & \leq\frac{1+\beta}{n}\sum_{t=1}^{n}V_1(\mathcal{F}_{t-1}) \big\Vert\vtheta^v-\widetilde{\vtheta}^v\big\Vert\\
	&=:B_n\big\Vert\vtheta^v-\widetilde{\vtheta}^v\big\Vert.
\end{align*}
From Assumption~\ref{ass:cons}~\ref{it:mom bounds cons} (which we prove below),
\[
\sup_{n\in\mathbb{N}}\E[B_n]= \frac{1+\beta}{n}\sum_{t=1}^{n}\E\big[V_1(\mathcal{F}_{t-1})\big]\leq (1+\beta)K<\infty,
\]
such that $B_n=O_{\P}(1)$ \citep[Theorem~12.11]{Dav94}. It follows that the first component is s.e.

It remains to prove that the second component of $\mS\Big(\big(v_t(\vtheta^v), c_t(\vtheta^c)\big)^\prime, (X_t, Y_t)^\prime\Big)$ is s.e., i.e., for all $\varepsilon>0$ there exists some $\delta>0$, such that
\begin{multline*}
	\limsup_{n\to\infty} \P\bigg\{\sup_{\vtheta\in\mTheta}\sup_{\Vert\widetilde{\vtheta}-\vtheta\Vert\leq\delta} \frac{1}{n}\sum_{t=1}^{n}\Big|S^{\CoVaR}\Big(\big(v_t(\vtheta^v), c_t(\vtheta^c)\big)^\prime, (X_t, Y_t)^\prime\Big) \\
	- S^{\CoVaR}\Big(\big(v_t(\widetilde{\vtheta}^v), c_t(\widetilde{\vtheta}^c)\big)^\prime, (X_t, Y_t)^\prime\Big)\Big|\geq\varepsilon\bigg\}<\varepsilon.
\end{multline*}
Write
\begin{align*}
	&\Big|S^{\CoVaR}\Big(\big(v_t(\vtheta^v), c_t(\vtheta^c)\big)^\prime, (X_t, Y_t)^\prime\Big) - S^{\CoVaR}\Big(\big(v_t(\widetilde{\vtheta}^v), c_t(\widetilde{\vtheta}^c)\big)^\prime, (X_t, Y_t)^\prime\Big)\Big|\\
	&= \Big|\1_{\{X_t> v_t(\vtheta^v)\}} \big(\1_{\{Y_t\leq c_t(\vtheta^c)\}}-\alpha\big)\big(c_t(\vtheta^c) - Y_t\big) - \1_{\{X_t> v_t(\widetilde{\vtheta}^v)\}} \big(\1_{\{Y_t\leq c_t(\widetilde{\vtheta}^c)\}}-\alpha\big)\big(c_t(\widetilde{\vtheta}^c) - Y_t\big)\Big|\\
	&\leq \Big|\1_{\{X_t> v_t(\vtheta^v)\}} \Big[\big(\1_{\{Y_t\leq c_t(\vtheta^c)\}}-\alpha\big)\big(c_t(\vtheta^c) - Y_t\big) - \big(\1_{\{Y_t\leq c_t(\widetilde{\vtheta}^c)\}}-\alpha\big)\big(c_t(\widetilde{\vtheta}^c) - Y_t\big)\Big]\Big|\\
	&\hspace{2cm} + \Big|\big(\1_{\{X_t> v_t(\vtheta^v)\}}- \1_{\{X_t> v_t(\widetilde{\vtheta}^v)\}}\big)  \big(\1_{\{Y_t\leq c_t(\widetilde{\vtheta}^c)\}}-\alpha\big)\big(c_t(\widetilde{\vtheta}^c) - Y_t\big)\Big|\\
	&=:A_{3t} + B_{3t}.
\end{align*} 
We show that both terms are s.e.~individually. For the first, we have that
\begin{align*}
	A_{3t}&\leq \Big|\big(\1_{\{Y_t\leq c_t(\vtheta^c)\}}-\alpha\big)\big(c_t(\vtheta^c) - Y_t\big) - \big(\1_{\{Y_t\leq c_t(\widetilde{\vtheta}^c)\}}-\alpha\big)\big(c_t(\widetilde{\vtheta}^c) - Y_t\big)\Big|\\
	&= \big|S^{\VaR}\big(c_t(\vtheta^c), Y_t\big) - S^{\VaR}\big(c_t(\widetilde{\vtheta}^c), Y_t\big)\big|,
\end{align*}
such that stochastic equicontinuity follows as above. 

For the second term, we obtain using Markov's inequality that
\begin{align*}
	&\P\bigg\{\sup_{\vtheta\in\mTheta}\sup_{\Vert\widetilde{\vtheta}-\vtheta\Vert\leq\delta}\frac{1}{n}\sum_{t=1}^{n}B_{3t}\geq\varepsilon\bigg\} \\
	&\leq \frac{1}{\varepsilon}\E\bigg[\sup_{\vtheta\in\mTheta}\sup_{\Vert\widetilde{\vtheta}-\vtheta\Vert\leq\delta}\frac{1}{n}\sum_{t=1}^{n}B_{3t}\bigg]\\
	&= \frac{1}{\varepsilon}\E\bigg[\sup_{\vtheta\in\mTheta}\sup_{\Vert\widetilde{\vtheta}-\vtheta\Vert\leq\delta}\frac{1}{n}\sum_{t=1}^{n}\Big|\big(\1_{\{X_t> v_t(\vtheta^v)\}}- \1_{\{X_t> v_t(\widetilde{\vtheta}^v)\}}\big)\big(\1_{\{Y_t\leq c_t(\widetilde{\vtheta}^c)\}}-\alpha\big)\big(c_t(\widetilde{\vtheta}^c) - Y_t\big)\Big|\bigg]\\
	&\underset{(\delta\downarrow0)}{\longrightarrow}\frac{1}{\varepsilon}\E\bigg[\sup_{\vtheta\in\mTheta}\frac{1}{n}\sum_{t=1}^{n}\Big|\1_{\{X_t= v_t(\vtheta^v)\}}\big(\1_{\{Y_t\leq c_t(\vtheta^c)\}}-\alpha\big)\big(c_t(\vtheta^c) - Y_t\big)\Big|\bigg]\\
	&=O(n^{-1})
\end{align*}
from Assumption~\ref{ass:an}~\ref{it:eq bound} (which we prove below). Therefore, the second term is also s.e., such that (b) follows.

\textbf{Assumption~\ref{ass:cons}~\ref{it:unique id}:} We first show that
\begin{equation}\label{eq:vt equiv}
	v_t(\vtheta^v)\overset{\text{a.s.}}{=}v_t(\vtheta_0^v)\text{ for all }t\in\mathbb{N}\quad\Longleftrightarrow\quad \vtheta^v=\vtheta_0^v.
\end{equation}
Since the direction ``$\Longleftarrow$'' is immediate, we only show the other direction.
Hence, suppose that $\P\big\{v_t(\vtheta^v)=v_t(\vtheta_0^v)\big\}=1$ for all $t\in\mathbb{N}$. 
This implies
\begin{align}
	& \P\big\{v_t(\vtheta^v)=v_t(\vtheta_0^v)\quad \forall t\in\mathbb{N}\big\} = 1\notag\\
	\overset{\eqref{eq:(3.11)}}{\Longrightarrow}\quad & \P\big\{\sigma_{X,t}(\vtheta^X)=\sigma_{X,t}(\vtheta_0^X)\quad \forall t\in\mathbb{N}\big\} = 1\label{eq:vola eq}\\
	\overset{\eqref{eq:CG}}{\Longrightarrow}\quad& \P\big\{\omega_{X} + \alpha_{X}|X_{t-1}| + \beta_{X}\sigma_{X,t-1}(\vtheta^X)=\omega_{X,0} + \alpha_{X,0}|X_{t-1}| + \beta_{X,0}\sigma_{X,t-1}(\vtheta_0^X)\quad \forall t\in\mathbb{N}\big\} = 1\notag\\
	\overset{\eqref{eq:vola eq}}{\Longrightarrow}\quad& \P\big\{\omega_{X} + \alpha_{X}|X_{t-1}| + \beta_{X}\sigma_{X,t-1}(\vtheta_0^X)=\omega_{X,0} + \alpha_{X,0}|X_{t-1}| + \beta_{X,0}\sigma_{X,t-1}(\vtheta_0^X)\quad \forall t\in\mathbb{N}\big\} = 1\notag\\
	\Longrightarrow\quad&\P\big\{\alpha_X = \alpha_{X,0}\ \wedge\ \omega_{X} + \beta_{X}\sigma_{X,t-1}(\vtheta_0^X)=\omega_{X,0} + \beta_{X,0}\sigma_{X,t-1}(\vtheta_0^X)\quad \forall t\in\mathbb{N}\big\} = 1\notag\\
	\Longrightarrow\quad& \alpha_{X}=\alpha_{X,0}\ \wedge\ \ \P\big\{\omega_{X} + \beta_{X}\sigma_{X,t-1}(\vtheta_0^X)=\omega_{X,0} + \beta_{X,0}\sigma_{X,t-1}(\vtheta_0^X)\quad \forall t\in\mathbb{N}\big\} = 1,\label{eq:vola eq2}
\end{align}
where the penultimate line follows from the fact that $X_{t-1}\mid \sigma_{X,t-1}(\vtheta_0^X)\sim F_{\varepsilon_X}\big(\cdot / \sigma_{X,t-1}(\vtheta_0^X)\big)$.

We now show that $\P\big\{\omega_{X} + \beta_{X}\sigma_{X,t-1}(\vtheta_0^X)=\omega_{X,0} + \beta_{X,0}\sigma_{X,t-1}(\vtheta_0^X)\quad \forall t\in\mathbb{N}\big\} = 1$ implies $\beta_X=\beta_{X,0}$. Proceed by contradiction and assume $\beta_{X}\neq\beta_{X,0}$. Then,
\begin{align*}
	1 &= \P\bigg\{\sigma_{X,t-1}(\vtheta_0^X) = \frac{\omega_{X,0}-\omega_X}{\beta_X-\beta_{X,0}}\quad \forall t\in\mathbb{N}\bigg\}\\
	&=\P\bigg\{\omega_{X,0} + \alpha_{X,0}|X_{t-2}| + \beta_{X,0}\sigma_{X,t-2}(\vtheta_0^X) = \frac{\omega_{X,0}-\omega_X}{\beta_X-\beta_{X,0}}\quad \forall t\in\mathbb{N}\bigg\}.
\end{align*}
This, however, contradicts the fact that $X_{t-2}\mid\sigma_{X,t-2}(\vtheta_0^X)\sim F_{\varepsilon_X}\big(\cdot / \sigma_{X,t-2}(\vtheta_0^X)\big)$. We conclude that $\beta_X=\beta_{X,0}$.

It remains to show that $\omega_{X}=\omega_{X,0}$. From $\beta_X=\beta_{X,0}$ and \eqref{eq:vola eq2}, it follows that
\[
\P\big\{\omega_{X} + \beta_{X,0}\sigma_{X,t-1}(\vtheta_0^X)=\omega_{X,0} + \beta_{X,0}\sigma_{X,t-1}(\vtheta_0^X)\quad \forall t\in\mathbb{N}\big\} = 1
\]
and, therefore, $\omega_X=\omega_{X,0}$.

Overall, we have obtained that $\vtheta^X=\vtheta_0^X$ and, hence, $\vtheta^v=\vtheta^{X}\odot(v,v,1)^\prime=\vtheta_0^{X}\odot(v,v,1)^\prime=\vtheta_0^v$, as desired.

It may be shown similarly that 
\begin{equation}\label{eq:ct equiv}
	c_t(\vtheta^c)\overset{\text{a.s.}}{=}c_t(\vtheta_0^c)\text{ for all }t\in\mathbb{N}\quad\Longleftrightarrow\quad \vtheta^c=\vtheta_0^c.
\end{equation}
The only difference is that one uses \eqref{eq:(3.2)} instead of \eqref{eq:(3.11)} to obtain the analog of \eqref{eq:vola eq}, i.e., $\P\big\{\sigma_{Y,t}(\vtheta^Y)=\sigma_{Y,t}(\vtheta_0^Y)\quad \forall t\in\mathbb{N}\big\} = 1$. The remainder of the proof follows along similar lines.

With the preliminary results in \eqref{eq:vt equiv} and \eqref{eq:ct equiv} at hand, we now proceed to establish (a) and (b).
In fact, we only prove (a), as the proof of (b) is analogous.
We proceed by contradiction and assume that there exists some $\xi>0$ such that for all $\tau>0$ it follows that
\[
\liminf_{n\to\infty}\frac{1}{n}\sum_{t=1}^{n}\P\big\{|v_t(\vtheta^v) - v_t(\vtheta_0^v)|>\tau\mid f_t^X\big(v_t(\vtheta_0^v)\big)>f_1\big\}=0
\]
whenever $\big\Vert\vtheta^v - \vtheta_0^v\big\Vert\geq\xi$.
Since all quantities involved in the probability are stationary by Lemma~\ref{lem:sd}, it necessarily follows that $\P\big\{|v_t(\vtheta^v) - v_t(\vtheta_0^v)|>\tau\mid f_t^X\big(v_t(\vtheta_0^v)\big)>f_1\big\}=0$ for all $t\in\mathbb{N}$ and all $\tau>0$.
This implies that $\P\big\{v_t(\vtheta^v) = v_t(\vtheta_0^v)\mid f_t^X\big(v_t(\vtheta_0^v)\big)>f_1\big\}=1$ for all $t\in\mathbb{N}$.
However, we know from \eqref{eq:vt equiv} that $v_t(\vtheta^v) = v_t(\vtheta_0^v)$ for $t\in\mathbb{N}$ is equivalent to $\vtheta^v=\vtheta_0^v$, whence $\P\big\{\vtheta^v = \vtheta_0^v\mid f_t^X\big(v_t(\vtheta_0^v)\big)>f_1\big\}=1$ for all $t\in\mathbb{N}$.
For this to be true, it necessarily has to hold that $\vtheta^v=\vtheta_0^v$, contradicting the fact that $\big\Vert\vtheta^v - \vtheta_0^v\big\Vert\geq\xi$.

\textbf{Assumption~\ref{ass:cons}~\ref{it:smooth}:} Combining \eqref{eq:(3.11)} and \eqref{eq:(3.2)} with \eqref{eq:lem51} gives
\begin{align}
	v_t(\vtheta^v) &= v\bigg[\frac{\omega_X}{1-\beta_X} + \alpha_X X_1(t,\beta_X)\bigg]\overset{\eqref{eq:para trafo 1}}{=}\frac{\omega_v}{1-\beta_v} + \alpha_v X_1(t,\beta_v),\label{eq:nice rep1}\\
	c_t(\vtheta^c) &= c\bigg[\frac{\omega_Y}{1-\beta_Y} + \alpha_Y Y_1(t,\beta_Y)\bigg]\overset{\eqref{eq:para trafo 2}}{=}\frac{\omega_c}{1-\beta_c} + \alpha_c Y_1(t,\beta_c).\label{eq:nice rep2}
\end{align}
From these representations it is obvious that $v_t(\vtheta^v)$ and $c_t(\vtheta^{c})$ are $\mathcal{F}_{t-1}$-measurable and a.s.\ continuous in $\vtheta^v$ and $\vtheta^{c}$, respectively.

\textbf{Assumption~\ref{ass:cons}~\ref{it:diff c}:} The continuous differentiability of $v_t(\vtheta^v)$ follows directly from \eqref{eq:nice rep1}.

\textbf{Assumption~\ref{ass:cons}~\ref{it:bound}:} The bounds can be obtained easily from Lemma~\ref{lem:9}.

\textbf{Assumption~\ref{ass:cons}~\ref{it:mom bounds cons}:} The conditions $\E|X_t|\leq K$ and $\E|Y_t|^{1+\iota}\leq K$ are immediate from CCC--GARCH Assumption~\ref{CGass:mom}.

By the Cauchy--Schwarz inequality, the remaining moment conditions hold if we can show that $\E\big[V(\mathcal{F}_{t-1})\big]\leq K$, $\E\big[V_1^2(\mathcal{F}_{t-1})\big]\leq K$ and $\E\big[C^2(\mathcal{F}_{t-1})\big]\leq K$. From Lemmas~\ref{lem:9} and~\ref{lem:10}, these conditions are satisfied when $\E|X_t|^2\leq K$ and $\E|Y_t|^2\leq K$. This is true by CCC--GARCH Assumption~\ref{CGass:mom}.

\subsection{Verification of Assumption~\ref{ass:an}}\label{Ass2}

We verify items~\ref{it:int}--\ref{it:mixing} of Assumption~\ref{ass:an} in this subsection. 

\textbf{Assumption~\ref{ass:an}~\ref{it:int}:} Follows directly from CCC--GARCH Assumption~\ref{CGass:param} and \eqref{eq:(D.6p)}.

\textbf{Assumption~\ref{ass:an}~\ref{it:diff}:} That $v_t(\cdot)$ and $c_t(\cdot)$ are a.s.\ twice continuously differentiable can be seen from \eqref{eq:nice rep1}--\eqref{eq:nice rep2}. 

Also, the gradients are different from $\vzero$ because from \eqref{eq:para trafo 1}--\eqref{eq:para trafo 2}, \eqref{eq:1.1}--\eqref{eq:1.2} and \eqref{eq:(8.2)}
\[
\nabla v_t(\vtheta^v)=\begin{pmatrix}\frac{1}{1-\beta_v}\\ X_1(t,\beta_v)\\ \frac{\omega_v}{(1-\beta_v)^2}+\alpha_v X_2(t,\beta_v)\end{pmatrix}\qquad\text{and}\qquad \nabla c_t(\vtheta^c)=\begin{pmatrix}\frac{1}{1-\beta_c}\\ Y_1(t,\beta_c)\\ \frac{\omega_c}{(1-\beta_c)^2}+\alpha_c Y_2(t,\beta_c)\end{pmatrix}.
\]

\textbf{Assumption~\ref{ass:an}~\ref{it:bound2.1}:} That $\norm{\nabla^2v_t(\vtheta^v)}\leq V_2(\mathcal{F}_{t-1})$ is immediate from Lemma~\ref{lem:9}. 

To prove the remaining Lipschitz-type condition, we first show that
\begin{equation}\label{eq:(10.1)}
	\big\Vert\nabla^2 \sigma_{X,t}(\vtheta^X) - \nabla^2 \sigma_{X,t}(\widetilde{\vtheta}^X)\big\Vert\leq \Sigma_X(\mathcal{F}_{t-1})\big\Vert\vtheta^X-\widetilde{\vtheta}^X\big\Vert
\end{equation}
for some $\Sigma_X(\mathcal{F}_{t-1})$.
Since the Frobenius norm is less than the sum of the absolute values of the matrix entries, we obtain from \eqref{eq:(p.48)} that
\begin{align*}
	\big\Vert\nabla^2 \sigma_{X,t}(\vtheta^X) - \nabla^2 \sigma_{X,t}(\widetilde{\vtheta}^X)\big\Vert &\leq 2\bigg|\frac{1}{(1-\beta_X)^2} - \frac{1}{(1-\widetilde{\beta}_X)^2}\bigg| + 2\big|X_2(t,\beta_X)- X_2(t,\widetilde{\beta}_X)\big| \\
	&\hspace{1cm} + \big|\alpha_X X_3(t,\beta_X) - \widetilde{\alpha}_X X_3(t,\widetilde{\beta}_X)\big| + 2\bigg|\frac{\omega_X}{(1-\beta_X)^3} - \frac{\widetilde{\omega}_X}{(1-\widetilde{\beta}_X)^3}\bigg|\\
	&=:A_{4t} + B_{4t} + C_{4t} + D_{4t}.
\end{align*}
In the following, we use extensively that
\[
\big|f(x)-f(y)\big|\leq\sup_{x\in[a,b]}\big|f^\prime(x)\big|\cdot\big|x-y\big|,\qquad a\leq x,y\leq b,
\]
for some differentiable (real-valued) function $f(\cdot)$.

For the first term, we get that
\[
A_{4t}\leq 2\sup_{\delta_1\leq\beta_X\leq 1-\delta_1}\bigg|\frac{2}{(1-\beta_X)^3}\bigg|\cdot\big|\beta_X-\widetilde{\beta}_X\big|=\frac{4}{\delta_1^3}\big|\beta_X-\widetilde{\beta}_X\big|.
\]
The second term may be bounded by
\[
B_{4t}\leq 2\sup_{\delta_1\leq\beta_X\leq 1-\delta_1}\big|X_3(t,\beta_X)\big|\cdot\big|\beta_X-\widetilde{\beta}_X\big|=2X_3(t,1-\delta_1)\big|\beta_X-\widetilde{\beta}_X\big|
\]
and the third by
\begin{align*}
	C_{4t} &\leq\big|\alpha_X[X_3(t,\beta_X) - X_3(t,\widetilde{\beta}_X)]\big| + \big|(\alpha_X-\widetilde{\alpha}_X)X_3(t,\widetilde{\beta}_X)\big|\\
	&\leq (1-\delta_1)\sup_{\delta_1\leq\beta_X\leq 1-\delta_1}\big|X_4(t,\beta_X)\big|\cdot \big|\beta_X-\widetilde{\beta}_X\big| + \sup_{\delta_1\leq\beta_X\leq 1-\delta_1}\big|X_3(t,\beta_X)\big|\cdot\big|\alpha_X-\widetilde{\alpha}_X\big|\\
	&=(1-\delta_1)X_4(t,1-\delta_1) \big|\beta_X-\widetilde{\beta}_X\big| + X_3(t,1-\delta_1)\big|\alpha_X-\widetilde{\alpha}_X\big|,
\end{align*}
where 
\begin{equation}\label{eq:(11.1)}
	X_4(t,\beta_X) = \sum_{i=4}^{\infty}(i-1)(i-2)(i-3)\beta_X^{i-4}|X_{t-i}|.
\end{equation}
Finally, 
\begin{align*}
	D_{4t} &\leq2\bigg|\frac{\omega_X-\widetilde{\omega}_X}{(1-\beta_X)^3} - \widetilde{\omega}_X\Big[\frac{1}{(1-\widetilde{\beta}_X)^3}-\frac{1}{(1-\beta_X)^3}\Big]\bigg|\\
	&\leq 2\sup_{\delta_1\leq\beta_X\leq 1-\delta_1}\bigg|\frac{1}{(1-\beta_X)^3}\bigg|\cdot\big|\omega_X-\widetilde{\omega}_X\big| + 2K_{\omega}\sup_{\delta_1\leq\beta_X\leq 1-\delta_1}\bigg|\frac{3}{(1-\beta_X)^4}\bigg|\cdot\big|\beta_X-\widetilde{\beta}_X\big|\\
	& =\frac{2}{\delta_1^3}\big|\omega_X-\widetilde{\omega}_X\big| + 6\frac{K_{\omega}}{\delta_1^4}\big|\beta_X-\widetilde{\beta}_X\big|
\end{align*}
Combining the inequalities for $A_{4t}$, $B_{4t}$, $C_{4t}$ and $D_{4t}$, the inequality in \eqref {eq:(10.1)} follows.

From \eqref{eq:(3.11)}, \eqref{eq:(10.1)} and norm equivalence, we obtain that
\begin{align}
	\big\Vert\nabla^2 v_t(\vtheta^v) - \nabla^2 v_t(\widetilde{\vtheta}^v)\big\Vert&\leq C\big[1+X_3(t,1-\delta_1) + X_4(t,1-\delta_1)\big]\big\Vert\vtheta^v-\widetilde{\vtheta}^v\big\Vert\notag\\
	&=:V_3(\mathcal{F}_{t-1}) \big\Vert\vtheta^v-\widetilde{\vtheta}^v\big\Vert,\label{eq:(12.1)}
\end{align}
which is the desired result.

\textbf{Assumption~\ref{ass:an}~\ref{it:bound2.2}:} The bounds for $\big\Vert\nabla c_t(\vtheta^c)\big\Vert$ and $\big\Vert\nabla^2 c_t(\vtheta^c)\big\Vert$ directly follow from Lemma~\ref{lem:9}. The Lipschitz-type condition
\[
\big\Vert\nabla^2 c_t(\vtheta^c) - \nabla^2 c_t(\widetilde{\vtheta}^c)\big\Vert\leq C_3(\mathcal{F}_{t-1})\big\Vert\vtheta^c-\widetilde{\vtheta}^c\big\Vert
\]
with $C_3(\mathcal{F}_{t-1})=C\big[1+Y_3(t,1-\delta_1) + Y_4(t,1-\delta_1)\big]$ (where $Y_4(t,\beta_Y)$ is defined in analogy to \eqref{eq:(11.1)}) follows similarly as \eqref{eq:(12.1)}:
The Lipschitz-type condition
\[
\big\Vert\nabla^2 \sigma_{Y,t}(\vtheta^Y) - \nabla^2 \sigma_{Y,t}(\widetilde{\vtheta}^Y)\big\Vert\leq \Sigma_Y(\mathcal{F}_{t-1})\big\Vert\vtheta^Y-\widetilde{\vtheta}^Y\big\Vert
\]
can be established similarly as \eqref{eq:(10.1)} and the conclusion then follows from \eqref{eq:(3.2)}.

\textbf{Assumption~\ref{ass:an}~\ref{it:mom bounds cons2}:} From CCC--GARCH Assumption~\ref{CGass:mom}, $\E|X_t|^{3+\iota}\leq K$ and $\E|Y_t|^{3+\iota}\leq K$. Then, Lemma~\ref{lem:10} applied to Lemma~\ref{lem:9} yields that $\E[V_i^{3+\iota}(\mathcal{F}_{t-1})]\leq C$ and $\E[C_i^{3+\iota}(\mathcal{F}_{t-1})]\leq C$ ($i=1,2$). 

In view of \eqref{eq:(11.1)} and \eqref{eq:(12.1)} it also follows from Lemma~\ref{lem:10} and $\E|X_t|^{3+\iota}\leq K$ that $\E[V_3^{3+\iota}(\mathcal{F}_{t-1})]\leq C$. 

By similar arguments, we also obtain that $\E[C_3^{3+\iota}(\mathcal{F}_{t-1})]\leq C$.

\textbf{Assumption~\ref{ass:an}~\ref{it:Lipschitz an}:} The Lipschitz continuity of $x\mapsto\partial_1 F_t(x,y)$ follows similarly as the Lipschitz continuity of $y\mapsto\partial_2 F_t(x,y)$ under Assumption~\ref{ass:cons}~\ref{it:Lipschitz cons}.

\textbf{Assumption~\ref{ass:an}~\ref{it:bound cdf}:} The bound $f_t^Y(\cdot)\leq K$ follows similarly as under \eqref{eq:(4.11)}.
Furthermore, from \eqref{eq:(14.1)},
\begin{align*}
	f_t(x,y) &= \frac{1}{\sigma_{X,t}\sigma_{Y,t}}f_{\varepsilon_X,\varepsilon_Y}\Big(\frac{x}{\sigma_{X,t}},\frac{y}{\sigma_{Y,t}}\Big)\\
	&\leq \frac{K}{\omega_{X,0}\omega_{Y,0}}\leq \frac{K}{\delta_1^2}.
\end{align*}
Also from \eqref{eq:(14.1)},
\begin{align*}
	\big|\partial_1F_t(x,y)\big| &= \frac{1}{\sigma_{X,t}}\bigg|\partial_1F_t\Big(\frac{x}{\sigma_{X,t}},\frac{y}{\sigma_{Y,t}}\Big)\bigg|\\
	&\leq \frac{K}{\omega_{X,0}}\leq \frac{K}{\delta_1},
\end{align*}
with a similar bound for $\partial_2F_t(x,y)$.

\textbf{Assumption~\ref{ass:an}~\ref{it:pd}:} Recall from Section~\ref{PR} that the matrices appearing in Assumption~\ref{ass:an}~\ref{it:pd} do not depend on $n$, such that we may drop the subindex $n$; see in particular Lemma~\ref{lem:sd}, which establishes stationarity of $v_t(\vtheta_0^v)$, $\nabla v_t(\vtheta_0^v)$, $c_t(\vtheta_0^c)$ and $\nabla c_t(\vtheta_0^c)$.
We only show that $\mLambda=\mLambda_n$ is positive definite. The remaining claims follow along similar lines. 
We proceed by contradiction and assume that $\mLambda$ is not positive definite. 
Since $\mLambda$ is clearly positive semi-definite, this implies the existence of some $\vx=(x_1,x_2,x_3)^\prime\in\mathbb{R}^3$ with $\vx\neq\vzeros$, such that $\vx^\prime\mLambda\vx=0$. By definition, $v_t(\vtheta_0^v)$ lies in the support of $X_t\mid\mathcal{F}_{t-1}$, such that $f_t^X\big(v_t(\vtheta_0^v)\big)>0$ by CCC--GARCH Assumption~\ref{CGass:error}.
Therefore, $\vx^\prime\mLambda\vx=0$ implies that a.s.
\begin{equation}\label{eq:(75)}
	0=\vx^\prime \nabla v_t(\vtheta_0^v)\overset{\eqref{eq:1.1}}{=}\bigg(\vx\odot\begin{pmatrix}1\\1\\v\end{pmatrix}\bigg)^\prime\nabla\sigma_{X,t}(\vtheta_0^X).
\end{equation}
By stationarity (see Lemma~\ref{lem:sd}), we also have a.s.
\begin{equation}\label{eq:(76)}
	0=\bigg(\vx\odot\begin{pmatrix}1\\1\\v\end{pmatrix}\bigg)^\prime\nabla\sigma_{X,t-1}(\vtheta_0^X).
\end{equation}
Computing \eqref{eq:(75)}$-\beta_{X,0}\cdot$\eqref{eq:(76)}, we obtain
\begin{eqnarray}
	0 &=& \bigg(\vx\odot\begin{pmatrix}1\\1\\v\end{pmatrix}\bigg)^\prime\nabla\sigma_{X,t}(\vtheta_0^X)-\beta_{X,0}\bigg(\vx\odot\begin{pmatrix}1\\1\\v\end{pmatrix}\bigg)^\prime\nabla\sigma_{X,t-1}(\vtheta_0^X)\notag\\
	&\overset{\eqref{eq:(8.2)}}{=}& x_1\bigg[\frac{1}{1-\beta_{X,0}}-\frac{\beta_{X,0}}{1-\beta_{X,0}}\bigg] + x_2\Big[X_1(t,\beta_{X,0}) - \beta_{X,0}X_1(t-1,\beta_{X,0})\Big]\notag\\
	&&\hspace{0.5cm} + vx_3\bigg[\frac{\omega_{X,0}}{(1-\beta_{X,0})^2} + \alpha_{X,0}X_2(t,\beta_{X,0})-\beta_{X,0}\bigg\{\frac{\omega_{X,0}}{(1-\beta_{X,0})^2} + \alpha_{X,0}X_2(t-1,\beta_{X,0})\bigg\}\bigg]\notag\\
	&=&x_1+x_2|X_{t-1}| + vx_3 \bigg[\omega_{X,0}\frac{1-\beta_{X,0}}{(1-\beta_{X,0})^2} + \alpha_{X,0}\big\{X_2(t,\beta_{X,0})-\beta_{X,0}X_2(t-1,\beta_{X,0})\big\}\bigg]\notag\\
	&=&x_1+x_2|X_{t-1}| + vx_3 \bigg[\frac{\omega_{X,0}}{1-\beta_{X,0}} + \alpha_{X,0}X_1(t-1,\beta_{X,0})\bigg]\notag\\
	&\overset{\eqref{eq:(8.1)}}{=}&x_1+x_2|X_{t-1}| + vx_3 \sigma_{X,t-1}(\vtheta_0^X).\notag
\end{eqnarray}
From this and the fact that $X_{t-1}\mid\sigma_{X,t-1}(\vtheta_0^X)\sim F_{\varepsilon_X}\big(\cdot/\sigma_{X,t-1}(\vtheta_0^X)\big)$, we can conclude that $x_1=x_2=x_3=0$.
This, however, contradicts the assumption $\vx\neq\vzeros$. Thus, positive definiteness of $\mLambda$ follows.

As mentioned above, the other claims follow similarly. 
We only mention that for the verification of the positive definiteness of $\mLambda_{(1)}$, one needs to notice that
\begin{align*}
	f_t^{Y}\big(c_t(\vtheta_0^c)\big) - \partial_2 F_{t}\big(v_t(\vtheta_0^v), c_t(\vtheta_0^c)\big) &= \int_{-\infty}^{\infty} f_t\big(x,c_t(\vtheta_0^c)\big)\D x - \int_{-\infty}^{v_t(\vtheta_0^v)}f_{t}\big(x, c_t(\vtheta_0^c)\big)\D x\\
	&= \int_{v_t(\vtheta_0^v)}^{\infty}f_{t}\big(x, c_t(\vtheta_0^c)\big)\D x>0
\end{align*}
is the analog condition of $f_t^X\big(v_t(\vtheta_0^v)\big)>0$ used to show positive definiteness of $\mLambda$ above.

\textbf{Assumption~\ref{ass:an}~\ref{it:lambda func}:} By Lemma~\ref{lem:sd}, $\Cc_t(\vtheta^c,\vtheta^v)$ is stationary, whence $\vlambda_n(\vtheta^c,\vtheta^v)=\frac{1}{n}\sum_{t=1}^{n}\E\big[\Cc_t(\vtheta^c,\vtheta^v)\big]$ is independent of $n$, and we may write $\vlambda(\vtheta^c,\vtheta^v)=\vlambda_n(\vtheta^c,\vtheta^v)$.

For the same reasons given in the proof of Lemma~C.\ref{lem:1 tilde}, we may apply the implicit function theorem \citep[e.g.,][Theorem~9.2]{Mun91} to $\vlambda(\cdot,\cdot)$, which satisfies $\vlambda(\vtheta_0^c,\vtheta_0^v)=\vzero$.
It yields that there exist a unique continuously differentiable function $\vtheta^c(\cdot)$ and some $\varepsilon>0$, such that for all $\vtheta^v$ with $\big\Vert\vtheta^v-\vtheta_0^v\big\Vert<\varepsilon$ it holds that $\vlambda\big(\vtheta^c(\vtheta^v), \vtheta^v\big)=\vzero$.
Additionally, we obtain by the MVT,
\begin{align*}
	\big\Vert\vtheta^c-\vtheta_0^c\big\Vert &= \big\Vert\vtheta^c(\vtheta^v)-\vtheta^c(\vtheta_0^v)\big\Vert\\
	&=\Big\Vert \frac{\partial\vtheta^c}{\partial\vtheta}(\vtheta^{v\ast})(\vtheta^v - \vtheta_0^v)\Big\Vert\\
	&= \Big\Vert \frac{\partial\vtheta^c}{\partial\vtheta}(\vtheta^{v\ast})\Big\Vert\cdot \big\Vert\vtheta^v - \vtheta_0^v\big\Vert\\
	&\leq K \big\Vert\vtheta^v - \vtheta_0^v\big\Vert,
\end{align*}	
where we used that $\sup_{\Vert\vtheta^v-\vtheta_0^v\Vert<\varepsilon}\big\Vert\partial\vtheta^c(\vtheta^v)/\partial\vtheta\big\Vert\leq K$ because $\partial\vtheta^c(\cdot)/\partial\vtheta$ is continuous.

\textbf{Assumption~\ref{ass:an}~\ref{it:eq bound}:} The outline of the proof is roughly similar to that given in Appendix~SA.2.7 of \citet{PZC19}. By definition of $v_t(\cdot)$ in \eqref{eq:(1.1)}, $X_t=v_t(\vtheta^v)$ is equivalent to
\begin{equation}\label{eq:hh11}
	X_t = \omega_v + \alpha_v|X_{t-1}| + \beta_v v_{t-1}(\vtheta^v),\qquad t=1,2,\ldots.
\end{equation}
Recall from \eqref{eq:(3.11)} that $v_{t}(\vtheta^v)$ has the same sign as $v$, because $\sigma_{X,t}(\vtheta^X)>0$. Assume without loss of generality that $v>0$; the case $v<0$ can be treated similarly. When $v>0$, $X_t=v_t(\vtheta^v)$ also implies that $X_t>0$. Using this and iterating \eqref{eq:hh11} for $t=1,\ldots,4$ gives
\begin{align}
	X_1 &= \omega_v + \alpha_v|X_{0}| + \beta_v v_{0}(\vtheta^v),\notag\\
	X_2 &= \omega_v + \alpha_v|X_{1}| + \beta_v v_{1}(\vtheta^v)\notag\\
	&= \omega_v + \alpha_v X_{1}  + \beta_v X_1\notag\\
	&= \omega_v + ( \alpha_v + \beta_v ) X_1,\label{eq:h2}\\
	X_3 &= \omega_v + ( \alpha_v + \beta_v ) X_2,\label{eq:h3}\\		
	X_4 &= \omega_v + ( \alpha_v + \beta_v ) X_3,\label{eq:h4}
\end{align}
From \eqref{eq:h2}--\eqref{eq:h3},
\begin{align*}
	\frac{X_2^2 - X_1 X_3}{X_2-X_1} &= \frac{X_2^2-\frac{X_2-\omega_v}{\alpha_v+\beta_v}(\omega_v+\alpha_vX_2+\beta_vX_2)}{X_2 - \frac{X_2-\omega_v}{\alpha_v+\beta_v}}\\
	&= \frac{(\alpha_v+\beta_v)X_2^2 - (X_2-\omega_v)(\omega_v+\alpha_vX_2 + \beta_v X_2)}{(\alpha_v+\beta_v)X_2-(X_2-\omega_v)}\\
	&= \frac{\alpha_vX_2^2+\beta_vX_2^2 - \omega_vX_2-\alpha_vX_2^2 - \beta_vX_2^2 + \omega_v^2 + \alpha_v\omega_vX_2+\beta_v\omega_vX_2}{\omega_v+\alpha_vX_2+\beta_vX_2-X_2}\\
	&= \omega_v\frac{\omega_v+\alpha_vX_2+\beta_vX_2-X_2}{\omega_v+\alpha_vX_2+\beta_vX_2-X_2}\\
	&=\omega_v.
\end{align*}
It follows similarly from \eqref{eq:h3}--\eqref{eq:h4} that $\omega_v=\frac{X_3^2 - X_2 X_4}{X_3-X_2}$, such that 
\[
\omega_v=\frac{X_2^2 - X_1 X_3}{X_2-X_1}=\frac{X_3^2 - X_2 X_4}{X_3-X_2}.
\]
Hence, a necessary condition for $X_t=v_t(\vtheta^v)$, $t=1,\ldots,n$ for $\vtheta^v\in\mTheta^v$ is that $(X_1,\ldots,X_n)^\prime$ lies in the set
\[
p^{-1}(0):=\big\{(X_1,\ldots,X_n)^\prime\in\mathbb{R}^{n}:\ p(X_1,\ldots,X_n)=0\big\},
\]
where $p(X_1,\ldots,X_n)=(X_2^2-X_1X_3)(X_3-X_2) - (X_3^2-X_2X_4)(X_2-X_1)$. The Hausdorff dimension of $p^{-1}(0)$ is less than $n$. Since, additionally, the distribution of $(X_1,\ldots,X_n)^\prime$ is absolutely continuous (CCC--GARCH Assumption~\ref{CGass:error}), it follows that a.s.
\[
\sup_{\vtheta^v\in\mTheta^v}\sum_{t=1}^{n}\1_{\{X_t=v_t(\vtheta^v)\}}\leq 4.
\]

Starting from the definition of $c_t(\cdot)$ in \eqref{eq:(1.2)}, $Y_t=c_t(\vtheta^c)$ is equivalent to
\[
Y_t=\omega_c+\alpha_c|Y_{t-1}|+ \beta_c c_{t-1}(\vtheta^c),
\]
which is the analog of \eqref{eq:hh11}. It then follows similarly as before that
\[
\sup_{\vtheta^c\in\mTheta^c}\sum_{t=1}^{n}\1_{\{Y_t=c_t(\vtheta^v)\}}\leq 4.
\]

\textbf{Assumption~\ref{ass:an}~\ref{it:mixing}:} By standard mixing inequalities \citep[e.g.,][]{Bra07}, $\beta$-mixing with exponential decay implies $\alpha$-mixing with exponential decay. The conclusion therefore follows from Lemma~\ref{lem:sd}.

\section{Verifying Assumptions~\ref{ass:cons} and \ref{ass:an} for VAR Models}
\label{sec:verificationVAR}

In this section, we verify Assumptions~\ref{ass:cons} and \ref{ass:an} for a VAR(1) model, i.e., a vector-autoregressive process of order one. Extensions to higher order VAR models are possible, though at the expense of more complicated notation. Subsection~\ref{VAR} presents the model and Subsection~\ref{PR VAR} collects some properties of the model that will be helpful in verifying Assumptions~\ref{ass:cons} and \ref{ass:an} in Subsections~\ref{Ass1 VAR} and \ref{Ass2 VAR}, respectively.

\subsection{The VAR Model}\label{VAR}

The VAR(1) model we consider throughout this section is given by the recursion
\begin{equation}\label{eq:VAR}
	\begin{pmatrix}
		X_t\\ Y_t
	\end{pmatrix} =
	\vphi_0 + \mPhi_1\begin{pmatrix}
		X_{t-1}\\ Y_{t-1}
	\end{pmatrix}
	+\vvarepsilon_t,\qquad 
	t\in\mathbb{Z},
\end{equation}
where 
\begin{equation*}
	\vvarepsilon_t=\begin{pmatrix} \varepsilon_{X,t}\\ \varepsilon_{Y,t} \end{pmatrix},\qquad \vphi_0=\begin{pmatrix}
		\phi_{X,0}\\ \phi_{Y,0}
	\end{pmatrix},\qquad
	\mPhi_1=\begin{pmatrix}
		\phi_{11} & \phi_{12}\\
		\phi_{21} & \phi_{22}
	\end{pmatrix}.
\end{equation*}
We make the following assumptions:

\begin{Vassumption}\label{Vass:error}
	The innovations $(\varepsilon_{X,t}, \varepsilon_{Y,t})^\prime$ are i.i.d., independent of 
	\[
	\mathcal{F}_{t-1}=\sigma\big\{(X_{t-1}, Y_{t-1})^\prime,(X_{t-2}, Y_{t-2})^\prime,\ldots \big\},
	\]
	and have mean zero. The Lebesgue density $f_{\varepsilon_X,\varepsilon_Y}(\cdot,\cdot)$ is positive for all $(x,y)^\prime\in\mathbb{R}^2$ such that $F_{\varepsilon_X,\varepsilon_Y}(x,y)\in(0,1)$. 
	Finally, 
	\begin{align*}
		f_{\varepsilon_X}(\cdot)&\leq K,& f_{\varepsilon_Y}(\cdot)&\leq K,\\
		\big|\partial_i F_{\varepsilon_X,\varepsilon_Y}(\cdot,\cdot)\big|&\leq K\ (i=1,2),& f_{\varepsilon_X,\varepsilon_Y}(\cdot,\cdot)&\leq K,\\
		\big|f_{\varepsilon_X}(x)-f_{\varepsilon_X}(x^\prime)\big|&\leq K |x-x^\prime|,& \big|f_{\varepsilon_Y}(y)-f_{\varepsilon_Y}(y^\prime)\big|&\leq K |y-y^\prime|,\\
		\big|\partial_1F_{\varepsilon_X,\varepsilon_Y}(x,y)-\partial_1 F_{\varepsilon_X,\varepsilon_Y}(x^\prime,y)\big|&\leq K |x-x^\prime|,& \big|\partial_2F_{\varepsilon_X,\varepsilon_Y}(x,y)-\partial_2 F_{\varepsilon_X,\varepsilon_Y}(x,y^\prime)\big|&\leq K |y-y^\prime|.
	\end{align*}
\end{Vassumption}

VAR Assumption~\ref{Vass:error} ensures that $\vphi_0 + \mPhi_1(X_t, Y_t)^\prime$ is the conditional mean of $(X_t, Y_t)^\prime\mid \mathcal{F}_{t-1}$, i.e.,
\[
\begin{pmatrix}
	\mu_{X,t}\\ \mu_{Y,t} 
\end{pmatrix}=\begin{pmatrix}
	\E[X_t\mid \mathcal{F}_{t-1}]\\ \E[Y_t\mid \mathcal{F}_{t-1}]
\end{pmatrix}=
\begin{pmatrix}
	\phi_{X,0} + \phi_{11}X_{t-1} + \phi_{12} Y_{t-1}\\
	\phi_{Y,0} + \phi_{21}X_{t-1} + \phi_{22} Y_{t-1}
\end{pmatrix}.
\]
We often exploit the decomposition
\begin{equation}\label{eq:VAR decomp}
	\begin{pmatrix}
		X_t\\ Y_t
	\end{pmatrix}=\begin{pmatrix}
		\mu_{X,t}\\ \mu_{Y,t} 
	\end{pmatrix}
	+\begin{pmatrix}
		\varepsilon_{X,t}\\ \varepsilon_{Y,t} 
	\end{pmatrix}
\end{equation}
in the following.

\begin{Vassumption}\label{Vass:param}
	All parameters $\big(\vphi_0^\prime, \vec^\prime(\mPhi_1)\big)^\prime$ from the compact parameter space $\mTheta^{(X,Y)}$ satisfy the following: 
	The eigenvalues $\lambda_i(\mPhi_1)$ ($i=1,2$) of $\mPhi_1$ are smaller than one in absolute value, i.e., $\big|\lambda_i(\mPhi_1)\big|<1$ for $i=1,2$.
	Furthermore, the true parameter vector $\vtheta_0=(\vtheta_0^{X}, \vtheta_0^{Y})$ lies in the interior of $\mTheta^{(X,Y)}$.
\end{Vassumption}

VAR Assumption~\ref{Vass:param} is mainly required to show that the $(X_t,Y_t)^\prime$ are strictly stationary.

\begin{Vassumption}\label{Vass:mom}
	$\E|X_t|^{3+\iota}<\infty$ and $\E|Y_t|^{3+\iota}<\infty$.
\end{Vassumption}

In VAR models, finite fourth moments are often required to show asymptotic normality of the standard least squares estimator \citep[e.g.,][Thm.~15.10]{Han22}, which contrasts with the much weaker moment requirements of VAR Assumption~\ref{CGass:mom}. 
This may be seen as an additional advantage of our modeling approach.

\begin{Vassumption}\label{Vass:mixing}
	The $(X_t, Y_t)^\prime$ are $\alpha$-mixing with mixing coefficients $\alpha(\cdot)$ satisfying
	\[
	\sum_{m=1}^{\infty}\alpha^{(\widetilde{q}-2)/\widetilde{q}}(m)<\infty\qquad\text{for some }\widetilde{q}>2.
	\] 
\end{Vassumption}

Since strong mixing may not hold under easily verifiable conditions for autoregressive models \citep{And84}, we follow, e.g., \citet{Han22} and simply assume it here.

\subsection{Preliminary Results}\label{PR VAR}

Before showing that Assumptions~\ref{ass:cons} and \ref{ass:an} hold under VAR Assumptions~\ref{Vass:error}--\ref{Vass:mixing} in Subsections~\ref{Ass1 VAR} and \ref{Ass2 VAR}, respectively, we explore some properties of the model that will be used throughout.
We assume the forecaster to be interested in conditioning on the $\sigma$-field
\[
\mathcal{F}_{t-1}=\sigma\big((X_{t-1},Y_{t-1})^\prime,\ldots,(X_{1},Y_{1})^\prime,\init\big)\qquad\text{with}\ \init=(X_0,Y_0,X_{-1},Y_{-1},\ldots)^\prime.
\]
We mention that our following verification would also go through verbatim when the initial information in the conditioning set would be given by $\init=(X_0,Y_0)^\prime$.

We now establish the (VaR, CoVaR) dynamics implied by \eqref{eq:VAR} with respect to $\mathcal{F}_{t-1}$. 
Since
\begin{align*}
	F_{X_t\mid\mathcal{F}_{t-1}}(x) &= \Pr\big\{X_t\leq x\mid\mathcal{F}_{t-1}\big\}\\
	&= \Pr\big\{\mu_{X,t} + \varepsilon_{X,t}\leq x\mid\mathcal{F}_{t-1}\big\}\\
	&= \Pr\big\{\varepsilon_{X,t}\leq x - \mu_{X,t} \mid\mathcal{F}_{t-1}\big\}\\
	&= F_{\varepsilon_X}(x - \mu_{X,t}),
\end{align*}
it easily follows that
\begin{align}
	\VaR_{t,\beta} &= F_{X_t\mid\mathcal{F}_{t-1}}^{-1}(\beta)\notag\\
	&=\mu_{X,t} + \VaR_{\beta}(F_{\varepsilon_X})\notag\\
	&=\big[\phi_{X,0}  + \VaR_{\beta}(F_{\varepsilon_X})\big] + \phi_{11}X_{t-1} + \phi_{12} Y_{t-1}.\label{eq:v VAR}
\end{align}
Similarly, by \eqref{eq:VAR decomp} and \eqref{eq:v VAR},
\begin{align*}
	F_{Y_t\mid X_t\geq\VaR_{t,\beta},\mathcal{F}_{t-1}}(y) &= \Pr\big\{Y_t\leq y\mid X_t\geq \VaR_{t,\beta},\ \mathcal{F}_{t-1}\big\}\\
	&= \Pr\big\{\varepsilon_{Y,t}\leq y-\mu_{Y,t}\mid \varepsilon_{X,t} \geq \VaR_{\beta}(F_{\varepsilon_X}),\ \mathcal{F}_{t-1}\big\}\\
	&= F_{\varepsilon_Y\mid\varepsilon_X \geq\VaR_{\beta}(F_{\varepsilon_X})}(y-\mu_{Y,t}),
\end{align*}
such that it easily follows that
\begin{align}
	\CoVaR_{t,\alpha\mid\beta} &= F_{Y_t\mid X_t\geq\VaR_{t,\beta},\ \mathcal{F}_{t-1}}^{-1}(\alpha)\notag\\
	&=\mu_{Y,t} + \CoVaR_{\alpha\mid \beta}(F_{\varepsilon_X,\varepsilon_Y})\notag\\
	&=\big[\phi_{Y,0}  + \CoVaR_{\alpha\mid \beta}(F_{\varepsilon_X,\varepsilon_Y})\big] + \phi_{21}X_{t-1} + \phi_{22} Y_{t-1}.\label{eq:c VAR}
\end{align}

In view of \eqref{eq:v VAR} and \eqref{eq:c VAR}, we now consider the following (VaR, CoVaR) model:
\begin{align}
	v_t(\vtheta^v) &= \theta_1^v + \theta_2^v X_{t-1} + \theta_3^v Y_{t-1},\qquad\vtheta^v=(\theta_1^v,\theta_2^v,\theta_3^v)^\prime,\label{eq:(V2.3)}\\
	c_t(\vtheta^c) &= \theta_1^c + \theta_2^c X_{t-1} + \theta_3^c Y_{t-1},\qquad\vtheta^c=(\theta_1^c,\theta_2^c,\theta_3^c)^\prime.\label{eq:(V2.4)}
\end{align}
In particular, the true values are given by
\begin{align}
	\vtheta_0^v &= (\phi_{X,0} + v,\ \phi_{11},\ \phi_{12})^\prime,&& v=\VaR_{\beta}(F_{\varepsilon_X})\label{eq:(V3.1)}\\
	\vtheta_0^c &= (\phi_{Y,0} + c,\ \phi_{21},\ \phi_{22})^\prime,&& c=\CoVaR_{\alpha\mid \beta}(F_{\varepsilon_X,\varepsilon_Y}).\label{eq:(V3.2)}
\end{align}
In view of \eqref{eq:(V2.3)} and \eqref{eq:(V2.4)}, the parameter space of the (VaR, CoVaR) model is
\begin{equation}\label{eq:(3.3)}
	\mTheta :=\mTheta^v\times\mTheta^c = \big\{\vtheta=(\vtheta^{v\prime}, \vtheta^{c\prime})^\prime=(\theta_1,\ldots,\theta_6)\in\mathbb{R}^{6}\colon (\theta_1-v,\, \theta_2,\, \theta_3,\, \theta_4-c,\, \theta_5,\, \theta_6)^\prime\in\mTheta^{(X,Y)}\big\}.
\end{equation}

Also from \eqref{eq:(V2.3)} and \eqref{eq:(V2.4)}, the gradients and Hessians of the models are
\begin{align}
	\nabla v_t(\vtheta^v) &= \begin{pmatrix}
		\frac{\partial v_t(\vtheta^v)}{\partial \theta_1^v}\\
		\frac{\partial v_t(\vtheta^v)}{\partial \theta_2^v}\\
		\frac{\partial v_t(\vtheta^v)}{\partial \theta_3^v}
	\end{pmatrix}
	=\begin{pmatrix}
		1\\ X_{t-1}\\  Y_{t-1}
	\end{pmatrix},\qquad
	\nabla^2 v_t(\vtheta^v)=\vzero,\label{eq:(1.P.1)}\\
	\nabla c_t(\vtheta^c) &= \begin{pmatrix}
		\frac{\partial c_t(\vtheta^c)}{\partial \theta_1^c}\\
		\frac{\partial c_t(\vtheta^c)}{\partial \theta_2^c}\\
		\frac{\partial c_t(\vtheta^c)}{\partial \theta_3^c}
	\end{pmatrix}
	=\begin{pmatrix}
		1\\ X_{t-1}\\  Y_{t-1}
	\end{pmatrix},\qquad
	\nabla^2 c_t(\vtheta^c)=\vzero.\label{eq:(1.P.2)}
\end{align}

\begin{lem}\label{lem:6tilde}
	Under VAR Assumptions~\ref{Vass:error}--\ref{Vass:mom}, it holds for all $\vtheta^v\in\mTheta^v$ and $\vtheta^c\in\mTheta^c$ that
	\begin{align*}
		|v_t(\vtheta^v)| &\leq V(\mathcal{F}_{t-1}):= K(1+|X_{t-1}| + |Y_{t-1}|),\\
		|c_t(\vtheta^v)| &\leq C(\mathcal{F}_{t-1}):= K(1+|X_{t-1}| + |Y_{t-1}|),\\
		\big\Vert\nabla v_t(\vtheta^v)\big\Vert &\leq V_1(\mathcal{F}_{t-1}):=K(1+|X_{t-1}| + |Y_{t-1}|),\\
		\big\Vert\nabla c_t(\vtheta^c)\big\Vert &\leq C_1(\mathcal{F}_{t-1}):=K(1+|X_{t-1}| + |Y_{t-1}|),\\
		\big\Vert\nabla^2 v_t(\vtheta^v)\big\Vert &\leq V_2(\mathcal{F}_{t-1}):=0,\\
		\big\Vert\nabla^2 c_t(\vtheta^c)\big\Vert &\leq C_2(\mathcal{F}_{t-1}):=0.
	\end{align*}
\end{lem}

\begin{proof}
	The first two upper bounds follow from \eqref{eq:(V2.3)}--\eqref{eq:(V2.4)}.
	The remaining bounds are immediate from \eqref{eq:(1.P.1)}--\eqref{eq:(1.P.2)} in combination with compactness of the parameter space $\mTheta$ (from compactness of $\mTheta^{(X,Y)}$ by VAR Assumption~\ref{Vass:param} and \eqref{eq:(3.3)}).
\end{proof}

\subsection{Verification of Assumption~\ref{ass:cons}}\label{Ass1 VAR}

We verify items~\ref{it:id}--\ref{it:mom bounds cons} of Assumption~\ref{ass:cons} in this subsection. 

\textbf{Assumption~\ref{ass:cons}~\ref{it:id}:} Recall from Section~\ref{PR VAR} that 
\begin{equation*}
	\begin{pmatrix}
		\VaR_{t,\beta}\\
		\CoVaR_{t,\alpha\mid\beta}
	\end{pmatrix}
	= \begin{pmatrix}
		v_t(\vtheta_{0}^{v})\\
		c_t(\vtheta_{0}^{c})
	\end{pmatrix}
\end{equation*}
for $\vtheta_0^v=(\phi_{X,0} + v,\ \phi_{11},\ \phi_{12})^\prime$ and $\vtheta_0^c=(\phi_{Y,0} + c,\ \phi_{21},\ \phi_{22})^\prime$, such that correct specification of the (VaR, CoVaR) model is immediate.

\textbf{Assumption~\ref{ass:cons}~\ref{it:str stat}:} Stationarity follows directly from Theorems 14.4 and 15.6 in \citet{Han22} together with the i.i.d.ness of the $\vvarepsilon_t$ from VAR Assumption~\ref{Vass:error}.

\textbf{Assumption~\ref{ass:cons}~\ref{it:cond dist}:} Exploit \eqref{eq:VAR decomp} and the fact that $\mu_{X,t},\mu_{Y,t}\in\mathcal{F}_{t-1}$ to write 
\begin{align}
	F_t(x,y) & =\P_{t-1}\big\{X_t\leq x,\ Y_t\leq y\big\}\notag\\
	&= \P_{t-1}\big\{\varepsilon_{X,t}\leq x-\mu_{X,t},\ \varepsilon_{Y,t}\leq y-\mu_{Y,t}\big\}\notag\\
	&= F_{\varepsilon_X,\varepsilon_Y}\big(x-\mu_{X,t},\ y-\mu_{Y,t}\big).\label{eq:(V5.1)}
\end{align}
The fact that $F_t(\cdot,\cdot)$ possesses a positive Lebesgue density now follows from VAR Assumption~\ref{Vass:error}.

\textbf{Assumption~\ref{ass:cons}~\ref{it:Lipschitz cons}:} From \eqref{eq:(V5.2)} below and VAR Assumption~\ref{Vass:error} it follows that
\begin{align*}
	\big|f_t^X(x)-f_t^X(x^\prime)\big| &=\big|f_{\varepsilon_X}(x-\mu_{X,t}) - f_{\varepsilon_X}(x^\prime-\mu_{X,t})\big|\\
	&\leq K\big|(x-\mu_{X,t}) - (x^\prime-\mu_{X,t})\big|\\
	&\leq K|x-x^\prime|.
\end{align*}
The Lipschitz continuity of $f_t^{Y}(\cdot)$ follows similarly.

From \eqref{eq:(V5.1)}, we have
\[
\partial_2 F_t(x,y)=\partial_2 F_{\varepsilon_X,\varepsilon_Y}\big(x-\mu_{X,t},\, y-\mu_{Y,t}\big),
\]
such that, by VAR Assumption~\ref{Vass:error},
\begin{align*}
	\big|\partial_2 F_t(x,y) - \partial_2 F_t(x,y^\prime)\big| &=\big|\partial_2 F_{\varepsilon_X,\varepsilon_Y}(x-\mu_{X,t},\, y-\mu_{Y,t}) - \partial_2 F_{\varepsilon_X,\varepsilon_Y}(x-\mu_{X,t},\, y^\prime-\mu_{Y,t})\big|\\
	&\leq K\big|(y-\mu_{Y,t}) - (y^\prime-\mu_{Y,t})\big| \\
	&\leq K|y-y^\prime|.
\end{align*}

\textbf{Assumption~\ref{ass:cons}~\ref{it:cond dens}:} Because $F_t^{X}(x)=F_{\varepsilon_X}\big(x-\mu_{X,t}\big)$ from \eqref{eq:(V5.1)}, it holds that
\begin{equation}\label{eq:(V5.2)}
	f_t^X(x) = f_{\varepsilon_X}\big(x-\mu_{X,t}\big).
\end{equation}
VAR Assumption~\ref{Vass:error} therefore implies $\sup_{x\in\mathbb{R}}f_t^X(x)\leq K$.

Moreover,
\begin{align*}
	f_t^X\big(v_t(\vtheta_0^v)\big) &= f_{\varepsilon_X}\big(v_t(\vtheta_0^v)-\mu_{X,t}\big)\\
	&= f_{\varepsilon_X}\big(\VaR_{\beta}(F_{\varepsilon_{X}})\big)>0,
\end{align*}
where the first equality follows from \eqref{eq:(V5.2)}, the second from the definition of $\mu_{X,t}$ and \eqref{eq:(V2.3)} and \eqref{eq:(V3.1)}, and the final inequality from VAR Assumption~\ref{Vass:error}.
Therefore, the lower bound for $f_t^X\big(v_t(\vtheta_0^v)\big)$ even holds with probability one.

Similarly, from VAR Assumption~\ref{Vass:error} with probability one,
\begin{align*}
	\int_{v_t(\vtheta_0^v)}^{\infty}f_t\big(x, c_t(\vtheta_0^c)\big)\D x &= \int_{v_t(\vtheta_0^v)}^{\infty}f_{\varepsilon_X,\varepsilon_Y}\big(x-\mu_{X,t}, c_t(\vtheta_0^c)-\mu_{Y,t}\big)\D x\\
	&= \int_{\VaR_{\beta}(F_{\varepsilon_X})}^{\infty}f_{\varepsilon_X,\varepsilon_Y}\big(x, c_t(\vtheta_0^c)-\mu_{Y,t}\big)\D x\\
	&= \int_{\VaR_{\beta}(F_{\varepsilon_X})}^{\infty}f_{\varepsilon_X,\varepsilon_Y}\big(x, \CoVaR_{\alpha\mid\beta}(F_{\varepsilon_X, \varepsilon_Y})\big)\D x>0,
\end{align*}
where the first equality follows from \eqref{eq:(V5.1)}, the second equality uses an obvious substitution, and the final equality exploits \eqref{eq:(V2.4)} and \eqref{eq:(V3.2)}.

\textbf{Assumption~\ref{ass:cons}~\ref{it:compact}:} The compactness of the parameter space directly follows from VAR Assumption~\ref{Vass:param} in combination with \eqref{eq:(3.3)}.

\textbf{Assumption~\ref{ass:cons}~\ref{it:ULLN}:} By Theorem~21.9 in \citet{Dav94} it suffices to show that
\begin{itemize}
	\item[(a)] $\mS\Big(\big(v_t(\vtheta^v), c_t(\vtheta^c)\big)^\prime, (X_t, Y_t)^\prime\Big)$ obeys a pointwise law of large numbers;
	\item[(b)] $\mS\Big(\big(v_t(\vtheta^v), c_t(\vtheta^c)\big)^\prime, (X_t, Y_t)^\prime\Big)$ is s.e.
\end{itemize}

Since the s.e.~property of (b) follows as for the CCC--GARCH model in Section~\ref{Ass1} (replacing Lemma~\ref{lem:9} with Lemma~\ref{lem:6tilde} where appropriate), we only show (a). 

Recall from Assumption~\ref{ass:cons}~\ref{it:str stat} that $(X_t,Y_t)^\prime$ is strictly stationary. By VAR Assumption~\ref{Vass:mixing} and standard mixing inequalities, $\mS\Big(\big(v_t(\vtheta^v), c_t(\vtheta^c)\big)^\prime, (X_t, Y_t)^\prime\Big)$ is strictly stationary and $\alpha$-mixing and, hence, ergodic by \citet[Proposition~3.44]{Whi01}. The ergodic theorem \citep[e.g.,][Theorem~3.34]{Whi01} then implies (a) if the (componentwise) moments of $\mS\Big(\big(v_t(\vtheta^v), c_t(\vtheta^c)\big)^\prime, (X_t, Y_t)^\prime\Big)$ exist.
Use the boundedness of the moments of $X_t$ and $Y_t$ (VAR Assumption~\ref{Vass:mom}) and the compactness of the parameter space (Assumption~\ref{ass:cons}~\ref{it:compact}) to deduce that
\begin{eqnarray*}
	\E\big|S^{\VaR}\big(v_t(\vtheta^v),X_t\big)\big|&\leq& \E|X_t| + \E|v_t(\vtheta^v)|\\
	& \overset{\eqref{eq:(V2.3)}}{\leq}& \E|X_t| + |\theta_1^v| + |\theta_2^v|\E|X_{t-1}| + |\theta_3^v|\E|Y_{t-1}|\leq C<\infty,\\
	\E\Big|S^{\CoVaR}\Big(\big(v_t(\vtheta^v), c_t(\vtheta^c)\big)^\prime, (X_t, Y_t)^\prime\Big)\Big|&\leq& \E|Y_t| + \E|c_t(\vtheta^c)|\\
	&\overset{\eqref{eq:(V2.4)}}{\leq}& \E|Y_t| + |\theta_1^c| + |\theta_2^c|\E|X_{t-1}| + |\theta_3^c|\E|Y_{t-1}| \leq C<\infty.
\end{eqnarray*}
Thus, the (componentwise) moments of $\mS\Big(\big(v_t(\vtheta^v), c_t(\vtheta^c)\big)^\prime, (X_t, Y_t)^\prime\Big)$ exist and the pointwise law of large numbers follows by the ergodic theorem, yielding (a).

\textbf{Assumption~\ref{ass:cons}~\ref{it:unique id}:} We first show the preliminary result that
\begin{equation}\label{eq:vt impl VAR}
	\P\big\{v_t(\vtheta^v)=v_t(\vtheta_0^v)\big\}=1\text{ for all }t\in\mathbb{N}\quad\Longrightarrow\quad \vtheta^v=\vtheta_0^v.
\end{equation}
Hence, suppose that $\P\big\{v_t(\vtheta^v)=v_t(\vtheta_0^v)\big\}=1$ for all $t\in\mathbb{N}$. This implies
\begin{align}
	& \P\big\{v_t(\vtheta^v)=v_t(\vtheta_0^v)\quad \forall t\in\mathbb{N}\big\} = 1\notag\\
	\overset{\eqref{eq:(V2.3)}}{\Longrightarrow}\quad & \P\big\{\theta_1^v + \theta_2^v X_{t-1} + \theta_3^v Y_{t-1}=\theta_{0,1}^v + \theta_{0,2}^v X_{t-1} + \theta_{0,3}^v Y_{t-1}\quad \forall t\in\mathbb{N}\big\} = 1\notag\\
	\Longrightarrow\quad & \P\big\{\theta_1^v + \theta_2^v \mu_{X,t-1} + \theta_2^v \varepsilon_{X,t-1} + \theta_3^v \mu_{Y,t-1} + \theta_3^v \varepsilon_{Y,t-1}\notag\\
	&\hspace{1cm} =\theta_{0,1}^v + \theta_{0,2}^v \mu_{X,t-1} + \theta_{0,2}^v \varepsilon_{X,t-1} + \theta_{0,3}^v \mu_{Y,t-1} + \theta_{0,3}^v \varepsilon_{Y,t-1}\quad \forall t\in\mathbb{N}\big\} = 1.\notag
\end{align}
Since $(\varepsilon_{X,t}, \varepsilon_{Y,t})^\prime\mid(\mu_{X,t-1},\mu_{Y,t-1})^\prime\sim F_{\varepsilon_X,\varepsilon_Y}$, this implies that
\begin{align}
	& \P\big\{\theta_2^v\varepsilon_{X,t} + \theta_3^v\varepsilon_{Y,t} = \theta_{0,2}^v\varepsilon_{X,t} + \theta_{0,3}^v\varepsilon_{Y,t} \quad \forall t\in\mathbb{N}\big\} = 1\qquad\text{and}\label{eq:(V4.1)}\\
	& \P\big\{\theta_1^v + \theta_2^v \mu_{X,t} + \theta_{3}^{v} \mu_{Y,t} = 
	\theta_{0,1}^v + \theta_{0,2}^v \mu_{X,t} + \theta_{0,3}^{v} \mu_{Y,t} \quad \forall t\in\mathbb{N}\big\} = 1.\label{eq:(V4.2)}
\end{align}
From \eqref{eq:(V4.1)},
\[
\P\Big\{ (\theta_2^v-\theta_{0,2}^v,\, \theta_3^v-\theta_{0,3}^v) \begin{pmatrix}\varepsilon_{X,t}\\ \varepsilon_{Y,t}\end{pmatrix} = 0 \quad \forall t\in\mathbb{N}\Big\} = 1.
\]
Since $(\varepsilon_{X,t}, \varepsilon_{Y,t})^\prime$ possesses a Lebesgue density, it must be the case that $\theta_2^v=\theta_{0,2}^v$ and $\theta_3^v=\theta_{0,3}^v$. Plugging this into \eqref{eq:(V4.2)} yields that $\theta_1^v=\theta_{0,1}^{v}$. 
Overall, $\vtheta^v=\vtheta_0^v$ follows, establishing \eqref{eq:vt impl VAR}.

It may be shown similarly that $\P\big\{c_t(\vtheta^c)=c_t(\vtheta_0^c)\big\}=1$ for all $t\in\mathbb{N}$ implies $\vtheta^c=\vtheta_0^c$.

For part (a), we proceed by contradiction and assume that there exists some $\xi>0$ such that for all $\tau>0$ it follows that
\[
\liminf_{n\to\infty}\frac{1}{n}\sum_{t=1}^{n}\P\big\{|v_t(\vtheta^v) - v_t(\vtheta_0^v)|>\tau\mid f_t^X\big(v_t(\vtheta_0^v)\big)>f_1\big\}=0
\]
whenever $\big\Vert\vtheta^v - \vtheta_0^v\big\Vert\geq\xi$.
In particular, this implies that $\P\big\{|v_t(\vtheta^v) - v_t(\vtheta_0^v)|>\tau\mid f_t^X\big(v_t(\vtheta_0^v)\big)>f_1\big\}=0$ for all $t\in\mathbb{N}$ and all $\tau>0$.
From the verification of Assumption~\ref{ass:cons}~\ref{it:cond dens}, we know that $f_t^{X}\big(v_t(\vtheta_0^v)\big)> f_1$ with probability one, such that $\P\big\{|v_t(\vtheta^v) - v_t(\vtheta_0^v)|>\tau\big\}=0$ for all $t\in\mathbb{N}$ and all $\tau>0$.
This, in turn, yields that $\P\big\{v_t(\vtheta^v) = v_t(\vtheta_0^v)\big\}=1$ for all $t\in\mathbb{N}$.
Therefore, \eqref{eq:vt impl VAR} implies that $\vtheta^v=\vtheta_0^v$, contradicting the fact that $\big\Vert\vtheta^v - \vtheta_0^v\big\Vert\geq\xi$.
The proof of part (b) follows along similar lines and, hence, is omitted.

\textbf{Assumption~\ref{ass:cons}~\ref{it:smooth}:} The $\mathcal{F}_{t-1}$-measurability and continuity of $v_t(\cdot)$ and $c_t(\cdot)$ are obvious from \eqref{eq:(V2.3)} and \eqref{eq:(V2.4)}, respectively.

\textbf{Assumption~\ref{ass:cons}~\ref{it:diff c}:} The continuous differentiability of $v_t(\vtheta^v)$ is immediate from \eqref{eq:(V2.3)}.

\textbf{Assumption~\ref{ass:cons}~\ref{it:bound}:} The bounds can be obtained easily from Lemma~\ref{lem:6tilde}.

\textbf{Assumption~\ref{ass:cons}~\ref{it:mom bounds cons}:} The conditions $\E|X_t|\leq K$ and $\E|Y_t|^{1+\iota}\leq K$ are immediate from VAR Assumption~\ref{Vass:mom}.

By the Cauchy--Schwarz inequality, the remaining moment conditions hold if we can show that $\E\big[V(\mathcal{F}_{t-1})\big]\leq K$, $\E\big[V_1^2(\mathcal{F}_{t-1})\big]\leq K$ and $\E\big[C^2(\mathcal{F}_{t-1})\big]\leq K$. From Lemma~\ref{lem:6tilde}, these conditions are satisfied when $\E|X_t^2|\leq K$ and $\E|Y_t^2|\leq K$. This is, however, true by VAR Assumption~\ref{Vass:mom}.

\subsection{Verification of Assumption~\ref{ass:an}}\label{Ass2 VAR}

We verify items~\ref{it:int}--\ref{it:mixing} of Assumption~\ref{ass:an} in this subsection. 

\textbf{Assumption~\ref{ass:an}~\ref{it:int}:} Follows directly from VAR Assumption~\ref{Vass:param}.

\textbf{Assumption~\ref{ass:an}~\ref{it:diff}:} That $v_t(\cdot)$ and $c_t(\cdot)$ are a.s.\ twice continuously differentiable can be seen from \eqref{eq:(V2.3)}--\eqref{eq:(V2.4)}. That the gradients are different from $\vzero$ follows from \eqref{eq:(1.P.1)}--\eqref{eq:(1.P.2)}.

\textbf{Assumption~\ref{ass:an}~\ref{it:bound2.1}:} That $\big\Vert\nabla^2v_t(\vtheta^v)\big\Vert\leq V_2(\mathcal{F}_{t-1})$ is immediate from Lemma~\ref{lem:6tilde}. 

Since $\nabla^2 v_t(\vtheta^v)\equiv\vzero$, the Lipschitz-type condition follows for $V_3(\mathcal{F}_{t-1})=0$.

\textbf{Assumption~\ref{ass:an}~\ref{it:bound2.2}:} The claim follows as the previous item with $C_3(\mathcal{F}_{t-1})=0$.

\textbf{Assumption~\ref{ass:an}~\ref{it:mom bounds cons2}:} Since $V_2(\mathcal{F}_{t-1})=V_3(\mathcal{F}_{t-1})=C_2(\mathcal{F}_{t-1})=C_3(\mathcal{F}_{t-1})=0$, we only have to show that $\E[V_1^{3+\iota}(\mathcal{F}_{t-1})]\leq K$ and $\E[C_1^{3+\iota}(\mathcal{F}_{t-1})]\leq K$ for some $\iota>0$. This, however, follows from Lemma~\ref{lem:6tilde}, the $c_r$-inequality and VAR Assumption~\ref{Vass:mom}.

\textbf{Assumption~\ref{ass:an}~\ref{it:Lipschitz an}:} The Lipschitz continuity of $x\mapsto\partial_1 F_t(x,y)$ follows similarly as that of $y\mapsto\partial_2 F_t(x,y)$ under Assumption~\ref{ass:cons}~\ref{it:Lipschitz cons}.

\textbf{Assumption~\ref{ass:an}~\ref{it:bound cdf}:} The bound $f_t^Y(\cdot)\leq K$ follows similarly as $f_t^X(\cdot)\leq K$ under Assumption~\ref{ass:cons}~\ref{it:cond dens}.

From \eqref{eq:(V5.1)},
\begin{equation*}
	f_t(x,y) = f_{\varepsilon_X,\varepsilon_Y}(x-\mu_{X,t},\, y-\mu_{Y,t})\leq K.
\end{equation*}
Also from \eqref{eq:(V5.1)},
\begin{equation*}
	\big|\partial_1F_t(x,y)\big| = \big|\partial_1F_t(x-\mu_{X,t},\, y-\mu_{Y,t})\big|\leq K,
\end{equation*}
with a similar bound for $\partial_2F_t(x,y)$.

\textbf{Assumption~\ref{ass:an}~\ref{it:pd}:} Recall from the verification of Assumption~\ref{ass:cons}~\ref{it:str stat} that $(X_t, Y_t)^\prime$ is strictly stationary. 
From \eqref{eq:(V2.3)}--\eqref{eq:(V2.4)} and \eqref{eq:(1.P.1)}--\eqref{eq:(1.P.2)} it then follows that $v_t(\vtheta_0^v)$, $\nabla v_t(\vtheta_0^v)$, $c_t(\vtheta_0^c)$ and $\nabla c_t(\vtheta_0^c)$ are strictly stationary.
Therefore, the matrices in Assumption~\ref{ass:an}~\ref{it:pd} do not depend on $n$, such that $\mLambda_n=\mLambda$, $\mLambda_{n,(1)}=\mLambda_{(1)}$, $\mV_n=\mV$, \ldots.

We again only show that the first matrix is positive definite. To do so, we proceed by contradiction and assume that $\mLambda$ is not positive definite.
Since $\mLambda$ is clearly positive semi-definite, this implies the existence of some $\vx=(x_1,x_2,x_3)^\prime\in\mathbb{R}^3$ with $\vx\neq\vzeros$, such that $\vx^\prime\mLambda\vx=0$. By definition, $v_t(\vtheta_0^v)$ is in the support of $X_t\mid\mathcal{F}_{t-1}$, such that $f_t^X\big(v_t(\vtheta_0^v)\big)>0$ by VAR Assumption~\ref{Vass:error}. Therefore, for any $t\in\mathbb{N}$,
\begin{align*}
	0 &= \vx^\prime\mLambda\vx\\
	&= \E\Big[f_t^{X}\big(v_t(\vtheta_0^v)\big)\vx^\prime\nabla v_t(\vtheta_0^v)\nabla^\prime v_t(\vtheta_0^v)\vx\Big]\\
	&= \E\Big[\underbrace{f_t^{X}\big(v_t(\vtheta_0^v)\big)}_{>0}\underbrace{(x_1+x_2 X_{t-1} + x_3 Y_{t-1})^2}_{\geq0}\Big],
\end{align*}
where the final step follows from \eqref{eq:(1.P.1)}. 
We conclude that $x_1+x_2 X_{t-1} + x_3 Y_{t-1}\overset{\text{a.s.}}{=}0$ for all $t\in\mathbb{N}$ or, equivalently,
\[
\Pr\big\{x_1+x_2 X_{t-1} + x_3 Y_{t-1}=0\quad \forall\ t\in\mathbb{N}\big\}=1.
\]
Arguing similarly as under Assumption~\ref{ass:cons}~\ref{it:id} (with the $x_i$ playing the role of $\theta_i^v-\theta_{0,i}^{v}$ for $i=1,2,3$), we obtain that $x_1=x_2=x_3=0$, contradicting $\vx\neq0$. This establishes positive definiteness of $\mLambda$.

\textbf{Assumption~\ref{ass:an}~\ref{it:lambda func}:} It is evident from \eqref{eq:(V2.3)}, \eqref{eq:(V2.4)} and \eqref{eq:(1.P.2)} that $v_t(\vtheta^v)$, $c_t(\vtheta^c)$ and $\nabla c_t(\vtheta^c)$ are fixed functions of $X_{t-1}$ and $Y_{t-1}$, which are stationary by Assumption~\ref{ass:cons}~\ref{it:str stat} (verified above).
Therefore, $v_t(\vtheta^v)$, $c_t(\vtheta^c)$ and $\nabla c_t(\vtheta^c)$ are stationary, which in turn implies that $\Cc_t(\vtheta^c,\vtheta^v)$ is stationary.
Hence, $\vlambda_n(\vtheta^c,\vtheta^v)=\frac{1}{n}\sum_{t=1}^{n}\E\big[\Cc_t(\vtheta^c,\vtheta^v)\big]$ is independent of $n$, and we may write $\vlambda(\vtheta^c,\vtheta^v)=\vlambda_n(\vtheta^c,\vtheta^v)$.
The remainder of the proof follows the verification of Assumption~\ref{ass:an}~\ref{it:lambda func} for CCC--GARCH models in Section~\ref{Ass2}.

\textbf{Assumption~\ref{ass:an}~\ref{it:eq bound}:} From \eqref{eq:(V2.3)}, 
\[
X_t=v_t(\vtheta^v)\qquad\Longleftrightarrow\qquad X_t=\theta_1^v + \theta_2^v X_{t-1} + \theta_3^v Y_{t-1}.
\]
Writing this out for $t=2,\ldots,5$ gives the linear system of equations:
\begin{align}
	X_2 & =\theta_1^v + \theta_2^v X_{1} + \theta_3^v Y_{1},\label{eq:(13.m.0)}\\
	X_3 & =\theta_1^v + \theta_2^v X_{2} + \theta_3^v Y_{2},\label{eq:(13.m.1)}\\
	X_4 & =\theta_1^v + \theta_2^v X_{3} + \theta_3^v Y_{3},\label{eq:(13.m.2)}\\
	X_5 & =\theta_1^v + \theta_2^v X_{4} + \theta_3^v Y_{4}.\label{eq:(13.m.3)}
\end{align}
Equations \eqref{eq:(13.m.0)}--\eqref{eq:(13.m.2)} imply
\[
\begin{pmatrix}
	X_2 \\ X_3\\ X_4
\end{pmatrix}=
\begin{pmatrix}
	1 & X_1& Y_1\\
	1 & X_2& Y_2\\
	1 & X_3& Y_3\\
\end{pmatrix}
\begin{pmatrix}
	\theta_1^v\\ \theta_2^v \\ \theta_3^v
\end{pmatrix}.
\]
This is a linear system of equations in three ``unknowns'' $\theta_1^v$, $\theta_2^v$ and $\theta_3^v$.
Lengthy and tedious but straightforward calculations yield that
\[
\theta_1^v = \frac{X_2^2Y_3 - X_2 X_3 Y_2 + X_3^2 Y_1 - X_1 X_3 Y_3 + X_1 X_4 Y_2 - X_2 X_4 Y_1}{X_2 Y_3 - X_3 Y_2 - X_1 Y_3 + X_1 Y_2 + X_3 Y_1 - X_2 Y_1}.
\]
The same arguments applied to \eqref{eq:(13.m.1)}--\eqref{eq:(13.m.3)} yield that
\[
\theta_1^v = \frac{X_3^2Y_4 - X_3 X_4 Y_3 + X_4^2 Y_2 - X_2 X_4 Y_4 + X_2 X_5 Y_3 - X_3 X_5 Y_2}{X_3 Y_4 - X_4 Y_3 - X_2 Y_4 + X_2 Y_3 + X_4 Y_2 - X_3 Y_2}.
\]

Hence, a necessary condition for $X_t=v_t(\vtheta^v)$ to occur (at least) 5 times is that $(X_1,Y_1,\ldots,X_n,Y_n)^\prime$ lies in the set
\[
p^{-1}(0):=\big\{(X_1,Y_1,\ldots,X_n,Y_n)^\prime\in\mathbb{R}^{2n}:\ p(X_1,Y_1,\ldots,X_n,Y_n)=0\big\},
\]
where 
\begin{align*}
	p(X_1,Y_1,\ldots,X_n,Y_n)&=(X_2^2Y_3 - X_2 X_3 Y_2 + X_3^2 Y_1 - X_1 X_3 Y_3 + X_1 X_4 Y_2 - X_2 X_4 Y_1)\\
	&\hspace{3cm} \times(X_3 Y_4 - X_4 Y_3 - X_2 Y_4 + X_2 Y_3 + X_4 Y_2 - X_3 Y_2) \\
	&\hspace{0.4cm} -(X_3^2Y_4 - X_3 X_4 Y_3 + X_4^2 Y_2 - X_2 X_4 Y_4 + X_2 X_5 Y_3 - X_3 X_5 Y_2)\\
	&\hspace{3cm} \times(X_2 Y_3 - X_3 Y_2 - X_1 Y_3 + X_1 Y_2 + X_3 Y_1 - X_2 Y_1).
\end{align*}
The Hausdorff dimension of $p^{-1}(0)$ is less than $n$. Since, additionally, the distribution of $(X_1,Y_1,\ldots,X_n,Y_n)^\prime$ is absolutely continuous (VAR Assumption~\ref{Vass:error}), it follows that a.s.
\[
\sup_{\vtheta^v\in\mTheta^v}\sum_{t=1}^{n}\1_{\{X_t=v_t(\vtheta^v)\}}\leq 5.
\]

We obtain similarly that
\[
\sup_{\vtheta^c\in\mTheta^c}\sum_{t=1}^{n}\1_{\{Y_t=c_t(\vtheta^v)\}}\leq 5.
\]

\textbf{Assumption~\ref{ass:an}~\ref{it:mixing}:} By \eqref{eq:(V2.3)}--\eqref{eq:(V2.4)} and \eqref{eq:(1.P.1)}--\eqref{eq:(1.P.2)}, the vector 
\[
\big(X_t,\ Y_t,\ v_t(\vtheta_0^v),\ \nabla^\prime v_t(\vtheta_0^v),\ c_t(\vtheta_0^c),\ \nabla^\prime c_t(\vtheta_0^c)\big)^\prime
\]
is a function of $(X_t,Y_t,X_{t-1}, Y_{t-1})^\prime$.
Since $(X_t,Y_t)^\prime$ is $\alpha$-mixing (VAR Assumption~\ref{Vass:mixing}), the same then holds for $\big(X_t,\ Y_t,\ v_t(\vtheta_0^v),\ \nabla^\prime v_t(\vtheta_0^v),\ c_t(\vtheta_0^c),\ \nabla^\prime c_t(\vtheta_0^c)\big)^\prime$ by \citet[Theorem~3.49]{Whi01}, establishing the assumption.

\section{Additional Simulation Results}
\label{sec:AddSimRes}

In this section, we present simulation results assessing the bandwidth choice in the asymptotic covariance estimator in Section \ref{sec:SimBandwidthChoice} and for two different methods that can be used for initialization of the dynamic models in Section 
\ref{sec:Initializations}.

\subsection{Bandwidth Choice in Estimating the Asymptotic Covariance}
\label{sec:SimBandwidthChoice}

The estimators for the asymptotic covariance matrix discussed in Section~\ref{avar} involve bandwidth choices $\widehat{b}_{n,x}$  and $\widehat{b}_{n,y}$ that are required for estimating the (conditional) densities and distribution functions that appear in the asymptotic covariance.
Here, we refine the choices discussed in Remark~\ref{rem:asvar} that are used in the simulation study and the empirical application of the main article by generalizing the bandwidth choices to
\begin{align*}
	&\widehat{b}_{n,x} = C \cdot \text{MAD} \Big[\big\{X_t-v_t(\widehat{\vtheta}_n^v) \big\}_{t=1,\ldots,n}\Big] \left[ \Phi^{-1} \big( \beta + m(n,\beta) \big) - \Phi^{-1} \big( \beta - m(n,\beta) \big) \right], \\
	&\widehat{b}_{n,y} = C \cdot \text{MAD} \Big[\big\{Y_t-c_t(\widehat{\vtheta}_n^c) \big\}_{t=1,\ldots,n}\Big] \left[ \Phi^{-1} \big( \alpha + m((1-\beta)n, \alpha) \big) - \Phi^{-1} \big(\alpha - m((1-\beta)n,\alpha)  \big) \right], \\
	&m(n, \tau) =  n^{-1/3} \left( \Phi^{-1} (0.975)\right)^{2/3} \left( \frac{1.5 (\phi(\Phi^{-1}(\tau)))^2}{2(\Phi^{-1}(\tau))^2 +1} \right)^{1/3},
\end{align*}
for some multiplicative factor $C \in \mathbb{R}$.
In the above formulas, $\text{MAD}(\cdot)$ refers to the sample median absolute deviation, and $\phi(\cdot)$ and $\Phi^{-1}(\cdot)$ denote the density and quantile functions of the standard normal distribution.
The choices given in Remark~\ref{rem:asvar} arise for $C=1$.

\setcounter{table}{4}
\begin{table}[tb]
	\centering
	\tiny 
	\begin{tabular}{rrlrrrrrlrrrrrlrrrrr}
		\toprule
		\multicolumn{2}{l}{\textbf{VaR}} & &  \multicolumn{5}{c}{$\omega_1$} & & \multicolumn{5}{c}{$A_{11}$} & & \multicolumn{5}{c}{$B_{11}$}  \\
		\cmidrule{4-8}  	\cmidrule{10-14} 	 \cmidrule{16-20} 
		$\alpha,\beta$ & $n$ & $C = $&
		0.25 & 0.5 & 1 & 2 & 4 & & 
		0.25 & 0.5 & 1 & 2 & 4 & & 
		0.25 & 0.5 & 1 & 2 & 4 \\ 
		\midrule
		\multirow{4}{*}{0.90} 
		& 500 &  & 0.97 & 0.98 & 0.98 & 0.98 & 0.98 &  & 0.90 & 0.92 & 0.94 & 0.94 & 0.93 &  & 0.95 & 0.96 & 0.97 & 0.96 & 0.96 \\ 
		& 1000 &  & 0.98 & 0.98 & 0.99 & 0.99 & 0.98 &  & 0.91 & 0.93 & 0.95 & 0.95 & 0.94 &  & 0.96 & 0.97 & 0.98 & 0.97 & 0.97 \\ 
		& 2000 &  & 0.99 & 0.99 & 0.99 & 0.99 & 0.99 &  & 0.93 & 0.95 & 0.95 & 0.95 & 0.94 &  & 0.98 & 0.98 & 0.99 & 0.99 & 0.98 \\ 
		& 4000 &  & 1.00 & 1.00 & 1.00 & 1.00 & 1.00 &  & 0.94 & 0.95 & 0.96 & 0.96 & 0.95 &  & 0.99 & 0.99 & 0.99 & 0.99 & 0.99 \\ 
		\addlinespace
		\multirow{4}{*}{0.95} 
		& 500 &  & 0.95 & 0.97 & 0.98 & 0.97 & 0.95 &  & 0.88 & 0.91 & 0.93 & 0.92 & 0.87 &  & 0.93 & 0.95 & 0.96 & 0.95 & 0.93 \\ 
		& 1000 &  & 0.97 & 0.98 & 0.98 & 0.98 & 0.97 &  & 0.89 & 0.92 & 0.93 & 0.94 & 0.89 &  & 0.95 & 0.96 & 0.97 & 0.97 & 0.95 \\ 
		& 2000 &  & 0.98 & 0.99 & 0.99 & 0.99 & 0.99 &  & 0.91 & 0.93 & 0.94 & 0.95 & 0.91 &  & 0.97 & 0.98 & 0.98 & 0.98 & 0.97 \\ 
		& 4000 &  & 0.99 & 0.99 & 1.00 & 1.00 & 0.99 &  & 0.92 & 0.94 & 0.95 & 0.95 & 0.93 &  & 0.98 & 0.99 & 0.99 & 0.99 & 0.99 \\ 
		\midrule
		\\ 
		\multicolumn{2}{l}{\textbf{CoVaR}} & &  \multicolumn{5}{c}{$\omega_2$} & & \multicolumn{5}{c}{$A_{22}$} & & \multicolumn{5}{c}{$B_{22}$}  \\
		\cmidrule{4-8}  	\cmidrule{10-14} 	 \cmidrule{16-20} 
		$\alpha,\beta$ & $n$ & $C = $&
		0.25 & 0.5 & 1 & 2 & 4 & & 
		0.25 & 0.5 & 1 & 2 & 4 & & 
		0.25 & 0.5 & 1 & 2 & 4 \\ 
		\midrule
		\multirow{4}{*}{0.90} 
		& 500 &  & 0.84 & 0.85 & 0.85 & 0.83 & 0.90 &  & 0.84 & 0.87 & 0.90 & 0.92 & 0.98 &  & 0.81 & 0.82 & 0.81 & 0.81 & 0.87 \\ 
		& 1000 &  & 0.85 & 0.86 & 0.86 & 0.84 & 0.86 &  & 0.85 & 0.88 & 0.90 & 0.91 & 0.94 &  & 0.83 & 0.84 & 0.83 & 0.82 & 0.84 \\ 
		& 2000 &  & 0.88 & 0.89 & 0.89 & 0.88 & 0.88 &  & 0.88 & 0.90 & 0.91 & 0.91 & 0.91 &  & 0.86 & 0.87 & 0.88 & 0.87 & 0.87 \\ 
		& 4000 &  & 0.94 & 0.94 & 0.94 & 0.94 & 0.93 &  & 0.90 & 0.93 & 0.94 & 0.94 & 0.92 &  & 0.92 & 0.93 & 0.93 & 0.93 & 0.92 \\ 
		\addlinespace
		\multirow{4}{*}{0.95} 
		& 500 &  & 0.80 & 0.82 & 0.82 & 0.79 & 0.86 &  & 0.80 & 0.85 & 0.88 & 0.89 & 0.94 &  & 0.77 & 0.78 & 0.77 & 0.74 & 0.81 \\ 
		& 1000 &  & 0.81 & 0.83 & 0.82 & 0.79 & 0.78 &  & 0.79 & 0.84 & 0.86 & 0.85 & 0.83 &  & 0.78 & 0.79 & 0.78 & 0.75 & 0.74 \\ 
		& 2000 &  & 0.82 & 0.84 & 0.84 & 0.81 & 0.79 &  & 0.82 & 0.85 & 0.87 & 0.86 & 0.80 &  & 0.80 & 0.83 & 0.82 & 0.79 & 0.76 \\ 
		& 4000 &  & 0.85 & 0.87 & 0.87 & 0.86 & 0.82 &  & 0.85 & 0.88 & 0.90 & 0.90 & 0.83 &  & 0.83 & 0.85 & 0.85 & 0.84 & 0.80 \\ 
		\bottomrule
	\end{tabular}
	\caption{
		Empirical coverage rates for the $95\%$ confidence intervals for the choices $C \in \{ 0.25, 0.5, 1, 2, 4\}$ in the bandwidth formulas given in the main text.
		The simulation results are based on the CCC–GARCH model in  \eqref{eqn:ECCCmodel} and $M = 5000$ simulation replications.
	}
	\label{tab:BandwidthChoice}
\end{table}

Table \ref{tab:BandwidthChoice} shows the coverage rates of the confidence intervals, akin to the column ``CI'' in Table~\ref{tab:SimResultsCoCAViaR6p} of the main manuscript, however, for the choices $C \in \{ 0.25, 0.5, 1, 2, 4\}$.
We observe that the coverage rates of the confidence intervals are almost unchanged for the different finite-sample correction factors $C$.

\subsection{Initialization of the Dynamic CoCAViaR Models}
\label{sec:Initializations}

Here, we assess how initializing our dynamic CoCAViaR models  either correctly with the model intercept parameters $\big(v_1(\vtheta^v), c_1(\vtheta^c \big)^\prime = \vomega$ or misspecified with a constant value of ones $\big(v_1(\vtheta^v), c_1(\vtheta^c \big)^\prime = (1,1)^\prime$ changes the estimation results.
For this, Table \ref{tab:Initialization} compares the estimators' empirical bias and standard deviation (over all simulation runs) for these two initialization choices.
We do not find a meaningful difference in either the bias or the standard deviation, even for the relatively small sample size of $n=500$.
This implies that our novel and theoretically attractive initialization method with a model parameter works in practice and delivers similar results to a constant initialization.

\begin{table}[tb]
	\centering
	\scriptsize
	\begin{tabular}{rrlrrrrrrlrrrrrr}
		\toprule
		\multicolumn{2}{l}{\textbf{VaR}} & &  \multicolumn{6}{c}{Bias}  & &  \multicolumn{6}{c}{Standard Deviation}  \\
		\cmidrule{4-9}  	\cmidrule{11-16} 
		& & & \multicolumn{2}{c}{$\omega_1$} & \multicolumn{2}{c}{$A_{11}$} & \multicolumn{2}{c}{$B{11}$} && \multicolumn{2}{c}{$\omega_1$} & \multicolumn{2}{c}{$A_{11}$} & \multicolumn{2}{c}{$B{11}$} \\
		\cmidrule{4-5}  	\cmidrule{6-7}  	\cmidrule{8-9}  \cmidrule{11-12}  	\cmidrule{13-14}  	\cmidrule{15-16}  
		$\alpha,\beta$ & $n$ & &
		$\vomega$ & $\boldsymbol{1}$ & $\vomega$ & $\boldsymbol{1}$   & $\vomega$ & $\boldsymbol{1}$    & &
		$\vomega$ & $\boldsymbol{1}$  & $\vomega$ & $\boldsymbol{1}$   & $\vomega$ & $\boldsymbol{1}$    \\
		\midrule 
		\midrule
		\multirow{4}{*}{0.90} 
		& 500 &  & 0.016 & 0.021 & 0.011 & 0.014 & $-$0.052 & $-$0.068 &  & 0.055 & 0.058 & 0.110 & 0.110 & 0.193 & 0.201 \\ 
		& 1000 &  & 0.010 & 0.013 & 0.009 & 0.011 & $-$0.036 & $-$0.045 &  & 0.041 & 0.044 & 0.078 & 0.078 & 0.145 & 0.150 \\ 
		& 2000 &  & 0.006 & 0.008 & 0.005 & 0.006 & $-$0.021 & $-$0.025 &  & 0.028 & 0.029 & 0.055 & 0.055 & 0.099 & 0.101 \\ 
		& 4000 &  & 0.002 & 0.003 & 0.002 & 0.002 & $-$0.007 & $-$0.009 &  & 0.017 & 0.017 & 0.039 & 0.039 & 0.062 & 0.063 \\
		\addlinespace
		\multirow{4}{*}{0.95} 
		& 500 &  & 0.024 & 0.031 & 0.011 & 0.014 & $-$0.059 & $-$0.075 &  & 0.083 & 0.089 & 0.147 & 0.147 & 0.211 & 0.222 \\ 
		& 1000 &  & 0.015 & 0.019 & 0.012 & 0.015 & $-$0.039 & $-$0.047 &  & 0.057 & 0.059 & 0.106 & 0.105 & 0.149 & 0.151 \\ 
		& 2000 &  & 0.009 & 0.010 & 0.006 & 0.008 & $-$0.021 & $-$0.026 &  & 0.037 & 0.037 & 0.076 & 0.076 & 0.099 & 0.100 \\ 
		& 4000 &  & 0.004 & 0.005 & 0.002 & 0.003 & $-$0.010 & $-$0.012 &  & 0.024 & 0.025 & 0.053 & 0.053 & 0.067 & 0.067 \\ 
		\midrule
		\\ 
		\multicolumn{2}{l}{\textbf{CoVaR}} & &  \multicolumn{6}{c}{Bias}  & &  \multicolumn{6}{c}{Standard Deviation}  \\
		\cmidrule{4-9}  	\cmidrule{11-16} 
		& & & \multicolumn{2}{c}{$\omega_2$} & \multicolumn{2}{c}{$A_{22}$} & \multicolumn{2}{c}{$B{22}$} && \multicolumn{2}{c}{$\omega_2$} & \multicolumn{2}{c}{$A_{22}$} & \multicolumn{2}{c}{$B{22}$} \\
		\cmidrule{4-5}  	\cmidrule{6-7}  	\cmidrule{8-9}  \cmidrule{11-12}  	\cmidrule{13-14}  	\cmidrule{15-16}  
		$\alpha,\beta$ & $n$ & &
		$\vomega$ & $\boldsymbol{1}$  & $\vomega$ & $\boldsymbol{1}$   & $\vomega$ & $\boldsymbol{1}$    & &
		$\vomega$ & $\boldsymbol{1}$  & $\vomega$ & $\boldsymbol{1}$   & $\vomega$ & $\boldsymbol{1}$    \\
		\midrule 
		\multirow{4}{*}{0.90} 
		& 500 &  & 0.134 & 0.132 & 0.057 & 0.061 & $-$0.414 & $-$0.412 &  & 0.223 & 0.219 & 0.453 & 0.447 & 0.628 & 0.623 \\ 
		& 1000 &  & 0.114 & 0.110 & 0.054 & 0.052 & $-$0.347 & $-$0.336 &  & 0.208 & 0.204 & 0.300 & 0.302 & 0.588 & 0.575 \\ 
		& 2000 &  & 0.078 & 0.074 & 0.041 & 0.041 & $-$0.236 & $-$0.226 &  & 0.173 & 0.167 & 0.203 & 0.202 & 0.492 & 0.475 \\ 
		& 4000 &  & 0.035 & 0.036 & 0.029 & 0.030 & $-$0.112 & $-$0.114 &  & 0.111 & 0.112 & 0.137 & 0.137 & 0.328 & 0.330 \\ 
		\addlinespace
		\multirow{4}{*}{0.95} 
		& 500 &  & 0.232 & 0.217 & 0.101 & 0.110 & $-$0.553 & $-$0.524 &  & 0.329 & 0.316 & 0.976 & 0.975 & 0.663 & 0.648 \\ 
		& 1000 &  & 0.215 & 0.209 & 0.109 & 0.107 & $-$0.490 & $-$0.477 &  & 0.316 & 0.313 & 0.673 & 0.671 & 0.640 & 0.633 \\ 
		& 2000 &  & 0.172 & 0.165 & 0.077 & 0.077 & $-$0.382 & $-$0.369 &  & 0.295 & 0.289 & 0.449 & 0.449 & 0.596 & 0.590 \\ 
		& 4000 &  & 0.123 & 0.117 & 0.071 & 0.070 & $-$0.276 & $-$0.261 &  & 0.254 & 0.246 & 0.305 & 0.305 & 0.527 & 0.509 \\ 
		\bottomrule
	\end{tabular}
	\caption{Empirical bias and standard deviation for the six parameter CoCAViaR model based on the CCC–GARCH model in  \eqref{eqn:ECCCmodel} and $M=5000$ simulation replications based on either the initialization with the parameter $\vomega$ (correct initialization) or a constant (misspecified) initialization with the value $\boldsymbol{1} = (1,1)^\prime$.}
	\label{tab:Initialization}
\end{table}

\section{Forecast Comparison with an Alternative Scoring Function}
\label{Forecast Comparison with Alternative Loss Function}

In this section, we reconsider the CoVaR forecast comparison from Section~\ref{sec:ForecastingApplication2} of the main paper. 
We keep the setup unchanged, except for the choice of the scoring function \eqref{eq:loss} used in the comparative backtest of \citet{FH24}, which we now replace by the 0-homogeneous choice
\begin{align}
	\label{eq:loss_0hom}
	\mS_\text{0-hom} \Bigg(\begin{pmatrix}v\\ c\end{pmatrix}, \begin{pmatrix}x\\ y\end{pmatrix}\Bigg) &=
	\begin{pmatrix}
		[\1_{\{x\leq v\}} - \beta] \log(v) - \1_{\{x > v\}} \log(x) \\
		\1_{\{x>v\}} \left\{ [ \1_{\{y\leq c\}} -\alpha] \log(c)  - \1_{\{y > c\}} \log(y)  \right\} 
	\end{pmatrix};
\end{align}
see \citet{NZ17} and the Supplemental Appendix S.7.3 of \citet{FH24} for details on the choice of 0-homogeneous scoring functions.
Notice that due to the logarithm function, the scoring function in \eqref{eq:loss_0hom} is only applicable if the forecasts and observations are strictly positive, which is however innocuous for our application to (extreme) risk forecasts of financial losses.

\begin{table}[p!]
	\centering
	\scriptsize
	\resizebox{0.85\columnwidth}{!}{
		\begin{tabular}{lllrrrlrrrllr}
			\toprule
			&&& \multicolumn{3}{c}{VaR} && \multicolumn{3}{c}{CoVaR} &&\multicolumn{2}{c}{Inference} \\
			\cmidrule{4-6}  	\cmidrule{8-10} 	 \cmidrule{12-13} 
			$X_t$ & model &  & score & rank & hits &  & score & rank & hits &  & zone & $p$-value \\ 
			\midrule
			\multirow{12}{*}{BAC} & CoCAViaR-SAV-fullA &  & 0.658 & 6 & 4.3 &  & 2.953 & 1 & 7.3 &  & green & 0.00 \\ 
			& CoCAViaR-AS-mixed &  & 0.647 & 2 & 4.2 &  & 3.123 & 2 & 7.5 &  & grey & 0.00 \\ 
			& CoCAViaR-AS-signs &  & 0.647 & 1 & 4.6 &  & 3.219 & 3 & 5.2 &  & grey & 0.00 \\ 
			& CoCAViaR-AS-pos &  & 0.653 & 5 & 4.7 &  & 3.448 & 4 & 7.5 &  & green & 0.01 \\ 
			& CoCAViaR-SAV-diag &  & 0.660 & 9 & 4.5 &  & 3.499 & 5 & 8.7 &  & green & 0.02 \\ 
			& CoCAViaR-SAV-full &  & 0.658 & 6 & 4.3 &  & 3.560 & 6 & 8.3 &  & green & 0.07 \\ 
			& DCC-n-Chol &  & 0.660 & 10 & 4.5 &  & 4.489 & 7 & 19.5 &  & yellow & 0.12 \\ 
			& DCC-gjr-t-Chol &  & 0.649 & 4 & 4.7 &  & 4.623 & 8 & 15.8 &  & grey & 0.00 \\ 
			& DCC-t-Chol &  & 0.662 & 11 & 4.8 &  & 4.631 & 9 & 17.2 &  &  &  \\ 
			& DCC-n-sym &  & 0.660 & 8 & 4.5 &  & 4.835 & 10 & 19.5 &  & yellow & 0.43 \\ 
			& DCC-gjr-t-sym &  & 0.649 & 3 & 4.8 &  & 4.858 & 11 & 16.5 &  & grey & 0.00 \\ 
			& DCC-t-sym &  & 0.662 & 12 & 4.9 &  & 5.092 & 12 & 17.7 &  & yellow & 0.95 \\ 
			\midrule
			\multirow{12}{*}{C} & CoCAViaR-SAV-fullA &  & 0.667 & 10 & 4.8 &  & 3.298 & 1 & 4.1 &  & yellow & 0.10 \\ 
			& CoCAViaR-SAV-full &  & 0.667 & 10 & 4.8 &  & 3.316 & 2 & 5.0 &  & green & 0.08 \\ 
			& CoCAViaR-SAV-diag &  & 0.667 & 8 & 4.9 &  & 3.403 & 3 & 4.1 &  & yellow & 0.14 \\ 
			& CoCAViaR-AS-mixed &  & 0.659 & 4 & 4.9 &  & 3.407 & 4 & 5.6 &  & green & 0.03 \\ 
			& CoCAViaR-AS-signs &  & 0.658 & 1 & 5.3 &  & 3.463 & 5 & 5.2 &  & green & 0.02 \\ 
			& CoCAViaR-AS-pos &  & 0.666 & 5 & 5.8 &  & 3.775 & 6 & 4.8 &  & yellow & 0.44 \\ 
			& DCC-n-sym &  & 0.666 & 7 & 4.9 &  & 3.947 & 7 & 9.6 &  & green & 0.10 \\ 
			& DCC-n-Chol &  & 0.666 & 6 & 5.0 &  & 4.059 & 8 & 11.0 &  & green & 0.03 \\ 
			& DCC-gjr-t-Chol &  & 0.658 & 3 & 5.0 &  & 4.086 & 9 & 11.0 &  & grey & 0.05 \\ 
			& DCC-t-sym &  & 0.668 & 12 & 5.1 &  & 4.133 & 10 & 10.9 &  & yellow & 0.33 \\ 
			& DCC-gjr-t-sym &  & 0.658 & 2 & 5.0 &  & 4.162 & 11 & 8.6 &  & grey & 0.05 \\ 
			& DCC-t-Chol &  & 0.667 & 11 & 5.1 &  & 4.195 & 12 & 11.6 &  &  &  \\ 
			\midrule
			\multirow{12}{*}{GS} & CoCAViaR-AS-pos &  & 0.606 & 3 & 4.3 &  & 3.161 & 1 & 5.5 &  & green & 0.03 \\ 
			& CoCAViaR-AS-mixed &  & 0.605 & 2 & 4.7 &  & 3.232 & 2 & 4.2 &  & grey & 0.01 \\ 
			& CoCAViaR-SAV-fullA &  & 0.611 & 6 & 4.7 &  & 3.380 & 3 & 9.2 &  & yellow & 0.12 \\ 
			& CoCAViaR-AS-signs &  & 0.603 & 1 & 5.0 &  & 3.441 & 4 & 4.0 &  & grey & 0.00 \\ 
			& CoCAViaR-SAV-full &  & 0.611 & 6 & 4.7 &  & 3.538 & 5 & 10.1 &  & yellow & 0.18 \\ 
			& CoCAViaR-SAV-diag &  & 0.614 & 8 & 4.9 &  & 3.659 & 6 & 6.5 &  & yellow & 0.16 \\ 
			& DCC-n-sym &  & 0.618 & 10 & 4.3 &  & 3.819 & 7 & 10.0 &  & green & 0.08 \\ 
			& DCC-n-Chol &  & 0.619 & 12 & 4.3 &  & 3.908 & 8 & 13.0 &  & green & 0.00 \\ 
			& DCC-t-sym &  & 0.617 & 9 & 4.6 &  & 4.012 & 9 & 9.4 &  & grey & 0.01 \\ 
			& DCC-gjr-t-sym &  & 0.609 & 5 & 4.1 &  & 4.072 & 10 & 11.7 &  & grey & 0.04 \\ 
			& DCC-gjr-t-Chol &  & 0.608 & 4 & 4.2 &  & 4.149 & 11 & 13.1 &  & grey & 0.02 \\ 
			& DCC-t-Chol &  & 0.619 & 11 & 4.6 &  & 4.167 & 12 & 12.0 &  &  &  \\ 
			\midrule
			\multirow{12}{*}{JPM} & CoCAViaR-AS-mixed &  & 0.547 & 1 & 4.0 &  & 2.918 & 1 & 6.9 &  & grey & 0.00 \\ 
			& CoCAViaR-AS-signs &  & 0.549 & 3 & 4.1 &  & 3.040 & 2 & 6.7 &  & grey & 0.00 \\ 
			& CoCAViaR-SAV-fullA &  & 0.558 & 6 & 4.1 &  & 3.086 & 3 & 12.5 &  & green & 0.01 \\ 
			& CoCAViaR-SAV-full &  & 0.558 & 6 & 4.1 &  & 3.146 & 4 & 10.6 &  & green & 0.01 \\ 
			& CoCAViaR-AS-pos &  & 0.548 & 2 & 4.7 &  & 3.379 & 5 & 8.3 &  & grey & 0.00 \\ 
			& CoCAViaR-SAV-diag &  & 0.564 & 8 & 4.3 &  & 3.485 & 6 & 8.2 &  & green & 0.09 \\ 
			& DCC-n-sym &  & 0.569 & 11 & 3.7 &  & 3.619 & 7 & 17.2 &  & green & 0.00 \\ 
			& DCC-t-sym &  & 0.571 & 12 & 3.8 &  & 3.893 & 8 & 20.8 &  & red & 0.00 \\ 
			& DCC-n-Chol &  & 0.566 & 9 & 4.0 &  & 3.968 & 9 & 18.8 &  & green & 0.00 \\ 
			& DCC-gjr-t-sym &  & 0.555 & 5 & 3.8 &  & 4.032 & 10 & 17.7 &  & grey & 0.01 \\ 
			& DCC-t-Chol &  & 0.567 & 10 & 4.1 &  & 4.242 & 11 & 21.4 &  &  &  \\ 
			& DCC-gjr-t-Chol &  & 0.553 & 4 & 3.9 &  & 4.418 & 12 & 19.2 &  & grey & 0.00 \\ 
			\midrule
			\multirow{12}{*}{SPF} & CoCAViaR-SAV-fullA &  & 0.457 & 12 & 4.2 &  & 3.058 & 1 & 9.3 &  & green & 0.00 \\ 
			& CoCAViaR-AS-mixed &  & 0.439 & 2 & 4.5 &  & 3.185 & 2 & 5.3 &  & grey & 0.00 \\ 
			& CoCAViaR-AS-signs &  & 0.437 & 1 & 4.7 &  & 3.296 & 3 & 6.8 &  & grey & 0.00 \\ 
			& CoCAViaR-AS-pos &  & 0.446 & 5 & 4.7 &  & 3.354 & 4 & 5.9 &  & green & 0.01 \\ 
			& CoCAViaR-SAV-diag &  & 0.456 & 8 & 4.1 &  & 3.473 & 5 & 10.5 &  & green & 0.01 \\ 
			& CoCAViaR-SAV-full &  & 0.457 & 12 & 4.2 &  & 3.651 & 6 & 10.3 &  & yellow & 0.20 \\ 
			& DCC-gjr-t-sym &  & 0.441 & 4 & 4.3 &  & 4.062 & 7 & 14.5 &  & grey & 0.00 \\ 
			& DCC-n-sym &  & 0.454 & 6 & 4.5 &  & 4.223 & 8 & 16.7 &  & grey & 0.00 \\ 
			& DCC-gjr-t-Chol &  & 0.440 & 3 & 4.8 &  & 4.475 & 9 & 17.4 &  & grey & 0.00 \\ 
			& DCC-n-Chol &  & 0.454 & 7 & 4.4 &  & 4.482 & 10 & 20.5 &  & grey & 0.01 \\ 
			& DCC-t-sym &  & 0.456 & 10 & 4.5 &  & 4.544 & 11 & 18.3 &  & yellow & 0.15 \\ 
			& DCC-t-Chol &  & 0.456 & 9 & 4.5 &  & 4.704 & 12 & 21.1 &  &  &  \\ 
			\bottomrule
		\end{tabular}
	}
	\caption{VaR and CoVAR forecasting results for $Y_t$ equaling S\&P~500 losses and various choices of $X_t$ as in Table \ref{tab:ApplForecastingResults} of the main article, but here we use the 0-homogeneous scoring function of \citet{FH24} given in \eqref{eq:loss_0hom}. Details on the table columns are given in the text of the main article.}
	\label{tab:ApplForecastingResults_b0}
\end{table}

Table \ref{tab:ApplForecastingResults_b0} reports the forecast comparison results equivalent to Table \ref{tab:ApplForecastingResults} but using the 0-homogeneous scoring function in \eqref{eq:loss_0hom}.
The overall results are qualitatively unchanged with the CoCAViaR models clearly exhibiting the better forecast performance, especially in the CoVaR component.

\section{CoVaR Forecasting with Multivariate GARCH Models}
\label{CoVaR forecasting with multivariate GARCH models}

This section points out one important difficulty of multivariate volatility models in forecasting CoVaR and, hence, refines the arguments given in the first paragraph of Section~\ref{sec:CoVaRModels} in the main article.
By doing so, it sheds some light on the advantages of our modeling approach.

\begin{example}\label{ex:DCC-GARCH}
	Suppose that financial losses are generated by $(X_t,Y_t)^\prime=\mSigma_t\vvarepsilon_t$ for some positive-definite, $\mathcal{F}_{t-1}$-measurable $\mSigma_t$ and $\vvarepsilon_t\overset{\text{i.i.d.}}{\sim}F(\vzeros, \mI)$, independent of $\mathcal{F}_{t-1}$, where $F$ is some generic distribution with mean $\vzero$ and variance-covariance matrix $\mI$. Then, the conditional variance-covariance matrix is $\Var\big((X_t,Y_t)^\prime\mid\mathcal{F}_{t-1}\big)=\mSigma_t\mSigma_t^\prime=:\mH_t$. Much like univariate GARCH processes model the conditional variance, multivariate GARCH models---such as the DCC--GARCH of \citet{Eng02} and the corrected DCC--GARCH of \citet{Aie13}---directly model the conditional variance-covariance matrix $\mH_t$. While for many applications, such as portfolio construction, $\mH_t$ is the object of interest, CoVaR forecasting explicitly requires $\mSigma_t$ to be specified.
	
	However, for a given $\mH_t$, there exist infinitely many choices of $\mSigma_t$ satisfying the decomposition $\mSigma_t\mSigma_t^\prime=\mH_t$.
	Examples include the symmetric square root implied by the eigenvalue decomposition ($\mSigma_t^{s}$) or the lower triangular matrix of the Cholesky decomposition ($\mSigma_t^{l}$). 
	In fact, $\mSigma_t$ multiplied by any (constant or time-varying) orthogonal matrix $\mO_t$ yields a valid decomposition of the variance-covariance matrix, because $(\mSigma_t\mO_t)(\mSigma_t\mO_t)^\prime=\mSigma_t\mSigma_t^\prime=\mH_t$.
	
	The problem is that each possibility, while implying the \textit{same} variance-covariance dynamics $\mH_t$, implies \textit{different} values for the CoVaR.
	E.g., consider 
	\[
	\mH_t=\begin{pmatrix}10 & 6\\ 6 & 10 \end{pmatrix},\qquad \mSigma_t^{s}= \begin{pmatrix}3 & 1\\ 1 & 3 \end{pmatrix},\qquad \mSigma_t^{l}= \begin{pmatrix}3.16... & 0\\ 1.89... & 2.52... \end{pmatrix},
	\]
	and $\vvarepsilon_t=(\varepsilon_{X,t}, \varepsilon_{Y,t})^\prime$ with standardized $t_{5}$-distributed $\varepsilon_{X,t}$ and $\varepsilon_{Y,t}$ that are independent of each other. Then, for the model $\mSigma_t^{s}\vvarepsilon_t$ ($\mSigma_t^{l}\vvarepsilon_t$) based on the symmetric (lower triangular) decomposition, we obtain $\CoVaR_{t,0.95}= 9.76...$ 
	($\CoVaR_{t,0.95}=9.10...$).
	Thus, CoVaR  forecasts depend intimately on the decomposition of $\mH_t$, on which typical multivariate GARCH models stay silent. Thus, in principle, a given single GARCH-type model for $\mH_t$ may be consistent with an infinite number of CoVaR forecasts (depending on the choice of $\mSigma_t$), thus creating an unsatisfactory ambiguity when applied to CoVaR forecasting. 
\end{example}

Although the precise form of $\mSigma_t$ is not estimable from the data, there are two ways to address this problem.
The first is to assume that $\vvarepsilon_t$ has a spherical distribution (i.e., $\mO\vvarepsilon_t\overset{d}{=}\vvarepsilon_t$ for any orthogonal matrix $\mO$).
In this case, the choice of $\mSigma_t$ does not matter because the conditional distribution is fully determined by $\mH_t$ (much like for the multivariate normal distribution); see \citet{KS14}. 
Although the assumption of spherical errors is testable \citep{FJM17}, it may not apply in any given situation and even when it seems credible, it may be invalidated by structural change. 

A second way to address the problem of non-unique $\mSigma_t$ is to let subject matter considerations guide the choice of $\mSigma_t$ in practice. 
For instance, \citet{NZ20} use \citeauthor{Eng02}'s \citeyearpar{Eng02} DCC--GARCH model for CoVaR forecasting employing a lower triangular $\mSigma_t^{l}$. 
They do so for modeling the loss of the broader market $X_{t}$ and that of some individual unit $Y_{t}$, i.e., the opposite of the scenario we consider in our empirical application. \citet{NZ20} use the model
\begin{equation*}
	\begin{pmatrix}
		X_{t}\\ Y_{t}
	\end{pmatrix} = \mSigma_t^l\vvarepsilon_t
	=\begin{pmatrix}
		\sigma_{X,t} & 0\\
		\sigma_{Y,t}\rho_{XY,t} & \sigma_{Y,t}\sqrt{1-\rho_{XY,t}^2}
	\end{pmatrix}
	\begin{pmatrix}
		\varepsilon_{X,t}\\
		\varepsilon_{Y,t}
	\end{pmatrix},
\end{equation*}
where $\mSigma_t^l$ is the lower-triangular matrix from the Cholesky decomposition of the variance-covariance matrix $\mH_t=\big((\sigma_{X,t}^2,\ \sigma_{X,t}\sigma_{Y,t}\rho_{XY,t})^\prime \mid (\sigma_{X,t}\sigma_{Y,t}\rho_{XY,t},\ \sigma_{Y,t}^2)^\prime\big)$. In such a situation it may be acceptable to assume that the individual shock $\varepsilon_{Y,t}$ does not impact the market loss, since the individual institution is very small relative to the market. 
However, when modeling systemic risk across business units of comparable size (e.g., different trading desks), such an assumption may be untenable, and a different decomposition may be needed. Yet, which one exactly may be difficult to say. Our approach sidesteps these difficulties by directly modeling CoVaR (and VaR).

\section{Computation of Risk Measure Forecasts}\label{Computation of Risk Measure Forecasts}

\subsection{Computing VaR \& CoVaR Forecasts for Lower Triangular Decomposition}

Consider the standard specification of $\mSigma_t$ in DCC--GARCH models as the lower-triangular matrix from the Cholesky decomposition of the variance-covariance matrix 
\[
\mH_t=\begin{pmatrix}
	\sigma_{X,t}^2 & \sigma_{XY,t}\\
	\sigma_{XY,t}  & \sigma_{Y,t}^2
\end{pmatrix};
\]
see Section~\ref{CoVaR forecasting with multivariate GARCH models}. Then,
\begin{equation}\label{eq:model lt}
	\begin{pmatrix}X_t\\ Y_t \end{pmatrix}
	=\mSigma_t\begin{pmatrix}\varepsilon_{X,t}\\ \varepsilon_{Y,t} \end{pmatrix}=
	\begin{pmatrix}
		\sigma_{X,t} & 0\\
		\sigma_{Y,t}\rho_{XY,t} & \sigma_{Y,t}\sqrt{1-\rho_{XY,t}^2}
	\end{pmatrix}
	\begin{pmatrix}\varepsilon_{X,t}\\ \varepsilon_{Y,t} \end{pmatrix},
\end{equation}
where $\rho_{XY,t}=\sigma_{XY,t}/(\sigma_{X,t}\sigma_{Y,t})$ is the conditional correlation.

Under \eqref{eq:model lt}, we have that $\VaR_{t,\b}=\sigma_{X,t}F_{\varepsilon_X}^{-1}(\b)$, where $F_{\varepsilon_X}^{-1}(\cdot)$ denotes the quantile function of the (i.i.d.) $\varepsilon_{X,t}$. In practice, $F_{\varepsilon_X}^{-1}(\b)$ can be estimated as the empirical $\beta$-quantile of the standardized residuals to yield forecasts $\widehat{\VaR}_{t,\b}$. 


CoVaR can be forecasted from the following relation
\begin{align*}
	\alpha&= \P_{t-1}\big\{Y_t\leq \CoVaR_{t,\a|\b}\mid X_t\geq\VaR_{t,\b}\big\}\\
	&= \P_{t-1}\big\{\sigma_{Y,t}\rho_{XY,t}\varepsilon_{X,t} + \sigma_{Y,t}(1-\rho_{XY,t}^2)^{1/2} \varepsilon_{Y,t}\leq \CoVaR_{t,\a|\b}\mid X_t\geq\VaR_{t,\b}\big\}.
\end{align*}
From this relation, it is obvious that $\CoVaR_{t,\a|\b}$ can be computed as the $\a$-quantile of the distribution of $\sigma_{Y,t}\rho_{XY,t}\varepsilon_{X} + \sigma_{Y,t}(1-\rho_{XY,t}^2)^{1/2} \varepsilon_{Y}\mid X_t\geq\VaR_{t,\b}$, where $(\varepsilon_{X}, \varepsilon_{Y})^\prime$ is a generic element of the sequence $(\varepsilon_{X,t}, \varepsilon_{Y,t})^\prime$. Thus, $\CoVaR_{t,\a|\b}$ can be estimated as the empirical $\alpha$-quantile of
\[
\sigma_{Y,t}\rho_{XY,t}\varepsilon_{X,i} + \sigma_{Y,t}(1-\rho_{XY,t}^2)^{1/2} \varepsilon_{Y,i},\qquad i \in\mathcal{I}=\big\{t=1,\ldots,n:\ X_t\geq\widehat{\VaR}_{t,\b}\big\}.
\]

This describes how VaR and CoVaR forecasts are obtained from DCC--GARCH models in the empirical application.

\subsection{Computing VaR \& CoVaR Forecasts for General Decompositions}

Here, we consider a generic
\begin{equation}\label{eq:general sigma}
	\mSigma_t=\begin{pmatrix}\widetilde{\sigma}_{X,t}  & \widetilde{\sigma}_{XY,t}\\
		\widetilde{\sigma}_{XY,t} & \widetilde{\sigma}_{Y,t}\end{pmatrix}
\end{equation}
satisfying $\mSigma_t\mSigma_t^\prime=\mH_t$. The leading special case is the symmetric square root
\[
\mSigma_t=\frac{1}{\tau}\begin{pmatrix}\sigma_{X,t}^2 + s & \sigma_{XY,t}\\
	\sigma_{XY,t} & \sigma_{Y,t}^2 + s\end{pmatrix},
\]
where $s=\sigma_{X,t}\sigma_{Y,t}(1-\rho_{XY,t}^2)^{1/2}$ and $\tau=(\sigma_{X,t}^2+\sigma_{Y,t}^2+2s)^{1/2}$ with $\rho_{XY,t}=\sigma_{XY,t}/(\sigma_{X,t}\sigma_{Y,t})$ again denoting the correlation coefficient. 

Under \eqref{eq:general sigma}, it holds in the model $(X_t, Y_t)^\prime=\mSigma_t(\varepsilon_{X,t}, \varepsilon_{Y,t})^\prime$ that $X_t=\widetilde{\sigma}_{X,t}\varepsilon_{X,t} + \widetilde{\sigma}_{XY,t}\varepsilon_{Y,t}$, such that $\VaR_{t,\b}$ is the $\b$-quantile of the distribution of $\widetilde{\sigma}_{X,t}\varepsilon_{X} + \widetilde{\sigma}_{XY,t}\varepsilon_{Y}$, where $(\varepsilon_{X}, \varepsilon_{Y})^\prime$ is a generic element of the sequence $(\varepsilon_{X,t}, \varepsilon_{Y,t})^\prime$. Thus, $\VaR_{t,\b}$ may be estimated as $\widehat{\VaR}_{t,\b}$, i.e., the empirical $\b$-quantile of
\[
\widetilde{\sigma}_{X,t}\varepsilon_{X,i} + \widetilde{\sigma}_{XY,t}\varepsilon_{Y,i},\qquad i=1,\ldots,n.
\]

For general $\mSigma_t$ from \eqref{eq:general sigma},
CoVaR solves
\begin{align*}
	\alpha&= \P_{t-1}\big\{Y_t\leq \CoVaR_{t,\a|\b}\mid X_t\geq\VaR_{t,\b}\big\}\\
	&= \P_{t-1}\big\{\widetilde{\sigma}_{XY,t}\varepsilon_{X,t} + \widetilde{\sigma}_{Y,t}\varepsilon_{Y,t}\leq \CoVaR_{t,\a|\b}\mid X_t\geq\VaR_{t,\b}\big\}.
\end{align*}
Thus, $\CoVaR_{t,\a|\b}$ can be estimated as the empirical $\a$-quantile of
\[
\widetilde{\sigma}_{XY,t}\varepsilon_{X,i} + \widetilde{\sigma}_{Y,t}\varepsilon_{Y,i},\qquad i \in\mathcal{I}=\big\{t=1,\ldots,n:\ X_t\geq\widehat{\VaR}_{t,\b}\big\}.
\]

\singlespacing
\footnotesize

\bibliographystyle{apalike}
\bibliography{bibCoQR}

\end{document}